\documentclass[11pt]{article}

\usepackage[T1]{fontenc}
\usepackage[sc]{mathpazo}
\usepackage{amsmath}
\usepackage{amssymb}
\usepackage{framed}
\usepackage{graphicx,color}
\usepackage{multirow}
\usepackage{enumerate}
\usepackage{amsthm}
\usepackage{amsfonts,mathrsfs}
\usepackage{geometry} 
\usepackage{eepic}
\usepackage{ifthen}
\usepackage[vcentermath]{youngtab}
\usepackage[unicode=true,pdfusetitle, bookmarks=true,bookmarksnumbered=false,bookmarksopen=false, breaklinks=false,pdfborder={0 0 0},backref=false,colorlinks=false] {hyperref}
\hypersetup{
colorlinks,linkcolor=myurlcolor,citecolor=myurlcolor,urlcolor=myurlcolor}
\definecolor{myurlcolor}{rgb}{0,0,0.7}

\usepackage{color}

\newcommand{\red}{\textcolor{red}}
\newcommand{\blue}{\textcolor{blue}}

\usepackage{hyperref}
\hypersetup{pdfpagemode=UseNone}

\newcommand{\tinyspace}{\mspace{1mu}}

\newcommand{\op}[1]{\operatorname{#1}}

\newcommand{\abs}[1]{\left\lvert\tinyspace #1 \tinyspace\right\rvert}

\newcommand{\norm}[1]{\left\lVert\tinyspace #1 \tinyspace\right\rVert}

\newcommand{\wtil}[1]{\widetilde{#1}}

\renewcommand{\det}{\operatorname{det}}
\renewcommand{\t}{{\scriptscriptstyle\mathsf{T}}}

\newcommand{\setft}[1]{\mathrm{#1}}

\newcommand{\density}[1]{\setft{D}\left(#1\right)}

\renewcommand{\vec}{\op{vec}}
\newcommand{\im}{\op{im}}
\newcommand{\rank}{\op{rank}}
\newcommand{\spn}{\op{span}}
\newcommand{\sign}{\op{sign}}

\def\g{\mathfrak{g}}

\def\liet{\mathfrak{t}}

\def\Ad{\mathrm{Ad}}
\def\ad{\mathrm{ad}}
\def\vol{\mathrm{vol}}

\def \dif {\mathrm{d}}
\def \diag {\mathrm{diag}}
\def \vol {\mathrm{vol}}
\def \re {\mathrm{Re}}
\def \im {\mathrm{Im}}

\def\complex{\mathbb{C}}
\def\real{\mathbb{R}}
\def\natural{\mathbb{N}}
\def\integer{\mathbb{Z}}
\def\I{\mathbb{1}}

\newenvironment{mylist}[1]{\begin{list}{}{
    \setlength{\leftmargin}{#1}
    \setlength{\rightmargin}{0mm}
    \setlength{\labelsep}{2mm}
    \setlength{\labelwidth}{8mm}
    \setlength{\itemsep}{0mm}}}
    {\end{list}}

\def\ot{\otimes}

\newcommand{\inner}[2]{\langle #1 , #2\rangle}
\newcommand{\iinner}[2]{\langle #1 | #2\rangle}

\newcommand{\Inner}[2]{\left\langle #1 , #2\right\rangle}
\newcommand{\Innerm}[3]{\left\langle #1 \left| #2 \right| #3 \right\rangle}
\newcommand{\defeq}{\stackrel{\smash{\textnormal{\tiny def}}}{=}}

\newcommand{\pa}[1]{(#1)}
\newcommand{\Pa}[1]{\left(#1\right)}

\newcommand{\Br}[1]{\left[#1\right]}
\newcommand{\set}[1]{\{#1\}}
\newcommand{\Set}[1]{\left\{#1\right\}}

\newcommand{\ket}[1]{|#1\rangle}

\DeclareMathOperator{\trace}{Tr}
\newcommand{\ptr}[2]{\trace_{#1}\pa{#2}}
\newcommand{\Ptr}[2]{\trace_{#1}\Pa{#2}}

\newcommand{\Tr}[1]{\Ptr{}{#1}}

\newcommand{\fontmapset}{\mathbf} 

\newcommand{\Mset}[2]{\fontmapset{#1}\Pa{#2}}

\newcommand{\Her}[1]{\Mset{H}{#1}}

\def\cB{\mathcal{B}}\def\cC{\mathcal{C}}\def\cD{\mathcal{D}}
\def\cG{\mathcal{G}}\def\cI{\mathcal{I}}
\def\cO{\mathcal{O}}
\def\cP{\mathcal{P}}
\def\cU{\mathcal{U}}\def\cX{\mathcal{X}}

\def\bS{\mathbf{S}}

\def\rD{\mathrm{D}}
\def\rH{\mathrm{H}}
\def\rO{\mathrm{O}}
\def\rS{\mathrm{S}}\def\rT{\mathrm{T}}
\def\rU{\mathrm{U}}

\def\sH{\mathscr{H}}\def\sI{\mathscr{I}}
\def\sL{\mathscr{L}}
\def\sS{\mathscr{S}}

\def\X{\textsf{X}}

\newtheorem{thrm}{Theorem}[section]
\newtheorem{lem}[thrm]{Lemma}
\newtheorem{prop}[thrm]{Proposition}
\newtheorem{cor}[thrm]{Corollary}
\theoremstyle{definition}
\newtheorem{definition}[thrm]{Definition}
\newtheorem{remark}[thrm]{Remark}
\newtheorem{exam}[thrm]{Example}

\numberwithin{equation}{section}

\newcounter{questionnumber}

\newcommand{\jpa}{J. Phys. A~}

\newcommand{\pra}{Phys. Rev. A~}

\begin{document}

\title{\textbf{Volumes of Orthogonal Groups and Unitary Groups}}

\author{\blue{Lin Zhang}\footnote{E-mail: godyalin@163.com; linyz@zju.edu.cn}\\
  {\it\small Institute of Mathematics, Hangzhou Dianzi University, Hangzhou 310018, PR~China}}

\date{}
\maketitle \mbox{}\hrule\mbox{}
\begin{abstract}
The matrix integral has many applications in diverse fields. This
review article begins by presenting detailed key background
knowledge about matrix integral. Then the volumes of orthogonal
groups and unitary groups are computed, respectively. As a
unification, we present Mcdonald's volume formula for a compact Lie
group. With this volume formula, one can easily derives the volumes
of orthogonal groups and unitary groups. Applications are also
presented as well. Specifically, The volume of the set of mixed
quantum states is computed by using the volume of unitary group. The
volume of a metric ball in unitary group is also computed as well.

There are no new results in this article, but only detailed and
elementary proofs of existing results. The purpose of the article is
pedagogical, and to collect in one place many, if not all, of the
quantum information applications of the volumes of orthogonal and
unitary groups.

\end{abstract}
\mbox{}\hrule\mbox{}
\newpage
\tableofcontents
\newpage
\section{Introduction}\label{sect:intro}

Volumes of orthogonal groups and unitary groups are very useful in
physics and mathematics \cite{Boya1,Boya2}. In 1949, Ponting and
Potter had already calculated the volume of orthogonal and unitary
group \cite{Pont49}. A complete treatment for group manifolds is
presented by Marinov \cite{Marinov}, who extracted the volumes of
groups by studying curved path integrals \cite{Terentyev}. There is
a general closed formula for any compact Lie group in terms of the
root lattice \cite{Mcdonald}. Clearly the methods used previously
are not followed easily. \.{Z}yczkowski in his paper \cite{Karol03}
gives a re-derivation on the volume of unitary group, but it is
still not accessible. The main goal of this paper is to compute the
volumes of orthogonal groups and unitary groups in a systematic and
elementary way. The obtained formulae is more applicable.

As an application, by re-deriving the volumes of orthogonal groups
and unitary groups, the authors in \cite{Karol03} computes the
volume of the convex $(n^2-1)$-dimensional set
$\density{\complex^n}$ of the density matrices of size $n$ with
respect to the Hilbert-Schmidt measure. Recently, the authors in
\cite{Wei} give the integral representation of the exact volume of a
metric ball in unitary group; and present diverse applications of
volume estimates of metric balls in manifolds in information and
coding theory.

Before proceeding, we recall some notions in group theory. Assume
that a group $\cG$ acts on the underlying vector space $\cX$ via
$g\ket{x}$ for all $g\in \cG$ and $x\in\cX$. Let $\ket{x}$ be any
nonzero vector in $\cX$. The subset $\cG\ket{x}:=\set{g\ket{x}:g\in
\cG}$ is called the $\cG$-\emph{orbit} of $\ket{x}\in\cX$. Denote
$$
x^\cG:=\set{g\in \cG: g\ket{x}=\ket{x}}.
$$
We have the following fact:
$$
\cG\ket{x}\sim \cG/x^\cG.
$$
If $x^\cG=\set{e}$, where $e$ is a unit element of $\cG$, then the
action of $\cG$ on $\ket{x}$ is called \emph{free}. In this case,
$$
\cG\ket{x}\sim \cG/\set{e}\Longleftrightarrow \cG\sim \cG\ket{x}.
$$
Now we review a fast track of the volume of a unitary group. Let us
consider the unitary groups $\cG=\cU(n+1),n=1,2,\ldots$ To establish
their structure, we look at spaces in which the group acts
transitively, and identify the isotropy subgroup.  The unitary group
$\cU(n+1)$ acts naturally in the complex vector space
$\cX=\complex^{n+1}$ through the vector or "defining"
representation; the image of any nonzero vector
$\ket{x}=\ket{\psi}\in\complex^{n+1}$ is contained in the maximal
sphere $\mathbb{S}^{2n+1}$ of radius $\norm{\psi}$ since
$\norm{\psi}=\norm{U\psi}$ for $U\in\cU(n+1)$, and in fact it is
easy to see that it sweeps the whole sphere when $U$ runs though
$\cU(n+1)$, i.e. the group $\cU(n+1)$ acts transitively in this
sphere. The isotropy group of the vector $\ket{n+1}$ is easily seen
to be the unitary group with an entry less, that is
$x^\cG=\cU(n)\oplus 1$. Indeed, the following map is \emph{onto}:
$$
\varphi:\cU(n+1)\longrightarrow\cU(n+1)\ket{n+1} =\mathbb{S}^{2n+1}.
$$
If we identity $\cU(n)$ with $\cU(n)\oplus 1$, as a subgroup of
$\cU(n+1)$, then
$$
\cU(n)\ket{n+1}=\ket{n+1},
$$
implying that $\ker\varphi=\cU(n)$, and thus
$$
\cU(n+1)/\ker\varphi\sim\cU(n+1)\ket{n+1}=\mathbb{S}^{2n+1}=\Pa{\cU(n+1)/\ker\varphi}
\ket{n+1}.
$$
Therefore we have the equivalence relation
\begin{eqnarray}
\cU(n+1)/\cU(n)=\mathbb{S}^{2n+1}.
\end{eqnarray}
This indicates that
\begin{eqnarray}
\vol\Pa{\cU(n+1)}=\vol\Pa{\mathbb{S}^{2n+1}}\cdot\vol\Pa{\cU(n)},~~n=1,2,\ldots
\end{eqnarray}
That is,
\begin{eqnarray}
\vol\Pa{\cU(n)}=\vol\Pa{\mathbb{S}^1}\times\vol\Pa{\mathbb{S}^3}\times\cdots\times
\vol\Pa{\mathbb{S}^{2n-1}}.
\end{eqnarray}
We can see this in \cite{Fujii}. The volume of the sphere of unit
radius embedded in $\real^n(n\geqslant1)$, is calculated from the
Gaussian integral (see also Appendix I for the details):
\begin{eqnarray*}
\sqrt{\pi}=\int^{+\infty}_{-\infty}e^{-t^2}\dif t.
\end{eqnarray*}
Now
\begin{eqnarray*}
\Pa{\sqrt{\pi}}^n&=&\Pa{\int^{+\infty}_{-\infty}e^{-t^2}\dif t}^n =
\int_{\real^n} e^{-\norm{v}^2}\dif v\\
&=&\int^{+\infty}_0 \int_{\mathbb{S}^{n-1}(r)}e^{-r^2}\dif\sigma\dif
r=\int^{+\infty}_0 \sigma_{n-1}(r)e^{-r^2}\dif r,
\end{eqnarray*}
where $\sigma_{n-1}(r)=\int_{\mathbb{S}^{n-1}(r)}\dif\sigma$ is the
volume of sphere $\mathbb{S}^{n-1}(r)$ of radius $r$. Since
$$
\sigma_{n-1}(r) = \sigma_{n-1}(1)\times r^{n-1},
$$
it follows that
\begin{eqnarray*}
\Pa{\sqrt{\pi}}^n = \sigma_{n-1}(1)\times\int^{+\infty}_0
r^{n-1}e^{-r^2}\dif r,
\end{eqnarray*}
implying that
\begin{framed}
\begin{eqnarray}
\vol\Pa{\mathbb{S}^{n-1}} : = \frac{2\pi^{\frac n2}}{\Gamma\Pa{\frac
n2}}.
\end{eqnarray}
\end{framed}
Finally, we get the volume formula of a unitary group:
\begin{framed}
\begin{eqnarray}
\vol\Pa{\cU(n)} : = \prod^{n}_{k=1}\frac{2\pi^k}{\Gamma(k)} =\frac{
2^n\pi^{\frac{n(n+1)}2}}{1!2!\cdots (n-1)!}.
\end{eqnarray}
\end{framed}

\section{Volumes of orthogonal groups}\label{sect:orthogonal-group}

\subsection{Preliminary}\label{sect:Jacobian}

The following standard notations will be used \cite{Mathai}. Scalars
will be denoted by lower-case letters, vectors and matrices by
capital letters. As far as possible variable matrices will be
denoted by $X,Y,\ldots$ and constant matrices by $A,B,\ldots$.

Let $A=[a_{ij}]$ be a $n\times n$ matrix, then
$\Tr{A}=\sum^n_{j=1}a_{jj}$ is the trace of $A$, and $\det(A)$ is
the determinant of $A$, and $^\t$ over a vector or a matrix will
denote its transpose. Let $X=[x_{ij}]$ be a $m\times n$ matrix of
independent real entries $x_{ij}$'s. We denote the matrix of
differentials by $\dif{X}$, i.e.
$$
\dif{X}:=[\dif{x_{ij}}] = \Br{\begin{array}{cccc}
                               \dif{x_{11}} & \dif{x_{12}} & \cdots & \dif{x_{1n}} \\
                               \dif{x_{21}} & \dif{x_{22}} & \cdots & \dif{x_{2n}} \\
                               \vdots & \vdots & \ddots & \vdots \\
                               \dif{x_{m1}} & \dif{x_{m2}} & \cdots &
                               \dif{x_{mn}}
                             \end{array}
}.
$$
Then $[\dif{X}]$ stands for the product of the $m\times n$
differential elements
\begin{eqnarray}
[\dif{X}] := \prod^m_{i=1}\prod^n_{j=1}\dif{x_{ij}}
\end{eqnarray}
and when $X$ is a real square symmetric matrix, that is, $m=n,
X=X^\t$, then $[\dif{X}]$ is the product of the $n(n+1)/2$
differential elements, that is,
\begin{eqnarray}
[\dif{X}] := \prod^n_{j=1}\prod^n_{i=j} \dif{x_{ij}} =
\prod_{i\geqslant j}\dif{x_{ij}}.
\end{eqnarray}
Throughout this paper, stated otherwise, we will make use of
conventions that the signs will be ignored in the product
$[\dif{X}]$ of differentials of independent entries. Our notation
will be the following: Let $X=[x_{ij}]$ be a $m\times n$ matrix of
independent real entries. Then
\begin{eqnarray}
\fbox{$[\dif{X}] = \wedge^m_{i=1}\wedge^n_{j=1}\dif{x_{ij}}$}
\end{eqnarray}
when $[\dif{X}]$ appears with integrals or Jacobians of
transformations;
\begin{eqnarray}
\fbox{$[\dif{X}] = \prod^m_{i=1}\prod^n_{j=1}\dif{x_{ij}}$}
\end{eqnarray}
when $[\dif{X}]$ appears with integrals involving \emph{density
functions} where \blue{the functions are nonnegative and the
absolute value of the Jacobian is automatically taken}.

As Edelman said \cite{Edelman} before, many researchers in Linear
Algebra have little known the fact that the familiar matrix
factorizations, which can be viewed as changes of variables, have
simple Jacobians. These Jacobians are used extensively in
applications of random matrices in multivariate statistics and
physics.

It is assumed that the reader is familiar with the calculation of
Jacobians when a vector of scalar variables is transformed to a
vector of scalar variables. The result is stated here for the sake
of completeness. Let the vector of scalar variables $X$ be
transformed to $Y$, where
$$
X=\Br{\begin{array}{c}
        x_1 \\
        \vdots \\
        x_n
      \end{array}
}~~\text{and}~~Y=\Br{\begin{array}{c}
        y_1 \\
        \vdots \\
        y_n
      \end{array}
},
$$
by a one-to-one transformation. Let the matrix of partial
derivatives be denoted by
\begin{eqnarray*}
\frac{\partial Y}{\partial X} = \Br{\frac{\partial y_i}{\partial
x_j}}.
\end{eqnarray*}
The the determinant of the matrix $\Br{\frac{\partial y_i}{\partial
x_j}}$ is known as the \emph{Jacobian} of the transformation $X$
going to $Y$ or $Y$ as a function of $X$ it is written as
\begin{eqnarray*}
J(Y:X) = \det\Pa{\Br{\frac{\partial y_i}{\partial
x_j}}}~~\text{or}~~[\dif Y]=J(Y:X)[\dif X],~~J\neq0
\end{eqnarray*}
and
\begin{eqnarray*}
J(Y:X) = \frac1{J(X:Y)}~~\text{or}~~1=J(Y:X)J(X:Y).
\end{eqnarray*}

\blue{Note that when transforming $X$ to $Y$ the variables can be
taken in any order because a permutation brings only a change of
sign in the determinant and the magnitude remains the same, that is,
$\abs{J}$ remains the same where $\abs{J}$ denotes the absolute
value of $J$}. When evaluating integrals involving functions of
matrix arguments one often needs only the absolute value of $J$.
Hence in all the statements of this notes the notation $[\dif
Y]=J[\dif X]$ means that the relation is written ignoring the sign.

\begin{prop}\label{prop:vector-transformation}
Let $X,Y\in\real^n$ be of independent real variables and
$A\in\real^{n\times n}$ be a nonsingular matrix of constants. If
$Y=AX$, then
\begin{eqnarray}\label{eq:1.2-1}
\fbox{$[\dif Y] = \det(A)[\dif X]$.}
\end{eqnarray}
\end{prop}

\begin{proof}
The result follows from the definition itself. Note that when
$Y=AX,A=[a_{ij}]$ one has
\begin{eqnarray*}
y_i = a_{i1}x_1+\cdots +a_{in}x_n, i=1,\ldots,n
\end{eqnarray*}
where $x_j$'s and $y_j$'s denote the components of the vectors $X$
and $Y$, respectively. Thus the partial derivative of $y_i$ with
respect to $x_j$ is $a_{ij}$, and then the determinant of the
Jacobian matrix is $\det(A)$.
\end{proof}

In order to see the results in the more complicated cases we need
the concept of a \emph{tensor product}.

\begin{definition}[Tensor product]\label{def:tensor-product}
Let $A=[a_{ij}]\in\real^{p\times q}$ and
$B=[b_{ij}]\in\real^{m\times n}$. Then the \emph{tensor product},
denoted by $\ot$, is a $pm\times qn$ matrix in $\real^{pm\times
qn}$, formed as follows:
\begin{eqnarray}
B\ot A = \Br{\begin{array}{cccc}
               a_{11}B & a_{12}B & \cdots & a_{1q}B \\
               a_{21}B & a_{22}B & \cdots & a_{2q}B \\
               \vdots & \vdots & \ddots & \vdots \\
               a_{p1}B & a_{p2}B & \cdots & a_{pq}B
             \end{array}
}
\end{eqnarray}
and
\begin{eqnarray}
A\ot B = \Br{\begin{array}{cccc}
               b_{11}A & b_{12}A & \cdots & b_{1n}A \\
               b_{21}A & b_{22}A & \cdots & b_{2n}A \\
               \vdots & \vdots & \ddots & \vdots \\
               b_{m1}A & b_{m2}A & \cdots & b_{mn}A
             \end{array}
}.
\end{eqnarray}
\end{definition}

\begin{definition}[Vector-matrix correspondence]\label{def:vector-matrix}
Let $X=[x_{ij}]\in\real^{m\times n}$ matrix. Let the $j$-th column
of $X$ be denoted by $X_j$. That is, $X=[X_1,\ldots,X_n]$, where
$$
X_j = \Br{\begin{array}{c}
            x_{1j} \\
            \vdots \\
            x_{mj}
          \end{array}
}.
$$
Consider an $mn$-dimensional vector in $\real^{mn}$, formed by
appending $X_1,\ldots,X_n$ and forming a long string. This vector
will be denoted by $\vec(X)$. That is,
\begin{eqnarray}
\vec(X) = \Br{\begin{array}{c}
            X_1 \\
            \vdots \\
            X_n
          \end{array}
}.
\end{eqnarray}
\end{definition}
From the above definition, we see that the vec mapping is a
one-to-one and onto correspondence from $\real^{m\times n}$ to
$\real^n\ot\real^m$. We also see that $\vec(AXB)=\Pa{A\ot B^\t}
\vec(X)$ if the product $AXB$ exists.

\begin{prop}\label{prop:matrix-transformation}
Let $X,Y\in\real^{m\times n}$ be of independent real variables and
$A\in\real^{m\times m}$ and $B\in\real^{n\times n}$ nonsingular
matrices of constants. If $Y=AXB$, then
\begin{eqnarray}
\fbox{$[\dif Y]=\det(A)^n\det(B)^m[\dif X].$}
\end{eqnarray}
\end{prop}

\begin{proof}
Since $Y=AXB$, it follows that $\vec(Y)=(A\ot B^\t)\vec(X)$. Then by
using Proposition~\ref{prop:vector-transformation}, we have
\begin{eqnarray*}
J(Y:X)&=&\det\Pa{\frac{\partial Y}{\partial X}}:=
\det\Pa{\frac{\partial (\vec(Y))}{\partial (\vec(X))}} \\
&=& \det(A\ot B^\t) = \det(A\ot\I_n)\det(\I_m\ot B^\t)\\
&=& \det(A)^n\det(B)^m,
\end{eqnarray*}
implying that
$$
[\dif Y]=J(Y:X)[\dif X] = \det(A)^n\det(B)^m[\dif X].
$$
This completes the proof.
\end{proof}

\begin{remark}
Another approach to the proof that $[\dif Z]=\det(A)^n[\dif X]$,
where $Z=AX$, is described as follows: we partition $Z$ and $X$,
respectively, as: $Z=[Z_1,\ldots,Z_n], X=[X_1,\ldots,X_n]$. Now
$Z=AX$ can be rewritten as $Z_j=AX_j$ for all $j$. So
\begin{eqnarray}
\frac{\partial Z}{\partial X} = \frac{\partial (Z_1,\ldots,
Z_n)}{\partial (X_1,\ldots, X_n)} = \Br{\begin{array}{cccc}
                                          A &  &  &  \\
                                           & A &  &  \\
                                           &  & \ddots &  \\
                                           &  &  & A
                                        \end{array}
},
\end{eqnarray}
implying that
$$
[\dif Z]=\det\Pa{\Br{\begin{array}{cccc}
                                          A &  &  &  \\
                                           & A &  &  \\
                                           &  & \ddots &  \\
                                           &  &  & A
                                        \end{array}
}}[\dif X] = \det(A)^n[\dif X].
$$
\end{remark}

\begin{prop}\label{prop:LR-lower-triangular}
Let $X,A, B\in\real^{n\times n}$ be lower triangular matrices where
$A=[a_{ij}]$ and $B=[b_{ij}]$ are constant matrices with
$a_{jj}>0,b_{jj}>0, j=1,\ldots,n$ and $X$ is a matrix of independent
real variables. Then
\begin{eqnarray}
Y=X+X^\t&\Longrightarrow& [\dif Y]=2^n[\dif X],\\
Y=AX&\Longrightarrow&
[\dif Y]=\Pa{\prod^n_{j=1}a^j_{jj}}[\dif X],\\
Y=XB&\Longrightarrow& [\dif Y]=\Pa{\prod^n_{j=1}b^{n-j+1}_{jj}}[\dif
X].
\end{eqnarray}
Thus
\begin{eqnarray}
Y=AXB&\Longrightarrow& [\dif
Y]=\Pa{\prod^n_{j=1}a^j_{jj}b^{n-j+1}_{jj}}[\dif X].
\end{eqnarray}
\end{prop}

\begin{proof}
$Y=X+X^\t$ implies that
\begin{eqnarray*}
\Br{\begin{array}{cccc}
      x_{11} & 0&\cdots & 0 \\
      x_{21} & x_{22}&\cdots & 0 \\
      \vdots & \vdots&\ddots & \vdots \\
      x_{n1} & x_{n2}&\cdots & x_{nn}
    \end{array}
}+\Br{\begin{array}{cccc}
      x_{11} & x_{21}&\cdots & x_{n1} \\
      0 & x_{22}&\cdots & x_{2n} \\
      \vdots & \vdots&\ddots & \vdots \\
      0 & 0&\cdots & x_{nn}
    \end{array}
} = \Br{\begin{array}{cccc}
      2x_{11} & x_{21} & \cdots & x_{n1} \\
      x_{21} & 2x_{22} & \cdots & x_{n2} \\
      \vdots & \vdots & \ddots & \vdots \\
      x_{n1} & x_{n2} & \cdots & 2x_{nn}
    \end{array}}.
\end{eqnarray*}
When taking the partial derivatives the $n$ diagonal elements give
$2$ each and others unities and hence $[\dif Y]=2^n[\dif X]$. If
$Y=AX$, then the matrices of the configurations of the partial
derivatives, by taking the elements in the orders
$(y_{11},y_{21}\ldots, y_{n1}); (y_{22},\ldots,y_{n2}); \ldots;
y_{nn}$ and $(x_{11},x_{21}\ldots, x_{n1}); (x_{22},\ldots,x_{n2});
\ldots; x_{nn}$ are the following:
\begin{eqnarray*}
\frac{\partial (y_{11},y_{21}\ldots,y_{n1})}{\partial
(x_{11},x_{21}\ldots,x_{n1})}=A,~~\frac{\partial(y_{22},\ldots,y_{n2})}{\partial(x_{22},\ldots,x_{n2})}
= A[\hat 1|\hat 1],\ldots,\\
\frac{\partial y_{nn}}{\partial x_{nn}} = A[\hat
1\cdots\widehat{n-1}|\hat 1\cdots\widehat{n-1}]=a_{nn},
\end{eqnarray*}
where $A[\hat i_1\cdots\hat i_\mu|\hat j_1\cdots \hat j_\nu]$ means
that the obtained submatrix from deleting both the
$i_1,\ldots,i_\mu$-th rows and the $j_1,\ldots,j_\nu$-th columns of
$A$. Thus
\begin{eqnarray*}
\frac{\partial Y}{\partial X} &=& \Br{\begin{array}{cccc}
                                  \frac{\partial
(y_{11},y_{21},\ldots,y_{n1})}{\partial(x_{11},x_{21},\ldots,x_{n1})} & 0 & \cdots & 0 \\
                                   0 & \frac{\partial(y_{22},\ldots,y_{n2})}{\partial(x_{22},\ldots,x_{n2})} & \cdots & 0 \\
                                   \vdots & \vdots & \ddots & \vdots \\
                                   0 & 0 & \cdots & \frac{\partial y_{nn}}{\partial x_{nn}}
                                \end{array}}\\
                                &=& \Br{\begin{array}{cccc}
                                  A & 0 & \cdots & 0 \\
                                   0 & A[\hat 1|\hat 1] & \cdots & 0 \\
                                   \vdots & \vdots & \ddots & \vdots \\
                                   0 & 0 & \cdots & A[\hat
1\cdots\widehat{n-1}|\hat 1\cdots\widehat{n-1}]
                                \end{array}}.
\end{eqnarray*}
We can also take another approach to this proof. In fact, we
partition $X,Y$ by columns, respectively, $Y=[Y_1,\ldots,Y_n]$ and
$X=[X_1,\ldots, X_n]$. Then $Y=AX$ is equivalent to $Y_j=AX_j,
j=1,\ldots,n$. Since $Y,X,A$ are lower triangular, it follows that
\begin{eqnarray*}
\Br{\begin{array}{c}
      y_{11} \\
      y_{21} \\
      \vdots \\
      y_{n1}
    \end{array}
} = A\Br{\begin{array}{c}
      x_{11} \\
      x_{21} \\
      \vdots \\
      x_{n1}
    \end{array}
},\Br{\begin{array}{c}
      y_{22} \\
      \vdots \\
      y_{n2}
    \end{array}
} = A[\hat 1|\hat 1]\Br{\begin{array}{c}
      x_{22} \\
      \vdots \\
      x_{n2}
    \end{array}
},\ldots,y_{nn} = A[\hat 1\cdots\widehat{n-1}|\hat 1\cdots\widehat
{n-1}]x_{nn}=a_{nn}x_{nn}.
\end{eqnarray*}
Now
\begin{eqnarray*}
[\dif Y]&=&\prod^n_{j=1}[\dif Y_j] = \det(A)\det(A[\hat 1|\hat
1])\cdots \det(A[\hat 1\cdots\widehat{n-1}|\hat 1\cdots\widehat
{n-1}])\prod^n_{j=1}[\dif X_j]\\
&=& \Pa{\prod^n_{j=1}a_{jj}^j}[\dif X].
\end{eqnarray*}
Next if $Y=XB$, that is,
\begin{eqnarray*}
Y=XB &=&\Br{\begin{array}{cccc}
      x_{11} & 0&\cdots & 0 \\
      x_{21} & x_{22}&\cdots & 0 \\
      \vdots & \vdots&\ddots & \vdots \\
      x_{n1} & x_{n2}&\cdots & x_{nn}
    \end{array}
}\Br{\begin{array}{cccc}
      b_{11} & 0&\cdots & 0 \\
      b_{21} & b_{22}&\cdots & 0 \\
      \vdots & \vdots&\ddots & \vdots \\
      b_{n1} & b_{n2}&\cdots & b_{nn}
    \end{array}
}\notag\\
&=&\Br{\begin{array}{cccc}
      x_{11}b_{11} & 0&\cdots & 0 \\
      x_{21}b_{11}+x_{22}b_{21} & x_{22}b_{22}&\cdots & 0 \\
      \vdots & \vdots&\ddots & \vdots \\
      \sum^n_{j=1}x_{nj}b_{j1} & \sum^n_{j=1}x_{nj}b_{j2}&\cdots & x_{nn}b_{nn}
    \end{array}
}
\end{eqnarray*}
The matrices of the configurations of the partial derivatives, by
taking the elements in the orders
$y_{11};(y_{21},y_{22});\ldots;(y_{n1},\ldots, y_{nn})$ and
$x_{11};(x_{21},x_{22});\ldots;(x_{n1},\ldots, x_{nn})$ are the
following:
\begin{eqnarray*}
\frac{\partial y_{11}}{\partial
x_{11}}=b_{11},~~\frac{\partial(y_{21},y_{22})}{\partial(x_{21},x_{22})}
= \Br{\begin{array}{cc}
        b_{11} & b_{21} \\
        0 & b_{22}
      \end{array}
}=\Br{\begin{array}{cc}
        b_{11} & 0 \\
        b_{21} & b_{22}
      \end{array}
}^\t,\\
\frac{\partial
(y_{31},y_{32},y_{33})}{\partial(x_{31},x_{32},x_{33})} =
\Br{\begin{array}{ccc}
                                                                b_{11} & b_{21} & b_{31} \\
                                                                0 & b_{22} & b_{32} \\
                                                                0 & 0 &
                                                                b_{33}
                                                              \end{array}
}= \Br{\begin{array}{ccc}
                                                                b_{11} & 0 & 0 \\
                                                                b_{21} & b_{22} & 0 \\
                                                                b_{31} & b_{32} &
                                                                b_{33}
                                                              \end{array}
}^\t,\ldots,\\
\frac{\partial
(y_{n1},y_{n2},\ldots,y_{nn})}{\partial(x_{n1},x_{n2},\ldots,x_{nn})}
= \Br{\begin{array}{cccc}
        b_{11} & 0 & \cdots & 0 \\
        b_{21} & b_{22} & \cdots & 0 \\
        \vdots & \vdots & \ddots & \vdots \\
        b_{n1} & b_{n2} & \cdots & b_{nn}
      \end{array}
}^\t.
\end{eqnarray*}
Thus
\begin{eqnarray*}
\frac{\partial Y}{\partial X} = \Br{\begin{array}{cccc}
                                  \frac{\partial y_{11}}{\partial x_{11}} & 0 & \cdots & 0 \\
                                   0 & \frac{\partial(y_{21},y_{22})}{\partial(x_{21},x_{22})} & \cdots & 0 \\
                                   \vdots & \vdots & \ddots & \vdots \\
                                   0 & 0 & \cdots & \frac{\partial
(y_{n1},y_{n2},\ldots,y_{nn})}{\partial(x_{n1},x_{n2},\ldots,x_{nn})}
                                \end{array}}
\end{eqnarray*}
Denote by $B[i_1\ldots i_\mu|j_1\ldots j_\nu]$ the sub-matrix formed
by the $i_1,\ldots,i_\mu$-th rows and $j_1,\ldots,j_\nu$-th columns
of $B$. Hence
\begin{eqnarray*}
\frac{\partial Y}{\partial X} = \Br{\begin{array}{cccc}
                                 B[1|1] & 0 & \cdots & 0 \\
                                   0 & B[12|12] & \cdots & 0 \\
                                   \vdots & \vdots & \ddots & \vdots \\
                                   0 & 0 & \cdots & B[1\cdots n|1\cdots n]
                                \end{array}}^\t.
\end{eqnarray*}
The whole configuration is a upper triangular matrix with $b_{11}$
appearing $n$ times and $b_{22}$ appearing $n-1$ times and so on in
the diagonal. Also we give another approach to derive the Jacobian
for $Y=XB$. Indeed, we partition $Y,X$ by rows, respectively,
$$
Y=\Br{\begin{array}{c}
        Y_1 \\
        Y_2 \\
        \vdots \\
        Y_n
      \end{array}
},~~X=\Br{\begin{array}{c}
        X_1 \\
        X_2 \\
        \vdots \\
        X_n
      \end{array}
},
$$
where $Y_j,X_j$ are row-vectors. So $Y=XB$ is equivalent to
$Y_j=X_jB, j=1,\ldots,n$. Moreover
\begin{eqnarray*}
&&y_{11}=x_{11}B[1|1]=x_{11}a_{11}, [y_{21}, y_{22}]=[x_{21},
x_{22}]B[12|12], \ldots, \\
&&~[y_{n1},\ldots,y_{nn}] = [y_{n1},\ldots,y_{nn}]B[1\ldots
n|1\ldots n].
\end{eqnarray*}
Therefore
\begin{eqnarray*}
[\dif Y]&=&\prod^n_{j=1}[\dif Y_j] =
\det(B[1|1])\det(B[12|12])\cdots\det(B[1\cdots n|1\cdots
n])\prod^n_{j=1}[\dif X_j]\\
&=& \Pa{\prod^n_{j=1}b^{n+j-1}_{jj}}[\dif X].
\end{eqnarray*}
We are done.
\end{proof}

\begin{prop}\label{prop:lower-triangular-to-transpose}
Let $X$ be a lower triangular matrix of independent real variables
and $A=[a_{ij}]$ and $B=[b_{ij}]$ be lower triangular matrices of
constants with $a_{jj}>0, b_{ij}>0,j=1,\ldots,n$. Then
\begin{eqnarray}
Y=AX+X^\t A^\t &\Longrightarrow&
[\dif Y]=2^n\Pa{\prod^n_{j=1}a^j_{jj}}[\dif X],\\
Y= XB + B^\t X^\t&\Longrightarrow& [\dif
Y]=2^n\Pa{\prod^n_{j=1}b^{n-j+1}_{jj}}[\dif X].
\end{eqnarray}
\end{prop}

\begin{proof}
Let $Z=AX$. Then $Y=Z+Z^\t$. Thus $[\dif Y]=2^n [dZ]$. Since $Z=AX$,
it follows from Proposition~\ref{prop:LR-lower-triangular} that
$$
[\dif Z] = \Pa{\prod^n_{j=1}a^j_{jj}}[\dif X],
$$
implying the result. The proof of the second identity goes
similarly.
\end{proof}

\begin{prop}\label{prop:symmetric-case}
Let $X$ and $Y$ be $n\times n$ symmetric matrices of independent
real variables and $A\in\real^{n\times n}$ nonsingular matrix of
constants. If $Y=AXA^\t$, then
\begin{eqnarray}
\fbox{$[\dif Y] = \det(A)^{n+1}[\dif X]$.}
\end{eqnarray}
\end{prop}

\begin{proof}
Since both $X$ and $Y$ are symmetric matrices and $A$ is nonsingular
we can split $A$ and $A^\t$ as products of elementary matrices and
write in the form
\begin{eqnarray*}
Y=\cdots E_2E_1XE^\t_1E^\t_2\cdots
\end{eqnarray*}
where $E_j,j=1,2,\ldots$ are elementary matrices. Write $Y=AXA^\t$
as a sequence of transformations of the type
\begin{eqnarray*}
Y_1=E_1XE^\t_1, Y_2=E_2Y_1E^\t_2,\ldots,\Longrightarrow\\
~[\dif Y_1] = J(Y_1:X)[\dif X],[\dif Y_2] = J(Y_2:Y_1)[\dif
Y_1],\ldots
\end{eqnarray*}
Now successive substitutions give the final result as long as the
Jacobians of the type $J(Y_k:Y_{k-1})$ are computed. Note that the
elementary matrices are formed by multiplying any row (or column) of
an identity matrix with a scalar, adding a row (column) to another
row (column) and combinations of these operations. Hence we need to
consider only these two basic elementary matrices. Let us consider a
$3\times 3$ case and compute the Jacobians. Let $E_1$ be the
elementary matrix obtained by multiplying the first row by $\alpha$
and $E_2$ by adding the first row to the second row of an identity
matrix. That is,
\begin{eqnarray*}
E_1=\Br{\begin{array}{ccc}
          \alpha & 0 & 0 \\
          0 & 1 & 0 \\
          0 & 0 & 1
        \end{array}
}, E_2=\Br{\begin{array}{ccc}
             1 & 0 & 0 \\
             1 & 1 & 0 \\
             0 & 0 & 1
           \end{array}
}
\end{eqnarray*}
and
\begin{eqnarray*}
E_1XE^\t_1&=&\Br{\begin{array}{ccc}
             \alpha^2x_{11} & \alpha x_{12} & \alpha x_{13} \\
             \alpha x_{21} & x_{22} & x_{23} \\
             \alpha x_{31} & x_{32} & x_{33}
           \end{array}
},\\
E_2 Y_1E^\t_2 &=& \Br{\begin{array}{ccc}
                      u_{11} & u_{11}+u_{12} & u_{13} \\
                      u_{11}+u_{21} & u_{11}+u_{21}+u_{12}+u_{22} & u_{13}+u_{23} \\
                      u_{31} & u_{31}+u_{32} & u_{33}
                    \end{array}
},
\end{eqnarray*}
where $Y_1=E_1XE^\t_1$ and $Y_2=E_2 Y_1E^\t_2$ and the elements of
$Y_1$ are denoted by $u_{ij}$'s for convenience. The matrix of
partial derivatives in the transformation $Y_1$ written as a
function of $X$ is then
\begin{eqnarray*}
\frac{\partial Y_1}{\partial X} = \Br{\begin{array}{cccccc}
                                        \alpha^2 & 0 & 0 & 0 & 0 & 0 \\
                                        0 & \alpha & 0 & 0 & 0 & 0 \\
                                        0 & 0 & 1 & 0 & 0 & 0 \\
                                        0 & 0 & 0 & \alpha & 0 & 0 \\
                                        0 & 0 & 0 & 0 & 1 & 0 \\
                                        0 & 0 & 0 & 0 & 0 & 1
                                      \end{array}
}.
\end{eqnarray*}
This is obtained by taking the $x_{ij}$'s in the order
$x_{11};(x_{21},x_{22}); (x_{31},x_{32},x_{33})$ and the $u_{ij}$'s
also in the same order. Thus the Jacobian is given by
\begin{eqnarray*}
J(Y_1:X)=\alpha^4 = \alpha^{3+1}=\det(E_1)^{3+1}.
\end{eqnarray*}
Or by definition it follows directly that
$$
[\dif Y_1] = \dif(\alpha^2 x_{11})\dif(\alpha x_{12})\dif (\alpha
x_{13})\dif x_{22}\dif x_{23}\dif x_{33} = \alpha^4 [\dif X].
$$
For a $n\times n$ matrix it will be $\alpha^{n+1}$. Let the elements
of $Y_2$ be denoted by $v_{ij}$'s. Then again taking the variables
in the order as in the case of $Y_1$ written as a function of $X$
the matrix of partial derivatives in this transformation is the
following:
\begin{eqnarray*}
\frac{\partial Y_2}{\partial Y_1} = \Br{\begin{array}{cccccc}
                                        1 & 0 & 0 & 0 & 0 & 0 \\
                                        1 & 1 & 0 & 0 & 0 & 0 \\
                                        1 & 2 & 1 & 0 & 0 & 0 \\
                                        0 & 0 & 0 & 1 & 0 & 0 \\
                                        0 & 0 & 0 & 1 & 1 & 0 \\
                                        0 & 0 & 0 & 0 & 0 & 1
                                      \end{array}
}.
\end{eqnarray*}
The determinant of this matrix is $1=1^{3+1}=\det(E_2)^{3+1}$. In
general such a transformation gives the Jacobian, in absolute value,
as $1=1^{n+1}$. Thus the Jacobian is given by
\begin{eqnarray*}
J(Y:X) = \det(\cdots E_2E_1)^{n+1}=\det(A)^{n+1}.
\end{eqnarray*}
We are done.
\end{proof}

\begin{exam}
Let $X\in\real^{n\times n}$ be a real symmetric positive definite
matrix having a matrix-variate gamma distribution with parameters
$(\alpha,B=B^\t>0)$. We show that
\begin{framed}
\begin{eqnarray}
\det(B)^{-\alpha} = \frac1{\Gamma_n(\alpha)}\int_{X>0}[\dif X]
\det(X)^{\alpha-\frac{n+1}2}e^{-\Tr{BX}},~~~\re(\alpha)>\frac{n-1}2.
\end{eqnarray}
\end{framed}
Indeed, since $B$ is symmetric positive definite there exists a
nonsingular matrix $C$ such that $B=CC^\t$. Note that
\begin{eqnarray*}
\Tr{BX}=\Tr{C^\t XC}.
\end{eqnarray*}
Let
\begin{eqnarray*}
U=C^\t XC\Longrightarrow [\dif U]=\det(C)^{n+1}[\dif X]
\end{eqnarray*}
from Proposition~\ref{prop:symmetric-case} and
$\det(X)=\det(B)^{-1}\det(U)$. The integral on the right reduces to
the following:
\begin{eqnarray*}
\int_{X>0}[\dif X] \det(X)^{\alpha-\frac{n+1}2}e^{-\Tr{BX}} =
\det(B)^{-\alpha}\int_{U>0}[\dif
U]\det(U)^{\alpha-\frac{n+1}2}e^{-\Tr{U}}.
\end{eqnarray*}
But
\begin{eqnarray*}
\int_{U>0}[\dif
U]\det(U)^{\alpha-\frac{n+1}2}e^{-\Tr{U}}=\Gamma_n(\alpha)
\end{eqnarray*}
for $\re(\alpha)>\frac{n-1}2$. The result is obtained.
\end{exam}

\begin{prop}\label{prop:skew-symmetric-case}
Let $X,Y\in\real^{n\times n}$ skew symmetric matrices of independent
real variables and $A\in\real^{n\times n}$ nonsingular matrix of
constants. If $Y=AXA^\t$, then
\begin{eqnarray}
\fbox{$[\dif Y] = \det(A)^{n-1}[\dif X]$.}
\end{eqnarray}
\end{prop}
Note that when $X$ is skew symmetric the diagonal elements are zeros
and hence there are only $\frac{n(n-1)}2$ independent variables in
$X$.

\begin{prop}\label{prop:lower-upper-triangular}
Let $X,A,B\in\real^{n\times n}$ be lower triangular matrices where
$A=[a_{ij}]$ and $B=[b_{ij}]$ are nonsingular constant matrices with
positive diagonal elements, respectively, and $X$ is a matrix of
independent real variables. Then
\begin{eqnarray}
Y = A^\t X+X^\t A&\Longrightarrow&
[\dif Y]=2^n\Pa{\prod^n_{j=1}a^j_{jj}}[\dif X],\\
Y = XB^\t+ B X^\t &\Longrightarrow& [\dif
Y]=2^n\Pa{\prod^n_{j=1}b^{n-j+1}_{jj}}[\dif X].
\end{eqnarray}
\end{prop}

\begin{proof}
Consider $Y = A^\t X+X^\t A$. Premultiply by $(A^\t)^{-1}$ and
postmultiply by $A^{-1}$ to get the following:
\begin{eqnarray*}
Y=A^\t X+X^\t A\Longrightarrow
(A^\t)^{-1}YA^{-1}=(A^\t)^{-1}X^\t+XA^{-1}.
\end{eqnarray*}
Let
\begin{eqnarray*}
Z = XA^{-1}+\Pa{XA^{-1}}^\t \Longrightarrow [\dif
Z]=2^n\Pa{\prod^n_{j=1}a^{-(n-j+1)}_{jj}}[\dif X]
\end{eqnarray*}
by Proposition~\ref{prop:lower-triangular-to-transpose} and
\begin{eqnarray*}
Z = \Pa{A^{-1}}^\t YA^{-1}\Longrightarrow [\dif
Z]=\det(A)^{-(n+1)}[\dif Y]
\end{eqnarray*}
by Proposition~\ref{prop:symmetric-case}. Now writing $[\dif Y]$ in
terms of $[\dif X]$ one has
\begin{eqnarray*}
[\dif Y] =
\Pa{\prod^n_{j=1}a^{-(n-j+1)}_{jj}}2^n\Pa{\prod^n_{j=1}a^{n+1}_{jj}}[\dif
X] =2^n\Pa{\prod^n_{j=1}a^j_{jj}}[\dif X]
\end{eqnarray*}
since $\det(A)=\prod^n_{j=1}a_{jj}$ because $A$ is lower triangular.
Thus the first result follows. The second is proved as follows.
Clearly,
$$
Y=BX^\t+XB^\t\Longrightarrow
B^{-1}Y(B^\t)^{-1}=B^{-1}X+\Pa{B^{-1}X}^\t:=Z.
$$
Thus $[\dif Z]=\det(B)^{-(n+1)}[\dif Y]$ and $[\dif
Z]=2^n\Pa{\prod^n_{j=1}a^{-j}_{jj}}[\dif X]$. Therefore expressing
$[\dif Y]$ in terms of $[\dif X]$ gives the second result.
\end{proof}

\begin{prop}\label{prop:TT'}
Let $X\in\real^{n\times n}$ be a symmetric positive definite matrix
of independent real variables and $T=[t_{ij}]$ a real
\red{lower-triangular} matrix with $t_{jj}>0,j=1,\ldots,n$, and
$t_{ij},i\geqslant j$ independent. Then
\begin{eqnarray}
X = T^\t T &\Longrightarrow& [\dif X] =
2^n\Pa{\prod^n_{j=1}t^j_{jj}}[\dif T],\\
X = TT^\t &\Longrightarrow& [\dif X] =
2^n\Pa{\prod^n_{j=1}t^{n-j+1}_{jj}}[\dif T].
\end{eqnarray}
\end{prop}

\begin{proof}
By considering the matrix of differentials one has
\begin{eqnarray*}
X=TT^\t\Longrightarrow\dif X=\dif T\cdot T^\t + T\cdot\dif T^\t.
\end{eqnarray*}
Now treat this as a linear transformation in the differentials, that
is, $\dif X$ and $\dif T$ as variables and $T$ a constant. This
completes the proof.
\end{proof}

\begin{exam}
Let $X\in\real^{n\times n}$ symmetric positive definite matrix and
$\re(\alpha)>\frac{n-1}2$. Show that
\begin{framed}
\begin{eqnarray}
\Gamma_n(\alpha) &:=& \int_{X>0} [\dif X]
\det(X)^{\alpha-\frac{n+1}2}e^{-\Tr{X}}\notag \\
&=&
\pi^{\frac{n(n-1)}4}\Gamma(\alpha)\Gamma\Pa{\alpha-\frac12}\cdots\Gamma\Pa{\alpha-\frac{n-1}2}.
\end{eqnarray}
\end{framed}
Let $T$ be a real lower triangular matrix with positive diagonal
elements. Then the unique representation (see
Theorem~\ref{th:unique-factorization})
\begin{eqnarray*}
X = TT^\t\Longrightarrow[\dif X] =
2^n\Pa{\prod^n_{j=1}t^{n-j+1}_{jj}}[\dif T].
\end{eqnarray*}
Note that
\begin{eqnarray*}
\Tr{X} &=& \Tr{TT^\t} =
t^2_{11}+(t^2_{21}+t^2_{22})+\cdots+\Pa{t^2_{n1}+\cdots+t^2_{nn}},\\
\det(X) &=& \det(TT^\t) = \prod^n_{j=1}t^2_{jj}.
\end{eqnarray*}
When $X>0$, we have $TT^\t>0$, but $t_{jj}>0,j=1,\ldots,n$ which
means that $-\infty<t_{ij}<\infty, i>j, 0<t_{jj}<\infty,
j=1,\ldots,n$. The integral splits into $n$ integrals on $t_{jj}$'s
and $\frac{n(n-1)}2$ integrals on $t_{ij}$'s, $i>j$. That is,
\begin{eqnarray*}
\Gamma_n(\alpha)=\Pa{\prod^n_{j=1}2\int^\infty_0
\Pa{t^2_{jj}}^{\alpha-\frac{n+1}2}t^{n-j+1}_{jj}e^{-t^2_{jj}}\dif
t_{jj}} \times\Pa{\prod_{i>j}\int^\infty_{-\infty}e^{-t^2_{ij}}\dif
t_{ij}}.
\end{eqnarray*}
But
\begin{eqnarray*}
2\int^\infty_0 \pa{t^2_{jj}}^{\alpha-\frac j2}e^{-t^2_{jj}}\dif
t_{jj} = \Gamma\Pa{\alpha-\frac{j-1}2},
\end{eqnarray*}
for $\re(\alpha)>\frac{j-1}2, j=1,\ldots,n$ and
\begin{eqnarray*}
\int^\infty_{-\infty}e^{-t^2_{jj}}\dif t_{jj} =\sqrt{\pi}.
\end{eqnarray*}
Multiplying them together the result follows. Note that
\begin{eqnarray*}
\re(\alpha)>\frac{j-1}2,
j=1,\ldots,n\Longrightarrow\re(\alpha)>\frac{n-1}2.
\end{eqnarray*}
\end{exam}

The next result is extremely useful and uses the fact (see
Theorem~\ref{th:unique-factorization}) that any positive definite
$n\times n$ matrix $X$ has a unique decomposition as $X=T^\t T$,
where $T$ is an upper-triangular $n\times n$ matrix with positive
diagonal elements.
\begin{prop}
If $X$ is an $n\times n$ positive definite matrix and $X=T^\t T$,
where $T$ is \red{upper-triangular} with positive diagonal elements,
then
$$
[\dif X] = 2^n\prod^n_{j=1}t^{n-j+1}_{jj}[\dif T].
$$
\end{prop}

\begin{proof}
Since $X=T^\t T$, now express each of the elements of $X$ on and
above the diagonal in terms of each of the elements of $T$ and take
differentials. Remember that we are going to take the exterior
product of these differentials and that products of repeated
differentials are zero; hence there is no need to keep track of
differentials in the elements of $T$ which have previously occurred.
We get
\begin{eqnarray*}
x_{11} = t^2_{11}, &\quad&\dif x_{11}=2t_{11}\dif t_{11},\\
x_{12} = t_{11}t_{12},&\quad&\dif x_{12}=t_{11}\dif t_{12}+\cdots\\
&\vdots &\\
x_{1n} =t_{11}t_{1n},&\quad&\dif x_{1n} =t_{11}\dif t_{1n}+\cdots\\
x_{22}=t^2_{12}+t^2_{22},&\quad&\dif x_{22} =2t_{22}\dif
t_{22}+\cdots\\
&\vdots &\\
x_{2n} = t_{12}t_{1n}+t_{22}t_{2n},&\quad&\dif x_{2n}=t_{22}\dif
t_{2n}+\cdots\\
&\vdots &\\
x_{nn}=t^2_{1n}+\cdots+t^2_{nn},&\quad&\dif x_{nn}=2t_{nn}\dif
t_{nn}+\cdots
\end{eqnarray*}
Hence taking exterior products gives
$$
[\dif X] = \bigwedge^n_{i\leqslant j}\dif x_{ij} = 2^n
t^n_{11}t^{n-1}_{22}\cdots t_{nn}\bigwedge^n_{i\leqslant j}\dif
t_{ij} = 2^n\prod^n_{j=1}t^{n-j+1}_{jj}[\dif T],
$$
as desired.
\end{proof}

\subsection{The computation of volumes}

\begin{definition}[Stiefel manifold]
Let $A$ be a $n\times m(n\geqslant m)$ matrix with real entries such
that $A^\t A=\I_m$, that is, the $m$ columns of $A$ are
\emph{orthonormal vectors}. The set of all such matrices $A$ is
known as the \emph{Stiefel manifold}, denoted by $\cO(m,n)$. That
is, for all $n\times m$ matrices $A$,
\begin{eqnarray}\label{eq:2.1}
\fbox{$\cO(m,n) = \Set{A\in\real^{n\times m}: A^\t A=\I_m},$}
\end{eqnarray}
where $\real^{n\times m}$ denotes the set of all $n\times m$ real
matrices.

\end{definition}
The equation $A^\t A=\I_m$ imposes $m(m+1)/2$ conditions on the
elements of $A$. Thus the number of independent entries in $A$ is
$mn-m(m+1)/2$.

If $m=n$, then $A$ is an orthogonal matrix. The set of such
orthogonal matrices form a group. This group is known as the
\emph{orthogonal group} of $m\times m$ matrices.

\begin{definition}[Orthogonal group]
Let $B$ be a $n\times n$ matrix with real elements such that $B^\t
B=\I_n$. The set of all $B$ is called an \emph{orthogonal group},
denoted by $\cO(n)$. That is,
\begin{eqnarray}\label{eq:2.2}
\fbox{$\cO(n)=\Set{B\in\real^{n\times n}: B^\t B=\I_n}.$}
\end{eqnarray}
\end{definition}
Clearly $\cO(n,n)=\cO(n)$. Note that $B^\t B=\I_n$ imposes
$n(n+1)/2$ conditions and hence the number of independent entries in
$B$ is only $n^2-n(n+1)/2=n(n-1)/2$.

\begin{definition}[A symmetric or a skew symmetric matrix]
Let $A\in\real^{n\times n}$. If $A= A^\t$, then $A$ is said to be
\emph{symmetric} and if $A^\t=- A$, then it is \emph{skew
symmetric}.
\end{definition}

\begin{prop}\label{prop:Haar-measure}
Let $V\in\cO(n)$ with independent entries and the diagonal entries
or the entries in the first row of $V$ all positive. Denote
$\dif{G}=V^\t \dif{V}$ where $V=[v_1,\ldots, v_n]$. Then
\begin{eqnarray}
[\dif{G}] &=& \prod^{n-1}_{i=1}\prod^n_{j=i+1}\inner{v_i}{\dif{v_j}}\label{eq:ortho-1}\\
&=& 2^{n(n-1)/2}\det(\I_n+X)^{-(n-1)}[\dif{X}],\label{eq:ortho-2}
\end{eqnarray}
where $X$ is a skew symmetric matrix such that the first row entries
of $(\I_n+X)^{-1}$, except the first entry, are negative.
\end{prop}

\begin{proof}
Let the columns of $V$ be denoted by $v_1,\ldots,v_n$. Since the
columns are orthonormal, we have $\inner{v_i}{v_j}=\delta_{ij}$.
Then
$$
\inner{v_i}{\dif{v_j}}+\inner{\dif{v_i}}{v_j}=0,
$$
implying $\inner{v_j}{\dif{v_j}}=0$ since $\inner{v_j}{\dif{v_j}}$
is a real scalar. We also have
$$
\inner{v_i}{\dif{v_j}} = - \inner{v_j}{\dif{v_i}}~\text{for}~i\neq
j.
$$
Then $V^\t\dif{V}$ is a skew symmetric matrix. That is,
\begin{eqnarray*}
\dif{G}&=&V^\t\dif{V} = \Br{\begin{array}{c}
                            v^\t_1 \\
                            v^\t_2 \\
                            \vdots \\
                            v^\t_n
                          \end{array}
}\Br{\dif{v_1},\dif{v_2},\cdots,\dif{v_n}}\\
&=&\Br{\begin{array}{cccc}
         \inner{v_1}{\dif{v_1}} & \inner{v_1}{\dif{v_2}} & \cdots & \inner{v_1}{\dif{v_n}} \\
         \inner{v_2}{\dif{v_1}} & \inner{v_2}{\dif{v_2}} & \cdots & \inner{v_2}{\dif{v_n}} \\
         \vdots & \vdots & \ddots & \vdots \\
         \inner{v_n}{\dif{v_1}} & \inner{v_n}{\dif{v_2}} & \cdots & \inner{v_n}{\dif{v_n}}
       \end{array}
}.
\end{eqnarray*}
This indicates that
\begin{eqnarray*}
\dif{G} = \Br{\begin{array}{cccc}
         0 & \inner{v_1}{\dif{v_2}} & \cdots & \inner{v_1}{\dif{v_n}} \\
         -\inner{v_1}{\dif{v_2}} & 0 & \cdots & \inner{v_2}{\dif{v_n}} \\
         \vdots & \vdots & \ddots & \vdots \\
         -\inner{v_1}{\dif{v_n}} & -\inner{v_2}{\dif{v_n}} & \cdots &
         0
       \end{array}
}
\end{eqnarray*}
Then there are only $n(n-1)/2$ independent entries in $G$. Then
$[\dif{G}]$ is the wedge product of the entries upper the leading
diagonal in the matrix $V^\t\dif{V}$:
$$
[\dif{G}] = \wedge^{n-1}_{i=1}\wedge^n_{j=i+1}\inner{v_i}{\dif{v_j}}
$$
This establishes \eqref{eq:ortho-1}. For establishing
\eqref{eq:ortho-2}, take a skew symmetric matrix $X$, then
$V=2(\I_n+X)^{-1}-\I_n$ is orthonormal such that $VV^\t=\I_n$.
Further the matrix of differentials in $V$ is given by $\dif{V} =
-2(\I_n+X)^{-1}\cdot\dif{X}\cdot(\I_n+X)^{-1}$, i.e.
\begin{eqnarray*}
\dif{V} = -\frac12(\I_n+V)\cdot\dif{X}\cdot(\I_n+V).
\end{eqnarray*}
Thus
\begin{eqnarray*}
\dif{G} = V^\t\dif{V} = -\frac12(\I_n+V^\t)\cdot\dif{X}\cdot(\I_n+V)
\end{eqnarray*}
and the wedge product is obtained
\begin{eqnarray*}
[\dif{G}] = \det\Pa{\frac{\I_n+V^\t}{\sqrt{2}}}^{n-1}[\dif{X}] =
\det\Pa{\sqrt{2}(\I_n+X)^{-1}}^{n-1}[\dif{X}].
\end{eqnarray*}
Therefore the desired identity \eqref{eq:ortho-2} is proved.
\end{proof}

\begin{prop}\label{prop:spectral-decom}
Let $X$ be a $n\times n$ symmetric matrix of independent real
entries and with distinct and nonzero eigenvalues
$\lambda_1>\cdots>\lambda_n$ and let
$D=\diag(\lambda_1,\ldots,\lambda_n)$. Let $V\in\cO(n)$ be a unique
such that $X=VDV^\t$. Then
\begin{eqnarray}\label{eq:real-vol-element}
\fbox{$[\dif{X}] =
\Pa{\prod^{n-1}_{i=1}\prod^n_{j=i+1}\abs{\lambda_i-\lambda_j}}[\dif{D}][\dif{G}],$}
\end{eqnarray}
where $\dif{G}=V^\t\dif{V}$.
\end{prop}

\begin{proof}
Take the differentials in $X=VDV^\t$ to get
$$
\dif{X} = \dif{V}\cdot D\cdot V^\t + V\cdot \dif{D}\cdot V^\t +
V\cdot D\cdot \dif{V}^\t,
$$
implying that
$$
V^\t\cdot \dif{X}\cdot V = V^\t\dif{V}\cdot D + \dif{D} + D\cdot
\dif{V}^\t V.
$$
Let $\dif{Y}=V^\t\cdot \dif{X}\cdot V$ for fixed $V$. Since
$\dif{G}=V^\t\dif{V}$, it follows that $[\dif{X}]=[\dif{Y}]$ and
\begin{eqnarray*}
\dif{Y} &=& \dif{G}\cdot D + \dif{D} + D\cdot\dif{G}^\t \\
&=& \dif{D} + \dif{G}\cdot D  - D\cdot\dif{G}\\
&=& \dif{D} + [\dif{G}, D]
\end{eqnarray*}
where we used the fact that $\dif{G}^\t = -\dif{G}$ and the concept
of commutator, defined by $[M,N]:=MN-NM$. Clearly the commutator of
$M$ and $N$ is skew symmetric. Now
\begin{eqnarray*}
\dif{y_{jj}} = \dif{\lambda_j}~~\text{and}~~\dif{y_{ij}} =
(\lambda_j-\lambda_i)\dif{g_{ij}},~~i<j=1,\ldots,n.
\end{eqnarray*}
Then
\begin{eqnarray*}
[\dif{Y}] &=&
\Pa{\prod^n_{j=1}\dif{y_{jj}}}\Pa{\prod^n_{i<j}\dif{y_{ij}}} =
\Pa{\prod^n_{j=1}\dif{\lambda_j}}\Pa{\prod_{i<j}\abs{\lambda_i-\lambda_j}\dif{g_{ij}}}\\
&=&\Pa{\prod_{i<j}\abs{\lambda_i-\lambda_j}}[\dif{D}][\dif{G}],
\end{eqnarray*}
where $[\dif{D}]=\prod^n_{j=1}\dif{\lambda_j}$ and
$[\dif{G}]=\prod_{i<j}\dif{g_{ij}}$, and the desired conclusion is
obtained.
\end{proof}

Clearly $X=VDV^\t$, where
$D=\diag(\lambda_1,\ldots,\lambda_n)~(\lambda_1>\cdots>\lambda_n),
VV^\t=\I_n$, is not a one-to-one transformation from $X$ to $(D,V)$
since $X$ determines $2^n$ matrices $[\pm v_1,\ldots,\pm v_n]$,
where $v_1,\ldots,v_n$ are the columns of $V$, such that $X=VDV^\t$.

This transformation can be shown to be unique if one entry from each
row and column are of a specified sign, for example, the diagonal
entries are positive. Once this is done we are integrating with
respect to $\dif{G}$, where $\dif{G}=V^\t\dif{V}$ over the full
orthogonal group $\cO(n)$, the result must be divided by $2^n$ to
get the result for a unique transformation $X=VDV^\t$.

\begin{remark}
Now we can try to compute the following integral based on
Eq.~\eqref{eq:real-vol-element}:
\begin{eqnarray*}
\int_{X>0:\Tr{X}=1} [\dif{X}] &=&
\int_{\lambda_1>\cdots>\lambda_n>0}
\delta\Pa{\sum^n_{j=1}\lambda_j-1}\prod_{i<j}\abs{\lambda_i-\lambda_j}\prod^n_{j=1}\dif\lambda_j\times
\int_{\cO_1(n)}[\dif{G}]\\
&=& \frac1{n!} \int_0^\infty
\delta\Pa{\sum^n_{j=1}\lambda_j-1}\prod_{i<j}\abs{\lambda_i-\lambda_j}\prod^n_{j=1}\dif\lambda_j\times
\frac1{2^n}\int_{\cO(n)}[\dif{G}]\\
&=& \frac1{2^n n!}\frac1{C^{(1,1)}_n}\vol\Pa{\cO(n)},
\end{eqnarray*}
that is,
\begin{eqnarray*}
\vol\Pa{\density{\real^n}}:=\int_{X>0:\Tr{X}=1} [\dif{X}]
=\frac{\pi^{\frac{n(n-1)}4}\Gamma\Pa{\frac{n+1}2}}{\Gamma\Pa{\frac{n(n+1)}4}\Gamma\Pa{\frac12}}\prod^n_{j=1}\Gamma\Pa{\frac
j2}.
\end{eqnarray*}
See below for the notation $C^{(1,1)}_n$ and $\vol\Pa{\cO(n)}$.
\end{remark}

\begin{prop}\label{prop:rectangular-matrix}
Let $X$ be a $p\times n (p\leqslant n)$ matrix of rank $p$ and let
$X=TU^\t_1$, where $T$ is a $p\times p$ lower triangular matrix with
distinct nonzero diagonal entries and $U_1$ is a unique $n\times p$
semi-orthogonal matrix, $U^\t_1 U_1=\I_p$, all are of independent
real entries. Let $U_2$ be an $n\times (n-p)$ semi-orthogonal matrix
such that $U_1$ augmented with $U_2$ is a full orthogonal matrix.
That is, $U=[U_1~~U_2], U^\t U=\I_n, U^\t_2 U_2=\I_{n-p},
U^\t_1U_2=0$. Let $u_j$ be the $j$-th column of $U$ and $\dif{u_j}$
its differential. Then
\begin{eqnarray}
\fbox{$[\dif{X}] =
\Pa{\prod^p_{j=1}\abs{t_{jj}}^{n-j}}[\dif{T}][\dif{U_1}],$}
\end{eqnarray}
where
$$
[\dif{U_1}]=\wedge^p_{j=1}\wedge^n_{i=j+1}\inner{u_i}{\dif{u_j}}.
$$
\end{prop}

\begin{proof}
Note that
$$
U^\t U  = \Br{\begin{array}{c}
                U^\t_1 \\
                U^\t_2
              \end{array}
}[U_1~~U_2] = \Br{\begin{array}{cc}
                U^\t_2 U_2 & U^\t_1 U_2 \\
                U^\t_2 U_1 & U^\t_2 U_2
              \end{array}} = \Br{\begin{array}{cc}
                                   \I_p & 0 \\
                                   0 & \I_{n-p}
                                 \end{array}
              }.
$$
Take the differentials in $X=TU^\t_1$ to get
$$
\dif{X} = \dif{T}\cdot U^\t_1+T\cdot\dif{U^\t_1}.
$$
Then
\begin{eqnarray*}
\dif{X}\cdot U &=& \dif{T}\cdot U^\t_1 U+T\cdot\dif{U^\t_1}U\\
&=& \dif{T}\cdot U^\t_1[U_1,~~U_2]+T\cdot\dif{U^\t_1}[U_1,~~U_2]\\
&=& \Br{\dif{T}+T\cdot\dif{U^\t_1}\cdot
U_1,~~T\cdot\dif{U^\t_1}\cdot U_2}
\end{eqnarray*}
since $U^\t_1 U_1=\I_p, U^\t_1 U_2=0$. Make the substitutions
$$
\dif{W}=\dif{X}\cdot U, \dif{Y}=\dif{U^\t_1}\cdot U_1,
\dif{S}=\dif{U^\t_1}\cdot U_2, \dif{H}=T\cdot \dif{S}.
$$
Now we have
$$
\dif{W} = [\dif{T}+T\cdot \dif{Y},~~\dif{H}].
$$
Thus
$$
\dif T+T\cdot \dif{Y} =\Br{\begin{array}{cccc}
                       \dif t_{11} & 0 & \cdots & 0 \\
                       \dif t_{21} & \dif t_{22} & \cdots & \vdots \\
                       \vdots & \vdots & \ddots & 0 \\
                       \dif t_{p1} & \dif t_{p2} & \cdots & \dif t_{pp}
                     \end{array}
} + \Br{\begin{array}{cccc}
                       t_{11} & 0 & \cdots & 0 \\
                       t_{21} & t_{22} & \cdots & \vdots \\
                       \vdots & \vdots & \ddots & 0 \\
                       t_{p1} & t_{p2} & \cdots & t_{pp}
                     \end{array}
}\Br{\begin{array}{cccc}
       0 & \dif{y_{12}} & \cdots & \dif{y_{1p}} \\
       -\dif{y_{12}} & 0 & \cdots & \dif{y_{2p}} \\
       \vdots & \vdots & \ddots & \vdots \\
       -\dif{y_{1p}} & -\dif{y_{2p}} & \cdots & 0
     \end{array}
}.
$$
Let us consider, for example, the case where $p=2,3$ in order for
computing the wedge product of $\dif T + T\cdot \dif Y$. Now for
$p=2$, we have
\begin{eqnarray*}
\dif T + T\cdot \dif Y &=& \Br{\begin{array}{cc}
                               \dif t_{11} & 0 \\
                               \dif t_{21} & \dif t_{22}
                             \end{array}
}+\Br{\begin{array}{cc}
                                t_{11} & 0 \\
                                t_{21} & t_{22}
                             \end{array}
}\Br{\begin{array}{cc}
                                0 & \dif y_{12} \\
                                -\dif y_{12} & 0
                             \end{array}
} \\
&=& \Br{\begin{array}{cc}
          \dif t_{11} & t_{11}\dif y_{12} \\
          \dif t_{21}-t_{22}\dif y_{12} & \dif t_{22}+t_{21}\dif
          y_{12}
        \end{array}
}
\end{eqnarray*}
Thus the wedge product of $\dif T+T\cdot\dif Y$ is:
\begin{eqnarray*}
[\dif T+T\cdot\dif Y] &=& \dif t_{11}\wedge (t_{11}\dif
y_{12})\wedge (\dif t_{21}-t_{22}\dif y_{12})\wedge (\dif
t_{22}+t_{21}\dif y_{12})\\
&=& t_{11}\dif t_{11}\wedge \dif y_{12}\wedge \dif t_{21}\wedge \dif
t_{22} = t_{11}[\dif T][\dif Y]\\
&=& \Pa{\prod^2_{j=1}\abs{t_{jj}}^{2-j}}[\dif T][\dif Y].
\end{eqnarray*}
For $p=3$, we have
\begin{eqnarray*}
\dif T + T\cdot \dif Y &=& \Br{\begin{array}{ccc}
                               \dif t_{11} & 0 &0\\
                               \dif t_{21} & \dif t_{22}&0\\
                               \dif t_{31} & \dif t_{32}&\dif t_{33}
                             \end{array}
}+\Br{\begin{array}{ccc}
                                t_{11} & 0&0 \\
                                t_{21} & t_{22}&0\\
                                t_{31} & t_{32}&t_{33}
                             \end{array}
}\Br{\begin{array}{ccc}
                                0 & \dif y_{12}&\dif y_{13} \\
                                -\dif y_{12} & 0 &\dif y_{23}\\
                                -\dif y_{13} & -\dif y_{23}&0
                             \end{array}
} \\
&=& \Br{\begin{array}{ccc}
          \dif t_{11} & t_{11}\dif y_{12}& t_{11}\dif y_{13}\\
          \dif t_{21}-t_{22}\dif y_{12} & \dif t_{22} + t_{21}\dif y_{12}& t_{21}\dif y_{13}+t_{22}\dif
          y_{23}\\
           \dif t_{31}-t_{32}\dif y_{12} - t_{33}\dif y_{13}   & \dif t_{32}+t_{31}\dif y_{12}-t_{33}\dif y_{23} & \dif t_{33} + t_{31}\dif y_{13}+t_{32}\dif y_{23}
        \end{array}
},
\end{eqnarray*}
implying the wedge product of $\dif T+T\cdot \dif Y$ is:
$$
t^2_{11}t_{22}[\dif T][\dif Y] =
\Pa{\prod^3_{j=1}\abs{t_{jj}}^{3-j}}[\dif T][\dif Y].
$$
For the general $p$, by straight multiplication, and remembering
that the variables are only
$$
\dif{y_{12}},\ldots,\dif{y_{1p}},\dif{y_{23}},\ldots,\dif{y_{2p}},\ldots,
\dif{y_{p-1p}}.
$$
Thus the wedge product of $\dif{T}+T\cdot\dif{Y}$
gives
$$
\Pa{\prod^p_{j=1}\abs{t_{jj}}^{p-j}}[\dif{Y}][\dif{T}],
$$
ignoring the sign, and
$$
[\dif{Y}] =
\wedge^{p-1}_{j=1}\wedge^p_{i=j+1}\inner{u_i}{\dif{u_j}}.
$$
Now consider $\dif{H}=T\cdot\dif{S}$. Since $\dif{S}$ is a $p\times
(n-p)$ matrix, we have
$$
[\dif{H}] = \det(T)^{n-p}[\dif{S}] =
\Pa{\prod^p_{j=1}\abs{t_{jj}}^{n-p}}[\dif{S}].
$$
The wedge product in $\dif{S}$ is the following:
$$
[\dif{S}]= \wedge^p_{j=1}\wedge^n_{i=p+j}\inner{u_i}{\dif{u_j}}.
$$
Hence from the above equations,
\begin{eqnarray*}
[\dif{X}]&=&[\dif{W}] = \wedge[\dif T+T\dif Y]\wedge[\dif H]\\
&=&\Pa{\prod^p_{j=1}\abs{t_{jj}}^{p-j}}[\dif{Y}][\dif{T}]
\Pa{\prod^p_{j=1}\abs{t_{jj}}^{n-p}}[\dif{S}]\\
&=&\Pa{\prod^p_{j=1}\abs{t_{jj}}^{n-j}}[\dif{Y}][\dif{T}][\dif{S}].
\end{eqnarray*}
Now
$$
[\dif Y][\dif S] =
\wedge^{p-1}_{j=1}\wedge^p_{i=j+1}\inner{u_i}{\dif{u_j}}\wedge^p_{j=1}\wedge^n_{i=p+j}\inner{u_i}{\dif{u_j}}.
$$
Substituting back one has
$$
[\dif{X}]=\Pa{\prod^p_{j=1}\abs{t_{jj}}^{n-j}}[\dif{T}]\wedge^p_{j=1}\wedge^n_{i=j+1}\inner{u_i}{\dif{u_j}}
$$
which establishes the result.
\end{proof}
If the triangular matrix $T$ is restricted to the one with positive
diagonal entries, that is, $t_{jj}>0,j=1,\ldots,p$, then while
integrating over $T$ using
Proposition~\ref{prop:rectangular-matrix}, the result must be
multiplied by $2^p$. Without the factor $2^p$, the $t_{jj}$'s must
be integrated over $-\infty<t_{jj}<\infty,j=1,\ldots,p$. If the
expression to be integrated contains both $T$ and $U$, then restrict
$t_{jj}>0,j=1,\ldots,p$ and integrate $U$ over the full Stiefel
manifold. If the rows of $U$ are $u_1,\ldots,u_p$, then $\pm
u_1,\ldots,\pm u_p$ give $2^p$ choices. Similarly
$t_{jj}>0,t_{jj}<0$ give $2^p$ choices. But there are not $2^{2p}$
choices in $X=TU$. There are only $2^p$ choices. Hence either
integrate out the $t_{jj}$'s over $-\infty<t_{jj}<\infty$ and a
unique $U$ or over $0<t_{jj}<\infty$ and the $U$ over the full
Stiefel manifold. For uniqueness of matrix factorization, see
Theorem~\ref{th:unique-for-rectangular}.

\begin{prop}\label{prop:x1t1u1}
Let $X_1$ be an $n\times p(n\geqslant p)$ matrix of rank $p$ of
independent real entries and let $X_1=U_1T_1$ where $T_1$ is a real
$p\times p$ upper triangular matrix with distinct nonzero diagonal
entries and $U_1$ is a unique real $n\times p$ semi-orthogonal
matrix, that is, $U^\t_1 U_1=\I_p$. Let $U=[U_1~~U_2]$ such that
$U^\t U=\I_n, U^\t_2 U_2=\I_{n-p}, U^\t_1 U_2=0$. Let $u_j$ be the
$j$-th column of $U$ and $\dif{u_j}$ its differential. Then
\begin{eqnarray}
\fbox{$[\dif{X_1}] =
\Pa{\prod^p_{j=1}\abs{t_{jj}}^{n-j}}[\dif{T_1}][\dif{U_1}],$}
\end{eqnarray}
where
$$
[\dif{U_1}] = \prod^p_{j=1}\prod^n_{i=j+1}\inner{u_i}{\dif{u_j}}.
$$
\end{prop}

\begin{prop}
If $X_1,T_1$ and $U_1$ are as defined in
Proposition~\ref{prop:x1t1u1}, then the surface area of the full
Stiefel manifold $\cO(p,n)$ or the total integral of the wedge
product $\wedge^p_{j=1}\wedge^n_{i=j+1}\inner{u_i}{\dif{u_j}}$ over
$\cO(p,n)$ is given by
\begin{framed}
\begin{eqnarray}
\int_{\cO(p,n)}\wedge^p_{j=1}\wedge^n_{i=j+1}\inner{u_i}{\dif{u_j}}
= \frac{2^p\pi^{\frac{pn}2}}{\Gamma_p\Pa{\frac n2}},
\end{eqnarray}
\end{framed}
\noindent where
$$
\Gamma_p(\alpha)=\pi^{\frac{p(p-1)}4}\Gamma(\alpha)\Gamma\Pa{\alpha-\frac12}\cdots\Pa{\alpha-\frac{p-1}2}
$$
for $\re(\alpha)>\frac{p-1}2$.
\end{prop}

\begin{proof}
Note that since $X_1$ is $n\times p$, the sum of squares of the $np$
variables in $X_1$ is given by
$$
\Tr{X^\t_1X_1} = \sum_{i=1}^n\sum^p_{j=1}x^2_{ij}.
$$
Then
\begin{eqnarray*}
\int_{X_1}[\dif{X_1}]e^{-\Tr{X^\t_1X_1}} &=&
\int^\infty_{-\infty}\cdots \int^\infty_{-\infty}
e^{-\sum_{i=1}^n\sum^p_{j=1}x^2_{ij}}\prod_{i=1}^n\prod^p_{j=1}\dif{x_{ij}} \\
&=&\prod_{i=1}^n\prod^p_{j=1}\int^\infty_{-\infty}e^{-x^2_{ij}}\dif{x_{ij}} \\
&=& \pi^{\frac{np}2}
\end{eqnarray*}
by direct evaluation of the exponential integrals. Make the
transformation as in Proposition~\ref{prop:x1t1u1}:
$$
X_1=U_1T_1\Longrightarrow X^\t_1X_1=T^\t_1T_1,
$$
where $T_1$ is a real $p\times p$ upper triangular matrix with
distinct nonzero diagonal entries and $U_1$ is a unique real
$n\times p$ semi-orthogonal matrix---$U^\t_1U_1=\I_p$, implying
$$
\Tr{X^\t_1X_1}=\Tr{T^\t_1T_1} = \sum_{i\leqslant j}t^2_{ij}.
$$
Note that $\dif{X_1}$ is available from
Proposition~\ref{prop:x1t1u1}. Now
\begin{eqnarray*}
\int_{X_1}[\dif{X_1}]e^{-\Tr{X^\t_1X_1}} =\Pa{
\int_{T_1}\prod^p_{j=1}\abs{t_{jj}}^{n-j}e^{-\sum_{i\leqslant
j}t^2_{ij}}[\dif{T_1}]}\Pa{\int_{\cO(p,n)}\wedge^p_{j=1}\wedge^n_{i=j+1}\inner{u_i}{\dif{u_j}}}.
\end{eqnarray*}
But for $0<t_{jj}<\infty, -\infty<t_{ij}<\infty(i<j)$ and $U_1$
unrestricted,
\begin{eqnarray*}
\int_{T_1}\prod^p_{j=1}\abs{t_{jj}}^{n-j}e^{-\sum_{i\leqslant
j}t^2_{ij}}[\dif{T_1}] = 2^{-p}\Gamma_p\Pa{\frac n2}
\end{eqnarray*}
observing that for $j=1,\ldots,p$, the $p$ integrals
\begin{eqnarray*}
\int_0^\infty\abs{t_{jj}}^{n-j}e^{-t^2_{jj}}\dif{t_{jj}} =
2^{-1}\Gamma\Pa{\frac n2 - \frac{j-1}2}, n>j-1,
\end{eqnarray*}
and each of the $p(p-1)/2$ integrals
$$
\int^\infty_{-\infty} e^{-t^2_{ij}}\dif{t_{ij}}=\sqrt{\pi}, i<j.
$$
Thus the result that follows.
\end{proof}

\begin{thrm}\label{th:vol-of-unitary-group}
Let $X$ be a full-ranked and $n\times n$ matrix of independent real
entries and let $X=UT$, where $T$ is a real $n\times n$ upper
triangular matrix with distinct nonzero diagonal entries and $U$ is
a unique real orthogonal matrix. Let $u_j$ be the $j$-th column of
$U$ and $\dif{u_j}$ its differential. Then the volume content of the
full orthogonal group $\cO(n)$ is given by
\begin{framed}
\begin{eqnarray}
\vol(\cO(n)) = \int_{\cO(n)} \wedge [U^\t\dif{U}] =
\frac{2^n\pi^{\frac{n^2}2}}{\Gamma_n\Pa{\frac n2}} =
\frac{2^n\pi^{\frac{n(n+1)}4}}{\prod^n_{k=1}\Gamma\Pa{\frac k2}}.
\end{eqnarray}
\end{framed}
\end{thrm}

\begin{prop}
Let $X$ be a $p\times n (p\leqslant n)$ matrix of rank $p$ and let
$X=TU^\t_1$, where $T$ is a $p\times p$ lower triangular matrix with
distinct positive diagonal entries $t_{jj}>0,j=1,\ldots, p$ and
$U_1$ is a unique $n\times p$ semi-orthogonal matrix, $U^\t_1
U_1=\I_p$, all are of independent real entries. Let $A=XX^\t=TT^\t$.
Then
\begin{eqnarray*}
[\dif X]=2^{-p}\det(A)^{\frac n2-\frac{p+1}2}[\dif
A]\wedge^p_{j=1}\wedge^n_{i=j+1}\inner{u_i}{\dif u_j}.
\end{eqnarray*}
\end{prop}

\begin{proof}
Since $A = TT^\t$, it follows that
$$
[\dif A]=2^p\Pa{\prod^p_{j=1}t^{p+1-j}_{jj}}[\dif T],
$$
i.e.
$$
[\dif T]=2^{-p}\Pa{\prod^p_{j=1}t^{-p-1+j}_{jj}}[\dif A].
$$
But
\begin{eqnarray*}
[\dif{X}] =
\Pa{\prod^p_{j=1}\abs{t_{jj}}^{n-j}}[\dif{T}][\dif{U_1}],
\end{eqnarray*}
where
$$
[\dif{U_1}] = \prod^p_{j=1}\prod^n_{i=j+1}\inner{u_i}{\dif{u_j}}.
$$
Note that $\det(A)=\det(T)^2=\prod^p_{j=1}t^2_{jj}$. The desired
conclusion is obtained.
\end{proof}

\begin{prop}\label{prop:polar-decom-1}
Let $X$ be a $m\times n$ matrix of rank $m(m\leqslant n)$,
$T=[t_{jk}]$ a $m\times m$ lower triangular matrix with $t_{jj}>0,
j=1,\ldots,m$ and $L$ a $n\times m$ matrix satisfying $L^\t L=\I_m$,
where the matrices are of independent real entries. Then show that,
if $X=TL^\t$, then
\begin{eqnarray}
\fbox{$[\dif X] = \Pa{\prod^m_{j=1}t^{n-j}_{jj}}[\dif T][\dif\hat
L]$}
\end{eqnarray}
where
$$
\dif\hat L = \prod^m_{j=1}\prod^n_{i=j+1}\inner{l_i}{\dif l_j},
$$
$l_j$ is the $j$-th column of $\hat L=[L~~L_1]\in\cO(n)$; $\dif l_i$
the differential of the $i$-th column of $L$.
\end{prop}

\begin{prop}[Polar decomposition]\label{prop:polar-decom-2}
Let $X$ be a $m\times n(m\leqslant n)$ matrix, $S$ a $m\times m$
symmetric positive definite matrix and $L$ a $n\times m$ matrix with
$L^\t L=\I_m$, all are of independent real entries. Then show that,
if $X=\sqrt{S}L^\t$, then
\begin{eqnarray}
[\dif X] = \Pa{\frac12}^m \Pa{\det(S)}^{\frac{n-m-1}2}[\dif S][\dif
\hat L]
\end{eqnarray}
where $\dif \hat L$ is defined in
Proposition~\ref{prop:polar-decom-1}.
\end{prop}

\begin{proof}
Now if $X=\sqrt{S}L^\t$ and $L^\t L=\I_m$, then $XX^\t=S$. By
Proposition~\ref{prop:polar-decom-1}, we have $X = TL^\t$, where
$T=[t_{jk}]$ is a $m\times m$ lower triangular matrix with
$t_{jj}>0, j=1,\ldots,m$ and $L$ a $n\times m$ matrix satisfying
$L^\t L=\I_m$. Denote $\hat L=[L~~L_1]\in\cO(n)$. Hence
$$
[\dif X] = \Pa{\prod^m_{j=1}t^{n-j}_{jj}}[\dif T][\dif \hat L].
$$
It also holds that $S=TT^\t$ implies
$$
[\dif S] = 2^m\Pa{\prod^m_{j=1}t^{m+1-j}_{jj}}[\dif T].
$$
Both expressions indicate that
$$
[\dif X] = 2^{-m}\Pa{\prod^m_{j=1}t^{n-m-1}_{jj}}[\dif S][\dif \hat
L].
$$
Since $\det(S)=\det(TT^\t)=\prod^m_{j=1}t^2_{jj}$, it follows that
\begin{eqnarray*}
[\dif X] = \Pa{\frac12}^m \Pa{\det(S)}^{\frac{n-m-1}2}[\dif S][\dif
\hat L].
\end{eqnarray*}
We are done.
\end{proof}

\begin{prop}\label{prop:polar-decom-3}
With the same notations as in Proposition~\ref{prop:polar-decom-2},
it holds that
\begin{framed}
\begin{eqnarray}
\int_{L^\t L=\I_m} [\dif\hat L] =
\frac{2^m\pi^{\frac{mn}2}}{\Gamma_m\Pa{\frac n2}}
\end{eqnarray}
\end{framed}
\noindent and for $m=n$
\begin{eqnarray}
\int_{\cO(n)} [\dif\hat V] =
\frac{2^n\pi^{\frac{n^2}2}}{\Gamma_n\Pa{\frac n2}}.
\end{eqnarray}
Define the normalized orthogonal measures as
\begin{eqnarray}
\dif\mu(V) := \Pa{\frac{\Gamma_n\Pa{\frac
n2}}{2^n\pi^{\frac{n^2}2}}}[\dif\hat V] =
\Pa{\frac{\Gamma_n\Pa{\frac n2}}{2^n\pi^{\frac{n^2}2}}}
\prod_{i>j}\inner{v_i}{\dif v_j}
\end{eqnarray}
or
\begin{eqnarray}
\dif\mu(V) := \Pa{\frac{\Gamma_n\Pa{\frac
n2}}{2^n\pi^{\frac{n^2}2}}}[\dif G],
\end{eqnarray}
where $\dif G = V^\t \dif V$ for $V=[v_1,\ldots, v_n]\in \cO(n)$. It
holds that $[\dif\hat V]=\wedge(V^\t\dif V)$ is invariant under
simultaneous translations $V\to UVW$, where $U,W\in\cO(n)$. That is,
$\dif\mu(V)$ is an invariant measure under both left and right
translations, i.e. Haar measure over $\cO(n)$.
\end{prop}

\begin{proof}
We know that
$$
\int_X [\dif X]e^{-\Tr{XX^\t}} = \pi^{\frac{mn}2}.
$$
From Proposition~\ref{prop:polar-decom-2}, via the transformation
$X=\sqrt{S}L^\t$, we see that
\begin{eqnarray*}
\int_X [\dif X]e^{-\Tr{XX^\t}} =
\Pa{\frac12}^m\int_{S>0}\det(S)^{\frac{n-m-1}2}e^{-\Tr{S}}[\dif
S]\times \int_{L^\t L=\I_m}[\dif \hat L].
\end{eqnarray*}
We also see from the definition of $\Gamma_p(\alpha)$ that
$$
\Gamma_m\Pa{\frac n2} = \int_{S>0}\det(S)^{\frac
n2-\frac{m+1}2}e^{-\Tr{S}}[\dif S].
$$
Then
\begin{eqnarray*}
\int_{L^\t L=\I_m}[\dif \hat L] =
\frac{2^m\pi^{\frac{mn}2}}{\Gamma_m\Pa{\frac n2}}.
\end{eqnarray*}
For $m=n$, the result follows easily.

For fixed $U,W\in\cO(n)$, we have
$$
(UVW)^\t\dif(UVW) = (W^\t V^\t U^\t)\Pa{U\cdot\dif V\cdot
W}=W^\t\cdot \dif G\cdot W,
$$
implying that
$$
\wedge\Br{(UVW)^\t\dif(UVW)} = [\dif G].
$$
That is, $\dif\mu(V)=\dif\mu(UVW)$ for all $U,V,W\in\cO(n)$,
$\dif\mu(V)$ is an invariant measure under both left and right
translations over $\cU(n)$.
\end{proof}

\section{Volumes of unitary groups}\label{sect:unitary-group}

\subsection{Preliminary}

In Section~\ref{sect:orthogonal-group}, we dealt with matrices where
the entries are either \emph{real} constants or \emph{real}
variables. Here we consider the matrices whose entries are complex
quantities. When the matrices are real, we will use the same
notations as in Section~\ref{sect:orthogonal-group}. In the complex
case, the matrix variable $X$ will be denoted by $\wtil{X}$ to
indicate that the entries in $X$ are complex variables so that the
entries of theorems in Section~\ref{sect:unitary-group} will not be
confused with those in Section~\ref{sect:orthogonal-group}. The
complex conjugate of a matrix $\wtil A$ will be denoted by
$\overline{\wtil A}$ and the conjugate transpose by $\wtil A^*$. The
determinant of $\wtil A$ will be denoted by $\det(\wtil A)$. The
absolute value of a scalar $a$ will also be denoted by $\abs{a}$.
The \emph{wedge product of differentials} in $\wtil{X}$ will be
denoted by $[\dif\wtil{X}]$ and the \emph{matrix of differentials}
by $\dif\wtil{X}$.

It is assumed that the reader is familiar with the basic properties
of real and complex matrices. Some properties of complex matrices
will be listed here for convenience.

A matrix $\wtil X$ with complex elements can always be written as
$\wtil X = X_1+\sqrt{-1}X_2$ where $X_1=\re(\wtil X)$ and
$X_2=\im(\wtil X)$ are real matrices. Let us examine the wedge
product of the differentials in $\wtil X$. In general, there are
$n^2$ real variables in $X_1$ and another $n^2$ real variables and
the wedge product of the
differentials will be denoted by the following:\\~\\
\textbf{Notation 1:}
\begin{eqnarray}
\fbox{$[\dif\wtil X]:=[\dif X_1][\dif X_2]$~~or~~$[\dif\wtil
X]:=\Br{\dif \Pa{\re(\wtil X)}}\Br{\dif \Pa{\im(\wtil X)}}$}
\end{eqnarray}
where $[\dif X_1]$ is the wedge product in $\dif X_1$ and $[\dif
X_2]$ is the wedge product in $\dif X_2$. In this notation an empty
product is interpreted as unity. That is, when the matrix $\wtil X$
is real then $X_2$ is null and $[\dif\wtil X]:=[\dif X_1]$. If
$\wtil X$ is a hermitian matrix, then $X_1$ is \emph{symmetric} and
$X_2$ is \emph{skew symmetric}, and in this case
\begin{eqnarray}
\fbox{$[\dif X_1]= \wedge_{j\geqslant k}\dif
x^{(1)}_{jk}~~~\text{and}~~~[\dif X_2]=\wedge_{j>k}\dif
x^{(2)}_{jk}$}
\end{eqnarray}
where $X_1=\Br{x^{(1)}_{jk}}$ and $X_2=\Br{x^{(2)}_{jk}}$. If $\wtil
Y$ is a scalar function of $\wtil X=X_1+\sqrt{-1}X_2$ then $\wtil Y$
can be written as $\wtil Y=Y_1+\sqrt{-1}Y_2$ where $Y_1$ and $Y_2$
are real. Thus if $\wtil Y=F(\wtil X)$ it is a transformation of
$(X_1,X_2)$ to $(Y_1,Y_2)$ or where $(Y_1,Y_2)$ is written as a
function of $(X_1,X_2)$ then we will use the following
notation for the Jacobian in the complex case.\\~\\
\textbf{Notation 2: (Jacobians in the complex case).}
$J(Y_1,Y_2:X_1,X_2)$: Jacobian of the transformation where $Y_1$ and
$Y_2$ are written as functions of $X_1$ and $X_2$ or where $\wtil
Y=Y_1+\sqrt{-1}Y_2$ is a function of $\wtil X=X_1+\sqrt{-1}X_2$.

\begin{lem}\label{lem:fact-on-det}
Consider a matrix $\wtil A\in\complex^{n\times n}$ and $(2n)\times
(2n)$ matrices $B$ and $C$ where
\begin{eqnarray*}
\wtil A = A_1+\sqrt{-1}A_2, B=\Br{\begin{array}{cc}
                              A_1 & A_2 \\
                              -A_2 & A_1
                            \end{array}
},C=\Br{\begin{array}{cc}
                              A_1 & -A_2 \\
                              A_2 & A_1
                            \end{array}
}
\end{eqnarray*}
where $A_1,A_2\in\real^{n\times n}$. Then for $\det(A_1)\neq0$
\begin{eqnarray}
\abs{\det(\wtil A)} = \abs{\det(B)}^{\frac12} =
\abs{\det(C)}^{\frac12}.
\end{eqnarray}
\end{lem}

\begin{proof}
Let $\det(A)=a+\sqrt{-1}b$ where $a$ and $b$ are real scalars. Then
the absolute value is available as
$\sqrt{(a+\sqrt{-1}b)(a-\sqrt{-1}b)}$. If
$\det(A_1+\sqrt{-1}A_2)=a+\sqrt{-1}b$, then
$\det(A_1-\sqrt{-1}A_2)=a-\sqrt{-1}b$. Hence
\begin{eqnarray*}
&&(a+\sqrt{-1}b)(a-\sqrt{-1}b) =
\det(A_1+\sqrt{-1}A_2)\det(A_1-\sqrt{-1}A_2) \\
&&= \det\Pa{\Br{\begin{array}{cc}
                                                             A_1+\sqrt{-1}A_2 & 0 \\
                                                             0 & A_1-\sqrt{-1}A_2
                                                           \end{array}
}}.
\end{eqnarray*}
Adding the last $n$ columns to the first $n$ columns and then adding
the last $n$ rows to the first $n$ rows we have
\begin{eqnarray*}
\det\Pa{\Br{\begin{array}{cc}
                                                             A_1+\sqrt{-1}A_2 & 0 \\
                                                             0 & A_1-\sqrt{-1}A_2
                                                           \end{array}
}}=\det\Pa{\Br{\begin{array}{cc}
                                                             2A_1 & A_1-\sqrt{-1}A_2 \\
                                                             A_1-\sqrt{-1}A_2 & A_1-\sqrt{-1}A_2
                                                           \end{array}
}}.
\end{eqnarray*}
Using similar steps we have
\begin{eqnarray*}
\det\Pa{\Br{\begin{array}{cc}
                                                             2A_1 & A_1-\sqrt{-1}A_2 \\
                                                             A_1-\sqrt{-1}A_2 & A_1-\sqrt{-1}A_2
                                                           \end{array}
}}&=&\det\Pa{\Br{\begin{array}{cc}
                                                             2A_1 & A_1-\sqrt{-1}A_2 \\
                                                            -\sqrt{-1}A_2 & \frac12A_1-\frac12\sqrt{-1}A_2
                                                           \end{array}
}}\\
&=&\det\Pa{\Br{\begin{array}{cc}
                                                             2A_1 & -\sqrt{-1}A_2 \\
                                                            -\sqrt{-1}A_2 & \frac12A_1
                                                           \end{array}
}}\\
&=&\det(A_1)\det(A_1+A_2A^{-1}_1A_2) \\
&=& \det\Pa{\Br{\begin{array}{cc}
                              A_1 & A_2 \\
                              -A_2 & A_1
                            \end{array}
}} = \det\Pa{\Br{\begin{array}{cc}
                              A_1 & -A_2 \\
                              A_2 & A_1
                            \end{array}
}}.
\end{eqnarray*}
by evaluating as the determinant of partitioned matrices. Thus the
absolute value of $\det(\wtil A)$ is given by
\begin{eqnarray*}
\abs{\det(\wtil A)} = \sqrt{\det(A_1)\det(A_1+A_2A^{-1}_1A_2)} =
\abs{\det(B)}^{\frac12} = \abs{\det(C)}^{\frac12}.
\end{eqnarray*}
This establishes the result.
\end{proof}

\begin{remark}
Now we denote $A_1=\re(\wtil A)$ and $A_2=\im(\wtil A)$. Clearly
both $\re(\wtil A)$ and $\im(\wtil A)$ are real matrices. Each
complex matrix $\wtil A:=\re(\wtil A)+\sqrt{-1}\im(\wtil A)$ can be
represented faithfully as a block-matrix
\begin{eqnarray}
\wtil A\longrightarrow\Br{\begin{array}{cc}
                            \re(\wtil A) & -\im(\wtil A) \\
                            \im(\wtil A) & \re(\wtil A)
                          \end{array}
}.
\end{eqnarray}
Thus
\begin{eqnarray}
\wtil A^*\longrightarrow\Br{\begin{array}{cc}
                            \re(\wtil A)^\t & \im(\wtil A)^\t \\
                            -\im(\wtil A)^\t & \re(\wtil A)^\t
                          \end{array}
}.
\end{eqnarray}
Then $\wtil Y = \wtil A\wtil X$ can be rewritten as, via
block-matrix technique,
\begin{eqnarray}\label{eq:complex-to-real-1}
\Br{\begin{array}{cc}
          \re(\wtil Y) & -\im(\wtil Y) \\
          \im(\wtil Y) & \re(\wtil Y)
        \end{array}
} = \Br{\begin{array}{cc}
          \re(\wtil A) & -\im(\wtil A) \\
          \im(\wtil A) & \re(\wtil A)
        \end{array}
}\Br{\begin{array}{cc}
          \re(\wtil X) & -\im(\wtil X) \\
          \im(\wtil X) & \re(\wtil X)
        \end{array}
}.
\end{eqnarray}
From the above, we see that the mentioned representation is an
injective ring homomorphism which is continuous. sometimes we use
the following representation:
\begin{eqnarray}\label{eq:complex-to-real-2}
\Br{\begin{array}{c}
                                 \re(\wtil Y) \\
                                 \im(\wtil Y)
                               \end{array}
} = \Br{\begin{array}{cc}
          \re(\wtil A) & -\im(\wtil A) \\
          \im(\wtil A) & \re(\wtil A)
        \end{array}
}\Br{\begin{array}{c}
                                 \re(\wtil X) \\
                                 \im(\wtil X)
                               \end{array}
}.
\end{eqnarray}
Lemma~\ref{lem:fact-on-det} can be reexpressed as
\begin{eqnarray}
\abs{\det\Pa{\Br{\begin{array}{cc}
          \re(\wtil A) & -\im(\wtil A) \\
          \im(\wtil A) & \re(\wtil A)
        \end{array}
}}} = \abs{\det(\wtil A)}^2 = \abs{\det(\wtil A\wtil A^*)}.
\end{eqnarray}
\end{remark}

\begin{prop}
Let $\wtil X,\wtil Y\in\complex^n$ be of $n$ independent complex
variables each, $\wtil A\in\complex^{n\times n}$ a nonsingular
matrix of constants. If $\wtil Y=\wtil A\wtil X$, then
\begin{eqnarray}
[\dif \wtil Y] = \abs{\det(\wtil A\wtil A^*)}[\dif\wtil X].
\end{eqnarray}
If $\wtil Y^*=\wtil X^*\wtil A^*$, then
\begin{eqnarray}
[\dif \wtil Y^*] = (-1)^n\abs{\det(\wtil A\wtil A^*)}[\dif\wtil X].
\end{eqnarray}
\end{prop}

\begin{proof}
Let $\wtil X = X_1+\sqrt{-1}X_2$, where $X_m\in\real^n,m=1,2$. Let
$\wtil Y=Y_1+\sqrt{-1}Y_2$, where $Y_m\in\real^n,m=1,2$ are real.
$\wtil Y=\wtil A\wtil X$ implies that $Y_m=AX_m,m=1,2$ if $\wtil
A=A$ is real. This transformation is such that the $2n$ real
variables in $(Y_1,Y_2)$ are written as functions of the $2n$ real
variables in $(X_1,X_2)$. Let
\begin{eqnarray*}
X^\t_1 = [x_{11},\ldots,x_{n1}],\quad
X^\t_2=[x_{12},\ldots,x_{n2}],\\
Y^\t_1 = [y_{11},\ldots,y_{n1}],\quad Y^\t_2=[y_{12},\ldots,y_{n2}].
\end{eqnarray*}
Then the Jacobian is the determinant of the following matrix of
partial derivatives:
\begin{eqnarray*}
\frac{\partial (Y_1,Y_2)}{\partial (X_1,X_2)} = \frac{\partial
(y_{11},\ldots,y_{n1},y_{12},\ldots,y_{n2})}{\partial
(x_{11},\ldots,x_{n1},x_{12},\ldots,x_{n2})}.
\end{eqnarray*}
Note that
\begin{eqnarray*}
\frac{\partial (Y_1,Y_2)}{\partial (X_1,X_2)} = \frac{\partial
(y_{11},\ldots,y_{n1},y_{12},\ldots,y_{n2})}{\partial
(x_{11},\ldots,x_{n1},x_{12},\ldots,x_{n2})}.
\end{eqnarray*}
Note that
\begin{eqnarray*}
\frac{\partial Y_1}{\partial X_1} &=& \frac{\partial
(y_{11},\ldots,y_{n1})}{\partial (x_{11},\ldots,x_{n1})}=A,\\
\frac{\partial Y_1}{\partial X_2} &=& 0 = \frac{\partial
Y_2}{\partial X_1},\\
\frac{\partial Y_2}{\partial X_2} &=&A.
\end{eqnarray*}
Thus the Jacobian is
\begin{eqnarray*}
J = \det\Pa{\Br{\begin{array}{cc}
              A & 0 \\
              0 & A
            \end{array}
}}=\det(A)^2.
\end{eqnarray*}
If $\wtil A$ is complex, then let $\wtil A=A_1+\sqrt{-1}A_2$ where
$A_1$ and $A_2$ are real. Then
\begin{eqnarray*}
\wtil Y &=& Y_1+\sqrt{-1}Y_2 = (A_1+\sqrt{-1}A_2)(X_1+\sqrt{-1}X_2) \\
&=& (A_1X_1 -A_2X_2) + \sqrt{-1}(A_1X_2+A_2X_1)
\end{eqnarray*}
implies that
\begin{eqnarray*}
Y_1=A_1X_1 -A_2X_2,\quad Y_2=A_1X_2+A_2X_1.
\end{eqnarray*}
Then
\begin{eqnarray*}
\frac{\partial Y_1}{\partial X_1} &=& A_1,~~~\frac{\partial
Y_1}{\partial X_2} = -A_2,\\
\frac{\partial Y_2}{\partial X_1} &=& A_2,~~~\frac{\partial
Y_2}{\partial X_2} = A_1.
\end{eqnarray*}
Thus the Jacobian is
\begin{eqnarray*}
J  =  \det\Pa{\Br{\begin{array}{cc}
                A_1 & -A_2 \\
                A_2 & A_1
              \end{array}
}} = \abs{\det(\wtil A)}^2 = \abs{\det(\wtil A\wtil A^*)},
\end{eqnarray*}
which establishes the result. The second result follows by noting
that $[\dif\wtil Y^*]=[\dif\wtil Y_1](-1)^n[\dif\wtil Y_2] =
(-1)^n[\dif\wtil Y]$.
\end{proof}

\begin{prop}
Let $\wtil X,\wtil Y\in\complex^{m\times n}$ of $mn$ independent
complex variables each. Let $\wtil A\in\complex^{m\times m}$ and
$\wtil B\in\complex^{n\times n}$ nonsingular matrices of constants.
If $\wtil Y=\wtil A\wtil X\wtil B$, then
\begin{eqnarray}
[\dif \wtil Y] = \abs{\det(\wtil A\wtil A^*)}^n\abs{\det(\wtil
B\wtil B^*)}^m[\dif \wtil X].
\end{eqnarray}
\end{prop}

\begin{proof}
Let $\wtil Y=Y_1+\sqrt{-1}Y_2$ and $\wtil X=X_1+\sqrt{-1}X_2$.
Indeed, let $\wtil Y=[\wtil Y_1,\ldots,\wtil Y_n]$ and $\wtil
X=[\wtil X_1,\ldots,\wtil X_n]$, then $\wtil Y_j = \wtil A\wtil
X_j,j=1,\ldots,n$ when $\wtil Y=\wtil A\wtil X$. Thus $[\dif \wtil
Y_j]=\abs{\det(\wtil A\wtil A^*)}[\dif\wtil X_j]$ for each $j$,
therefore, ignoring the signs,
\begin{eqnarray*}
[\dif\wtil Y] = \prod_{j=1}^n [\dif\wtil Y_j] =\abs{\det(\wtil
A\wtil A^*)}^n [\dif\wtil X].
\end{eqnarray*}
Denoting $\wtil A=A_1+\sqrt{-1}A_2$, the determinant is
\begin{eqnarray*}
J = \det\Pa{\Br{\begin{array}{cc}
              A_1 & -A_2 \\
              A_2 & A_1
            \end{array}
}}^n.
\end{eqnarray*}
Hence the Jacobian in this case, denoting $\wtil
B=B_1+\sqrt{-1}B_2$, is given by
\begin{eqnarray*}
J = \det\Pa{\Br{\begin{array}{cc}
              B_1 & B_2 \\
              -B_2 & B_1
            \end{array}
}}^m.
\end{eqnarray*}
For establishing our result, write $\wtil Y=\wtil A\wtil Z$ where
$\wtil Z=\wtil X\wtil B$. That is,
\begin{eqnarray*}
[\dif\wtil Y]=\abs{\det(\wtil A\wtil A^*)}^n[\dif\wtil Z] =
\abs{\det(\wtil A\wtil A^*)}^n\abs{\det(\wtil B\wtil
B^*)}^m[\dif\wtil X].
\end{eqnarray*}
This completes the proof.
\end{proof}

\begin{remark}
Another approach to the fact that $[\dif\wtil Y]=\abs{\det(\wtil
A\wtil A^*)}^n[\dif\wtil Z]$, where $\wtil Y=\wtil A\wtil Z$, is
described as follows:
$$
\begin{cases}\re(\wtil Y) &= \re(\wtil A)\re(\wtil Z) - \im(\wtil A)\im(\wtil
Z)\\\im(\wtil Y) &= \im(\wtil A)\re(\wtil Z) + \re(\wtil A)\im(\wtil
Z)
\end{cases}
$$
leading to
$$
\frac{\partial \Pa{\re(\wtil Y),\im(\wtil Y)}}{\partial
\Pa{\re(\wtil Z),\im(\wtil Z)}} = \Br{\begin{array}{cc}
                                        \re(\wtil A)^{(n)} & -\im(\wtil A)^{(n)} \\
                                        \im(\wtil A)^{(n)} & \re(\wtil A)^{(n)}
                                      \end{array}
},
$$
where
$$
\re(\wtil A)^{(n)}:=\Br{\begin{array}{ccc}
                                           \re(\wtil A) &  &  \\
                                            & \ddots &  \\
                                            &  & \re(\wtil A)
                                         \end{array}},~~\im(\wtil
                                         A)^{(n)}:=\Br{\begin{array}{ccc}
                                           \im(\wtil A) &  &  \\
                                            & \ddots &  \\
                                            &  & \im(\wtil A)
                                         \end{array}}.
$$
Then the Jacobian of this transformation can be computed as
\begin{eqnarray}
J\Pa{\re(\wtil Y),\im(\wtil Y):\re(\wtil Z),\im(\wtil Z)} &=&
\det\Pa{\Br{\begin{array}{cc}
                                        \re(\wtil A)^{(n)} & -\im(\wtil A)^{(n)} \\
                                        \im(\wtil A)^{(n)} & \re(\wtil A)^{(n)}
                                      \end{array}
}}\\
&=&\abs{\det\Pa{\re(\wtil A)^{(n)}+\sqrt{-1}\im(\wtil A)^{(n)}}}^2.
\end{eqnarray}
That is,
$$
J\Pa{\re(\wtil Y),\im(\wtil Y):\re(\wtil Z),\im(\wtil Z)} =
\abs{\det\Pa{\wtil A^{(n)}}}^2 =\abs{\det(\wtil A)}^{2n}=
\abs{\det(\wtil A\wtil A^*)}^n.
$$
\end{remark}

\begin{prop}
Let $\wtil X, \wtil A, \wtil B\in\complex^{n\times n}$ be lower
triangular matrices where $\wtil X$ is matrix of $\frac{n(n+1)}2$
independent complex variables, $\wtil A,\wtil B$ are nonsingular
matrices of constants. Then
\begin{eqnarray}
\wtil Y = \wtil X+\wtil X^\t&\Longrightarrow&
[\dif\wtil Y]=2^{2n}[\dif\wtil X],\label{eq:3aaa}\\
&\Longrightarrow& [\dif\wtil Y]=2^n[\dif\wtil
X]~\text{if the}~\wtil x_{jj}\text{'s are real};\label{eq:3bbb}\\
\wtil Y = \wtil A\wtil X&\Longrightarrow&[\dif\wtil
Y]=\Pa{\prod^n_{j=1}\abs{\wtil a_{jj}}^{2j}}[\dif\wtil X],\label{eq:3ccc}\\
&\Longrightarrow&[\dif\wtil Y]=\Pa{\prod^n_{j=1}\abs{\wtil
a_{jj}}^{2j-1}}[\dif\wtil
X]~\text{if the}~\wtil a_{jj}\text{'s and}~\wtil x_{jj}\text{'s are real};\label{eq:3ddd}\\
\wtil Y=\wtil X\wtil B&\Longrightarrow&[\dif\wtil
Y]=\Pa{\prod^n_{j=1}\abs{\wtil b_{jj}}^{2(n-j+1)}}[\dif\wtil
X],\label{eq:3eee}\\
&\Longrightarrow& [\dif\wtil Y]=\Pa{\prod^n_{j=1}\abs{\wtil
b_{jj}}^{2(n-j)+1}}[\dif\wtil X]~\text{if the}~\wtil b_{jj}\text{'s
and}~\wtil x_{jj}\text{'s are real};\label{eq:3fff}
\end{eqnarray}
\end{prop}

\begin{proof}
Results \eqref{eq:3aaa} and \eqref{eq:3bbb} are trivial. Indeed,
note that
\begin{eqnarray*}
\wtil y_{jk} =\begin{cases}2\wtil x_{jj},&\text{if}~j=k\\
\wtil x_{jk},&\text{if}~j> k\end{cases}.
\end{eqnarray*}
By the definition, ignoring the sign, we have
\begin{eqnarray*}
[\dif\wtil Y]=\wedge_{j\geqslant k}\dif\wtil y_{jk} =
\wedge_{j=1}^n\dif\wtil y_{jj}\wedge_{j>k}\dif\wtil y_{jk},
\end{eqnarray*}
where $\dif\wtil y_{jk}:=\dif y^{(1)}_{jk}\dif y^{(2)}_{jk}$ for
$\wtil y_{jk}=y^{(1)}_{jk}+\sqrt{-1}y^{(2)}_{jk}$. So for
$j=1,\ldots,n$, we get $y^{(m)}_{jj}=2x^{(m)}_{jj},m=1,2$. Hence the
result. If $\wtil x_{jj}$'s are real, the result follows easily by
definition. Let
\begin{eqnarray*}
&&\wtil Y=Y_1+\sqrt{-1}Y_2, \wtil X=X_1+\sqrt{-1}X_2,
\wtil A=A_1+\sqrt{-1}A_2, \wtil B=B_1+\sqrt{-1}B_2,\\
&&Y_m=[y^{(m)}_{jk}],X_m=[x^{(m)}_{jk}],A_m=[a^{(m)}_{jk}],B_m=[b^{(m)}_{jk}],
m=1,2.
\end{eqnarray*}
where $Y_m,X_m,A_m,B_m,m=1,2$ are all real.

When $\wtil Y= \wtil A\wtil X$ we have $Y_1=A_1X_1-A_2X_2$ and
$Y_2=A_1X_2+A_2X_1$. The matrix of partial derivative of $Y_1$ with
respect to $X_1$, that is $\frac{\partial Y_1}{\partial X_1}$, can
be seen to be a lower triangular matrix with $a^{(1)}_{jj}$ repeated
$j$ times, $j=1,\ldots,n$, on the diagonal. Let this matrix be
denoted by $G_1$. That is,
$$
\frac{\partial Y_1}{\partial X_1}=\frac{\partial Y_2}{\partial
X_2}:=G_1 = \Br{\begin{array}{cccc}
             A_1 &  &  &  \\
              & A_1[\hat 1|\hat 1] &  &  \\
              &  & \ddots &  \\
              &  &  & A_1[\hat1\cdots\widehat{n-1}|\hat1\cdots\widehat{n-1}]
           \end{array}
}.
$$
Let $G_2$ be a matrix of the same structure with $a^{(2)}_{jj}$'s on
the diagonal. Similarly,
$$
-\frac{\partial Y_1}{\partial X_2}=\frac{\partial Y_2}{\partial
X_1}:=G_2 = \Br{\begin{array}{cccc}
             A_2 &  &  &  \\
              & A_2[\hat 1|\hat 1] &  &  \\
              &  & \ddots &  \\
              &  &  & A_2[\hat1\cdots\widehat{n-1}|\hat1\cdots\widehat{n-1}]
           \end{array}
}.
$$
Then the Jacobian matrix is given by
\begin{eqnarray*}
\frac{\partial (Y_1,Y_2)}{\partial(X_1,X_2)}=\Br{\begin{array}{cc}
                                                   G_1 & -G_2 \\
                                                   G_2 & G_1
                                                 \end{array}
}.
\end{eqnarray*}
Let $\wtil G=G_1+\sqrt{-1}G_2$. Then
$$
\wtil G=\Br{\begin{array}{cccc}
             \wtil A &  &  &  \\
              & \wtil A[\hat 1|\hat 1] &  &  \\
              &  & \ddots &  \\
              &  &  & \wtil A[\hat1\cdots\widehat{n-1}|\hat1\cdots\widehat{n-1}]
           \end{array}
},
$$
where
\begin{eqnarray*}
\wtil A&=&A_1+\sqrt{-1}A_2,\wtil A[\hat 1|\hat 1]=A_1[\hat 1|\hat
1]+\sqrt{-1}A_2[\hat 1|\hat 1],\ldots, \\
\wtil A[\hat1\cdots\widehat{n-1}|\hat1\cdots\widehat{n-1}] &=&
A_1[\hat1\cdots\widehat{n-1}|\hat1\cdots\widehat{n-1}]+\sqrt{-1}
A_2[\hat1\cdots\widehat{n-1}|\hat1\cdots\widehat{n-1}].
\end{eqnarray*}
Thus
\begin{eqnarray*}
\det\Pa{\frac{\partial
(Y_1,Y_2)}{\partial(X_1,X_2)}}=\det\Pa{\Br{\begin{array}{cc}
                                                   G_1 & -G_2 \\
                                                   G_2 & G_1
                                                 \end{array}
}}.
\end{eqnarray*}
From Lemma~\ref{lem:fact-on-det}, the determinant is available as
$\abs{\det(\wtil G)}^2$ where $\wtil G=G_1+\sqrt{-1}G_2$. Since
$\wtil G$ is triangular the absolute value of the determinant is
given by
\begin{eqnarray*}
\abs{\det(\wtil G)}^2 &=& \abs{\det(\wtil A)}^2\abs{\det(\wtil
A[\hat 1|\hat 1])}^2\cdots\abs{\det(\wtil A[\hat1\cdots
\widehat{n-1}|\hat1\cdots \widehat{n-1}])}^2\\
&=&\prod^n_{j=1}(\abs{a_{jj}}^2)^j.
\end{eqnarray*}
This establishes \eqref{eq:3ccc}. Another approach is presented
also: Let $\wtil Y=[\wtil Y_1,\ldots,\wtil Y_n]$, where $\wtil
Y_j,j=1,\ldots,n$, is the $j$-th column of the matrix $\wtil Y$.
Similarly for $\wtil X=[\wtil X_1,\ldots,\wtil X_n]$. Now $\wtil
Y=\wtil A\wtil X$ implies that
\begin{eqnarray*}
\wtil Y_j = \wtil A\wtil X_j,~~j=1,\ldots,n.
\end{eqnarray*}
That is,
\begin{eqnarray*}
\Br{\begin{array}{c}
      \wtil y_{11} \\
      \wtil y_{21} \\
      \vdots \\
      \wtil y_{n1}
    \end{array}
} = \wtil A\Br{\begin{array}{c}
      \wtil x_{11} \\
      \wtil x_{21} \\
      \vdots \\
      \wtil x_{n1}
    \end{array}
},\Br{\begin{array}{c}
       0 \\
      \wtil y_{22} \\
      \vdots \\
      \wtil y_{n2}
    \end{array}
} = \wtil A\Br{\begin{array}{c}
       0 \\
      \wtil x_{22} \\
      \vdots \\
      \wtil x_{n2}
    \end{array}
},\ldots,\Br{\begin{array}{c}
      0 \\
      0 \\
      \vdots \\
      \wtil y_{nn}
    \end{array}
} = \wtil A\Br{\begin{array}{c}
       0 \\
       0 \\
      \vdots \\
      \wtil x_{nn}
    \end{array}
}.
\end{eqnarray*}
Since $\wtil Y,\wtil X,\wtil A$ are all lower triangular, it follows
that
\begin{eqnarray*}
\Br{\begin{array}{c}
      \wtil y_{11} \\
      \wtil y_{21} \\
      \vdots \\
      \wtil y_{n1}
    \end{array}
} &=& \wtil A\Br{\begin{array}{c}
      \wtil x_{11} \\
      \wtil x_{21} \\
      \vdots \\
      \wtil x_{n1}
    \end{array}
},\Br{\begin{array}{c}
      \wtil y_{22} \\
      \vdots \\
      \wtil y_{n2}
    \end{array}
} = \wtil A[\hat1|\hat1]\Br{\begin{array}{c}
      \wtil x_{22} \\
      \vdots \\
      \wtil x_{n2}
    \end{array}
},\ldots,\\
      \wtil y_{nn}&=& \wtil A[\hat1\ldots \widehat{n-1}|\widehat1\ldots \widehat{n-1}]
      \wtil x_{nn},
\end{eqnarray*}
where $A[\hat i \hat j|\hat i\hat j]$ stands for a matrix obtained
from deleting the $i,j$-th rows and columns of $\wtil A$,
respectively. We can now draw the conclusion that
\begin{eqnarray*}
[\dif\wtil Y_j] = \abs{\det(A[\hat 1\ldots \widehat {j-1}]|\hat
1\ldots \widehat {j-1}])A[\hat 1\ldots \widehat {j-1}]|\hat 1\ldots
\widehat {j-1}])^*}[\dif\wtil X_j],
\end{eqnarray*}
that indicates that
\begin{eqnarray*}
[\dif\wtil Y] &=& \prod^n_{j=1}[\dif\wtil Y_j] =
\prod^n_{j=1}\abs{\det(A[\hat 1\ldots \widehat {j-1}]|\hat 1\ldots
\widehat {j-1}])A[\hat 1\ldots \widehat {j-1}]|\hat 1\ldots \widehat
{j-1}])^*}[\dif\wtil X_j]\\
&=& \abs{\det(\wtil A\wtil A^*)}\abs{\det(\wtil A[\hat1|\hat1]\wtil
A[\hat1|\hat1]^*)}\cdots\abs{\wtil a_{nn}\wtil a_{nn}^*}[\dif\wtil
X]\\
&=& \abs{\wtil a_{11}\wtil a_{22}\cdots \wtil
a_{nn}}^2\times\abs{\wtil a_{22}\wtil a_{33}\cdots
\wtil a_{nn}}^2\times \cdots\times\abs{\wtil a_{nn}}^2[\dif\wtil X]\\
&=& \Pa{\prod^n_{j=1}\abs{\wtil a_{jj}}^{2j}}[\dif\wtil X].
\end{eqnarray*}
If the $\wtil x_{jj}$'s and $\wtil a_{jj}$'s are real then note that
the $x_{jk}$'s for $j>k$ contribute $\wtil a_{jj}$ twice that is,
corresponding to $x^{(1)}_{jk}$ and $x^{(2)}_{jk}$, whereas the
$\wtil a_{jj}$'s appear only once corresponding to the
$x^{(1)}_{jj}$'s since the $x^{(2)}_{jj}$'s are zeros. This
establishes \eqref{eq:3ddd}. If $\wtil Y=\wtil X\wtil B$ and if a
matrix $H_1$ is defined corresponding to $G_1$ then note that the
$b^{(1)}_{jj}$'s appear $n-j+1$ times on the diagonal for
$j=1,\ldots,n$. Results \eqref{eq:3eee} and \eqref{eq:3fff} are
established by using similar steps as in the case of \eqref{eq:3ccc}
and \eqref{eq:3ddd}.
\end{proof}

\begin{prop}\label{prop:adjoint-action}
Let $\wtil X\in\complex^{n\times n}$ be hermitian matrix of
independent complex entries and $\wtil A\in\complex^{n\times n}$ be
a nonsingular matrix of constants. If $\wtil Y=\wtil A\wtil X \wtil
A^*$, then
\begin{eqnarray}
\fbox{$[\dif\wtil Y] =\abs{\det(\wtil A\wtil A^*)}^n[\dif\wtil X]$.}
\end{eqnarray}
\end{prop}

\begin{proof}
Since $\wtil A$ is nonsingular it can be written as a product of
elementary matrices. Let $\wtil E_1,\ldots,\wtil E_k$ be elementary
matrices such that
\begin{eqnarray*}
\wtil A=\wtil E_k \wtil E_{k-1}\cdots \wtil E_1\Longrightarrow \wtil
A^*=\wtil E^*_1 \wtil E^*_2\cdots \wtil E^*_k.
\end{eqnarray*}
For example, let $\wtil E_1$ be such that the $j$-th row of an
identity matrix is multiplied by a scalar $\wtil c=a+\sqrt{-1}b$
where $a,b\in\real$. Then $\wtil E_1\wtil X \wtil E^*_1$ means that
the $j$-th row of $\wtil X$ is multiplied by $a+\sqrt{-1}b$ and the
$j$-th column of $\wtil X$ is multiplied by $a-\sqrt{-1}b$. Let
$$
\wtil U_1=\wtil E_1\wtil X\wtil E^*_1,~~\wtil U_2 = \wtil E_2\wtil
U_1\wtil E_2^*,~~\ldots,~~\wtil U_k = \wtil E_k\wtil U_{k-1}\wtil
E_k^*.
$$
Then the Jacobian of $\wtil Y$ written as a function $\wtil X$ is
given by
\begin{eqnarray*}
J(\wtil Y:\wtil X) = J(\wtil Y:\wtil U_{k-1})\cdots J(\wtil
U_1:\wtil X).
\end{eqnarray*}
Let us evaluate $[\dif\wtil U_1]$ in terms of $[\dif\wtil X]$ by
direct computation. Since $\wtil X$ is hermitian its diagonal
elements are real and the elements above the leading diagonal are
the complex conjugates of those below the leading diagonal, and
$\wtil U_1$ is also of the same structure as $\wtil X$. Let $\wtil
U_1=U+\sqrt{-1}V$ and $\wtil X=Z+\sqrt{-1}W$ where
$U=[u_{jk}],V=[v_{jk}],Z=[z_{jk}],W=[w_{jk}]$ are all real and the
diagonal elements of $V$ and $W$ are zeros. Take the $u_{jj}$'s and
$z_{jj}$'s separately. The matrix of partial derivatives of
$u_{11},\ldots,u_{nn}$ with respect to $z_{11},\ldots,z_{nn}$ is a
diagonal matrix with the $j$-th element $a^2+b^2$ and all other
elements unities. That is,
$$
\frac{\partial (\diag(U))}{\partial (\diag(Z))}=\frac{\partial
(u_{11},\ldots,u_{nn})}{\partial
(z_{11},\ldots,z_{nn})}=\Br{\begin{array}{ccccc}
                              1 &  &  &  &  \\
                               & \ddots &  &  &  \\
                               &  & a^2+b^2=\abs{\wtil c}^2 &  &  \\
                               &  &  & \ddots &  \\
                               &  &  &  & 1
                            \end{array}
} :=C,
$$
where $\diag(X)$ means the diagonal matrix, obtained by keeping the
diagonal entries of $X$ and ignoring the off-diagonal entries.

The remaining variables produce a $\frac{n(n-1)}2\times
\frac{n(n-1)}2$ matrix of the following type
\begin{eqnarray*}
\frac{\partial (U_0,V_0)}{\partial (Z_0,W_0)} =
\Br{\begin{array}{cc}
                  A_0 & B_0 \\
                  -B_0 & A_0
                \end{array}
}
\end{eqnarray*}
where $U_0,V_0,Z_0,W_0$ mean that the diagonal elements are deleted,
$A_0$ is a diagonal matrix with $n-1$ of the diagonal elements equal
to $a$ and the remaining unities and $B_0$ is a diagonal matrix such
that corresponding to every $a$ in $A_0$ there is a $b$ or $-b$ with
$j-1$ of them equal to $-b$ and $n-j$ of them equal to $b$. Thus the
Jacobian of this transformation is:
\begin{eqnarray*}
J(\wtil U_1:\wtil X) &=& \frac{\partial (\diag(U),
U_0,V_0)}{\partial (\diag(Z),Z_0,W_0)} =
\det\Pa{\Br{\begin{array}{ccc}
                            C & 0 & 0 \\
                            0 & A_0 & B_0 \\
                            0 & -B_0 & A_0
                          \end{array}
}} \\
&=& \det(C)\det\Pa{\Br{\begin{array}{cc}
                  A_0 & B_0 \\
                  -B_0 & A_0
                \end{array}
}}
\end{eqnarray*}
From Lemma~\ref{lem:fact-on-det}, the determinant is
$\abs{\det(A_0+\sqrt{-1}B_0)}^2$. That is,
\begin{eqnarray*}
&&\det\Pa{\frac{\partial (U_0,V_0)}{\partial (Z_0,W_0)}} =
\det\Pa{\Br{\begin{array}{cc}
                  A_0 & B_0 \\
                  -B_0 & A_0
                \end{array}
}} \\
&&=\abs{\det((A_0+\sqrt{-1}B_0)(A_0+\sqrt{-1}B_0)^*)} \\
&&= (a^2+b^2)^{n-1}.
\end{eqnarray*}
Thus
\begin{eqnarray*}
[\dif\wtil U_1] = (a^2+b^2)^n[\dif\wtil X] =\abs{\det(\wtil E_1\wtil
E^*_1)}^n [\dif\wtil X].
\end{eqnarray*}
Note that interchanges of rows and columns can produce only a change
in the sign in the determinant, the addition of a row (column) to
another row (column) does not change the determinant and elementary
matrices of the type $\wtil E_1$ will produce $\abs{\det(\wtil
E_1\wtil E^*_1)}^n$ in the Jacobian. Thus by computing $J(\wtil
U_1:\wtil X),J(\wtil U_2:\wtil U_1)$ etc we have
\begin{eqnarray*}
[\dif\wtil Y] = \abs{\det(\wtil A\wtil A^*)}^n [\dif\wtil X].
\end{eqnarray*}

As a specific example, the configuration of the partial derivatives
for $n=3$ with $j=2$ is the following: If
\begin{eqnarray*}
\wtil X = \Br{\begin{array}{ccc}
                     z_{11} & * & * \\
                     z_{21}+\sqrt{-1}w_{21} & z_{22} & * \\
                     z_{31}+\sqrt{-1}w_{31} & z_{32}+\sqrt{-1}w_{32} & z_{33}
                   \end{array}
}~\text{and}~\wtil E_1=\Br{\begin{array}{ccc}
                       1 & 0 & 0 \\
                       0 & \wtil c & 0 \\
                       0 & 0 & 1
                     \end{array}
},
\end{eqnarray*}
then
\begin{eqnarray*}
\wtil U_1 &=& \Br{\begin{array}{ccc}
                     z_{11} & * & * \\
                     \wtil c(z_{21}+\sqrt{-1}w_{21}) & \abs{\wtil c}^2z_{22} & * \\
                     z_{31}+\sqrt{-1}w_{31} & \overline{\wtil c}(z_{32}+\sqrt{-1}w_{32}) & z_{33}
                   \end{array}
}\\
&=&\Br{\begin{array}{ccc}
       u_{11} & * & * \\
       u_{21}+\sqrt{-1}v_{21} & u_{22} & * \\
       u_{31}+\sqrt{-1}v_{31} & u_{32}+\sqrt{-1}v_{32} & u_{33}
     \end{array}
},
\end{eqnarray*}
thus
\begin{eqnarray*}
u_{11}=z_{11},u_{22}=\abs{\wtil c}^2z_{22},u_{33}=z_{33},\\
u_{21}=az_{21}-bw_{21},u_{31}=z_{31},u_{32}=az_{32}+bw_{32},\\
v_{21}=aw_{21}+bz_{21},v_{31}=w_{31},v_{32}=aw_{32}-bz_{32}.
\end{eqnarray*}
Now
\begin{eqnarray*}
\Br{\begin{array}{c}
      u_{11} \\
      u_{22} \\
      u_{33} \\
      u_{21} \\
      u_{31} \\
      u_{32} \\
      v_{21} \\
      v_{31} \\
      v_{32}
    \end{array}
}=\Br{\begin{array}{ccccccccc}
        1 & 0 & 0 & 0 & 0 & 0 & 0 & 0 & 0 \\
        0 & \abs{\wtil c}^2 & 0 & 0 & 0 & 0 & 0 & 0 & 0 \\
        0 & 0 & 1 & 0 & 0 & 0 & 0 & 0 & 0 \\
        0 & 0 & 0 & a & 0 & 0 & -b & 0 & 0 \\
        0 & 0 & 0 & 0 & 1 & 0 & 0 & 0 & 0 \\
        0 & 0 & 0 & 0 & 0 & a & 0 & 0 & b \\
        0 & 0 & 0 & b & 0 & 0 & a & 0 & 0 \\
        0 & 0 & 0 & 0 & 0 & 0 & 0 & 1 & 0 \\
        0 & 0 & 0 & 0 & 0 & -b & 0 & 0 & a
      \end{array}
}\Br{\begin{array}{c}
      z_{11} \\
      z_{22} \\
      z_{33} \\
      z_{21} \\
      z_{31} \\
      z_{32} \\
      w_{21} \\
      w_{31} \\
      w_{32}
    \end{array}
}.
\end{eqnarray*}
Now
\begin{eqnarray*}
C= \Br{\begin{array}{ccc}
         1 &  &  \\
          & \abs{\wtil c}^2 &  \\
          &  & 1
       \end{array}
}, A_0 = \Br{\begin{array}{ccc}
               a &  &  \\
                & 1 &  \\
                &  & a
             \end{array}
}, B_0 = \Br{\begin{array}{ccc}
               -b &  &  \\
                & 0 &  \\
                &  & b
             \end{array}
}.
\end{eqnarray*}
We are done.
\end{proof}

\begin{remark}
If $\wtil X$ is skew hermitian then the diagonal elements are purely
imaginary, that is, the real parts are zeros. It is easy to note
that the structure of the Jacobian matrix for a transformation of
the type $\wtil Y=\wtil A\wtil X\wtil A^*$, where $\wtil X^*=-\wtil
X$, remains the same as that in the hermitian case of
Proposition~\ref{prop:adjoint-action}. The roles of
$(u_{jj},z_{jj})$'s and $(v_{jj},w_{jj})$'s are interchanged. Thus
the next theorem will be stated without proof.
\end{remark}

\begin{prop}\label{prop:}
Let $\wtil X\in\complex^{n\times n}$ skew hermitian matrix of
independent complex entries. Let $\wtil A\in\complex^{n\times n}$ be
a nonsingular matrix of constants. If $\wtil Y=\wtil A\wtil X \wtil
A^*$, then
\begin{eqnarray}
\fbox{$[\dif\wtil Y] =\abs{\det(\wtil A\wtil A^*)}^n[\dif\wtil X]$.}
\end{eqnarray}
\end{prop}

Some simple nonlinear transformations will be considered here. These
are transformations which become linear transformations in the
differentials so that the Jacobian of the original transformation
becomes the Jacobian of the linear transformation where the matrices
of differentials are treated as the new variables and everything
else as constants.

\begin{prop}\label{prop:TT-complex-case}
Let $\wtil X\in\complex^{n\times n}$ be hermitian positive definite
matrix of independent complex variables. Let $\wtil
T\in\complex^{n\times n}$ be lower triangular and $\wtil
Q\in\complex^{n\times n}$ be upper triangular matrices of
independent complex variables with real and positive diagonal
elements. Then
\begin{eqnarray}
\wtil X=\wtil T\wtil T^* &\Longrightarrow& [\dif\wtil X] =
2^n\Pa{\prod^n_{j=1}t^{2(n-j)+1}_{jj}}[\dif\wtil T],\\
\wtil X=\wtil Q\wtil Q^* &\Longrightarrow& [\dif\wtil X] =
2^n\Pa{\prod^n_{j=1}t^{2(j-1)+1}_{jj}}[\dif\wtil Q].
\end{eqnarray}
\end{prop}

\begin{proof}
When the diagonal elements of the triangular matrices are real and
positive there exist unique representations $\wtil X=\wtil T\wtil
T^*$ and $\wtil X=\wtil Q\wtil Q^*$. Let $\wtil X=X_1+\sqrt{-1}X_2$
and $\wtil T=T_1+\sqrt{-1}T_2$, where $\wtil X=[\wtil x_{jk}],\wtil
T=[\wtil t_{jk}],\wtil t_{jk}=0,j<k,
X_m=[x^{(m)}_{jk}],T_m=[t^{(m)}_{jk}],m=1,2$. Note that $X_1$ is
\emph{symmetric} and $X_2$ is \emph{skew symmetric}. The diagonal
elements of $X_2$ and $T_2$ are zeros. Hence when considering the
Jacobian we should take $\wtil x_{jj},j=1,\ldots,p$ and $\wtil
x_{jk},j>k$ separately.
\begin{eqnarray*}
\wtil X = \wtil T\wtil T^*&\Longrightarrow&
X_1+\sqrt{-1}X_2 = (T_1+\sqrt{-1}T_2)(T^\t_1-\sqrt{-1}T^\t_2)\\
&\Longrightarrow&
\begin{cases}X_1=T_1T^\t_1+T_2T^\t_2\\X_2=T_2T^\t_1-T_1T^\t_2\end{cases},
\end{eqnarray*}
with $t^{(1)}_{jj}=t_{jj},t^{(2)}_{jj}=0,j=1,\ldots,n$. Note that
\begin{eqnarray*}
x^{(1)}_{jj} = \Pa{\Pa{t^{(1)}_{j1}}^2 + \cdots +
\Pa{t^{(1)}_{jj}}^2} + \Pa{\Pa{t^{(2)}_{j1}}^2 + \cdots +
\Pa{t^{(2)}_{j,j-1}}^2}
\end{eqnarray*}
implies that
\begin{eqnarray*}
\frac{\partial x^{(1)}_{jj}}{\partial t^{(1)}_{jj}} =
2t^{(1)}_{jj}=2t_{jj}, j=1,\ldots,n.
\end{eqnarray*}
So
\begin{eqnarray*}
\frac{\partial (x^{(1)}_{11},\ldots,x^{(1)}_{nn})}{\partial
(t^{(1)}_{11},\ldots,t^{(1)}_{nn})} = \Br{\begin{array}{ccc}
                                  2t_{11} &  &  \\
                                   & \ddots &  \\
                                   &  & 2t_{nn}
                                \end{array}
}:=Z.
\end{eqnarray*}
Now consider the $x^{(1)}_{jk}$'s for $j>k$. It is easy to note that
\begin{eqnarray*}
\frac{\partial (X_{10},X_{20})}{\partial (T_{10},T_{20})} =
\Br{\begin{array}{cc}
                  U & V \\
                  W & Y
                \end{array}
},
\end{eqnarray*}
where a zero indicates that the $x^{(1)}_{jj}$'s are removed and the
derivatives are taken with respect to the $t^{(1)}_{jk}$'s and
$t^{(2)}_{jk}$'s for $j>k$. $U$ and $Y$ are lower triangular
matrices with $t_{jj}$ repeated $n-j$ times along the diagonal and
$V$ is of the same form as $U$ but with $t^{(2)}_{jj}=0$ along the
diagonal and the $t^{(1)}_{jk}$'s replaced by the $t^{(2)}_{jk}$'s.
For example, take the $x^{(1)}_{jk}$'s in the order
$x^{(1)}_{21},x^{(1)}_{31},\ldots,x^{(1)}_{n1},x^{(1)}_{32},\ldots,x^{(1)}_{n,n-1}$
and $t^{(1)}_{jk}$'s also in the same order. Then we get the
$\frac{n(n-1)}2\times \frac{n(n-1)}2$ matrix
\begin{eqnarray*}
U = \frac{\partial X_{10}}{\partial T_{10}} =
\Br{\begin{array}{cccc}
                                                    t_{11} & 0 & \cdots & 0 \\
                                                    * & t_{11} & \cdots & 0 \\
                                                    \vdots & \vdots & \ddots & 0 \\
                                                    * & * & \cdots &
                                                    t_{n-1,n-1}
                                                  \end{array}
}
\end{eqnarray*}
where the $*$'s indicate the presence of elements some of which may
be zeros. Since $U$ and $V$ are lower triangular with the diagonal
elements of $V$ being zeros, one can make $W$ null by adding
suitable combinations of the rows of $(U,V)$. This will not alter
the lower triangular nature or the diagonal elements of $Y$. Then
the determinant is given by
\begin{eqnarray*}
\det\Pa{\Br{\begin{array}{cc}
              U & V \\
              W & Y
            \end{array}
}} = \det(U)\det(Y) = \prod^n_{j=1}t^{2(n-j)}_{jj}.
\end{eqnarray*}
Multiply with the $2t_{jj}$'s for $j=1,\ldots,n$ to establish the
result. As a specific example, we consider the case where $n=3$. Let
$\wtil X\in\complex^{3\times 3}$. Denote $\wtil X=X_1+\sqrt{-1}X_2$.
Thus
$$
X_1=\Br{\begin{array}{ccc}
          x^{(1)}_{11} & x^{(1)}_{21} & x^{(1)}_{31} \\
          x^{(1)}_{21} & x^{(1)}_{22} & x^{(1)}_{32} \\
          x^{(1)}_{31} & x^{(1)}_{32} & x^{(1)}_{33}
        \end{array}
},~~X_2=\Br{\begin{array}{ccc}
          0 & -x^{(2)}_{21} & -x^{(2)}_{31} \\
          x^{(2)}_{21} & 0 & -x^{(2)}_{32} \\
          x^{(2)}_{31} & x^{(2)}_{32} & 0
        \end{array}
}.
$$
Similarly, let $\wtil T=T_1+\sqrt{-1}T_2$. We also have:
$$
T_1=\Br{\begin{array}{ccc}
          t^{(1)}_{11} & 0 & 0 \\
          t^{(1)}_{21} & t^{(1)}_{22} & 0 \\
          t^{(1)}_{31} & t^{(1)}_{32} & t^{(1)}_{33}
        \end{array}
},~~T_2=\Br{\begin{array}{ccc}
          0 & 0 & 0 \\
          t^{(2)}_{21} & 0 & 0 \\
          t^{(2)}_{31} & t^{(2)}_{32} & 0
        \end{array}
}.
$$
Now $X_1=T_1T^\t_1+T_2T^\t_2$ can be expanded as follows:
\begin{eqnarray*}
x^{(1)}_{11} &=& \Pa{t^{(1)}_{11}}^2,~x^{(1)}_{21} =
t^{(1)}_{21}t^{(1)}_{11},~x^{(1)}_{31} = t^{(1)}_{31}t^{(1)}_{11};\\
x^{(1)}_{22} &=&
\Pa{t^{(1)}_{21}}^2+\Pa{t^{(1)}_{22}}^2+\Pa{t^{(2)}_{21}}^2,~
x^{(1)}_{32} =
t^{(1)}_{31}t^{(1)}_{21}+t^{(1)}_{32}t^{(1)}_{22}+t^{(2)}_{31}t^{(2)}_{21};\\
x^{(1)}_{33}
&=&\Pa{t^{(1)}_{31}}^2+\Pa{t^{(1)}_{32}}^2+\Pa{t^{(1)}_{33}}^2+\Pa{t^{(2)}_{31}}^2+\Pa{t^{(2)}_{32}}^2.
\end{eqnarray*}
Then $X_2=T_2T^\t_1-T_1T^\t_2$ can be expanded as follows:
\begin{eqnarray*}
x^{(2)}_{21} = t^{(2)}_{21}t^{(1)}_{11},~~x^{(2)}_{31} =
t^{(2)}_{31}t^{(1)}_{11},~~x^{(2)}_{32} =
t^{(2)}_{31}t^{(1)}_{21}+t^{(2)}_{32}t^{(1)}_{22}.
\end{eqnarray*}
From the above, we see that
\begin{eqnarray*}
\Br{\begin{array}{c}
      \dif x^{(1)}_{11} \\
      \dif x^{(1)}_{22} \\
      \dif x^{(1)}_{33} \\
      \dif x^{(1)}_{21} \\
      \dif x^{(1)}_{31} \\
      \dif x^{(1)}_{32} \\
      \dif x^{(2)}_{21} \\
      \dif x^{(2)}_{31} \\
      \dif x^{(2)}_{32}
    \end{array}
}=\Br{\begin{array}{ccccccccc}
        2t^{(1)}_{11} & 0 & 0 & 0 & 0 & 0 & 0 & 0 & 0 \\
        0 & 2t^{(1)}_{22} & 0 & 2t^{(1)}_{21} & 0 & 0 & 2t^{(2)}_{21} & 0 & 0 \\
        0 & 0 & 2t^{(1)}_{33} & 0 & 2t^{(1)}_{31} & 2t^{(1)}_{32} & 0 & 2t^{(2)}_{31} & 2t^{(2)}_{32} \\
        t^{(1)}_{21} & 0 & 0 & t^{(1)}_{11} & 0 & 0 & 0 & 0 & 0 \\
        t^{(1)}_{31} & 0 & 0 & 0 & t^{(1)}_{11} & 0 & 0 & 0 & 0 \\
        0 & t^{(1)}_{32} & 0 & t^{(1)}_{31} & t^{(1)}_{21} & t^{(1)}_{22} & t^{(2)}_{31} & t^{(2)}_{21} & 0 \\
        t^{(2)}_{21} & 0 & 0 & 0 & 0 & 0 & t^{(1)}_{11} & 0 & 0 \\
        t^{(2)}_{31} & 0 & 0 & 0 & 0 & 0 & 0 & t^{(1)}_{11} & 0 \\
        0 & t^{(2)}_{32} & 0 & t^{(2)}_{31} & 0 & 0 & 0 & t^{(1)}_{21} & t^{(1)}_{22}
      \end{array}
}\Br{\begin{array}{c}
      \dif t^{(1)}_{11} \\
      \dif t^{(1)}_{22} \\
      \dif t^{(1)}_{33} \\
      \dif t^{(1)}_{21} \\
      \dif t^{(1)}_{31} \\
      \dif t^{(1)}_{32} \\
      \dif t^{(2)}_{21} \\
      \dif t^{(2)}_{31} \\
      \dif t^{(2)}_{32}
    \end{array}
}.
\end{eqnarray*}
In what follows, we compute the Jacobian of this transformation:
\begin{eqnarray*}
J(\wtil X:\wtil T) &=& \det\Pa{\Br{\begin{array}{ccccccccc}
        2t^{(1)}_{11} & 0 & 0 & 0 & 0 & 0 & 0 & 0 & 0 \\
        0 & 2t^{(1)}_{22} & 0 & 2t^{(1)}_{21} & 0 & 0 & 2t^{(2)}_{21} & 0 & 0 \\
        0 & 0 & 2t^{(1)}_{33} & 0 & 2t^{(1)}_{31} & 2t^{(1)}_{32} & 0 & 2t^{(2)}_{31} & 2t^{(2)}_{32} \\
        t^{(1)}_{21} & 0 & 0 & t^{(1)}_{11} & 0 & 0 & 0 & 0 & 0 \\
        t^{(1)}_{31} & 0 & 0 & 0 & t^{(1)}_{11} & 0 & 0 & 0 & 0 \\
        0 & t^{(1)}_{32} & 0 & t^{(1)}_{31} & t^{(1)}_{21} & t^{(1)}_{22} & t^{(2)}_{31} & t^{(2)}_{21} & 0 \\
        t^{(2)}_{21} & 0 & 0 & 0 & 0 & 0 & t^{(1)}_{11} & 0 & 0 \\
        t^{(2)}_{31} & 0 & 0 & 0 & 0 & 0 & 0 & t^{(1)}_{11} & 0 \\
        0 & t^{(2)}_{32} & 0 & t^{(2)}_{31} & 0 & 0 & 0 & t^{(1)}_{21} & t^{(1)}_{22}
      \end{array}
}}\\
&=&2t^{(1)}_{11}\det\Pa{\Br{\begin{array}{ccccccccc}
        1 & 0 & 0 & 0 & 0 & 0 & 0 & 0 & 0 \\
        0 & 2t^{(1)}_{22} & 0 & 2t^{(1)}_{21} & 0 & 0 & 2t^{(2)}_{21} & 0 & 0 \\
        0 & 0 & 2t^{(1)}_{33} & 0 & 2t^{(1)}_{31} & 2t^{(1)}_{32} & 0 & 2t^{(2)}_{31} & 2t^{(2)}_{32} \\
        t^{(1)}_{21} & 0 & 0 & t^{(1)}_{11} & 0 & 0 & 0 & 0 & 0 \\
        t^{(1)}_{31} & 0 & 0 & 0 & t^{(1)}_{11} & 0 & 0 & 0 & 0 \\
        0 & t^{(1)}_{32} & 0 & t^{(1)}_{31} & t^{(1)}_{21} & t^{(1)}_{22} & t^{(2)}_{31} & t^{(2)}_{21} & 0 \\
        t^{(2)}_{21} & 0 & 0 & 0 & 0 & 0 & t^{(1)}_{11} & 0 & 0 \\
        t^{(2)}_{31} & 0 & 0 & 0 & 0 & 0 & 0 & t^{(1)}_{11} & 0 \\
        0 & t^{(2)}_{32} & 0 & t^{(2)}_{31} & 0 & 0 & 0 & t^{(1)}_{21} & t^{(1)}_{22}
      \end{array}
}},
\end{eqnarray*}
by adding the corresponding multiples of the first row to the second
row through the last one, respectively, we get
\begin{eqnarray*}
J(\wtil X:\wtil T) =
2t^{(1)}_{11}\det\Pa{\Br{\begin{array}{ccccccccc}
        1 & 0 & 0 & 0 & 0 & 0 & 0 & 0 & 0 \\
        0 & 2t^{(1)}_{22} & 0 & 2t^{(1)}_{21} & 0 & 0 & 2t^{(2)}_{21} & 0 & 0 \\
        0 & 0 & 2t^{(1)}_{33} & 0 & 2t^{(1)}_{31} & 2t^{(1)}_{32} & 0 & 2t^{(2)}_{31} & 2t^{(2)}_{32} \\
        0 & 0 & 0 & t^{(1)}_{11} & 0 & 0 & 0 & 0 & 0 \\
        0 & 0 & 0 & 0 & t^{(1)}_{11} & 0 & 0 & 0 & 0 \\
        0 & t^{(1)}_{32} & 0 & t^{(1)}_{31} & t^{(1)}_{21} & t^{(1)}_{22} & t^{(2)}_{31} & t^{(2)}_{21} & 0 \\
        0 & 0 & 0 & 0 & 0 & 0 & t^{(1)}_{11} & 0 & 0 \\
        0 & 0 & 0 & 0 & 0 & 0 & 0 & t^{(1)}_{11} & 0 \\
        0 & t^{(2)}_{32} & 0 & t^{(2)}_{31} & 0 & 0 & 0 & t^{(1)}_{21} & t^{(1)}_{22}
      \end{array}
}},
\end{eqnarray*}
iteratively, finally we get the final result. We also take a simple
approach (i.e. by definition) with a tedious computation as follows:
\begin{eqnarray*}
\dif x^{(1)}_{11} &=& 2t_{11}\dif t_{11},~\dif x^{(1)}_{21} =
t^{(1)}_{21}\dif t_{11}+t_{11}\dif t^{(1)}_{21},~\dif x^{(1)}_{31} = t^{(1)}_{31}\dif t_{11}+t_{11}\dif t^{(1)}_{31};\\
\dif x^{(1)}_{22} &=& 2t_{22}\dif t_{22}+ 2t^{(1)}_{21}\dif
t^{(1)}_{21}+2t^{(2)}_{21}\dif t^{(2)}_{21},\\
\dif x^{(1)}_{32} &=&
t^{(1)}_{31}\dif t^{(1)}_{21}+ t^{(1)}_{21}\dif t^{(1)}_{31}+t^{(1)}_{32}\dif t_{22}+ t_{22}\dif t^{(1)}_{32}+t^{(2)}_{31}\dif t^{(2)}_{21}+t^{(2)}_{21}\dif t^{(2)}_{31};\\
\dif x^{(1)}_{33} &=& 2t^{(1)}_{31}\dif t^{(1)}_{31} +
2t^{(1)}_{32}\dif t^{(1)}_{32} + 2t_{33}\dif t_{33} +
2t^{(2)}_{31}\dif t^{(2)}_{31} + 2t^{(2)}_{32}\dif t^{(2)}_{32},
\end{eqnarray*}
and
\begin{eqnarray*}
\dif x^{(2)}_{21} &=& t^{(2)}_{21}\dif t_{11}+t_{11}\dif
t^{(2)}_{21},~~\dif x^{(2)}_{31} = t^{(2)}_{31}\dif
t_{11}+t_{11}\dif t^{(2)}_{31},\\
\dif x^{(2)}_{32} &=& t^{(2)}_{31}\dif t^{(1)}_{21}+
t^{(1)}_{21}\dif t^{(2)}_{31}+t_{22}\dif t^{(2)}_{32}+
t^{(2)}_{32}\dif t_{22}.
\end{eqnarray*}
Hence we can compute the Jacobian by definition as follows:
\begin{eqnarray*}
[\dif \wtil X] &=&\dif x^{(1)}_{11}\wedge\dif x^{(1)}_{21}\wedge\dif
x^{(1)}_{31}\wedge\dif x^{(2)}_{21}\wedge\dif x^{(2)}_{31}\wedge\dif
x^{(1)}_{22}\wedge\dif x^{(2)}_{32}\wedge\dif x^{(1)}_{32}\wedge\dif
x^{(1)}_{33}\\
&=& 2^3\Pa{\prod^3_{j=1}t^{2(n-j)+1}_{jj}}[\dif \wtil T].
\end{eqnarray*}

The proof in the case of $\wtil X=\wtil Q\wtil Q^*$ is similar but
in this case it can be seen that the triangular matrices
corresponding to $U$ and $Y$ will have $t_{jj}$ repeated $j-1$ times
along the diagonal for $j=1,\ldots,n$.
\end{proof}

\begin{exam}\label{exam:complex-gamma}
If $\wtil X^*=\wtil X\in\complex^{n\times n}$ is positive definite ,
and $\re(\alpha)>n-1$,
\begin{eqnarray}
\wtil\Gamma_n(\alpha) &:=& \int_{\wtil X>0}[\dif\wtil X]
\abs{\det\Pa{\wtil X}}^{\alpha-n}e^{-\Tr{\wtil X}} \notag\\
&=& \pi^{\frac{n(n-1)}2}\Gamma(\alpha)\Gamma(\alpha-1)\cdots
\Gamma(\alpha-n+1).
\end{eqnarray}
Indeed, let $\wtil T=[\wtil t_{jk}],\wtil t_{jk}=0,j<k$ be a lower
triangular matrix with real and positive diagonal elements
$t_{jj}>0,j=1,\ldots,n$ such that $\wtil X=\wtil T\wtil T^*$. Then
from Proposition~\ref{prop:TT-complex-case}
\begin{eqnarray*}
[\dif\wtil X] = 2^n\Pa{\prod^n_{j=1}t^{2(n-j)+1}_{jj}}[\dif\wtil T]
\end{eqnarray*}
Note that
\begin{eqnarray*}
\Tr{\wtil X} = \Tr{\wtil T\wtil T^*} = \sum^n_{j=1}t^2_{jj}+
\abs{\wtil t_{21}}^2+\cdots+\abs{\wtil t_{n1}}^2+\cdots+\abs{\wtil
t_{n,n-1}}^2
\end{eqnarray*}
and
\begin{eqnarray*}
\abs{\det(\wtil X)}^{\alpha-n}[\dif\wtil X] =
2^n\Pa{\prod^n_{j=1}t^{2\alpha-2j+1}_{jj}}[\dif\wtil T].
\end{eqnarray*}
The integral over $\wtil X$ splits into $n$ integrals over the
$t_{jj}$'s and $\frac{n(n-1)}2$ integrals over the $\wtil t_{jk}$'s,
$j>k$. Note that $0<t_{jj}<\infty, -\infty<t^{(1)}_{jk}<\infty,
-\infty<t^{(2)}_{jk}<\infty$, where $\wtil
t_{jk}=t^{(1)}_{jk}+\sqrt{-1}t^{(2)}_{jk}$. But
\begin{eqnarray*}
2\int^\infty_0 t^{2\alpha-2j+1}_{jj}e^{-t^2_{jj}}\dif t_{jj}
=\Gamma(\alpha-j+1),~~\re(\alpha)>j-1,
\end{eqnarray*}
for $j=1,\ldots,n$, so $\re(\alpha)>n-1$ and
\begin{eqnarray*}
\int_{\wtil t_{jk}} e^{-\abs{\wtil t_{jk}}^2}\dif\wtil t_{jk} =
\int^\infty_{-\infty}\int^\infty_{-\infty}
e^{-\Pa{\Pa{t^{(1)}_{jk}}^2+\Pa{t^{(2)}_{jk}}^2}}\dif
t^{(1)}_{jk}\dif t^{(2)}_{jk}=\pi.
\end{eqnarray*}
The desired result is obtained.
\end{exam}

\begin{definition}[$\wtil \Gamma_n(\alpha)$: complex
matrix-variate gamma]\label{def:tilde-gamma} It is defined as stated
in Example~\ref{exam:complex-gamma}. We will write with a
\emph{tilde} over $\Gamma$ to distinguish it from the matrix-variate
gamma in the real case.
\end{definition}

\begin{exam}\label{exam:complex-matrix-variate-density}
Show that
\begin{eqnarray*}
f(\wtil X) = \frac{\abs{\det(\wtil B)}^\alpha\abs{\det(\wtil
X)}^{\alpha-n}e^{-\Tr{\wtil B\wtil X}}}{\wtil \Gamma_n(\alpha)}
\end{eqnarray*}
for $\wtil B^*=\wtil B>0,\wtil X^*=\wtil X>0,\re(\alpha)>n-1$ and
$f(\wtil X)=0$ elsewhere, is a density function for $\wtil X$ where
$\wtil B$ is a constant matrix, with $\wtil \Gamma_n(\alpha)$ as
given in Definition~\ref{def:tilde-gamma}. Indeed, evidently
$f(\wtil X)\geqslant0$ for all $\wtil X$ and for all $\wtil X$ and
it remains to show that the total integral is unity. Since $\wtil B$
is hermitian positive definite there exists a nonsingular $\wtil C$
such that $\wtil B=\wtil C^*\wtil C$. Then
\begin{eqnarray*}
\Tr{\wtil B\wtil X} = \Tr{\wtil C\wtil X\wtil C^*}.
\end{eqnarray*}
Hence from Proposition~\ref{prop:TT-complex-case}
\begin{eqnarray*}
\wtil Y = \wtil C\wtil X\wtil C^*\Longrightarrow [\dif\wtil Y] =
\abs{\det\Pa{\wtil C\wtil C^*}}^n[\dif\wtil X]=\abs{\det(\wtil
B)}^n[\dif\wtil X],
\end{eqnarray*}
and
\begin{eqnarray*}
\wtil X = C^{-1}\wtil YC^{*-1}\Longrightarrow\abs{\det(\wtil
X)}=\abs{\det(CC^*)}^{-1}\abs{\det(\wtil Y)}
\end{eqnarray*}
Then
\begin{eqnarray*}
\int_{\wtil X>0}f(\wtil X)[\dif\wtil X] = \int_{\wtil Y>0}[\dif\wtil
Y]\frac{\abs{\det(\wtil Y)}^{\alpha-n}e^{-\Tr{\wtil Y}}}{\wtil
\Gamma_n(\alpha)}=1.
\end{eqnarray*}
But from Example~\ref{exam:complex-gamma}, the right side is unity
for $\re(\alpha)>n-1$. This density $f(\wtil X)$ is known as the
\emph{complex matrix-variate density with the parameters $\alpha$
and $\wtil B$}.
\end{exam}

\subsection{The computation of volumes}

\begin{definition}[Semiuniatry and unitary matrices]
A $p\times n$ matrix $\wtil U$ is said to be \emph{semiunitary} if
$\wtil U\wtil U^*=\I_p$ for $p<n$ or $\wtil U^*\wtil U=\I_n$ for
$p>n$. When $n=p$ and $\wtil U \wtil U^*=\I_n$, then $\wtil U$ is
called a \emph{unitary matrix}. The set of all $n\times n$ unitary
matrices is denoted by $\cU(n)$. That is,
\begin{eqnarray}
\fbox{$\cU(n):=\Set{\wtil U\in\complex^{n\times n}: \wtil U\wtil
U^*=\I_n}$,}
\end{eqnarray}
where $\complex^{n\times n}$ denotes the set of all $n\times n$
complex matrices.
\end{definition}

\begin{definition}[A hermitian or a skew hermitian matrix]
Let $\wtil A\in\complex^{n\times n}$. If $\wtil A=\wtil A^*$, then
$\wtil A$ is said to be \emph{hermitian} and if $\wtil A^*=-\wtil
A$, then it is \emph{skew hermitian}.
\end{definition}

When dealing with unitary matrices a basic property to be noted is
the following:
\begin{eqnarray*}
\wtil U\wtil U^*=\I\Longrightarrow \wtil U^* \dif\wtil U = -
\dif\wtil U^*\wtil U.
\end{eqnarray*}
But $\Pa{\wtil U^* \dif\wtil U}^* = \dif\wtil U^*\wtil U$, which
means that $\wtil U^* \dif\wtil U$ is a skew hermitian matrix. The
wedge product of $\wtil U^* \dif\wtil U$, namely, $\wedge\Pa{\wtil
U^* \dif\wtil U}$ enters into the picture when evaluating the
Jacobians involving unitary transformations. Hence this will be
denoted by $\dif\wtil G$ for convenience. Starting from $\wtil
U^*\wtil U=\I_n$ one has $\dif\wtil U\cdot \wtil U^*$.

Assume that $\wtil U=[\wtil u_{ij}]\in\cU(n)$ where $\wtil
u_{ij}\in\complex$. Let $\wtil u_{jj}=\abs{\wtil
u_{jj}}e^{\sqrt{-1}\theta_j}$ by Euler's formula, where
$\theta_j\in[-\pi,\pi]$. Then
\begin{eqnarray}
\wtil U=\Br{\begin{array}{cccc}
               \abs{\wtil u_{11}} & \wtil u_{12}e^{-\sqrt{-1}\theta_2} & \cdots & \wtil u_{1n}e^{-\sqrt{-1}\theta_n} \\
               \wtil u_{21}e^{-\sqrt{-1}\theta_1} & \abs{\wtil u_{22}} & \cdots & \wtil u_{2n}e^{-\sqrt{-1}\theta_n} \\
               \vdots  & \vdots  & \ddots  & \vdots  \\
               \wtil u_{n1}e^{-\sqrt{-1}\theta_1} & \wtil u_{n2}e^{-\sqrt{-1}\theta_2} & \cdots & \abs{\wtil
               u_{nn}}
             \end{array}
}\Br{\begin{array}{cccc}
       e^{\sqrt{-1}\theta_1}& 0 & \cdots & 0 \\
       0 & e^{\sqrt{-1}\theta_2} & \cdots & 0 \\
       \vdots & \vdots & \ddots & \vdots \\
       0 & 0 & \cdots & e^{\sqrt{-1}\theta_n}
     \end{array}
}.
\end{eqnarray}
This indicates that any $\wtil U\in\cU(n)$ can be factorized into a
product of a unitary matrix with diagonal entries being nonnegative
and a diagonal unitary matrix. It is easily seen that such
factorization of a given unitary matrix is unique. In fact, we have
a correspondence which is one-to-one:
\begin{center}
\fbox{$\cU(n)\sim \Pa{\cU(n)/\cU(1)^{\times n}}\times \cU(1)^{\times
n}$.}
\end{center}
\textbf{Notation.} When $\wtil U$ is a $n\times n$ unitary matrix of
independent complex entries, $\wtil U^*$ its conjugate transpose and
$\dif\wtil U$ the matrix of differentials then the wedge product in
$\dif\wtil G:=\dif\wtil U\cdot\wtil U^*$ will be denoted by
$[\dif\wtil G]$. That is, ignoring the sign,
\begin{eqnarray*}
\fbox{$[\dif\wtil G] :=\wedge \Pa{\dif\wtil U\cdot\wtil U^*} =
\wedge\Pa{\wtil U\cdot\dif\wtil U^*}.$}
\end{eqnarray*}
If the diagonal entries or the entries in one row of this unitary
matrix $\wtil U$ are assumed to be real, then the skew hermitian
matrix $\dif\wtil U\cdot\wtil U^*$ will be denoted by $\dif\wtil
G_1$ and its wedge product by
\begin{eqnarray*}
[\dif\wtil G_1] := \wedge\Pa{\dif\wtil U\cdot\wtil U^*}.
\end{eqnarray*}
Indeed, $[\dif\wtil G]$ here means the wedge product over $\cU(n)$,
but however $[\dif\wtil G_1]$ means the wedge product over
$\cU(n)/\cU(1)^{\times n}$. Therefore, for any measurable function
$f$ over $\cU(n)$, $[\dif\wtil G]=[\dif\wtil G_1][\dif\wtil D]$,
$$
\int_{\cU(n)} f(\wtil U)[\dif\wtil G] =
\int_{\cU_1(n)}\int_{\cU(1)^{\times n}}f(\wtil V\wtil D)[\dif \wtil
G_1][\dif \wtil D],
$$
where $\wtil U=\wtil V\wtil D$ for $\wtil V\in\cU_1(n)$, $\dif\wtil
G=\wtil U^*\dif \wtil U$ and $\dif\wtil G_1=\wtil V^*\dif \wtil V$.
Furthermore, let $\wtil V=[\wtil v_{ij}]$ for $\wtil
v_{ij}\in\complex$ and $\wtil v_{jj}=v_{jj}\in\real^+$. Then it
holds still that $\wtil V\in\cU(n)$, thus $v_{jj}$ is not an
independent variable, for example, $v_{11}=\sqrt{1-\abs{\wtil
v_{21}}^2-\cdots-\abs{\wtil v_{n1}}^2}$. From this, we see that
$$
[\dif \wtil G_1]=\prod_{i<j} \dif\Pa{\re(\wtil V^*\dif\wtil
V)_{ij}}\dif\Pa{\im(\wtil V^*\dif\wtil V)_{ij}}
$$
and
$$
[\dif \wtil G]=\Pa{\prod^n_{j=1}\im(\wtil U^*\dif \wtil
U)_{jj}}\times \prod_{i<j} \re(\wtil U^*\dif\wtil U)_{ij}\im(\wtil
U^*\dif\wtil U)_{ij}.
$$
\begin{prop}\label{prop:LU-complex-case}
Let $\wtil T\in\complex^{n\times n}$ be lower triangular and $\wtil
U\in\cU(n)$ be of independent complex variables. Let $\wtil X=\wtil
T\wtil U$. Then:
\begin{enumerate}[(i)]
\item for all the diagonal entries $t_{jj},j=1,\ldots,n$ of $\wtil
T$ being real and positive,
\begin{eqnarray}
\fbox{$[\dif\wtil X] = \Pa{\prod^n_{j=1}t^{2(n-j)+1}_{jj}}[\dif\wtil
T][\dif\wtil G]$}
\end{eqnarray}
where $\dif\wtil G = \dif\wtil U\cdot\wtil U^*$; and
\item for all the
diagonal entries in $\wtil U$ being real,
\begin{eqnarray}
[\dif\wtil X] = \Pa{\prod^n_{j=1}\abs{\wtil
t_{jj}}^{2(n-j)}}\cdot[\dif\wtil T][\dif\wtil G_1]
\end{eqnarray}
where $\dif\wtil G_1 = \dif\wtil U\cdot\wtil U^*$.
\end{enumerate}
\end{prop}

\begin{proof}
Taking differentials in $\wtil X=\wtil T\wtil U$ one has
\begin{eqnarray*}
\dif\wtil X = \dif\wtil T\cdot\wtil U + \wtil T\cdot\dif\wtil U.
\end{eqnarray*}
Postmultiplying by $\wtil U^*$ and observing that $\wtil U\wtil
U^*=\I_n$ we have
\begin{eqnarray}\label{eq:aaaa}
\dif\wtil X\cdot \wtil U^* = \dif\wtil T + \wtil T\cdot\dif\wtil
U\cdot\wtil U^*.
\end{eqnarray}
(i). Let the diagonal elements in $\wtil T$ be real and positive and
all other elements in $\wtil T$ and $\wtil U$ be complex. Let
\begin{eqnarray}\label{eq:bbbb}
\dif\wtil V = \dif\wtil X\cdot \wtil U^*\Longrightarrow [\dif\wtil
V]=[\dif\wtil X]
\end{eqnarray}
ignoring the sign, since $\wtil U$ is unitary. Let $\dif\wtil
G=\dif\wtil U\cdot \wtil U^*$ and its wedge product be $[\dif\wtil
G]$. Then
\begin{eqnarray}\label{eq:cccc}
\dif\wtil V = \dif\wtil T+\wtil T\cdot\dif\wtil G
\end{eqnarray}
where $\dif\wtil G$ is skew hermitian. Write
\begin{eqnarray*}
\wtil V &=& [\wtil v_{jk}], \quad \wtil v_{jk} =
v^{(1)}_{jk}+\sqrt{-1}v^{(2)}_{jk},\\
\wtil T &=& [\wtil t_{jk}],\quad\wtil t_{jk} =
t^{(1)}_{jk}+\sqrt{-1}t^{(2)}_{jk}, j>k,~~~t^{(1)}_{jj}=t_{jj}>0,~~t^{(2)}_{jj}=0,\\
\dif\wtil G &=& [\dif\wtil g_{jk}],\quad \dif\wtil g_{jk} = \dif
g^{(1)}_{jk}+\sqrt{-1}\dif g^{(2)}_{jk},~~~\dif g^{(1)}_{jk}=-\dif
g^{(1)}_{kj},~~\dif g^{(2)}_{jk}=\dif g^{(2)}_{kj}.
\end{eqnarray*}
From Eq.~\eqref{eq:cccc},
\begin{eqnarray*}
\dif\wtil v_{jk} =
\begin{cases}
\dif\wtil t_{jk} + \Pa{\wtil t_{j1}\dif\wtil g_{1k}+\cdots+\wtil
t_{jj}\dif\wtil g_{jk}},&j\geqslant k\\
\Pa{\wtil t_{j1}\dif\wtil g_{1k}+\cdots+\wtil t_{jj}\dif\wtil
g_{jk}},&j<k.
\end{cases}
\end{eqnarray*}
The new variables are $\dif v^{(m)}_{jk} = \widehat v^{(m)}_{jk},
\dif t^{(m)}_{jk}=\widehat t^{(m)}_{jk}, \dif g^{(m)}_{jk}=\widehat
g^{(m)}_{jk}, j\geqslant k,m=1,2$. The matrices of partial
derivatives are easily seen to be the following:
\begin{eqnarray*}
&&\Br{\frac{\partial \widehat v^{(1)}_{jj}}{\partial \widehat
t_{jj}}}=\I,\quad \Br{\frac{\partial \widehat v^{(m)}_{jk}}{\partial
\widehat t^{(m)}_{jk}},j>k}=\I, \quad m=1,2,\\
&&\Br{\frac{\partial \widehat v^{(1)}_{kj}}{\partial \widehat
g^{(1)}_{kj}},j>k}=A,\quad \Br{\frac{\partial \widehat
v^{(2)}_{kj}}{\partial \widehat g^{(2)}_{kj}},j>k}=B
\end{eqnarray*}
where $A$ and $B$ are triangular matrices with $t_{jj}$ repeated
$n-j$ times,
\begin{eqnarray*}
\Br{\frac{\partial \widehat v^{(2)}_{jj}}{\partial \widehat
g^{(2)}_{jj}}}=\diag(t_{11},\ldots,t_{nn}).
\end{eqnarray*}
By using the above identity matrices one can wipe out other
submatrices in the same rows and columns and using the triangular
blocks one can wipe out other blocks below it when evaluating the
determinant of the Jacobian matrix and finally the determinant in
absolute value reduces to the form
\begin{eqnarray*}
\det(A)\det(B)t_{11}\cdots t_{nn} = \prod^n_{j=1}t^{2(n-j)+1}_{jj}.
\end{eqnarray*}
Hence the result. As a specific example, we consider the case where
$n=3$. We expand the expression: $\dif \wtil V=\dif \wtil T+\wtil
T\cdot\dif \wtil G$. That is,
\begin{eqnarray*}
\dif V_1 &=& \dif T_1 + T_1\dif \Pa{\re(\wtil G)} - T_2\dif \Pa{\im(\wtil G)},\\
\dif V_2 &=& \dif T_2 + T_2\dif \Pa{\re(\wtil G)} + T_1\dif\Pa{
\im(\wtil G)}.
\end{eqnarray*}
Furthermore,
\begin{eqnarray*}
&&\Br{\begin{array}{ccc}
      \dif v^{(1)}_{11} & \dif v^{(1)}_{12} & \dif v^{(1)}_{13} \\
      \dif v^{(1)}_{21} & \dif v^{(1)}_{22} & \dif v^{(1)}_{23} \\
      \dif v^{(1)}_{31} & \dif v^{(1)}_{32} & \dif v^{(1)}_{33}
    \end{array}
} \\
&&= \Br{\begin{array}{ccc}
      \dif t^{(1)}_{11} & 0 & 0 \\
      \dif t^{(1)}_{21} & \dif t^{(1)}_{22} & 0 \\
      \dif t^{(1)}_{31} & \dif t^{(1)}_{32} & \dif t^{(1)}_{33}
    \end{array}
}+\Br{\begin{array}{ccc}
       t^{(1)}_{11} & 0 & 0 \\
       t^{(1)}_{21} &  t^{(1)}_{22} & 0 \\
       t^{(1)}_{31} &  t^{(1)}_{32} &  t^{(1)}_{33}
    \end{array}
}\Br{\begin{array}{ccc}
      0 & \dif g^{(1)}_{12} & \dif g^{(1)}_{13} \\
      -\dif g^{(1)}_{12} & 0 & \dif g^{(1)}_{23} \\
      -\dif g^{(1)}_{13} & -\dif g^{(1)}_{23} & 0
    \end{array}
} \\
&&~~~~~~- \Br{\begin{array}{ccc}
       0 & 0 & 0 \\
       t^{(2)}_{21} &  0 & 0 \\
       t^{(2)}_{31} &  t^{(2)}_{32} &  0
    \end{array}
}\Br{\begin{array}{ccc}
      \dif g^{(2)}_{11} & \dif g^{(2)}_{12} & \dif g^{(2)}_{13} \\
      \dif g^{(2)}_{12} & \dif g^{(2)}_{22} & \dif g^{(2)}_{23} \\
      \dif g^{(2)}_{13} & \dif g^{(2)}_{23} & \dif g^{(2)}_{33}
    \end{array}
}
\end{eqnarray*}
and
\begin{eqnarray*}
&&\Br{\begin{array}{ccc}
      \dif v^{(2)}_{11} & \dif v^{(2)}_{12} & \dif v^{(2)}_{13} \\
      \dif v^{(2)}_{21} & \dif v^{(2)}_{22} & \dif v^{(2)}_{23} \\
      \dif v^{(2)}_{31} & \dif v^{(2)}_{32} & \dif v^{(2)}_{33}
    \end{array}
} \\
&&= \Br{\begin{array}{ccc}
      0 & 0 & 0 \\
      \dif t^{(2)}_{21} & 0 & 0 \\
      \dif t^{(2)}_{31} & \dif t^{(2)}_{32} & 0
    \end{array}
}+\Br{\begin{array}{ccc}
       0 & 0 & 0 \\
       t^{(2)}_{21} &  0 & 0 \\
       t^{(2)}_{31} &  t^{(2)}_{32} &  0
    \end{array}
}\Br{\begin{array}{ccc}
      0 & \dif g^{(1)}_{12} & \dif g^{(1)}_{13} \\
      -\dif g^{(1)}_{12} & 0 & \dif g^{(1)}_{23} \\
      -\dif g^{(1)}_{13} & -\dif g^{(1)}_{23} & 0
    \end{array}
} \\
&&~~~~~~+ \Br{\begin{array}{ccc}
       t^{(1)}_{11} & 0 & 0 \\
       t^{(1)}_{21} &  t^{(1)}_{22} & 0 \\
       t^{(1)}_{31} &  t^{(1)}_{32} &  t^{(1)}_{33}
    \end{array}
}\Br{\begin{array}{ccc}
      \dif g^{(2)}_{11} & \dif g^{(2)}_{12} & \dif g^{(2)}_{13} \\
      \dif g^{(2)}_{12} & \dif g^{(2)}_{22} & \dif g^{(2)}_{23} \\
      \dif g^{(2)}_{13} & \dif g^{(2)}_{23} & \dif g^{(2)}_{33}
    \end{array}
}.
\end{eqnarray*}
Thus
\begin{eqnarray*}
\dif v^{(1)}_{11} &=& \dif t_{11}, ~~\dif v^{(1)}_{12} = t_{11}\dif
g^{(1)}_{12},~~\dif v^{(1)}_{13} = t_{11}\dif g^{(1)}_{13},\\
\dif v^{(1)}_{21} &=& \dif t^{(1)}_{21} - t_{22}\dif g^{(1)}_{12} -
t^{(2)}_{21}\dif g^{(2)}_{11},~~\dif v^{(1)}_{22} = \dif t_{22} +
t^{(1)}_{21}\dif g^{(1)}_{12} - t^{(2)}_{21}\dif g^{(2)}_{12},\\
\dif v^{(1)}_{23}&=& t^{(1)}_{21}\dif g^{(1)}_{13} + t_{22}\dif
g^{(1)}_{23} - t^{(2)}_{21}\dif g^{(2)}_{13},\\
\dif v^{(1)}_{31}&=& \dif t^{(1)}_{31} -  t^{(1)}_{32}\dif
g^{(1)}_{12} - t_{33}\dif g^{(1)}_{13} - t^{(2)}_{31}\dif
g^{(2)}_{11} - t^{(2)}_{32}\dif g^{(2)}_{12},\\
\dif v^{(1)}_{32}&=& \dif t^{(1)}_{32} + t^{(1)}_{31}\dif
g^{(1)}_{12} - t_{33}\dif g^{(1)}_{23} - t^{(2)}_{31}\dif
g^{(2)}_{12} - t^{(2)}_{32}\dif g^{(2)}_{22},\\
\dif v^{(1)}_{33}&=& \dif t_{33} + t^{(1)}_{31}\dif g^{(1)}_{13} +
t^{(1)}_{32}\dif g^{(1)}_{23} - t^{(2)}_{31}\dif g^{(2)}_{13} -
t^{(2)}_{32}\dif g^{(2)}_{23}
\end{eqnarray*}
and
\begin{eqnarray*}
\dif v^{(2)}_{11} &=& t_{11} \dif g^{(2)}_{11}, ~~\dif v^{(2)}_{12}
= t_{11}\dif
g^{(2)}_{12},~~\dif v^{(2)}_{13} = t_{11}\dif g^{(2)}_{13},\\
\dif v^{(2)}_{21} &=& \dif t^{(2)}_{21} + t^{(1)}_{21}\dif
g^{(2)}_{11}+ t_{22}\dif g^{(2)}_{12}- t^{(2)}_{21}\dif
g^{(1)}_{12},\\
\dif v^{(2)}_{22} &=& t^{(2)}_{21}\dif g^{(1)}_{12}+
t^{(1)}_{21}\dif g^{(2)}_{12} + t_{22}\dif g^{(2)}_{22},\\
\dif v^{(2)}_{23}&=& t^{(2)}_{21}\dif g^{(1)}_{13}+
t^{(1)}_{21}\dif g^{(2)}_{13} + t_{22}\dif g^{(2)}_{23},\\
\dif v^{(2)}_{31}&=& \dif t^{(2)}_{31} -t^{(2)}_{32}\dif g^{(1)}_{12}
+ t^{(1)}_{31}\dif g^{(2)}_{11} + t^{(1)}_{32}\dif g^{(2)}_{12} + t_{33}\dif g^{(2)}_{13},\\
\dif v^{(2)}_{32}&=& \dif t^{(2)}_{32} + t^{(2)}_{31}\dif g^{(1)}_{12}+ t^{(1)}_{31}\dif g^{(2)}_{12} + t^{(1)}_{32}\dif g^{(2)}_{22} + t_{33}\dif g^{(2)}_{23},\\
\dif v^{(2)}_{33}&=&t^{(2)}_{31}\dif g^{(1)}_{13}+ t^{(2)}_{32}\dif
g^{(1)}_{23} + t^{(1)}_{31}\dif g^{(2)}_{13}+ t^{(1)}_{32}\dif
g^{(2)}_{23} + t_{33}\dif g^{(2)}_{33}.
\end{eqnarray*}
According to the definition, we now have $\dif\wtil v_{jk}=\dif
v^{(1)}_{jk}\wedge\dif v^{(2)}_{jk}$, then
\begin{eqnarray*}
[\dif \wtil V] &=& \dif \wtil v_{11}\wedge\dif \wtil
v_{12}\wedge\dif \wtil v_{13}\wedge\dif \wtil v_{21}\wedge\dif \wtil
v_{22}\wedge\dif \wtil v_{23}\wedge\dif \wtil v_{31}\wedge\dif \wtil
v_{32}\wedge\dif \wtil v_{33}\\
&=&\Pa{\prod^3_{j=1}t^{2(3-j)+1}_{jj}}[\dif \wtil T][\dif \wtil G].
\end{eqnarray*}
(ii). Let the diagonal elements of $\wtil U$ be real and all other
elements in $\wtil U$ and $\wtil T$ complex. Starting from
Eq.~\eqref{eq:cccc}, observing that $\dif\wtil G$ is $\dif\wtil G_1$
in this case with the wedge product $[\dif\wtil G_1]$, and taking
the variables $\dif\wtil v_{jk}$'s in the order $\dif v^{(1)}_{jk},
j\geqslant k, \dif v^{(1)}_{jk}, j< k, \dif v^{(2)}_{jk}, j\geqslant
k, \dif v^{(2)}_{jk}, j< k$ and the other variables in the order
$\dif t^{(1)}_{jk},j\geqslant k,\dif g^{(1)}_{jk}, j>k, \dif
t^{(2)}_{jk},j\geqslant k,\dif g^{(2)}_{jk},j>k$, we have the
following configuration in the Jacobian matrix:
\begin{eqnarray*}
\Br{\begin{array}{cccc}
      \I & * & * & * \\
      0 & -A_1 & 0 & A_2 \\
      0 & * & \I & * \\
      0 & -A_2 & 0 & -A_1
    \end{array}
}
\end{eqnarray*}
where the matrices marked by $*$ can be made null by operating with
the first and third column submatrices when taking the determinant.
Thus they can be taken as null matrices, and $A_1$ and $A_2$ are
triangular matrices with respectively $t^{(1)}_{jj}$ and
$t^{(2)}_{jj}$ repeated $n-j$ times in the diagonal. The Jacobian
matrix can be reduced to the form
\begin{eqnarray*}
\Br{\begin{array}{cc}
      A & B \\
      -B & A
    \end{array}
}, A=\Br{\begin{array}{cc}
      \I & 0 \\
      0 & -A_1
    \end{array}
}, B=\Br{\begin{array}{cc}
      0 & 0 \\
      0 & A_2
    \end{array}
}.
\end{eqnarray*}
Then the determinant is given by
\begin{eqnarray*}
\abs{\begin{array}{cc}
      A & B \\
      -B & A
    \end{array}} &=& \abs{\det((A+\sqrt{-1}B)(A+\sqrt{-1}B)^*)}\\
    &=& \abs{\det((-A_1+\sqrt{-1}A_2)(-A_1+\sqrt{-1}A_2)^*)}\\
    &=&\prod^p_{j=1}\abs{\wtil t_{jj}}^{2(n-j)}
\end{eqnarray*}
since $-A_1+\sqrt{-1}A_2$ is triangular with the diagonal elements
$-t^{(1)}_{jj}+\sqrt{-1}t^{(2)}_{jj}$ repeated $n-j$ times, giving
$\Pa{t^{(1)}_{jj}}^2+\Pa{t^{(2)}_{jj}}^2=\abs{\wtil t_{jj}}^2$
repeated $n-j$ times in the final determinant and hence the result.
\end{proof}

By using Proposition~\ref{prop:LU-complex-case} one can obtain
expressions for the integral over $\wtil U$ of $[\dif\wtil G]$ and
$[\dif\wtil G_1]$. These will be stated as corollaries here and the
proofs will be given after stating both the corollaries.
\begin{thrm}\label{th:vol-of-unitary-group}
Let $\dif\wtil G=\dif\wtil U\cdot \wtil U^*$, where $\wtil
U\in\cU(n)$. Then
\begin{framed}
\begin{eqnarray}
\vol\Pa{\cU(n)}=\int_{\cU(n)} [\dif\wtil G] =
\frac{2^n\pi^{n^2}}{\wtil \Gamma_n(n)}=\frac{
2^n\pi^{\frac{n(n+1)}2}}{1!2!\cdots (n-1)!}.
\end{eqnarray}
\end{framed}
\end{thrm}

\begin{proof}
Let $\wtil X$ be a $n\times n$ matrix of independent complex
variables. Let
\begin{eqnarray*}
B = \int_{\wtil X} [\dif\wtil X]e^{-\Tr{\wtil X\wtil X^*}} =
\int_{\wtil X} [\dif\wtil X]e^{-\sum_{j,k}\abs{\wtil x_{jk}}^2} =
\pi^{n^2}
\end{eqnarray*}
since
\begin{eqnarray*}
\int_{\wtil x_{jk}} e^{-\abs{\wtil x_{jk}}^2} \dif\wtil
x_{jk}=\int^{+\infty}_{-\infty}\int^{+\infty}_{-\infty}
e^{-\Pa{\Pa{x^{(1)}_{jk}}^2 + \Pa{x^{(2)}_{jk}}^2}}\dif
x^{(1)}_{jk}\dif x^{(2)}_{jk}=\pi.
\end{eqnarray*}
Consider the transformation used in
Proposition~\ref{prop:LU-complex-case} with $t_{jj}$'s real and
positive. Then
\begin{eqnarray*}
\Tr{\wtil X\wtil X^*} = \Tr{\wtil T\wtil T^*}
\end{eqnarray*}
and let
\begin{eqnarray*}
B = \int_{\wtil X} [\dif\wtil X]e^{-\Tr{\wtil X\wtil X^*}} =
\int_{\wtil T}\int_{\wtil U}
\Pa{\prod^n_{j=1}t^{2(n-j)+1}_{jj}}e^{-\sum_{j\geqslant k}\abs{\wtil
t_{jk}}^2}[\dif\wtil T][\dif\wtil G].
\end{eqnarray*}
But
\begin{eqnarray*}
\int^\infty_0 t^{2(n-j)+1}_{jj}e^{-t^2_{jj}}\dif t_{jj} =
\frac12\Gamma(n-j+1)~~~\text{for}~~~n-j+1>0
\end{eqnarray*}
and for $j>k$
\begin{eqnarray*}
\int_{\wtil t_{jk}}e^{-\abs{\wtil t_{jk}}^2}\dif\wtil t_{jk}=\pi.
\end{eqnarray*}
Then the integral over $\wtil T$ gives
\begin{eqnarray*}
2^{-n}\pi^{\frac{n(n-1)}2} \prod^n_{j=1}\Gamma(n-j+1)=2^{-n}\wtil
\Gamma_n(n).
\end{eqnarray*}
Hence
\begin{eqnarray*}
\int_{\cU(n)} [\dif\wtil G] = \frac{2^n\pi^{n^2}}{\wtil
\Gamma_n(n)}.
\end{eqnarray*}
\end{proof}

\begin{thrm}\label{thrm:vol-of-sub-unitary-group}
Let $\wtil U_1\in\cU(n)$ with the diagonal elements real. Let
$\dif\wtil G_1=\dif\wtil U_1\cdot \wtil U^*_1$. Let the full unitary
group of such $n\times n$ matrices $\wtil U_1$ be denoted by
$\cU_1(n)=\cU(n)/\cU(1)^{\times n}$. Then
\begin{eqnarray}
\vol\Pa{\cU_1(n)} = \vol\Pa{\cU(n)/\cU(1)^{\times n}}=
\int_{\cU_1(n)} [\dif\wtil G_1] = \frac{\pi^{n(n-1)}}{\wtil
\Gamma_n(n)}= \frac{ \pi^{\frac{n(n-1)}2}}{1!2!\cdots (n-1)!}.
\end{eqnarray}
\end{thrm}

\begin{proof}
Now consider the transformation used in
Proposition~\ref{prop:LU-complex-case} with all the elements in
$\wtil T$ complex and the diagonal elements of $\wtil U$ real. Then
\begin{eqnarray*}
B=\int_{\wtil T}\int_{\wtil U_1}
\Pa{\prod^n_{j=1}t^{2(n-j)}_{jj}}e^{-\sum_{j\geqslant k}\abs{\wtil
t_{jk}}^2}[\dif\wtil T][\dif\wtil G_1].
\end{eqnarray*}
Note that
\begin{eqnarray*}
\prod_{j>k} \Pa{\int_{\wtil t_{jk}}e^{-\abs{\wtil
t_{jk}}^2}\dif\wtil t_{jk}} = \pi^{\frac{n(n-1)}2}.
\end{eqnarray*}
Let $\wtil t_{jj}=\wtil t=t_1+\sqrt{-1}t_2$. Put $t_1=r\cos\theta$
and $t_2=r\sin\theta$. Let the integral over $\wtil t_{jj}$ be
denoted by $a_j$. Then
\begin{eqnarray*}
a_j &=& \int_{\wtil t} \abs{\wtil t}^{2(n-j)}e^{-\abs{\wtil
t}^2}\dif\wtil t =
\int^{+\infty}_{-\infty}\int^{+\infty}_{-\infty}(t^2_1+t^2_2)^{n-j}e^{-(t^2_1+t^2_2)}\dif t_1\dif t_2\\
&=&4
\int^{+\infty}_0\int^{+\infty}_0(t^2_1+t^2_2)^{n-j}e^{-(t^2_1+t^2_2)}\dif t_1\dif t_2\\
&=&4\int^{\frac\pi2}_{\theta=0}\int^\infty_{r=0} (r^2)^{n-j}\cdot
e^{-r^2}\cdot r\cdot\dif r\dif\theta\\
&=&\pi\Gamma(n-j+1)~~\text{for}~~n-j+1>0.
\end{eqnarray*}
Then
\begin{eqnarray*}
B = \pi^n\wtil \Gamma_n(n)\int_{\wtil U_1}[\dif\wtil G_1]=\pi^{n^2}
\end{eqnarray*}
implies that
\begin{eqnarray*}
\int_{\wtil U_1}[\dif\wtil G_1]=\frac{\pi^{n(n-1)}}{\wtil
\Gamma_n(n)},
\end{eqnarray*}
which establishes the result.
\end{proof}

\begin{exam}
Evaluate the integral
\begin{framed}
\begin{eqnarray}
\Delta(\alpha) = \int_{\wtil X} [\dif\wtil X]\abs{\det(\wtil X\wtil
X^*)}^\alpha \cdot e^{-\Tr{\wtil X\wtil X^*}} = \pi^{n^2}\frac{\wtil
\Gamma_n(\alpha+n)}{\wtil \Gamma_n(n)}
\end{eqnarray}
\end{framed}
\noindent for $\re(\alpha)>-1$, where $\wtil X\in\complex^{n\times
n}$ matrix of independent complex variables. Indeed, put $\wtil
X=\wtil T\wtil U$, where $\wtil U\in\cU(n)$, $\wtil T$ is lower
triangular with real distinct and positive diagonal elements. Then
\begin{eqnarray*}
\abs{\det(\wtil X\wtil X^*)}^\alpha &=&
\prod^n_{j=1}t^{2\alpha}_{jj},\\
\Tr{\wtil X\wtil X^*} &=& \Tr{\wtil T\wtil T^*} = \sum_{j\geqslant
k} \abs{\wtil t_{jk}}^2 =
\sum^n_{j=1}t^2_{jj}+\sum_{j>k}\abs{\wtil t_{jk}}^2,\\
~[\dif\wtil X] &=& \Pa{\prod^n_{j=1}t^{2(n-j)+1}_{jj}}[\dif\wtil
T][\dif\wtil G],\\
\int_{\wtil U}[\dif\wtil G] &=& \frac{2^n\pi^{n^2}}{\wtil
\Gamma_n(n)}.
\end{eqnarray*}
Thus
\begin{eqnarray*}
\Delta(\alpha) &=& \int_{\wtil T,\wtil U}[\dif\wtil T][\dif\wtil
G]\Pa{\prod^n_{j=1}t^{2\alpha+2(n-j)+1}e^{-t^2_{jj}}}e^{-\sum_{j>k}\abs{\wtil
t_{jk}}^2}\\
&=& \Pa{\prod^n_{j=1}\int^\infty_0
t^{2\alpha+2(n-j)+1}e^{-t^2_{jj}}\dif t_{jj}}
\cdot\Pa{\prod_{j>k}\int^{+\infty}_{-\infty} e^{-\abs{\wtil
t_{jk}}^2}\dif\wtil t_{jk}}\cdot\int_{\wtil
U}[\dif\wtil G]\\
&=& \Pa{2^{-n}\prod^n_{j=1}\Gamma(\alpha+n-j+1)}\cdot
\pi^{\frac{n(n-1)}2}\cdot\frac{2^n\pi^{n^2}}{\wtil \Gamma_n(n)}
\end{eqnarray*}
for $\re(\alpha+n)>n-1$ or $\re(\alpha)>-1$. That is,
\begin{eqnarray*}
\Delta(\alpha) = \pi^{n^2}\frac{\wtil \Gamma_n(\alpha+n)}{\wtil
\Gamma_n(pn)}~~\text{for}~~\re(\alpha)>-1.
\end{eqnarray*}
\end{exam}

\begin{prop}
Let $\wtil X\in\complex^{n\times n}$ be a hermitian matrix of
independent complex entries with real distinct eigenvalues
$\lambda_1>\lambda_2>\cdots>\lambda_n$. Let $\wtil U\in\cU(n)$ with
real diagonal entries and let $\wtil X = \wtil UD\wtil U^*$, where
$D=\diag(\lambda_1,\ldots,\lambda_n)$. Then
\begin{eqnarray}\label{eq:dif-psd}
\fbox{$[\dif\wtil X] =
\Pa{\prod_{j>k}\abs{\lambda_k-\lambda_j}^2}\cdot[\dif D][\dif\wtil
G_1]$,}
\end{eqnarray}
where $\dif\wtil G_1=\wtil U^*\cdot\dif\wtil U$.
\end{prop}

\begin{proof}
Take the differentials in $\wtil X=\wtil UD\wtil U^*$ to get
\begin{eqnarray*}
\dif\wtil X = \dif\wtil U\cdot D\cdot\wtil U^* + \wtil U\cdot\dif
D\cdot \wtil U^*+ \wtil U\cdot D\cdot \dif\wtil U^*.
\end{eqnarray*}
Premultiply by $\wtil U^*$, postmultiply by $\wtil U$ and observe
that $\dif\wtil G_1$ is skew hermitian. Then one has
\begin{eqnarray*}
\dif\wtil W = \dif\wtil G_1\cdot D + \dif D - D\cdot \dif\wtil G_1
\end{eqnarray*}
where $\dif\wtil W  = \wtil U^*\cdot\dif\wtil X\wtil U$ with
$[\dif\wtil W]=[\dif\wtil X]$. Using the same steps as in the proof
of Proposition~\ref{prop:spetral-decom-complex}, we have
\begin{eqnarray*}
[\dif\wtil W] = \Pa{\prod_{j>k}\abs{\lambda_k-\lambda_j}^2}[\dif
D][\dif\wtil G_1].
\end{eqnarray*}
Hence the result follows.
\end{proof}

\begin{exam}\label{exam:two-integrals}
Let $D=\diag(\lambda_1,\ldots,\lambda_n)$ where the $\lambda_j$'s
are real distinct and positive or let
$\lambda_1>\cdots>\lambda_n>0$. Show that
\begin{enumerate}[(i)]
\item
\begin{eqnarray*}
\int_{\lambda_1>\cdots>\lambda_n>0} [\dif
D]\Pa{\prod_{j>k}\abs{\lambda_k-\lambda_j}^2}e^{-\Tr{D}} =
\frac{\Pa{\wtil
\Gamma_n(n)}^2}{\pi^{n(n-1)}}=\Pa{\prod^n_{j=1}\Gamma(j)}^2;
\end{eqnarray*}
\item \begin{eqnarray*}
\int_{\lambda_1>\cdots>\lambda_n>0} [\dif
D]\Pa{\prod_{j>k}\abs{\lambda_k-\lambda_j}^2}\cdot
\Pa{\prod^n_{j=1}\abs{\lambda_j}^{\alpha-n}}\cdot e^{-\Tr{D}} =
\frac{\wtil \Gamma_n(\alpha)\wtil \Gamma_n(n)}{\pi^{n(n-1)}}.
\end{eqnarray*}
\end{enumerate}
In fact, let $\wtil Y$ be a $n\times n$ hermitian positive definite
matrix, $\wtil U$ a unitary matrix with real diagonal elements such
that
\begin{eqnarray*}
\wtil U^* \wtil Y\wtil U = D=\diag(\lambda_1,\ldots,\lambda_n).
\end{eqnarray*}
From the matrix-variate gamma integral
\begin{eqnarray*}
\wtil \Gamma_n(\alpha) = \int_{\wtil Y=\wtil Y^*>0}[\dif\wtil Y]
\abs{\det(\wtil Y)}^{\alpha-n}\cdot e^{-\Tr{\wtil
Y}}~~\text{for}~~\re(\alpha)>n-1.
\end{eqnarray*}
Hence
\begin{eqnarray*}
\wtil \Gamma_n(n) = \int_{\wtil Y=\wtil Y^*>0}[\dif\wtil Y]
e^{-\Tr{\wtil Y}}.
\end{eqnarray*}
Then $\wtil Y=\wtil UD\wtil U^*$ implies that
\begin{eqnarray*}
[\dif\wtil Y]=\Pa{\prod_{j>k}\abs{\lambda_k-\lambda_j}^2}[\dif
D][\dif\wtil G_1],\quad \dif\wtil G_1=\wtil U^*\cdot \dif\wtil U.
\end{eqnarray*}
Note that
\begin{eqnarray*}
\Tr{\wtil Y} &=& \Tr{\wtil UD\wtil U^*} = \Tr{D},\\
\abs{\det(\wtil Y)}^{\alpha-n}&=&
\Pa{\prod^n_{j=1}\abs{\lambda_j}}^{\alpha-n}.
\end{eqnarray*}
Then
\begin{eqnarray*}
\wtil \Gamma_n(n) = \int_{\lambda_1>\cdots>\lambda_n>0} [\dif
D]\Pa{\prod_{j>k}\abs{\lambda_k-\lambda_j}^2}e^{-\Tr{D}}\times
\int_{\wtil U} [\dif\wtil G_1].
\end{eqnarray*}
Clearly
\begin{eqnarray*}
\int_{\wtil U} [\dif\wtil G_1] = \frac{\pi^{n(n-1)}}{\wtil
\Gamma_n(n)}.
\end{eqnarray*}
Substituting this, results (i) and (ii) follow.
\end{exam}

\begin{remark}
We can try to compute the following integral in \eqref{eq:dif-psd}:
\begin{eqnarray*}
\int_{\wtil X:\Tr{\wtil X}=1} [\dif\wtil X] &=&
\int_{\lambda_1>\lambda_2>\cdots>\lambda_n>0}\delta\Pa{1-\sum^n_{j=1}\lambda_j}\prod_{i<j}(\lambda_i-\lambda_j)^2\prod^n_{j=1}\dif\lambda_j\times
\int_{\cU_1(n)}[\dif\wtil G_1]
\end{eqnarray*}
which is equivalent to the following
\begin{eqnarray*}
\int_{\wtil X:\Tr{\wtil X}=1} [\dif\wtil X] &=&
\frac1{n!}\int\delta\Pa{1-\sum^n_{j=1}\lambda_j}\prod_{i<j}(\lambda_i-\lambda_j)^2\prod^n_{j=1}\dif\lambda_j\times
\int_{\cU_1(n)}[\dif\wtil G_1]\\
&=&
\frac1{n!}\frac{\prod^{n-1}_{j=0}\Gamma(n-j)\Gamma(n-j+1)}{\Gamma(n^2)}\times
\frac{\pi^{\frac{n(n-1)}2}}{\prod^n_{j=1}\Gamma(j)}\\
&=&\frac1{n!}\pi^{\frac{n(n-1)}2}\frac{\Gamma(1)\Gamma(2)\cdots\Gamma(n+1)}{\Gamma(n^2)},
\end{eqnarray*}
where we used the integral formula:
\begin{eqnarray*}
\int\delta\Pa{1-\sum^n_{j=1}\lambda_j}\prod_{i<j}(\lambda_i-\lambda_j)^2\prod^n_{j=1}\dif\lambda_j=
\frac{\prod^{n-1}_{j=0}\Gamma(n-j)\Gamma(n-j+1)}{\Gamma(n^2)}.
\end{eqnarray*}
This means that
\begin{framed}
\begin{eqnarray}\label{eq:vol-formula-of-state}
\vol\Pa{\density{\complex^n}} =
\pi^{\frac{n(n-1)}2}\frac{\Gamma(1)\Gamma(2)\cdots\Gamma(n)}{\Gamma(n^2)}.
\end{eqnarray}
\end{framed}
We make some remarks here: to obtain the volume of the set of mixed
states acting on $\complex^n$, one has to integrate the volume
element $[\dif \wtil X]$. By definition, the first integral gives
$\frac1{n!}\frac1{C^{\rH\rS}_n}$, where $C^{\rH\rS}_n=C^{(1,2)}_n$
and
\begin{eqnarray*}
\frac1{C^{(\alpha,\beta)}_n}
&=&\overbrace{\int^\infty_0\cdots\int^\infty_0}^n\delta\Pa{1-\sum^n_{j=1}\lambda_j}\prod^n_{k}\lambda^{\alpha-1}_k\prod_{i<j}(\lambda_i-\lambda_j)^\beta\prod^n_{j=1}\dif\lambda_j\\
&=&\frac1{\Gamma\Pa{\alpha
n+\beta\frac{n(n-1)}2}}\prod^n_{j=1}\frac{\Gamma\Pa{1+j\frac{\beta}2}\Gamma\Pa{\alpha+(j-1)\frac{\beta}2}}{\Gamma\Pa{1+\frac{\beta}2}},
\end{eqnarray*}
while the second is equal to the volume of the flag manifold. To
make the diagonalization unique, one has to restrict to a certain
order of eigenvalues, say $\lambda_1>\lambda_2>\cdots>\lambda_n$ (a
generic density matrix is not degenerate), which corresponds to a
choice of a certain Weyl chamber of the eigenvalue simplex
$\Delta_{n-1}$. In other words, different permutations of the vector
of $n$ generically different permutations (Weyl chambers) equals to
$n!$. This is why the factor $\frac1{n!}$ appears in the right hand
side in the above identity. In summary, the transformation
$$
\wtil
X\mapsto (D,\wtil U)
$$
such that $\wtil X=\wtil UD\wtil U^*$, is one-to-one if and only if
$D=\diag(\lambda_1,\ldots,\lambda_n)$, where
$\lambda_1>\cdots>\lambda_n$ and $\wtil U\in \cU(n)/\cU(1)^{\times
n}$.
\end{remark}

\begin{remark}
In this remark, we will discuss the connection between two integrals
\cite{Santosh2011}: for $\alpha,\beta>0$,
\begin{eqnarray*}
\cI^{(1)}_n(\alpha, \beta)
&=&\overbrace{\int^\infty_0\cdots\int^\infty_0}^n\delta\Pa{1-\sum^n_{j=1}\lambda_j}
\prod^n_{k=1}\lambda^{\alpha-1}_k\prod_{1\leqslant i<j\leqslant n}(\lambda_i-\lambda_j)^\beta\prod^n_{j=1}\dif\lambda_j,\\
\cI^{(\mathrm{w})}_n(\alpha, \beta)
&=&\overbrace{\int^\infty_0\cdots\int^\infty_0}^n
\exp\Pa{-\sum^n_{j=1}x_j}\prod^n_{k=1}x^{\alpha-1}_k\prod_{1\leqslant
i<j\leqslant n}(x_i-x_j)^\beta\prod^n_{j=1}\dif x_j.
\end{eqnarray*}
We introduce an auxiliary variable $t$ in the expression of
$\cI^{(1)}_n(\alpha,\beta)$, and define $I(t)$ as
\begin{eqnarray}
I(t)=\overbrace{\int^\infty_0\cdots\int^\infty_0}^n\delta\Pa{t-\sum^n_{j=1}\lambda_j}
\prod^n_{k=1}\lambda^{\alpha-1}_k\prod_{1\leqslant i<j\leqslant
n}(\lambda_i-\lambda_j)^\beta\prod^n_{j=1}\dif\lambda_j.
\end{eqnarray}
Then $I(1)=\cI^{(1)}_n(\alpha,\beta)$. Taking the Laplace transform,
denoted by $\sL$, of $I(t)$, we obtain
\begin{eqnarray}
\sL(I)&:=&
\overbrace{\int^\infty_0\cdots\int^\infty_0}^n\Br{\int\delta\Pa{t-\sum^n_{j=1}\lambda_j}e^{-st}\dif
t} \prod^n_{k=1}\lambda^{\alpha-1}_k\prod_{1\leqslant i<j\leqslant
n}(\lambda_i-\lambda_j)^\beta\prod^n_{j=1}\dif\lambda_j\\
&=&\overbrace{\int^\infty_0\cdots\int^\infty_0}^n
\exp\Pa{-s\sum^n_{j=1}\lambda_j}
\prod^n_{k=1}\lambda^{\alpha-1}_k\prod_{1\leqslant i<j\leqslant
n}(\lambda_i-\lambda_j)^\beta\prod^n_{j=1}\dif\lambda_j\\
&=&s^{-\alpha
n-\beta\binom{n}{2}}\overbrace{\int^\infty_0\cdots\int^\infty_0}^n
\exp\Pa{-\sum^n_{j=1}x_j}
\prod^n_{k=1}x^{\alpha-1}_k\prod_{1\leqslant i<j\leqslant
n}(x_i-x_j)^\beta\prod^n_{j=1}\dif x_j,
\end{eqnarray}
which means that
\begin{eqnarray}
\wtil I(s) := \sL(I) = s^{-\alpha
n-\beta\binom{n}{2}}\cdot\cI^{(\mathrm{w})}(\alpha,\beta).
\end{eqnarray}
Therefore
\begin{eqnarray}
I(t) = \sL^{-1}(\sL(I)) = \sL^{-1}(\wtil I) = \frac{t^{\alpha
n+\beta\binom{n}{2}-1}}{\Gamma\Pa{\alpha
n+\beta\binom{n}{2}}}\cdot\cI^{(\mathrm{w})}(\alpha,\beta).
\end{eqnarray}
Letting $t=1$ gives the conclusion:
\begin{eqnarray}
\cI^{(1)}_n(\alpha,\beta) = \frac{1}{\Gamma\Pa{\alpha
n+\beta\binom{n}{2}}}\cdot\cI^{(\mathrm{w})}(\alpha,\beta).
\end{eqnarray}
In order to calculate the integral $\cI^{(1)}_n(\alpha,\beta)$, it
suffices to calculate the integral
$\cI^{(\mathrm{w})}(\alpha,\beta)$, which is derived in
Corollary~\ref{cor:Lag-int} via Selberg's integral (See Appendix):
\begin{eqnarray}
\cI^{(\mathrm{w})}(\alpha,\beta)=\prod^n_{j=1}\frac{\Gamma\Pa{1+j\frac{\beta}2}\Gamma\Pa{\alpha+(j-1)\frac{\beta}2}}{\Gamma\Pa{1+\frac{\beta}2}}.
\end{eqnarray}
\end{remark}

Matrix integrals, especially over unitary groups, are very
important. We recently present some results of this aspect
\cite{Zhang1}. Generally speaking, the computation of matrix
integrals is very difficult from the first principle. Thus
frequently we need to perform variable substitution in computing
integrals. The first step in substitution is to compute Jacobians of
this transformation, this is what we present. We also apply the
matrix integrals over unitary groups to a problem
\cite{Zhang2,Zhang3} in quantum information theory.

\section{The volume of a compact Lie group}

The content of this section is mainly from \cite{Macdonald1980}.
Some missing details are provided. Note that Macdonald's result
presents unifying treatment for computing the volume of a compact
Lie group. Before proceeding, it is necessary to recall the notion
of root system and its related properties.

\begin{definition}
A root system $(E,R)$ is a finite-dimensional \blue{real} vector
space $E$ with an inner product $\Inner{\cdot}{\cdot}$, together
with a finite collection $R$ of nonzero vectors in $E$ satisfying
the following properties:
\begin{enumerate}[(i)]
\item $\spn \set{R}=E$.
\item If $\alpha\in R$, then $L_{\alpha}\cap R=\set{\pm\alpha}$, where $L_\alpha=\set{r\alpha:
r\in\real}$.
\item If $\alpha,\beta\in R$, then $s_\alpha(\beta)\in R$, where $s_\alpha$ is the linear transformation of $E$ defined by
$$
s_\alpha(x)=\beta-2\frac{\inner{x}{\alpha}}{\inner{\alpha}{\alpha}}\alpha,\quad
x\in E.
$$
\item For all $\alpha,\beta\in R$,
$2\frac{\inner{\beta}{\alpha}}{\inner{\alpha}{\alpha}}\in\integer$.
\end{enumerate}
The dimension of $E$ is called the \emph{rank} of the root system
and the elements of $R$ are called \emph{roots}.
\end{definition}

\begin{definition}
If $(E,R)$ is a root system, the \emph{Weyl group} $W$ of $R$ is the
subgroup of the orthogonal group of $E$ generated by the reflections
$s_\alpha(\alpha\in R)$. That is, for arbitrary positive integer
$k\in\natural$ and any non-negative integer $n_1,\ldots,n_k$,
$s^{n_1}_{\alpha_1}\cdots s^{n_k}_{\alpha_k}\in W$, i.e.,
$$
W = \Set{s^{n_1}_{\alpha_1}\cdots s^{n_k}_{\alpha_k}:
s_{\alpha_j}\in R\text{ for }j=1,\ldots,k\in \natural}.
$$
\end{definition}

\begin{definition}
If $(E,R)$ is a root system, then for each root $\alpha\in R$, the
\emph{co-root} $\alpha^\vee$ is the vector given by
$$
\alpha^\vee := 2\frac{\alpha}{\inner{\alpha}{\alpha}}.
$$
The set of all co-roots is denoted $R^\vee$ and is called the
\emph{dual root system} to $R$.
\end{definition}
It is easily seen that
$\Inner{\beta}{\alpha^\vee}=2\frac{\Inner{\beta}{\alpha}}{\inner{\beta}{\beta}}\in\integer$
for $\alpha,\beta\in R$. Besides, we see that
$$
\Inner{\alpha^\vee}{\alpha^\vee} = \frac4{\inner{\alpha}{\alpha}}.
$$
In fact, if $R$ is a rooty system, then $R^\vee$ is also a root
system and the Weyl group for $R^\vee$ is the same as the Weyl group
for $R$. Furthermore, $R^{\vee\vee}=R$. Thus
$$
\alpha = 2\frac{\alpha^\vee}{\inner{\alpha^\vee}{\alpha^\vee}}.
$$
\begin{definition}
If $(E,R)$ is a root system, a subset $\Delta$ of $R$ is called a
\emph{base} if the following conditions hold:
\begin{enumerate}[(i)]
\item $\Delta$ is a basis for $E$ as a vector space.
\item Each root $\alpha\in R$ can be expressed as a linear combination of elements
of $\Delta$ with integer coefficients and in such a way that the
coefficients are either all non-negative or all non-positive.
\end{enumerate}
The roots for which the coefficients are non-negative are called
\emph{positive roots} and the others are called \emph{negative
roots} (relative to the base $\Delta$). The set of positive roots
relative to a fixed base $\Delta$ is denoted by $R^+$ and the set of
negative roots is denote by $R^-$. Thus $R=R^+\sqcup R^-$. The
elements of $\Delta$ are called the \emph{positive simple roots}.
\end{definition}
Note that if $\Delta$ is a base for $R$, then the set of all
co-roots $\alpha^\vee(\alpha\in\Delta)$ is a base for the dual root
system $R^\vee$. We also see that if $\Delta$ is a base, then $W$ is
generated by the reflections $s_\alpha$ with $\alpha\in \Delta$.

\begin{definition}
An element $\mu$ of $E$ is an \emph{integral element} if for all
$\alpha\in R$, $\Inner{\mu}{\alpha^\vee}\in\integer$. If $\Delta$ is
a base for $R$, an element $\mu$ of $E$ is \emph{dominant} (relative
to $\Delta$) if $\Inner{\mu}{\alpha}\geqslant0$ for all
$\alpha\in\Delta$ and strictly dominant if $\Inner{\mu}{\alpha}>0$
for all $\alpha\in\Delta$.
\end{definition}

\begin{definition}
Let $\Delta=\Set{\alpha_1,\ldots,\alpha_n}$ be a base. Then the
\emph{fundamental weights} (relative to $\Delta$) are the elements
$\mu_1,\ldots,\mu_n$ with the property that
$$
\Inner{\mu_i}{\alpha^\vee_j}=\delta_{ij},\quad i,j=1,\ldots,n.
$$
That is, $\set{\mu_1,\ldots,\mu_n}$ can be viewed as the dual base
of $\Delta^\vee:=\Set{\alpha^\vee_1,\ldots,\alpha^\vee_n}$.
\end{definition}

\begin{definition}
Let $\Delta=\Set{\alpha_1,\ldots,\alpha_n}$ be a base for $R$ and
$R^+$ the associated set of positive roots. We then let $\rho$
denote \emph{half the sum of the positive roots} and let $\sigma$
denote \emph{half the sum of the positive co-roots}:
$$
\rho = \frac12\sum_{\alpha\in R^+}\alpha,\quad \sigma
=\frac12\sum_{\alpha\in R^+}\alpha^\vee.
$$
They are often called the \emph{Weyl vectors}.
\end{definition}
One important result for Weyl vector is that $\rho$ are strictly
dominant integral element; indeed,
\begin{eqnarray}
\Inner{\rho}{\alpha^\vee}=1,\quad \forall \alpha^\vee\in
\Delta^\vee.
\end{eqnarray}
Similarly,
\begin{eqnarray}
\Inner{\sigma}{\alpha}=1,\quad \forall \alpha\in \Delta.
\end{eqnarray}\label{eq:weyl-vec-basic-weight}
It is well-known that the Weyl vector can be expressed as the sum of
fundamental weights:
\begin{eqnarray}
\rho = \sum^n_{j=1}\mu_j.
\end{eqnarray}

\begin{definition}
Let $\Delta=\set{\alpha_1,\ldots,\alpha_n}$ be a base for $R$. For a
positive root $\alpha\in R^+$, the \emph{height} of $\alpha$ is
defined by
\begin{eqnarray*}
\mathrm{ht}(\alpha) = \sum^n_{j=1}k_j\in\integer_+,
\end{eqnarray*}
where $\alpha = \sum^n_{j=1}k_j\alpha_j$ for non-negative integers
$k_1,\ldots,k_n$. That is,
\begin{eqnarray*}
\mathrm{ht}\Pa{\sum^n_{j=1}k_j\alpha_j} =
\sum^n_{j=1}k_j\in\integer_+.
\end{eqnarray*}
Apparently, $\mathrm{ht}(\alpha)=1$ if and only if
$\alpha\in\Delta$.
\end{definition}
With these notations, we can show that
$\Inner{\rho}{\alpha^\vee}=\mathrm{ht}(\alpha^\vee)$ for any
$\alpha\in R^+$. Note that the base
$\Delta=\set{\alpha_1,\ldots,\alpha_n}$ for the root system $R$ is
given, then $\Delta^\vee=\set{\alpha^\vee_1,\ldots,\alpha^\vee_n}$
is known for the dual root system $R^\vee$. Now for any $\alpha\in
R^+$, $\alpha^\vee$ is also positive co-root and it is can be
written as
$$
\alpha^\vee = \sum^n_{j=1}k_j\alpha^\vee_j,\quad k_j\in\natural,
j=1,\ldots,n.
$$
By using \eqref{eq:weyl-vec-basic-weight}, we see that
\begin{eqnarray*}
\Inner{\rho}{\alpha^\vee} =
\Inner{\sum^n_{i=1}\mu_i}{\sum^n_{j=1}k_j\alpha^\vee_j} =
\sum^n_{i,j=1}k_j\Inner{\mu_i}{\alpha^\vee_j} =
\sum^n_{i,j=1}k_j\delta_{ij} =
\sum^n_{j=1}k_j=\mathrm{ht}(\alpha^\vee).
\end{eqnarray*}
Similarly, $\Inner{\sigma}{\alpha}=\mathrm{ht}(\alpha)$ for
$\alpha\in R^+$.

We also have to recall the following polynomial function
$P:E\to\real$ given by
$$
P(X) = \prod_{\alpha\in R^+}\Inner{\alpha}{X}.
$$
Note that $P$ has an important property:
$$
P(w\cdot X) = \det(w)P(X), \quad\forall w\in W, \forall X\in E.
$$
Indeed, for any $w\in W$, we have
\begin{eqnarray*}
P(w\cdot X) =\prod_{\alpha\in R^+}\Inner{\alpha}{w\cdot X} =
\prod_{\alpha\in R^+}\Inner{w^{-1}\cdot\alpha}{X}.
\end{eqnarray*}
Suppose first that $w=w^{-1}=s_\beta$, where $\beta$ is a positive
simple root. Since $s_\beta$ permutes the positive roots different
from $\beta$, whereas $s_\beta\cdot\beta=-\beta$. Thus
\begin{eqnarray*}
P(s_\beta\cdot X) =\prod_{\alpha\in
R^+\backslash\set{\beta}}\Inner{\alpha}{X}\times \Inner{-\beta}{X} =
-\prod_{\alpha\in R^+}\Inner{\alpha}{X}.
\end{eqnarray*}
This because the determinant of a reflection is $-1$. Note $W$ is
generated by all positive simple roots $\alpha$. Thus for $w\in W$,
it can be written as the product of some positive simple roots:
$w=s_{\alpha_1}\cdots s_{\alpha_k}$. Then
\begin{eqnarray*}
P(w\cdot X) &=& P(s_{\alpha_1}\cdots s_{\alpha_k}\cdot X) =
-P(s_{\alpha_1}\cdots s_{\alpha_{k-1}}\cdot X) \\
&=& \cdots\\
&=&(-1)^kP(X)=\det(s_{\alpha_1}\cdots s_{\alpha_k})P(X)\\
&=&\det(w)P(X).
\end{eqnarray*}
This means that $P$ is alternating, or skew-symmetric.

Next we can start our focus, i.e., computing the volume of a compact
Lie group. Let $G$ be a compact Lie group and let $\g$ be its Lie
algebra, thought of as the tangent space $\rT_eG$ to $G$ at the
identity element $e$. Choose a \blue{Lebesgue measure} $\mu_L$ on
the vector space $\g$. By means of a chart of $G$ at $e$, we can
construct from $\mu_L$ a translation-invariant measure on a
neighborhood of $e$ in $G$, and then we can extend this by
translation in $G$ to a \blue{Haar measure} $\mu_{\mathrm{Haar}}$ on
$G$. The purpose of this section is to establish a formula
(Eq.~\eqref{eq:A1} below), obtained by Macdonald in 1980, for
$\mu_{\mathrm{Haar}}(G)$, the volume of $G$ relative to the measure
$\mu_{\mathrm{Haar}}$, as a function of $\mu_L$.

There are two ingredients in the formula. Firstly, from a suitably
chosen \emph{Chevalley basis} (i.e. root vectors) of the
complexification of $\g$, we can construct an "integer lattice"
$\g_\integer$, which is a lattice in $\g$ and a Lie algebra over
$\integer$. By abuse of notation, let $\mu_L(\g/\g_\integer)$ denote
\red{the volume (with respect to $\mu_L$) of a fundamental
parallelepiped} for $\g_\integer$ in $\g$. Secondly, it is
well-known that the manifold $G$ has the same cohomology, apart from
torsion, as a product of odd-dimensional spheres, say of dimensions
$r_1,\ldots,r_n$. We shall prove the following result:
\begin{thrm}[\cite{Macdonald1980}]
It holds that
\begin{eqnarray}\label{eq:A1}
\mu_{\mathrm{Haar}}(G) =
\mu_L(\g/\g_\integer)\prod^n_{j=1}\vol\Pa{\mathbb{S}^{r_j}}
\end{eqnarray}
where $\vol\Pa{\mathbb{S}^{r_j}}$ is the superficial measure of the
unit sphere $\mathbb{S}^{r_j}$ in $\real^{r_j+1}$, that is (since
$r_j=2m_j+1$ is odd)
$\vol\Pa{\mathbb{S}^{r_j}}=\vol\Pa{\mathbb{S}^{2m_j+1}}=\frac{2\pi^{m_j+1}}{\Gamma(m_j+1)}$.
By abuse of notation, we record the above fact as
$$
\vol(G) =
\vol(\g/\g_\integer)\prod^n_{j=1}\vol\Pa{\mathbb{S}^{r_j}}.
$$
\end{thrm}

\begin{proof}
Let $T$ be a maximal torus in $G$ and let $\liet\subset\g$ be its
Lie algebra. Let $d=\dim(G), n=\dim(T)$. The Lebesgue measure
$\mu_L$ on $\g$ determines a Lebesgue measure (also denoted by
$\mu_L$) and hence a Haar measure (also denoted by
$\mu_{\mathrm{Haar}}$) on $T$\footnote{Note: an inner product on
$\g$ does determine Lebesgue measures on both $\g$ and $\liet$.}.
Let $\liet_\integer$ be the lattice in $\liet$ such that the kernel
of $\exp: \liet \to T$ is $\ker\exp=2\pi\liet_\integer$, then
clearly
\begin{eqnarray}
\mu_{\mathrm{Haar}}(T)=\mu_L(\liet/2\pi\liet_\integer)=(2\pi)^n\mu_L(\liet/\liet_\integer),\\
\vol(T)=\vol(\liet/2\pi\liet_\integer)=(2\pi)^n\vol(\liet/\liet_\integer).\label{eq:ration-vol}
\end{eqnarray}
Let $R$ be the set of roots of $G$ relative to $T$, and let $W$ be
the Weyl group. The roots are \red{real} linear forms on $\liet$,
integer-valued on the lattice $\liet_\integer$. \red{Fix a system of
positive roots, and let $P = \prod_{\alpha>0}\alpha$, a homogeneous
polynomial function on $\liet$ of degree $N=\frac12(d-n)$, i.e., the
number of all positive roots.}

Let $\Inner{\xi}{\eta}$ be a positive definite inner product on $\g$
which is invariant under the adjoint action of $G$ and such that the
cube generated by an orthonormal basis of $\g$ has \textbf{unit
volume} relative to $\mu_L$. Let
$\norm{\xi}=\sqrt{\Inner{\xi}{\xi}}$.

If $\phi$ is a suitable $G$-invariant function on $\g$, we have
\cite{Duistermaat2000}
\begin{eqnarray}
\int_\g \phi(\xi)\dif\mu_L(\xi) =
\frac{\vol(G/T)}{\abs{W}}\int_\liet
\phi(\tau)P(\tau)^2\dif\mu_L(\tau)
\end{eqnarray}
be the counterpart for $\g$ of Weyl's integration formula. By taking
$\phi(\xi) = \exp\Pa{-\frac12\Inner{\xi}{\xi}}$, we obtain
\begin{eqnarray}\label{eq:vital-equation}
\sqrt{(2\pi)^d} =\int_\g
\exp\Pa{-\frac12\Inner{\xi}{\xi}}\dif\mu_L(\xi)=
\frac1{\abs{W}}\frac{\vol(G)}{\vol(T)}\int_\liet
\exp\Pa{-\frac12\Inner{\tau}{\tau}}P(\tau)^2\dif\mu_L(\tau)
\end{eqnarray}
To calculate this integral, we proceed as follows. Let
$\Inner{x}{y}=\sum_jx_jy_j$ be the usual inner product on $\real^n$.
This extends to a scalar product on the algebra $S(\real^n)$ of
polynomial functions on $\real^n$, such that
$\Inner{x^\alpha}{x^\beta}=\alpha!\delta_{\alpha\beta}$ for any two
multi-indices
$\alpha=(\alpha_1,\ldots,\alpha_n),\beta=(\beta_1,\ldots,\beta_n)\in\natural^n$.
Here $\alpha!:=\prod^n_{j=1}\alpha_j!$ and
$\delta_{\alpha\beta}:=\prod^n_{j=1}\delta_{\alpha_j\beta_j}$.
Indeed, such scalar product is equivalently defined
\cite{Iwasaki1997} by
$$
\Inner{p}{q}_\partial := p(\partial)q|_{x=0},
$$
where $p,q\in\real[x_1,\ldots,x_n]$ are homogeneous polynomials of
the same order and $p(\partial)$ means that $x_j$ is replaced by
$\partial_j$ in $p(x)\equiv p(x_1,\ldots,x_n)$. For instance, a
generic polynomial in $\real[x_1,\ldots,x_n]$ can be written as
$$
p(x) = \sum_{\alpha} c_\alpha x^\alpha,\quad x^\alpha:=\prod^n_{j=1}
x^{\alpha_j}_j.
$$
Thus we know that $\left.\frac{\partial^{\alpha_j}}{\partial
x_j^{\alpha_j}}x^{\beta_j}_j\right|_{x_j=0}=\alpha_j!\delta_{\alpha_j\beta_j}$,
and
\begin{eqnarray*}
\Inner{x^\alpha}{x^\beta} =\left.\Pa{
\prod^n_{j=1}\frac{\partial^{\alpha_j}}{\partial
x_j^{\alpha_j}}}\Pa{\prod^n_{j=1}
x^{\beta_j}_j}\right|_{(x_1,\ldots,x_n)=0}
=\prod^n_{j=1}\frac{\partial^{\alpha_j}}{\partial
x_j^{\alpha_j}}x^{\beta_j}_j
=\prod^n_{j=1}\alpha_j!\delta_{\alpha_j\beta_j}=\alpha!\delta_{\alpha\beta}.
\end{eqnarray*}
Let $\gamma$ be the (Gaussian) measure on $\real^n$ defined by
$$
\dif\gamma(x) =
\frac1{\sqrt{(2\pi)^n}}\exp\Pa{-\frac12\Inner{x}{x}}[\dif x]
$$
where $[\dif x]$ is Lebesgue measure. For a function $f$ on
$\real^n$, let $f^*=f*\gamma$, i.e.,
$$
f^*(x) = \int_{\real^n}f(x-y)\dif\gamma(y)
$$
whenever the integral is defined. Then for all $f,g\in S(\real^n)$
we have
\begin{eqnarray}\label{eq:conv-with-gaussian}
f^* = e^{\frac\Delta2}f
\end{eqnarray}
where $\Delta$ is the Laplace operator, and
\begin{eqnarray}\label{eq:inner-lap}
\int_{\real^n} f^*(\mathrm{i}x)
\overline{g^*(\mathrm{i}x)}\dif\gamma(x) = \Inner{f}{g}_\partial.
\end{eqnarray}
To prove Eq.~\eqref{eq:conv-with-gaussian}, we may assume by
linearity that $f$ is a monomial $x^\alpha$, i.e. the coefficient of
$\frac{\xi^\alpha}{\alpha!}$ in $e^{\Inner{x}{\xi}}$. Indeed,
$$
e^{\Inner{x}{\xi}} =
\exp\Pa{\sum^n_{j=1}x_j\xi_j}=\sum^\infty_{m=0}\frac{(\sum^n_{j=1}x_j\xi_j)^m}{m!},
$$
where
$$
\Pa{\sum^n_{j=1}x_j\xi_j}^m = \sum_{\alpha:\abs{\alpha}=m}
\binom{m}{\alpha}x^\alpha\xi^\alpha
$$
for $\binom{m}{\alpha} = \frac{m!}{\alpha_1!\cdots\alpha_n!},
x^\alpha=x^{\alpha_1}_1\cdots x^{\alpha_n}_n$, and
$\xi^\alpha=\xi^{\alpha_1}_1\cdots \xi^{\alpha_n}_n$;
$\abs{\alpha}=\sum^n_{j=1}\alpha_j$ and
$\alpha!=\alpha_1!\cdots\alpha_n!$. Thus
$$
e^{\Inner{x}{\xi}}
=\sum^\infty_{m=0}\sum_{\alpha:\abs{\alpha}=m}\frac{\xi^\alpha}{\alpha!}x^\alpha.
$$
This means that $e^{\Inner{x}{\xi}}$ is a linear combination of
monomials $x^\alpha$. Hence it is enough to verify
Eq.~\eqref{eq:conv-with-gaussian} when $f(x)=e^{\Inner{x}{\xi}}$;
indeed,
\begin{eqnarray*}
f^*(x) &=& \int_{\real^n} f(x-y)\dif\gamma(y)\\
&=& \int_{\real^n}e^{\Inner{x-y}{\xi}}
\frac1{\sqrt{(2\pi)^n}}\exp\Pa{-\frac12\Inner{y}{y}}[\dif y]\\
&=&e^{\Inner{x}{\xi}}\int_{\real^n}e^{\Inner{-y}{\xi}}
\frac1{\sqrt{(2\pi)^n}}\exp\Pa{-\frac12\Inner{y}{y}}[\dif y]
\end{eqnarray*}
implying
\begin{eqnarray*}
f^*(x) &=&
e^{\Inner{x}{\xi}}\prod^n_{j=1}\frac1{\sqrt{2\pi}}\int_{\real^n}
\exp\Pa{-\frac12(y_j^2+2\xi_jy_j)}\dif y_j\\
&=&
e^{\Inner{x}{\xi}}\prod^n_{j=1}e^{\frac12\xi^2_j}\frac1{\sqrt{2\pi}}\int_{\real^n}
\exp\Pa{-\frac12(y_j+\xi_j)^2}\dif (y_j+\xi_j)
\end{eqnarray*}
then a simple calculation shows that
$f^*(x)=e^{\Inner{x}{\xi}}\exp\Pa{\frac12\Inner{\xi}{\xi}}$. Then
\begin{eqnarray*}
\Delta e^{\Inner{x}{\xi}} &=&
\Pa{\sum^n_{j=1}\frac{\partial^2}{\partial
x_j^2}}\exp\Pa{\sum^n_{j=1}\xi_jx_j} =\sum^n_{j=1}
\Pa{\frac{\partial^2}{\partial
x_j^2}\exp\Pa{\sum^n_{j=1}\xi_jx_j}}\\
&=& \sum^n_{j=1}\xi^2_je^{\Inner{x}{\xi}} =
\Inner{\xi}{\xi}e^{\Inner{x}{\xi}},
\end{eqnarray*}
which means that $f^*(x)=e^{\frac\Delta2}f(x)$ since
\begin{eqnarray*}
e^{\frac\Delta2}f(x) &=& e^{\frac\Delta2}e^{\Inner{x}{\xi}} =
\sum^\infty_{m=0}\frac1{m!}\Pa{\frac\Delta2}^me^{\Inner{x}{\xi}}\\
&=& \sum^\infty_{m=0}\frac1{m!2^m}\Delta^me^{\Inner{x}{\xi}} =
\sum^\infty_{m=0}\frac1{m!2^m}\Inner{\xi}{\xi}^me^{\Inner{x}{\xi}}\\
&=& e^{\frac12\Inner{\xi}{\xi}}e^{\Inner{x}{\xi}} = f^*(x).
\end{eqnarray*}
Likewise, it is enough to verify Eq.~\eqref{eq:inner-lap} when
$f(x)$ is replaced by $e^{\Inner{x}{\xi}}$ and $g(x)$ by
$e^{\Inner{x}{\eta}}$ (and $\Inner{f}{g}_\partial$ by
$e^{\Inner{\xi}{\eta}}$). Indeed, for this case,
\begin{eqnarray*}
\Inner{f}{g}_\partial &=&
\Inner{e^{\Inner{x}{\xi}}}{e^{\Inner{x}{\eta}}}_\partial =
\Inner{\sum^\infty_{i=0}\sum_{\alpha:\abs{\alpha}=i}\frac{\xi^\alpha}{\alpha!}x^\alpha}{\sum^\infty_{j=0}\sum_{\beta:\abs{\beta}=j}\frac{\eta^\beta}{\beta!}x^\beta}_\partial\\
&=&\sum^\infty_{i=0}\sum_{\alpha:\abs{\alpha}=i}\sum^\infty_{j=0}\sum_{\beta:\abs{\beta}=j}\frac{\xi^\alpha}{\alpha!}\frac{\eta^\beta}{\beta!}\Inner{x^\alpha}{x^\beta}_\partial\\
&=&\sum^\infty_{j=0}\sum_{\alpha:\abs{\alpha}=j}\frac{\xi^\alpha
\eta^\alpha}{\alpha!} = e^{\Inner{\xi}{\eta}}
\end{eqnarray*}
and
\begin{eqnarray*}
\int_{\real^n} f^*(\mathrm{i}x)
\overline{g^*(\mathrm{i}x)}\dif\gamma(x) &=& \int_{\real^n}
e^{\frac12\Inner{\xi}{\xi}}e^{\Inner{\mathrm{i}x}{\xi}}\overline{e^{\frac12\Inner{\eta}{\eta}}e^{\Inner{\mathrm{i}x}{\eta}}}\dif\gamma(x)\\
&=&e^{\frac12\Inner{\xi}{\xi}}e^{\frac12\Inner{\eta}{\eta}}\int_{\real^n}e^{\mathrm{i}\Inner{x}{\xi-\eta}}\dif\gamma(x)\\
&=&e^{\frac12\Inner{\xi}{\xi}}e^{\frac12\Inner{\eta}{\eta}}e^{-\frac12\Inner{\xi-\eta}{\xi-\eta}}=e^{\Inner{\xi}{\eta}}
\end{eqnarray*}
where we used the fact that
$$
\frac1{\sqrt{2\pi}}\int^\infty_{-\infty}
e^{-at^2}e^{-\mathrm{i}ty}\dif t =
\frac1{\sqrt{2a}}\exp\Pa{-\frac{y^2}{4a}}.
$$
Suppose now that $f$ is a \emph{harmonic} homogeneous polynomial,
i.e. $\Delta f=0$. Then $\Delta^kf=0$ for $k\geqslant1$ and
$$
f^*=e^{\frac\Delta2}f=\sum^\infty_{k=0}\frac{\Delta^k}{2^kk!}f = f
$$
by Eq.~\eqref{eq:conv-with-gaussian}, and therefore from
Eq.~\eqref{eq:inner-lap} we obtain
\begin{eqnarray*}
\Inner{f}{f}_\partial &=&\int_{\real^n} f^*(\mathrm{i}x)
\overline{f^*(\mathrm{i}x)}\dif\gamma(x)= \int_{\real^n}
f(\mathrm{i}x)
\overline{f(\mathrm{i}x)}\dif\gamma(x)\notag\\
&=&\int_{\real^n}f(x)^2\dif\gamma(x).
\end{eqnarray*}
Now we have already proven that if $f$ is a \emph{harmonic}
homogeneous polynomial, then we have
\begin{eqnarray}\label{eq:ff}
\Inner{f}{f}_\partial=\int_{\real^n}f(x)^2\dif\gamma(x)
=\frac1{\sqrt{(2\pi)^n}}\int_{\real^n}f(x)^2e^{-\frac12\inner{x}{x}}[\dif
x].
\end{eqnarray}
The polynomial $P$ is skew-symmetric with respect to the Weyl group
$W$, and the inner product (on $\liet$ or $\liet^*$) is
$W$-invariant. It follows that $\Delta P=0$, and hence from
Eq.~\eqref{eq:ff} that
$$
\int_\liet
\exp\Pa{-\frac12\Inner{\tau}{\tau}}P(\tau)^2\dif\mu_L(\tau) =
\sqrt{(2\pi)^n}\Inner{P}{P}_\partial,
$$
where
\begin{eqnarray*}
\Inner{P}{P}_\partial&=&\Inner{\prod_{\alpha>0}\alpha}{\prod_{\alpha>0}\alpha}
=\Inner{\prod^N_{j=1}\alpha_j}{\prod^N_{j=1}\alpha_j} \\
&\defeq& \sum_{\sigma\in
S_N}\Inner{\alpha_1\ot\cdots\ot\alpha_N}{\alpha_{\sigma^{-1}(1)}\ot\cdots\ot\alpha_{\sigma^{-1}(N)}}.
\end{eqnarray*}
Now substituting this integral value in
Eq.~\eqref{eq:vital-equation}, we obtain
\begin{eqnarray*}
\sqrt{(2\pi)^d} = \frac1{\abs{W}}\frac{\vol(G)}{\vol(T)}
\sqrt{(2\pi)^n}\Inner{P}{P}_\partial.
\end{eqnarray*}
That is,
\begin{eqnarray}
\frac{\vol(G)}{\vol(T)}= (2\pi)^{\frac{d-n}2}
\frac{\abs{W}}{\Inner{P}{P}_\partial}=
(2\pi)^N\frac{\abs{W}}{\Inner{P}{P}_\partial}.
\end{eqnarray}
Note here that $N=\frac{d-n}2$. Now the number $\Inner{P}{P}$ has
already been calculated by Steinberg \cite{Harder}:
$$
\Inner{P}{P}_\partial = \abs{W}P(\rho) =
2^{-N}\abs{W}\prod^n_{j=1}m_j!\prod_{\alpha>0}
\Inner{\alpha}{\alpha}
$$
where $\rho:=\frac12\sum_{\alpha>0}\alpha$ is half the sum of the
positive roots, and the $m_j$ are the \blue{exponents} of $G$.
Indeed, from Weyl's denominator formula:
$$
\sum_{w\in W}\sign(w)e^{w\rho} =
\prod_{\alpha>0}\Pa{e^{\frac12\alpha} - e^{-\frac12\alpha}},
$$
we see that for any chosen $X$ such that $\Inner{\alpha}{X}\neq0$
for all $\alpha>0$, we have that
\begin{eqnarray*}
\sum_{w\in W}\sign(w)e^{\Inner{w\rho}{tX}} =e^{-\Inner{\rho}{tX}}
\prod_{\alpha>0}\Pa{e^{\Inner{\alpha}{tX}} - 1},\quad\forall
t\in\real.
\end{eqnarray*}
That is,
\begin{eqnarray*}
\sum_{w\in W}\sign(w)e^{t\Inner{w\rho}{X}} &=& e^{-t\Inner{\rho}{X}}
\prod_{\alpha>0}\Pa{e^{t\Inner{\alpha}{X}} - 1}\\
&=&t^N e^{-t\Inner{\rho}{X}}
\prod_{\alpha>0}\frac{e^{t\Inner{\alpha}{X}} - 1}{t}.
\end{eqnarray*}
Now using the Taylor expansion of $e^x$:
$$
e^x=\sum^\infty_{k=0}\frac{x^k}{k!},
$$
we see that
$$
\frac{e^{t\Inner{\alpha}{X}} - 1}{t}
=\Inner{\alpha}{X}+\sum^\infty_{k=2}\frac{t^{k-1}\Inner{\alpha}{X}^k}{k!}
= \Inner{\alpha}{X}+o(t).
$$
And thus
\begin{eqnarray*}
t^N e^{-t\Inner{\rho}{X}}
\prod_{\alpha>0}\frac{e^{t\Inner{\alpha}{X}} - 1}{t} =
t^N(1+o(t))(\prod_{\alpha>0}\Inner{\alpha}{X}+o(t))=t^N(P(X)+o(t)).
\end{eqnarray*}
Taking the $N$-th derivative with respect to $t$ on both sides, we
get that
$$
\sum_{w\in W}\sign(w)\Inner{w\rho}{X}^Ne^{t\Inner{w\rho}{X}} =
N!P(X)+o(t).
$$
Now taking the limit for $t\to0$, we get that
$$
\sum_{w\in W}\sign(w)\Inner{w\rho}{X}^N =
N!P(X)\Longleftrightarrow\sum_{w\in W} \sign(w)(w\rho)^N = N!P.
$$
Thus
\begin{eqnarray*}
\Inner{P}{P}_\partial &=& \frac1{N!}\Inner{\sum_{w\in W}
\sign(w)(w\rho)^N}{\prod_{\alpha>0}\alpha}  = \frac1{N!}\sum_{w\in
W} \sign(w)\Inner{(w\rho)^N}{\prod_{\alpha>0}\alpha}\\
&=& \frac1{N!}\sum_{w\in W}
\sign(w)^2\Inner{\rho^N}{\prod_{\alpha>0}\alpha}=\frac1{N!}\abs{W}N!\prod_{\alpha>0}\Inner{\rho}{\alpha}
= \abs{W}P(\rho)
\end{eqnarray*}
where we used the fact that the polynomial $P$ is skew-symmetric
with respect to the Weyl group $W$. We can also present another
approach to the identity $\Inner{P}{P}_\partial=\abs{W}P(\rho)$.
Denote $q_\rho(X)=\sum_{w\in W}\sign(w)e^{\Inner{w\rho}{X}}$ Note
that
\begin{eqnarray*}
P(\partial) q_\rho(0) &=& \sum_{w\in
W}\sign(w)\left.P(\partial)e^{\Inner{w\rho}{X}}\right|_{X=0}=
\sum_{w\in W}\sign(w)\left.\prod_{\alpha>0}\partial_\alpha
e^{\Inner{w\rho}{X}}\right|_{X=0}\\
&=& \sum_{w\in W}\sign(w)\prod_{\alpha>0}\frac{\dif}{\dif
t}\big|_{t=0} e^{\Inner{w\rho}{t\alpha}} = \sum_{w\in
W}\sign(w)\prod_{\alpha>0} \Inner{w\rho}{\alpha},
\end{eqnarray*}
i.e.,
\begin{eqnarray*}
P(\partial) q_\rho(0) &=& \sum_{w\in W}\sign(w)\prod_{\alpha>0}
\Inner{\alpha}{w\rho} = \sum_{w\in W}\sign(w) P(w\rho)\\
&=&\sum_{w\in W}\sign(w)^2 P(\rho)=\abs{W}P(\rho).
\end{eqnarray*}
In addition,
\begin{eqnarray*}
\left.P(\partial)\prod_{\alpha>0}\Pa{e^{\frac12\Inner{\alpha}{X}} -
e^{-\frac12\Inner{\alpha}{X}}}\right|_{X=0} = P(\partial)P(0) =
\Inner{P}{P}_\partial.
\end{eqnarray*}
Indeed,
\begin{eqnarray*}
\prod_{\alpha>0}\Pa{e^{\frac12\alpha} - e^{-\frac12\alpha}} =
\prod_{\alpha>0}\Pa{\alpha +
\sum^\infty_{n=1}\frac{\alpha^{2n+1}}{2^{2n}(2n+1)!}} =
\prod_{\alpha>0}\alpha + \text{terms of degree at least } N+1.
\end{eqnarray*}
Note that $P(\partial)$ is the linear partial operator of order $N$,
hence
\begin{eqnarray*}
\left.P(\partial)\prod_{\alpha>0}\Pa{e^{\frac12\Inner{\alpha}{X}} -
e^{-\frac12\Inner{\alpha}{X}}}\right|_{X=0} = P(\partial)P(0) +0=
\Inner{P}{P}_\partial.
\end{eqnarray*}
Thus
$$
\Inner{P}{P}_\partial = \abs{W}P(\rho).
$$
Now for any root $\alpha$,
$$
\frac{2\Inner{\rho}{\alpha}}{\Inner{\alpha}{\alpha}}
=\Inner{\rho}{\alpha^\vee}= \mathrm{ht}(\alpha^\vee),
$$
the height of $\alpha^\vee$, and it is known that if $d_1\leqslant
d_2\leqslant \cdots\leqslant d_n$ are the \blue{degrees of the basic
invariants}, where $d_j=m_j+1$ \cite{Humphreys1990}, the number of
roots of height $k$ minus the number of roots of height $k+1$ is
just the number of $m_j$'s equal to $k$ (or the number of $d_j$'s
equal to $k+1$)
\cite{Kostant1959,Steinberg1959,Macdonald1972,Humphreys1990}. That
is,
\begin{eqnarray*}
\abs{\Set{\alpha^\vee>0: \mathrm{ht}(\alpha^\vee)=k}} -
\abs{\Set{\alpha^\vee>0: \mathrm{ht}(\alpha^\vee)=k+1}} =
\abs{\Set{j\in[n]: m_j=k}}.
\end{eqnarray*}
Indeed, in order to show this result, we partition the number $N$ of
all positive roots via the height of all positive co-roots:
$$
\Set{\alpha^\vee|\alpha>0} =\bigsqcup^{h-1}_{j=1}\Set{\alpha^\vee>0:
\mathrm{ht}(\alpha^\vee)=j},
$$
where $h$ is the Coxeter number \cite{Humphreys1990} for Weyl group
$W$. Denote by $n_k=\abs{\Set{\alpha^\vee>0:
\mathrm{ht}(\alpha^\vee)=k}}$. Thus we see that
$$
N = \sum^{h-1}_{k=1}n_k,\quad n_1\geqslant
n_2\geqslant\cdots\geqslant n_{h-1}.
$$
It is shown that the Poincar\'{e} polynomial $W(t):=\sum_{w\in
W}t^{\ell(w)}$, where $\ell(w)$ is the length of $w\in W$, can be
factorized by two ways \cite{Macdonald1972}:
\begin{eqnarray}
W(t) =
\prod_{\alpha>0}\frac{1-t^{1+\mathrm{ht}(\alpha)}}{1-t^{\mathrm{ht}(\alpha)}}
= \prod^n_{j=1}\frac{1-t^{m_j+1}}{1-t}.
\end{eqnarray}
This implies that
$$
\abs{W} = \lim_{t\to1}W(t) =
\prod_{\alpha>0}\frac{1+\mathrm{ht}(\alpha)}{\mathrm{ht}(\alpha)} =
\prod^n_{j=1}(m_j+1).
$$
Note the above identity holds also for dual root system since $W$
corresponds to both the original root system $R$ and its dual root
system $R^\vee$, it follows that
$$
\prod_{\alpha>0}\frac{1+\mathrm{ht}(\alpha^\vee)}{\mathrm{ht}(\alpha^\vee)}
= \prod^n_{j=1}(m_j+1).
$$
Now partition $[n]$ by the number of times $k$ as exponents of $W$:
$$
[n] = \bigsqcup^{h-1}_{k=1}\cI_k,\quad \cI_k:=\Set{j\in[n]: m_j=k}.
$$
Then
$$
\prod_{\alpha>0}\frac{1+\mathrm{ht}(\alpha^\vee)}{\mathrm{ht}(\alpha^\vee)}
= 2^{n_1}\Pa{\frac32}^{n_2}\cdots \Pa{\frac{h}{h-1}}^{n_{h-1}} =
2^{n_1-n_2}3^{n_2-n_3}\cdots h^{n_{h-1}}
$$
and
$$
\prod^n_{j=1}(m_j+1) = 2^{\abs{\cI_1}}3^{\abs{\cI_2}}\cdots
h^{\abs{\cI_{h-1}}}.
$$
Therefore we have
$$
n_k - n_{k+1} = \abs{\cI_k},\quad k=1,\ldots, h-2; \quad
n_{h-1}=\abs{\cI_{h-1}}.
$$
In what follows, we calculate the value of
$\prod_{\alpha>0}\mathrm{ht}(\alpha^\vee)$. In fact, we have
$\mathrm{ht}(\alpha^\vee)\in[1,h-1]\cap\natural$ and
\begin{eqnarray}
\prod_{\alpha>0}\mathrm{ht}(\alpha^\vee) = 1^{n_1}2^{n_2}\cdots
(h-1)^{n_{h-1}}.
\end{eqnarray}
Clearly,
\begin{eqnarray*}
n_1 &=& \abs{\cI_{h-1}}+\cdots+ \abs{\cI_2}+\abs{\cI_1},\\
n_2 &=& \abs{\cI_{h-1}}+\cdots+ \abs{\cI_2},\\
&\vdots&\\
n_{h-1}&=&\abs{\cI_{h-1}}.
\end{eqnarray*}
Furthermore,
\begin{eqnarray*}
\prod_{\alpha>0}\mathrm{ht}(\alpha^\vee) &=&
(1!)^{\abs{\cI_1}}(2!)^{\abs{\cI_2}}\cdots
[(h-1)!]^{\abs{\cI_{h-1}}}\\
&=& \prod_{j\in \cI_1}m_j!\times \prod_{j\in
\cI_2}m_j!\times\cdots\times \prod_{j\in \cI_{h-1}}m_j!\\
&=&\prod_{j\in \cI_1\sqcup\cdots\sqcup \cI_{h-1}}m_j! =
\prod_{j\in[n]}m_j! = \prod^n_{j=1}m_j!.
\end{eqnarray*}
Thus
$$
\prod_{\alpha>0}\frac{2\Inner{\rho}{\alpha}}{\Inner{\alpha}{\alpha}}
=\prod_{\alpha>0}\Inner{\rho}{\alpha^\vee}=\prod^n_{j=1}(d_j-1)!=\prod^n_{j=1}m_j!.
$$
Therefore,
\begin{eqnarray*}
P(\rho) = \prod_{\alpha>0}\Inner{\alpha}{\rho} =
\prod_{\alpha>0}\frac{2\Inner{\rho}{\alpha}}{\Inner{\alpha}{\alpha}}\prod_{\alpha>0}\frac{\Inner{\alpha}{\alpha}}2
= 2^{-N}\prod^n_{j=1}m_j!\prod_{\alpha>0}\Inner{\alpha}{\alpha}.
\end{eqnarray*}
That is,
\begin{eqnarray}
P(\rho) =
2^{-N}\prod^n_{j=1}m_j!\prod_{\alpha>0}\Inner{\alpha}{\alpha}.
\end{eqnarray}
Since also $\abs{W}=\prod^n_{j=1}d_j=\prod^n_{j=1}(m_j+1)$
\cite{Humphreys1990}, it follows that
\begin{eqnarray*}
\Inner{P}{P} = \abs{W}P(\rho) =
2^{-N}\abs{W}\prod^n_{j=1}m_j!\prod_{\alpha>0}\Inner{\alpha}{\alpha}.
\end{eqnarray*}
That is,
\begin{eqnarray}
\Inner{P}{P} =
2^{-N}\prod^n_{j=1}d_j!\prod_{\alpha>0}\Inner{\alpha}{\alpha}.
\end{eqnarray}
Hence we have
\begin{eqnarray}
\frac{\vol(G)}{\vol(T)}=\frac{(2\pi)^N}{P(\rho)} =
\prod_{\alpha>0}\frac{2\pi}{\Inner{\alpha}{\rho}}
\end{eqnarray}
a formula due to Harish-Chandra, and also, using
Eq.~\eqref{eq:ration-vol},
\begin{eqnarray*}
\vol(G) &=& \frac{(2\pi)^N\vol(T)}{P(\rho)}=
\frac{(2\pi)^{N+n}\vol(\liet/\liet_\integer)}{2^{-N}\prod^n_{j=1}m_j!\prod_{\alpha>0}\Inner{\alpha}{\alpha}}\\
&=&
\vol(\liet/\liet_\integer)\prod_{\alpha>0}\frac4{\Inner{\alpha}{\alpha}}\prod^n_{j=1}\frac{2\pi^{m_j+1}}{m_j!}\\
&=&\vol(\liet/\liet_\integer)\prod_{\alpha>0}\Inner{\alpha^\vee}{\alpha^\vee}\prod^n_{j=1}\vol\Pa{\mathbb{S}^{r_j}}
\end{eqnarray*}
where $r_j=2m_j+1$ and $\sum^n_{j=1}m_j=N$ and
$\frac4{\Inner{\alpha}{\alpha}}=\Inner{\alpha^\vee}{\alpha^\vee}$
for $\alpha^\vee$ is the coroot corresponding to $\alpha$.

Let $\g_\complex,\liet_\complex$ be the complexifications of
$\g,\liet$. For $X=\xi+\mathrm{i}\eta\in\g_\complex(\xi,\eta\in\g)$
Let $\bar X = \xi-\mathrm{i}\eta$. There exists a Chevalley basis of
$\g_\complex$ relative to $\liet_\complex$ in which the root vectors
$X_\alpha$ satisfy $X_{-\alpha}=\bar X_\alpha$. Let
$\xi_\alpha=X_\alpha+X_{-\alpha},\eta_\alpha=\mathrm{i}(X_\alpha-X_{-\alpha})$,
then the "integer lattice" $\g_\integer$ is spanned by the
$\xi_\alpha$ and $\eta_\alpha$ for $\alpha>0$, together with a basis
of $\liet_\integer$; the $2N$ vectors $\xi_\alpha,\eta_\alpha$ are
orthogonal to each other and to $\liet$, and
$\norm{\xi_\alpha}=\norm{\eta_\alpha}=\norm{\alpha^\vee}$. It
follows that
$$
\vol(\g/\g_\integer) =
\vol(\liet/\liet_\integer)\prod_{\alpha>0}\Inner{\alpha^\vee}{\alpha^\vee};
$$
substituting this, we derive the desired formula:
\begin{eqnarray*}
\vol(G) =\vol(\g/\g_\integer)\prod^n_{j=1}\vol\Pa{\mathbb{S}^{r_j}}.
\end{eqnarray*}
We are done.
\end{proof}

\begin{remark}
Let $(W,S)$ be an irreducible finite Coxeter system of rank $n$ with
$S=\Set{s_{\alpha_1},\ldots,s_{\alpha_n}}$ its set of simple
reflections. The element $c:=s_{\alpha_1}\cdots s_{\alpha_n}\in W$
is called Coxeter transformation. Since all Coxeter transformations
are conjugate, they have the same order, characteristic polynomial
and eigenvalues. The order $h$ of Coxeter elements is called the
\emph{Coxeter number} of $W$. For a fixed Coxeter element $c\in W$,
if its eigenvalues are of the form:
$$
\exp\Pa{2\pi\mathrm{i}m_1/h},\quad,\exp\Pa{2\pi\mathrm{i}m_n/h}
$$
with $0<m_1\leqslant \cdots\leqslant m_n<h$, then the integers
$m_1,\ldots,m_n$ are called the \emph{exponents} of $W$. We have
already known that for any irreducible root system $R$ of rank $n$,
\begin{enumerate}[(i)]
\item $m_j+m_{n-j+1}=h$ for $1\leqslant j\leqslant n$;
\item $m_1=1,m_n=h-1$;
\item the height of the highest root in $R^+$ is given by $m_n=h-1$.
\end{enumerate}
Let $k_j:=\abs{\Set{\alpha\in R^+|\mathrm{ht}(\alpha)=j}}$. The
height distribution of $R^+$ is defined as a multiset of positive
integers:
$$
\Set{k_1,\ldots,k_{h-1}}.
$$
Apparently, $(k_1,\ldots,k_{h-1})\vdash N=\abs{R^+}$. It is
well-known that the exponents of the Weyl group $W$ are given by the
\emph{dual partition} of the height distribution of $R^+$ is given
by a multiset of non-negative integers
\cite{Kostant1959,Steinberg1959,Macdonald1972,Abc2016}:
$$
\Set{(0)^{n-k_1}, (1)^{k_1-k_2},\ldots, (h-2)^{k_{h-2}-k_{h-1}},
(h-1)^{k_{h-1}}},
$$
where $(a)^b$ means the integer $a$ appears exactly $b$ times. If
$k_j-k_{j+1}>0$, then $m_j$ appears exactly $k_j-k_{j+1}$ times.
Otherwise, $m_j$ does not appear if $k_j=k_{j+1}$.

Note that in the Macdonald's method to the volume of a compact Lie
group, we see that $d=\sum^n_{j=1}r_j=2\sum^n_{j=1}m_j+n$. It
follows that $\sum^n_{j=1}m_j=\frac{d-n}2=N$. Besides, since
$m_j+m_{n-j+1}=h$, it follows that
$$
N=\sum^n_{j=1}m_j = \frac12nh.
$$
\end{remark}

\begin{remark}
Denote all polynomials in $x=(x_1,\ldots,x_d)$ with coefficients in
$\mathbb{F}$, where $\mathbb{F}=\complex$ or $\real$, by
$\mathbb{F}[x_1,\ldots,x_d]$. A polynomial is called
\emph{homogeneous} if all the monomials appearing in this polynomial
have the same total degree. Denote the space of homogeneous
polynomials of degree $n$ in $d$ variables by $\cP_n(\mathbb{F}^d)$.
That is,
\begin{eqnarray*}
\cP_n(\mathbb{F}^d) = \Set{P\mid
P(x)=\sum_{\alpha:\abs{\alpha}=n}c_\alpha x^\alpha, \text{where
}x\in\mathbb{F}^d}.
\end{eqnarray*}
Note that $\dim \cP_n(\real^d)=\binom{n+d-1}{n}$. For any two
homogeneous polynomials $P,Q\in \cP_n(\mathbb{F}^d)$, a scalar
product may be defined as follows:
\begin{eqnarray*}
\Inner{P}{Q} := \sum_{\alpha:\abs{\alpha}=n}\bar p_\alpha q_\alpha
\alpha!,
\end{eqnarray*}
where $P(x)=\sum_{\alpha:\abs{\alpha}=n}p_\alpha x^\alpha$ and
$Q(x)=\sum_{\alpha:\abs{\alpha}=n}q_\alpha x^\alpha$. From this
definition, we see that
$\Inner{x^\alpha}{x^\beta}=\alpha!\delta_{\alpha\beta}$. In
particular, if $P,Q\in\cP_n(\real^d)$, the scalar product can be
rewritten as \cite{Iwasaki1997,Knapp2002}
\begin{eqnarray*}
\Inner{P}{Q} = (P(\partial)Q)(0).
\end{eqnarray*}
Indeed, since $(\partial^\alpha x^\beta)_{x=0}
=\alpha!\delta_{\alpha\beta}$ and
\begin{eqnarray*}
P(\partial)Q &=& \Pa{\sum_{\alpha:\abs{\alpha}=n}p_\alpha
\partial^\alpha}\Pa{\sum_{\beta:\abs{\beta}=n}q_\beta x^\beta} = \sum_{\alpha,\beta:\abs{\alpha}=\abs{\beta}=n}p_\alpha
q_\beta \Pa{\partial^\alpha x^\beta},
\end{eqnarray*}
it follows that
\begin{eqnarray*}
(P(\partial)Q)(0) =
\sum_{\alpha,\beta:\abs{\alpha}=\abs{\beta}=n}p_\alpha q_\beta
\Pa{\alpha!\delta_{\alpha\beta}} = \sum_{\alpha:\abs{\alpha}=n}
p_\alpha q_\alpha \alpha!,
\end{eqnarray*}
which means that $\Inner{P}{Q} = (P(\partial)Q)(0)$.

As an example, let $P(x)=\Inner{\alpha_1}{x}\Inner{\alpha_2}{x}$ and
$Q(x)=\Inner{\beta_1}{x}\Inner{\beta_2}{x}$ be real polynomials in
$\cP_2(\real^3)$, where $\alpha_j=(a^1_j,a^2_j,a^3_j)$ and
$\beta_j=(b^1_j,b^2_j,b^3_j)$ for $j=1,2$. In what follows, we
calculate the scalar product $\Inner{P}{Q}$. Apparently,
\begin{eqnarray*}
P(\partial) &=& \sum^3_{i,j=1}a^i_1a^j_2\partial_i\partial_j =
\sum^3_{i=1} a^i_1a^i_2\partial^2_i +
\sum_{i<j}\Pa{a^i_1a^j_2+a^j_1a^i_2}\partial_i\partial_j,\\
Q(x) &=& \sum^3_{i,j=1}b^i_1b^j_2x_ix_j= \sum^3_{i=1}b^i_1b^i_2x^2_i
+ \sum_{i<j}\Pa{b^i_1b^j_2+b^j_1b^i_2}x_ix_j.
\end{eqnarray*}
It follows that
\begin{eqnarray*}
P(\partial)Q(x) &=& 2\sum^3_{i=1}a^i_1a^i_2b^i_1b^i_2 +
\sum_{i<j}\Pa{a^i_1a^j_2+a^j_1a^i_2}\Pa{b^i_1b^j_2+b^j_1b^i_2}
\end{eqnarray*}
Now we see that
\begin{eqnarray*}
&&\Inner{\alpha_1}{\beta_1}\Inner{\alpha_2}{\beta_2}+
\Inner{\alpha_1}{\beta_2}\Inner{\alpha_2}{\beta_1}\\
&&= \Pa{\sum^3_{i=1}a^i_1b^i_1}\Pa{\sum^3_{j=1}a^j_2b^j_2} +
\Pa{\sum^3_{i=1}a^i_1b^i_2}\Pa{\sum^3_{j=1}a^j_2b^j_1}\\
&&=\Br{\sum^3_{i=1}a^i_1b^i_1a^i_2b^i_2 + \sum_{i\neq
j}a^i_1b^i_1a^j_2b^j_2} + \Br{\sum^3_{i=1}a^i_1b^i_2a^i_2b^i_1 +
\sum_{i\neq j}a^i_1b^i_2a^j_2b^j_1}\\
&&= 2\sum^3_{i=1}a^i_1a^i_2b^i_1b^i_2 + \sum_{i\neq
j}a^i_1b^i_1a^j_2b^j_2 + \sum_{i\neq j}a^i_1b^i_2a^j_2b^j_1.
\end{eqnarray*}
That is,
\begin{eqnarray*}
&&\Inner{\alpha_1}{\beta_1}\Inner{\alpha_2}{\beta_2}+
\Inner{\alpha_1}{\beta_2}\Inner{\alpha_2}{\beta_1}=
2\sum^3_{i=1}a^i_1a^i_2b^i_1b^i_2 + \sum_{i\neq
j}a^i_1b^i_1a^j_2b^j_2 + \sum_{i\neq j}a^i_1b^i_2a^j_2b^j_1\\
&&= 2\sum^3_{i=1}a^i_1a^i_2b^i_1b^i_2 + \sum_{i<
j}\Pa{a^i_1b^i_1a^j_2b^j_2+a^j_1b^j_1a^i_2b^i_2} + \sum_{i<
j}\Pa{a^i_1b^i_2a^j_2b^j_1+a^j_1b^j_2a^i_2b^i_1}\\
&&= 2\sum^3_{i=1}a^i_1a^i_2b^i_1b^i_2 + \sum_{i<
j}\Pa{a^i_1b^i_1a^j_2b^j_2+a^j_1b^j_1a^i_2b^i_2+a^i_1b^i_2a^j_2b^j_1+a^j_1b^j_2a^i_2b^i_1}\\
&&=2\sum^3_{i=1}a^i_1a^i_2b^i_1b^i_2 +
\sum_{i<j}\Pa{a^i_1a^j_2+a^j_1a^i_2}\Pa{b^i_1b^j_2+b^j_1b^i_2}.
\end{eqnarray*}
Therefore, we can conclude that
\begin{eqnarray*}
(P(\partial)Q)(0) &=&
\Inner{\alpha_1}{\beta_1}\Inner{\alpha_2}{\beta_2}+
\Inner{\alpha_1}{\beta_2}\Inner{\alpha_2}{\beta_1}\\
&=& \sum_{\sigma\in S_2}
\Inner{\alpha_1\ot\alpha_2}{\beta_{\sigma^{-1}(1)}\ot
\beta_{\sigma^{-1}(2)}}.
\end{eqnarray*}
We can generalize this result to the case where
$P(x)=\prod^n_{j=1}\Inner{\alpha_j}{x}$ and
$Q(x)=\prod^k_{j=1}\Inner{\beta_j}{x}$, in short,
$P=\prod^k_{j=1}\alpha_j$ and $Q=\prod^k_{j=1}\beta_j$. We have the
following:
\begin{eqnarray*}
\Inner{P}{Q} = (P(\partial)Q)(0) = \sum_{\sigma\in
S_k}\Inner{\alpha_1\otimes
\cdots\otimes\alpha_k}{\beta_{\sigma^{-1}(1)}\otimes
\cdots\beta_{\sigma^{-1}(k)}}.
\end{eqnarray*}

\end{remark}

\begin{remark}
Consider the real vector space $V$ generated by the roots.
$\Inner{\cdot}{\cdot}$ denotes the inner product on $V$ which is
invariant under the operation of the Weyl group $W$. The given inner
product on $V$ extends to one on the symmetric algebra of $V$ by the
formula:
\begin{eqnarray}
\Inner{\prod^k_{j=1}\alpha_j}{\prod^k_{j=1}\beta_j} &=&
\sum_{\sigma\in S_k}
\prod^k_{j=1}\Inner{\alpha_j}{\beta_{\sigma^{-1}(j)}} =
\sum_{\sigma\in
S_k}\Inner{\alpha_1}{\beta_{\sigma^{-1}(1)}}\cdots\Inner{\alpha_k}{\beta_{\sigma^{-1}(k)}}\\
&=&\sum_{\sigma\in S_k}\Inner{\alpha_1\otimes
\cdots\otimes\alpha_k}{\beta_{\sigma^{-1}(1)}\otimes
\cdots\beta_{\sigma^{-1}(k)}}
\end{eqnarray}
where $\alpha_j,\beta_j\in V$. In view of this, we see that
$$
\Inner{\prod_{\alpha>0}\alpha}{\rho^N} =\sum_{\sigma\in
S_N}\prod_{\alpha>0}\Inner{\alpha}{\rho} =
N!\prod_{\alpha>0}\Inner{\alpha}{\rho}.
$$
In the following, we show that \cite{Macdonald1982}
\begin{eqnarray}
\Inner{\prod_{\alpha>0}\alpha}{\prod_{\alpha>0}\alpha^\vee}
=\prod^n_{j=1}d_j!,
\end{eqnarray}
where $\alpha^\vee = \frac{2\alpha}{\Inner{\alpha}{\alpha}}$ such
that $\Inner{\alpha}{\alpha^\vee}=2$. Indeed, by replacing each
$\alpha^\vee_{\sigma^{-1}(j)}$ by
$\frac{2\alpha_{\sigma^{-1}(j)}}{\Inner{\alpha_{\sigma^{-1}(j)}}{\alpha_{\sigma^{-1}(j)}}}$
in the definition, we see that
\begin{eqnarray*}
\Inner{\prod_{\alpha>0}\alpha}{\prod_{\alpha>0}\alpha^\vee}
&\defeq&\sum_{\sigma\in S_N}\Inner{\alpha_1\otimes
\cdots\alpha_N}{\alpha^\vee_{\sigma^{-1}(1)}\otimes
\cdots\alpha^\vee_{\sigma^{-1}(N)}}\\
&=&\sum_{\sigma\in S_N}\Inner{\alpha_1\otimes
\cdots\alpha_N}{\alpha_{\sigma^{-1}(1)}\otimes
\cdots\alpha_{\sigma^{-1}(N)}}\frac{2^N}{\prod^N_{j=1}\Inner{\alpha_{\sigma^{-1}(j)}}{\alpha_{\sigma^{-1}(j)}}}\\
&=&\frac{2^N}{\prod_{\alpha>0}\Inner{\alpha}{\alpha}}
\sum_{\sigma\in S_N}\Inner{\alpha_1\otimes
\cdots\alpha_N}{\alpha_{\sigma^{-1}(1)}\otimes
\cdots\alpha_{\sigma^{-1}(N)}}.
\end{eqnarray*}
Therefore
\begin{eqnarray*}
\Inner{\prod_{\alpha>0}\alpha}{\prod_{\alpha>0}\alpha^\vee}
=\frac{2^N}{\prod_{\alpha>0}\Inner{\alpha}{\alpha}}\Inner{P}{P}=\prod^n_{j=1}d_j!.
\end{eqnarray*}
\end{remark}

Note that in deriving the volume formula of a compact Lie group, we
used the following fact which is necessarily recorded here:
\begin{thrm}[\cite{Duistermaat2000}]
The mapping that assigns to $\phi$ in $C_0(\g)$, the space of
continuous functions with compact support on $\g$, the function on
$G/T\times \liet$
\begin{eqnarray}
(gT,X)\longmapsto \phi(\Ad_g(X))\abs{\det(\ad_X)_{\g/\liet}}
\end{eqnarray}
extends to a topological isomorphism: $L^1(\g)\to L^1(G/T\times
\liet)^W$, where now $sT\in W$ acts on $(gT,X)$ by sending it to
$(gs^{-1}T, \Ad_s(X))$. Moreover, if $\phi\in L^1(\g)$, then:
\begin{eqnarray}
\int_\g \phi(X)\dif X  = \frac1{\abs{W}}\int_\liet
\Pa{\int_{G/T}\phi(\Ad_g(X))\dif(gT)}\abs{\det(\ad_X)_{\g/\liet}}\dif
X.
\end{eqnarray}
\end{thrm}
It is easily seen that if $\phi$ is $G$-invariant, i.e.,
$\phi(\Ad_g(X))=\phi(X)$ for all $g\in G$, then
$$
\int_{G/T}\phi(\Ad_g(X))\dif(gT) = \phi(X)\vol(G/T)=
\phi(X)\frac{\vol(G)}{\vol(T)}.
$$
Thus
\begin{eqnarray*}
\int_\g \phi(X)\dif X  =
\frac1{\abs{W}}\frac{\vol(G)}{\vol(T)}\int_\liet
\phi(X)\abs{\det(\ad_X)_{\g/\liet}}\dif X,
\end{eqnarray*}
where Jacobian $\abs{\det(\ad_X)_{\g/\liet}}=P(X)^2$ for
$P=\prod_{\alpha>0} \alpha$.

More simpler proof about Macdonald's volume formula can be found in
\cite{Hashimoto}. With Macdonald's volume formula, one can derive
all volume formulas for orthogonal groups and unitary groups.
Although these methods is a digression to the subject of the present
paper, I also collect these material together for reference.

Once again, we will use the Macdonald's volume formula for a compact
Lie group to derive the volume formula for the special unitary group
$G:=\rS\rU(n)$. The special unitary group $\rS\rU(n)$ is of rank
$n-1$. Its root system consists of $2\binom{n}{2}$ roots spanning a
$(n-1)$-dimensional Euclidean space. Here, we use $n$ redundant
coordinates instead of $n-1$ to emphasize the symmetries of the root
system (the $n$ coordinates have to add up to zero). In other words,
we are embedding this $n-1$ dimensional vector space in an
$n$-dimensional one. The maximal torus $T$ of $G$ is equivalently
identified with $\rS\rU(1)^{\times n}$, and its Lie algebra is
$\liet\cong\real^{n-1}$ of rank $n-1$. The lattice
$\liet_\integer\cong \integer^{n-1}$. Thus
$\vol(\liet/\liet_\integer)=1$.

The set of all positive roots of $\g$ is given by
\begin{eqnarray}
R^+ = \Set{\alpha_{ij}\in\liet^*: \alpha_{ij}(X) = x_i-x_j\text{ for
any } X=\diag(x_1,\ldots,x_n), i<j}.
\end{eqnarray}
Thus for any $\alpha\in R^+$, we have $\Inner{\alpha}{\alpha}=2$.
Note that $N=\abs{R^+}=\binom{n}{2}$. Then the Coxeter number
$h=\frac{2N}{n-1}=n$. The exponents of the special unitary group
$\rS\rU(n)$ are: $\set{1,\ldots,n-1}$. Note that
\begin{eqnarray*}
\vol_{\rH\rS}(G) =
\vol(\liet/\liet_\integer)\prod_{\alpha>0}\frac4{\Inner{\alpha}{\alpha}}\prod^{n-1}_{j=1}\vol\Pa{\mathbb{S}^{r_j}},
\end{eqnarray*}
where $r_j=2j+1$. It follows that
\begin{eqnarray*}
\vol_{\rH\rS}(G)
=2^{\binom{n}{2}}\prod^{n-1}_{j=1}\vol\Pa{\mathbb{S}^{2j+1}}.
\end{eqnarray*}
Since $\rU(n)=\rU(1)\times \rS\rU(n)$, it follows from
$\vol_{\rH\rS}(\rU(1))=\vol\Pa{\mathbb{S}^{1}}$ that
$$
\vol_{\rH\rS}(\rU(n)) =
\vol\Pa{\mathbb{S}^{1}}\vol_{\rH\rS}(\rS\rU(n)) =
2^{\binom{n}{2}}\prod^n_{j=1}\vol\Pa{\mathbb{S}^{2j-1}}.
$$
That is,
$$
\vol_{\rH\rS}(\rU(n)) =
\frac{(2\pi)^{\frac{n(n+1)}2}}{\prod^n_{j=1}\Gamma(j)}.
$$
Besides, we can also employ Harish-Chandra's volume formula for flag
manifold, we get that
\begin{eqnarray*}
\vol_{\rH\rS}(\rU(n)/\mathbb{T}^n) &=&
\prod_{\alpha>0}\frac{2\pi}{\inner{\alpha}{\rho}} =
\frac{(2\pi)^N}{P(\rho)} =
\frac{(2\pi)^N}{2^{-N}\prod^n_{j=1}m_j!\prod_{\alpha>0}\Inner{\alpha}{\alpha}}\\
&=& \frac{(2\pi)^N}{2^{-N}0!1!\cdots (n-1)!2^N} =
\frac{(2\pi)^{\frac{n(n-1)}2}}{\prod^n_{j=1}\Gamma(j)},
\end{eqnarray*}
implying that
\begin{eqnarray*}
\vol_{\rH\rS}(\rU(n)) = \vol(\mathbb{T}^n)
\frac{(2\pi)^{\frac{n(n-1)}2}}{\prod^n_{j=1}\Gamma(j)} = (2\pi)^n
\frac{(2\pi)^{\frac{n(n-1)}2}}{\prod^n_{j=1}\Gamma(j)}.
\end{eqnarray*}
Again, we also see that
$$
\vol_{\rH\rS}(\rU(n)) =
\frac{(2\pi)^{\frac{n(n+1)}2}}{\prod^n_{j=1}\Gamma(j)}.
$$

\begin{remark}
We can rewrite the above results as two partitions of $N$, the
number of all positive roots, $(n_1,\ldots,n_{h-1})\vdash N$ and
$(m_n,\ldots,m_1)\vdash N$, where $n_j$ denotes the number of $j$-th
row; $m_i$ the number of $n-j+1$-th column or the $n-j+1$-th row of
the conjugate partition of the original partition
$(m_n,\ldots,m_1)\vdash N$.
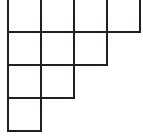
$$
\young(~~~~,~~~,~~,~)
$$
\end{remark}

\section{Applications}

The present section is directly written based on
\cite{Wei,Karol01,Karol03}. The results were already obtained. We
just here add some interpretation, from my angle, about them since
the details concerning computation therein are almost ignored.

\subsection{Hilbert-Schmidt volume of the set of mixed quantum states}

Any unitary matrix may be considered as an element of the
Hilbert-Schmidt space of operators with the scalar product
$\inner{\wtil U}{\wtil V}_{\rH\rS}=\Tr{\wtil U^*\wtil V}$. This
suggests the following definition of an invariant metric of the
unitary group $\cU(n)$: denote $\dif\wtil G:=\wtil U^*\dif\wtil U$,
then
\begin{eqnarray}
\dif s^2 := \Inner{\dif\wtil G}{\dif\wtil G}_{\rH\rS} =
-\Tr{\dif\wtil G^2},
\end{eqnarray}
implying
\begin{eqnarray}
\dif s^2 = \sum^n_{i,j=1} \abs{\dif\wtil G_{ij}}^2 = \sum^n_{j=1}
\abs{\dif\wtil G_{jj}}^2 + 2\sum^n_{i<j}\abs{\dif\wtil G_{ij}}^2.
\end{eqnarray}
Since $\dif\wtil G^*=-\dif\wtil G$, it follows that
\begin{eqnarray}
\dif s^2 = \sum^n_{j=1} \abs{\dif\wtil G_{jj}}^2 +
2\sum^n_{i<j}\Pa{\dif(\re(\wtil
G_{ij}))}^2+2\sum^n_{i<j}\Pa{\dif(\im(\wtil G_{ij}))}^2.
\end{eqnarray}
This indicates that the Hilbert-Schmidt volume element is given by
\begin{eqnarray}
\dif \nu= 2^{\frac{n(n-1)}2}\prod^n_{j=1}\dif\Pa{\im(\wtil
G_{jj})}\times \prod_{i<j}\dif(\re(\wtil G_{ij}))\dif(\im(\wtil
G_{ij})) = 2^{\frac{n(n-1)}2}[\dif\wtil G],
\end{eqnarray}
that is,
\begin{eqnarray}
\vol_{\rH\rS}\Pa{\cU(n)} &:=& \int_{\cU(n)}\dif \nu=
2^{\frac{n(n-1)}2}\int_{\cU(n)}[\dif\wtil G]\\
&=&2^{\frac{n(n-1)}2}\times\frac{
2^n\pi^{\frac{n(n+1)}2}}{1!2!\cdots (n-1)!}\\
&=&\frac{ (2\pi)^{\frac{n(n+1)}2}}{1!2!\cdots (n-1)!}.
\end{eqnarray}
Finally we have obtained the Hilbert-Schmidt volume of unitary
group:
\begin{framed}
\begin{eqnarray}
\vol_{\rH\rS}\Pa{\cU(n)} = 2^{\frac{n(n-1)}2}\vol\Pa{\cU(n)} =
\frac{ (2\pi)^{\frac{n(n+1)}2}}{1!2!\cdots (n-1)!}.
\end{eqnarray}
\end{framed}

We compute the volume of the convex $(n^2-1)$-dimensional set
$\density{\complex^n}$ of density matrices of size $n$ with respect
to the Hilbert-Schmidt measure.

The set of mixed quantum states $\density{\complex^n}$ consists of
Hermitian, positive matrices of size $n$, normalized by the trace
condition
\begin{eqnarray}
\density{\complex^n} = \set{\wtil\rho:\complex^n\to\complex^n|
\wtil\rho^*=\wtil\rho, \wtil\rho\geqslant 0,\Tr{\wtil\rho}=1}.
\end{eqnarray}
It is a compact convex set of dimensionality $(n^2-1)$. Any density
matrix may be diagonalized by a unitary rotation, $\wtil\rho =\wtil
U\Lambda\wtil U^*$, where $\wtil U\in\cU(n)$ and
$\Lambda=\diag(\lambda_1,\ldots,\lambda_n)$ for
$\lambda_j\in\real^+$. Since $\Tr{\wtil\rho}=1$, it follows that
$\sum^n_{j=1}\lambda_j=1$, so the spectra space is isomorphic with a
$(n-1)$-dimensional probability simplex
$\Delta_{n-1}:=\Set{p\in\real^+: \sum^n_{j=1}p_j=1}$.

Let $\wtil B$ be a diagonal unitary matrix. Since $\wtil \rho=\wtil
U\wtil B\Lambda \wtil B^* \wtil U^*$, in the generic case of a
\emph{non-degenerate spectrum} (i.e. with distinct non-negative
eigenvalues), the unitary matrix $\wtil U$ is determined up to $n$
arbitrary phases entering $\wtil B$. On the other hand, the matrix
$\Lambda$ is defined up to a permutation of its entries. The form of
the set of all such permutations depends on the character of the
degeneracy of the spectrum of $\wtil \rho$.

Representation $\wtil \rho=\wtil U\wtil B\Lambda \wtil B^* \wtil
U^*$ makes the description of some topological properties of the
$(n^2-1)$-dimensional space $\density{\complex^n}$ easier.
Identifying points in $\Delta_{n-1}$ which have the same components
(but ordered in a different way), we obtain an asymmetric simplex
$\wtil \Delta_{n-1}$. Equivalently, one can divide $\Delta_{n-1}$
into $n!$ identical simplexes and take any one of them. The
asymmetric simplex $\wtil \Delta_{n-1}$ can be decomposed in the
following natural way:
\begin{eqnarray}
\wtil\Delta_{n-1} = \bigcup_{d_1+\cdots+d_k=n}
\delta_{d_1,\ldots,d_k},
\end{eqnarray}
where $k=1,\ldots,n$ denotes the number of different coordinates of
a given point of $\wtil \Delta_{n-1}$, $d_1$ the number of
occurrences of the largest coordinate, $d_2$ the number of
occurrences of the second largest etc. Observe that
$\delta_{d_1,\ldots,d_k}$ is homeomorphic with the set $G_k$, where
$G_1$ is a single point, $G_2$ is a half-closed interval, $G_3$ an
open triangle with one edge but without corners and, generally,
$G_k$ is an $(k-1)$-dimensional simplex with one $(k-2)$-dimensional
hyperface without boundary (the latter is homeomorphic with an
$(k-2)$-dimensional open simplex). There are $n$ ordered
eigenvalues: $\lambda_1\geqslant\lambda_2\geqslant\cdots\geqslant
\lambda_n$, and $n-1$ independent relation operators "larger(>) or
equal(=)", which makes altogether $2^{n-1}$ different possibilities.
Thus, $\wtil \Delta_{n-1}$ consists of $2^{n-1}$ parts, out of which
$\binom{n-1}{m-1}$ parts are homeomorphic with $G_m$, when $m$
ranges from $1$ to $n$.

Let us denote the part of the space $\density{\complex^n}$ related
to the spectrum in $\delta_{d_1,\ldots,d_k}$ ($k$ different
eigenvalues; the largest eigenvalue has $d_1$ multiplicity, the
second largest $d_2$ etc) by $\rD_{d_1,\ldots,d_k}$. A mixed state
$\wtil\rho$ with this kind of the spectrum remains invariant under
arbitrary unitary rotations performed in each of the
$d_j$-dimensional subspaces of degeneracy. Therefore the unitary
matrix $\wtil B$ has a block diagonal structure with $k$ blocks of
size equal to $d_1,\ldots,d_k$ and
\begin{eqnarray}
\rD_{d_1,\ldots,d_k}\sim \Br{\cU(n)/(\cU(d_1)\times \cdots\times
\cU(d_k))} \times G_k,
\end{eqnarray}
where $d_1+\cdots+d_k=n$ and $d_j>0$ for $j=1,\ldots,k$. Thus
$\density{\complex^n}$ has the structure
\begin{eqnarray}
\density{\complex^n}\sim\bigcup_{d_1+\cdots+d_k=n}\rD_{d_1,\ldots,d_k}\sim
\bigcup_{d_1+\cdots+d_k=n}\Br{\cU(n)/(\cU(d_1)\times \cdots\times
\cU(d_k))} \times G_k,
\end{eqnarray}
where the sum ranges over all partitions $(d_1,\cdots,d_k)\vdash n$
of $n$. The group of rotation matrices $\wtil B$ equivalent to
$\cU(d_1)\times \cdots\times \cU(d_k)$ is called the \emph{stability
group} of $\cU(n)$.

Note also that the part of $\rD_{1,\ldots,1}$ represents a generic,
non-degenerate spectrum. In this case all elements of the spectrum
of $\wtil \rho$ are different and the stability group is equivalent
to an $n$-torus
\begin{eqnarray}
\rD_{1,\ldots,1}\sim \Br{\cU(n)/(\cU(1)^{\times n}} \times G_n.
\end{eqnarray}
The above representation of generic states enables us to define a
\emph{product measure} in the space $\density{\complex^n}$ of mixed
quantum states. To this end, one can take the uniform (Haar) measure
on $\cU(n)$ and a certain measure on the simplex $\Delta_{n-1}$.

The other $2^{n-1}-1$ parts of $\density{\complex^n}$ represent
various kinds of degeneracy and have \emph{measure zero}. The number
of non-homeomorphic parts is equal to the number $P(n)$ of different
representations of the number $n$ as the sum of positive natural
numbers. Thus $P(n)$ gives the number of different topological
structures present in the space $\density{\complex^n}$.

To specify uniquely the unitary matrix of eigenvectors $\wtil U$, it
is thus sufficient to select a point on the coset space
$$
\mathrm{Fl}^{(n)}_\complex : = \cU(n)/\cU(1)^{\times n},
$$
called the \emph{complex flag manifold}. The volume of this complex
flag manifold is:
\begin{framed}
\begin{eqnarray}
\vol_{\rH\rS}\Pa{\mathrm{Fl}^{(n)}_\complex} =
\frac{\vol_{\rH\rS}\Pa{\cU(n)}}{\vol_{\rH\rS}\Pa{\cU(1)}^n} =\frac{
(2\pi)^{\frac{n(n-1)}2}}{1!2!\cdots (n-1)!}.
\end{eqnarray}
\end{framed}
The generic density matrix is thus determined by $(n-1)$
parameters determining eigenvalues and $(n^2-n)$ parameters related
to eigenvectors, which sum up to the dimensionality $(n^2-1)$ of
$\density{\complex^n}$. Although for degenerate spectra the
dimension of the flag manifold decreases, these cases of measure
zero do not influence the estimation of the volume of the entire set
of density matrices. In this subsection, we shall use the
Hilbert-Schmidt metric. The infinitesimal distance takes a
particularly simple form
\begin{eqnarray}
\dif s_{\rH\rS}^2
=\norm{\dif\wtil\rho}^2_{\rH\rS}=\inner{\dif\wtil\rho}{\dif\wtil\rho}_{\rH\rS}
\end{eqnarray}
valid for any dimension $n$. Making use of the diagonal form
$\wtil\rho=\wtil U\Lambda \wtil U^\dagger$, we may write
\begin{eqnarray}
\dif\wtil\rho = \wtil U\Pa{\dif\Lambda + [\wtil U \dif \wtil U,
\Lambda]}\wtil U^*
\end{eqnarray}
Thus the infinitesimal distance can be rewritten as
\begin{eqnarray}
\dif s_{\rH\rS}^2 &=& \sum^n_{j=1}\dif\lambda_j^2 + 2
\sum^n_{i<j}(\lambda_i-\lambda_j)^2\abs{\Innerm{i}{\wtil
U^*\dif\wtil U}{j}}^2\\
&=&\sum^n_{j=1}\dif\lambda_j^2 + 2
\sum^n_{i<j}(\lambda_i-\lambda_j)^2\abs{\Innerm{i}{\dif\wtil
G}{j}}^2,
\end{eqnarray}
where $\dif\wtil G=\wtil U^*\dif\wtil U$. Apparently,
$\sum^n_{j=1}\dif\lambda_j=0$ since $\sum^n_{j=1}\lambda_j=1$. Thus
\begin{eqnarray}
\dif s_{\rH\rS}^2 =\sum^{n-1}_{i,j=1}\dif\lambda_i
m_{ij}\dif\lambda_j + 2
\sum^n_{i<j}(\lambda_i-\lambda_j)^2\abs{\dif\wtil G_{ij}}^2.
\end{eqnarray}
The corresponding volume element gains a factor $\sqrt{\det(M)}$,
where $M=[m_{ij}]$ is the metric in the $(n^2-n)$-dimensional
simplex $\Delta_{n-1}$ of eigenvalues. Note that
$$
M = \I_n+\Br{\begin{array}{cccc}
               1 & 1 & \cdots & 1 \\
               1 & 1 & \cdots & 1 \\
               \vdots & \vdots & \ddots & \vdots \\
               1 & 1 & \cdots & 1
             \end{array}
}.
$$
Therefore the Hilbert-Schmidt volume element is given by
\begin{eqnarray}
\dif V_{\rH\rS} =
\sqrt{n}\prod^{n-1}_{j=1}\dif\lambda_j\prod_{i<j}(\lambda_i-\lambda_j)^2\abs{\prod_{i<j}2\dif\Pa{\re(\wtil
G_{ij})}\dif\Pa{\im(\wtil G_{ij})}}.
\end{eqnarray}
Then
\begin{eqnarray*}
\int \dif V_{\rH\rS} &=&
\sqrt{n}2^{\frac{n(n-1)}2}\int\prod_{i<j}(\lambda_i-\lambda_j)^2\prod^{n-1}_{j=1}\dif\lambda_j\times
\int[\dif\wtil G_1]\\
&=&\sqrt{n}2^{\frac{n(n-1)}2}\vol\Pa{\density{\complex^n}}.
\end{eqnarray*}
That is, respect to Hilbert-Schmidt measure, the volume of the set
of mixed quantum states is
\begin{eqnarray*}
\vol_{\rH\rS}\Pa{\density{\complex^n}}
=\sqrt{n}2^{\frac{n(n-1)}2}\vol\Pa{\density{\complex^n}},
\end{eqnarray*}
i.e.
\begin{framed}
\begin{eqnarray}\label{eq:hs-vol}
\vol_{\rH\rS}\Pa{\density{\complex^n}}
=\sqrt{n}(2\pi)^{\frac{n(n-1)}2}\frac{\Gamma(1)\Gamma(2)\cdots\Gamma(n)}{\Gamma(n^2)}.
\end{eqnarray}
\end{framed}
We see from the above discussion that the obtained formula of volume
depends the used measure. If we used the Hilbert-Schmidt measure,
then we get the Hilbert-Schmidt volume of the set of quantum states
\cite{Karol03}; if we used the Bures measure, then we get the Bures
volume of the set of quantum states \cite{Sommers}.

A special important problem is to compute the volume of the set of
all separable quantum states, along this line, some investigation on
this topic had already been made \cite{Karol98,Karol99}. There are
some interesting topics for computing volumes of the set of some
kinds of states, for instance, Milz also considered the volumes of
conditioned bipartite state spaces \cite{Milz}, Link gave the
geometry of Gaussian quantum states \cite{Link} as well. We can also
propose some problems like this. Consider the following set of all
states being of the form:
$$
\cC\Pa{\frac12\I_2,\frac12\I_2}:=\Set{\wtil\rho_{12}\in\density{\complex^2\ot\complex^2}:
\ptr{1}{\wtil\rho_{12}}=\frac12\I_2=\ptr{2}{\wtil\rho_{12}}}.
$$
Paratharathy characterized the extremal points of this set
\cite{Parthasarathy}. He obtained that all the extremal points of
this convex set is maximal entangled states. That is,
$$
\frac{\ket{0}\ket{\psi_0}+\ket{1}\ket{\psi_1}}{\sqrt{2}}.
$$
It remains open to compute the volume of this convex set
$\cC\Pa{\frac12\I_2,\frac12\I_2}$.

\subsection{Area of the boundary of the set of mixed states}

The boundary of the set of mixed states is far from being trivial.
Formally it may be written as a solution of the equation
$$
\det(\wtil\rho)=0
$$
which contains all matrices of a lower rank. The boundary $\partial
\density{\complex^n}$ contains orbits of different dimensionality
generated by spectra of different rank and degeneracy. Fortunately
all of them are of measure \emph{zero} besides the generic orbits
created by unitary rotations of diagonal matrices with all
eigenvalues different and one of them equal to zero;
$$
\Lambda=\Set{0,\lambda_2<\cdots<\lambda_n}.
$$
Such spectra form the $(n-2)$-dimensional simplex $\Delta_{n-2}$,
which contains $(n-1)!$ the Weyl chambers---this is the number of
possible permutations of elements of $\Lambda$ which all belong to
the same unitary orbits.

Hence the hyper-area of the boundary may be computed in a way
analogous to \eqref{eq:hs-vol}:
\begin{eqnarray*}
&&\int_{\rank(\wtil X)=n-1} [\dif\wtil X] =
\int_{0<\lambda_2<\cdots<\lambda_n}
\delta\Pa{\sum^n_{j=2}\lambda_j-1}\prod_{2=i<j\leqslant
n}\abs{\lambda_i-\lambda_j}^2\prod^n_{j=2}(\lambda^2_j\dif\lambda_j)\times\int_{\cU_1(n)}[\dif\wtil
G_1]\\
&&= \frac1{(n-1)!}\overbrace{\int^\infty_0\cdots\int^\infty_0}^{n-1}
\delta\Pa{\sum^n_{j=2}\lambda_j-1}\prod_{2=i<j\leqslant
n}\abs{\lambda_i-\lambda_j}^2\prod^n_{j=2}(\lambda^2_j\dif\lambda_j)\times\int_{\cU_1(n)}[\dif\wtil
G_1]\\
&&=\frac1{\Gamma(n)}\frac{\Gamma(1)\cdots\Gamma(n)\Gamma(1)\cdots\Gamma(n+1)}{\Gamma(n^2-1)}\frac{\pi^{\frac{n(n-1)}2}}{\Gamma(1)\cdots\Gamma(n)}\\
&&=\pi^{\frac{n(n-1)}2}\frac{\Gamma(1)\cdots\Gamma(n+1)}{\Gamma(n)\Gamma(n^2-1)}
= \vol^{(n-1)}\Pa{\density{\complex^d}},
\end{eqnarray*}
i.e.
\begin{framed}
\begin{eqnarray*}
\vol^{(n-1)}_{\rH\rS} =
\sqrt{n-1}2^{\frac{n(n-1)}2}\vol^{(n-1)}\Pa{\density{\complex^d}}
=\sqrt{n-1}(2\pi)^{\frac{n(n-1)}2}\frac{\Gamma(1)\cdots\Gamma(n+1)}{\Gamma(n)\Gamma(n^2-1)}.
\end{eqnarray*}
\end{framed}
In an analogous way, we may find the volume of edges, formed by the
unitary orbits of the vector of eigenvalues with two zeros. More
generally, states of rank $N - n$ are unitarily similar to diagonal
matrices with $n$ eigenvalues vanishing,
$$
\Lambda=\Set{\lambda_1=\cdots=\lambda_m=0,\lambda_{m+1}<\cdots<\lambda_n}.
$$
These edges of order $m$ are $n^2-m^2-1$ dimensional, since the
dimension of the set of such spectra is $n-m-1$, while the orbits
have the structure of $\cU(n)/\Pa{\cU(m)\times \cU(1)^{n-m}}$ and
dimensionality $n^2-m^2-(n-m)$. We obtain the volume of the
hyperedges
\begin{eqnarray*}
\vol^{(n-m)}_{\rH\rS} =
\frac{\sqrt{n-m}}{(n-m)!}\frac1{C^{(1+2m,2)}_{n-m}}
\frac{\vol_{\rH\rS}\Pa{\mathrm{Fl}^{(n)}_\complex}}{\vol_{\rH\rS}\Pa{\mathrm{Fl}^{(m)}_\complex}}
\end{eqnarray*}

\subsection{Volume of a metric ball in unitary group}

Consider a metric ball around the identity $\I_n$ in the
$n$-dimensional unitary group $\cU(n)$ with Euclidean distance
$\epsilon$,
\begin{eqnarray}
B_\epsilon:=\Set{\wtil U\in\cU(n): \norm{\wtil U-\I_n}_2\leqslant
\epsilon},
\end{eqnarray}
where $\norm{*}_p$ is the $p$-norm for $p=2$. We consider the
invariant Haar-measure $\mu$, a unform distribution defined over
$\cU(n)$. Denote the eigenvalues of $\wtil U$ by
$e^{\sqrt{-1}\theta_j}$. The joint density of the angles $\theta_j$
is given by
\begin{eqnarray}
p(\theta_1,\ldots,\theta_n)=\frac1{(2\pi)^nn!}\prod_{1\leqslant
i<j\leqslant n}\abs{e^{\sqrt{-1}\theta_i} -
e^{\sqrt{-1}\theta_j}}^2,
\end{eqnarray}
where $\theta_j\in[-\pi,\pi], j=1,\ldots,n$. In what follows, we
check the correctness of the integral formula:
$$
\int p(\theta)\dif\theta=1.
$$
Indeed, set $J(\theta)=\prod_{1\leqslant i<j\leqslant
n}\abs{e^{\sqrt{-1}\theta_i} - e^{\sqrt{-1}\theta_j}}^2$ and
$\zeta_j=e^{\sqrt{-1}\theta_j}$, so
\begin{eqnarray}
J(\theta) &=& \prod_{i<j}\abs{\zeta_i-\zeta_j}^2 = \prod_{i<j}
(\zeta_i-\zeta_j)(\zeta^{-1}_i-\zeta^{-1}_j)\\
&=&
(\sign\tau)(\zeta_1\cdots\zeta_n)^{-(n-1)}\prod_{i<j}(\zeta_i-\zeta_j)^2,
\end{eqnarray}
where $\tau=(n\cdots 21)$, i.e. $\tau(j)=n+1-j$ or $\tau$ is written
as
$$
\tau:=\left(
  \begin{array}{cccc}
    1 & 2 & \cdots & n \\
    n & n-1 & \cdots & 1 \\
  \end{array}
\right).
$$
Note that $\sign\tau=(-1)^{\frac{n(n-1)}2}$. We see that the
integral is the constant term in
\begin{eqnarray}
C_n(\sign\tau)(\zeta_1\cdots\zeta_n)^{-(n-1)}\prod_{i<j}(\zeta_i-\zeta_j)^2.
\end{eqnarray}
Thus our task is to identify the constant term in this Laurent
polynomial. To work on the last factor, we recognize
\begin{eqnarray}
V(\zeta) = \prod_{i<j} (\zeta_i-\zeta_j)
\end{eqnarray}
as a Vandermonde determinant; hence
\begin{eqnarray}
V(\zeta) = \sum_{\sigma\in S_n}
(\sign\sigma)\zeta^{\sigma(1)-1}_1\cdots\zeta^{\sigma(n)-1}_n.
\end{eqnarray}
Hence
\begin{eqnarray}
\prod_{i<j}(\zeta_i-\zeta_j)^2 = V(\zeta)^2 = \sum_{\sigma,\pi\in
S_n}
(\sign\sigma)(\sign\pi)\zeta^{\sigma(1)+\pi(1)-2}_1\cdots\zeta^{\sigma(n)+\pi(n)-2}_n.
\end{eqnarray}
Let us first identify the constant term in
\begin{eqnarray}
J(\theta) = (\sign\tau)(\zeta_1\cdots\zeta_n)^{-(n-1)}V(\zeta)^2.
\end{eqnarray}
We see this constant term is equal to
\begin{eqnarray*}
&&\frac1{(2\pi)^n}
\overbrace{\int^{2\pi}_0\cdots\int^{2\pi}_0}^{n}J(\theta)\dif\theta \\
&&= (\sign\tau)\frac1{(2\pi)^n}
\overbrace{\int^{2\pi}_0\cdots\int^{2\pi}_0}^{n}\Pa{\sum_{\sigma,\pi\in
S_n}(\sign\sigma)(\sign\pi)\zeta^{\sigma(1)+\pi(1)-n-1}_1\cdots\zeta^{\sigma(n)+\pi(n)-n-1}_n}\dif\theta\\
&&=(\sign\tau) \sum_{\sigma,\pi\in
S_n}(\sign\sigma)(\sign\pi)\Pa{\frac1{2\pi}\int^{2\pi}_0
\zeta^{\sigma(1)+\pi(1)-n-1}_1\dif\theta_1}\times\cdots\times
\Pa{\frac1{2\pi}\int^{2\pi}_0
\zeta^{\sigma(n)+\pi(n)-n-1}_n\dif\theta_n}\\
&&=(\sign\tau) \sum_{(\sigma,\pi)\in S_n\times S_n:\forall
j,\sigma(j)+\pi(j)=n+1}(\sign\sigma)(\sign\pi)=(\sign\tau)
\sum_{(\sigma,\pi)\in S_n\times
S_n:\pi=\tau\sigma}(\sign\sigma)(\sign\pi).
\end{eqnarray*}
Note that the sum is over all $(\sigma,\pi)\in S_n\times S_n$ such
that $\sigma(j)+\pi(j)=n+1$ for each $j\in\set{1,\ldots,n}$. In
other words, we get $\pi(j)=n+1-\sigma(j)=\tau(\sigma(j))$ for all
$j\in\set{1,\ldots,n}$, i.e. $\pi=\tau\sigma$. Thus the sum is equal
to
$$
(\sign\tau)\sum_{\sigma\in S_n}(\sign\sigma)(\sign\tau\sigma) = n!,
$$
which gives rise to
\begin{eqnarray*}
\frac1{(2\pi)^n}
\overbrace{\int^{2\pi}_0\cdots\int^{2\pi}_0}^{n}J(\theta)\dif\theta
=\frac1{n!}.
\end{eqnarray*}

Now the condition on the distance measure $\norm{\wtil
U-\I_n}_2\leqslant \epsilon$ is equivalent to
$$
\sum^n_{j=1}\abs{e^{\sqrt{-1}\theta_j}-1}^2\leqslant \epsilon^2.
$$
Using Euler's formula
$e^{\sqrt{-1}\theta}=\cos\theta+\sqrt{-1}\sin\theta$ and the fact
that $(\cos\theta-1)^2+\sin^2\theta=4\sin^2\Pa{\frac\theta2}$, we
get
\begin{eqnarray}
\norm{\wtil U-\I_n}_2\leqslant \epsilon \Longleftrightarrow
\sum^n_{j=1}\sin^2\Pa{\frac{\theta_j}2}\leqslant \frac{\epsilon^2}4.
\end{eqnarray}
Thus the (normalized) volume of the metric ball $B_\epsilon$ equals
the following:
\begin{eqnarray}
\vol(B_\epsilon):=\mu(B_\epsilon)=\int_{B_\epsilon} \dif \mu(\wtil
U).
\end{eqnarray}
By spectral decomposition of unitary matrix, we have
$$
\wtil U=\wtil V\wtil D\wtil V^*,~~\wtil
D=e^{\sqrt{-1}\Theta},~~\Theta=\Br{\begin{array}{cccc}
         \theta_1 & 0 & \cdots & 0 \\
         0 & \theta_2 & \cdots & 0 \\
         \vdots & \vdots & \ddots & 0 \\
         0 & 0 & \cdots & \theta_n
       \end{array}
},~~\theta_j\in[-\pi,\pi](j=1,\ldots,n).
$$
Hence
$$
\wtil V^*\cdot\dif \wtil U\cdot \wtil V = \dif\wtil D + \Br{\wtil
V^*\dif\wtil V,\wtil D},
$$
implying that
$$
\wtil V^*\cdot\wtil U^*\dif \wtil U\cdot \wtil V = \wtil
D^*\Pa{\dif\wtil D + \Br{\wtil V^*\dif\wtil V,\wtil D}}.
$$
Let $\dif \wtil G=\wtil U^*\dif \wtil U, \dif \wtil G_1=\wtil
V^*\dif \wtil V$ and $\dif \wtil X = \wtil V^*\cdot\dif \wtil G\cdot
\wtil V$. Thus
\begin{eqnarray}
[\dif \wtil X] = [\dif \wtil G]
\end{eqnarray}
because of $\wtil V\in\cU(n)/\cU(1)^{\times n}$. We also have
$$
\dif\wtil X = \wtil D^*\cdot\Pa{\dif\wtil D + \Br{\dif\wtil
G_1,\wtil D}} = \wtil D^*\cdot\wtil V^*\cdot\dif \wtil U\cdot \wtil
V,
$$
it follows that
\begin{eqnarray}
[\dif\wtil X] = [\dif\wtil U].
\end{eqnarray}
Apparently,
\begin{eqnarray}
B_\epsilon=\Set{\wtil V\wtil D\wtil V^*\in\cU(n): \wtil D\in\cU(n),
\wtil V\in\cU(n)/\cU(1)^{\times n},\norm{\wtil
D-\I_n}_2\leqslant\epsilon}.
\end{eqnarray}
But
\begin{eqnarray}
[\dif\wtil U] = \prod_{1\leqslant i<j\leqslant
n}\abs{e^{\sqrt{-1}\theta_i} - e^{\sqrt{-1}\theta_j}}^2 [\dif\wtil
D][\dif\wtil G_1],
\end{eqnarray}
therefore
\begin{eqnarray}
[\dif\wtil G] = \prod_{1\leqslant i<j\leqslant
n}\abs{e^{\sqrt{-1}\theta_i} - e^{\sqrt{-1}\theta_j}}^2 [\dif\wtil
D][\dif\wtil G_1],
\end{eqnarray}
together with the facts that the region in which
$(\theta_1,\ldots,\theta_n)$ lies is symmetric and $\wtil V\in
\cU(n)/\cU(1)^{\times n}$, implying
\begin{eqnarray*}
\int_{B_\epsilon}[\dif\wtil G] &=&\frac1{n!}
\int\cdots\int_{\stackrel{\theta_j\in[-\pi,\pi](j=1,\ldots,n);}{
\sum^n_{j=1}\sin^2(\theta_j/2)\leqslant
\frac{\epsilon^2}4}}\prod_{1\leqslant i<j\leqslant
n}\abs{e^{\sqrt{-1}\theta_i} - e^{\sqrt{-1}\theta_j}}^2 [\dif\wtil
D]\times \int_{\cU(n)}[\dif\wtil G_1].
\end{eqnarray*}
That is,
\begin{eqnarray*}
\int_{B_\epsilon}\dif \mu(\wtil U) &=&\frac1{(2\pi)^nn!}
\int\cdots\int_{\stackrel{\theta_j\in[-\pi,\pi](j=1,\ldots,n);}{
\sum^n_{j=1}\sin^2(\theta_j/2)\leqslant
\frac{\epsilon^2}4}}\prod_{1\leqslant i<j\leqslant
n}\abs{e^{\sqrt{-1}\theta_i} - e^{\sqrt{-1}\theta_j}}^2 [\dif\wtil
D],
\end{eqnarray*}
where
$$
\dif\mu(\wtil U) = \frac{[\dif\wtil G]}{\int_{\cU(n)}[\dif\wtil
G]}~~\text{and}~~\int_{\cU(n)}[\dif\wtil G] =
(2\pi)^n\int_{\cU_1(n)}[\dif\wtil G_1].
$$
From this, we get
\begin{eqnarray}\label{eq:vol-ball-integral}
\vol\Pa{B_\epsilon} =
\int\cdots\int_{\stackrel{\theta_j\in[-\pi,\pi](j=1,\ldots,n);}{
\sum^n_{j=1}\sin^2(\theta_j/2)\leqslant \frac{\epsilon^2}4}}
p(\theta_1,\ldots,\theta_n)\prod^n_{j=1}\dif{\theta_j},
\end{eqnarray}
where $\epsilon\in[0,2\sqrt{n}]$. For the maximal distance
$\epsilon=2\sqrt{n}$, the restriction
$\sum^n_{j=1}\sin^2(\theta_j/2)\leqslant \frac{\epsilon^2}4$ becomes
irrelevant and $\vol\Pa{B_{2\sqrt{n}}}=1$.

We start by rewriting the $n$-dimensional integral
\eqref{eq:vol-ball-integral}, with the help of a Dirac delta
function, as
\begin{eqnarray}\label{eq:vol-ball-integral-1}
\vol\Pa{B_\epsilon} =
\int^{\frac{\epsilon^2}4}_0\int^\pi_{-\pi}\cdots\int^\pi_{-\pi}
\delta\Pa{t- \sum^n_{j=1}\sin^2\Pa{\frac{\theta_j}2}}
p(\theta_1,\ldots,\theta_n)\prod^n_{j=1}\dif{\theta_j}\dif t.
\end{eqnarray}
We know that
$$
\int^\infty_{-\infty} \delta(t-a)f(t)\dif t = f(a).
$$
If $f(t)\equiv\I_{\set{t:\alpha\leqslant t\leqslant \beta}}$, then
$$
\int^\beta_\alpha \delta(t-a)\dif t = \int^\infty_{-\infty}
\delta(t-a)\I_{\set{t:\alpha\leqslant t\leqslant \beta}}\dif t =
\I_{\set{t:\alpha\leqslant t\leqslant \beta}}.
$$
By using the Fourier representation of Dirac Delta function
$$
\delta(t-a) =
\frac1{2\pi}\int^\infty_{-\infty}e^{\sqrt{-1}(t-a)s}\dif s,
$$
we get
\begin{eqnarray*}
\int^{\frac{\epsilon^2}4}_0 \delta(t-a)\dif t &=&
\frac1{2\pi}\int^\infty_{-\infty}\Pa{\int^{\frac{\epsilon^2}4}_0
e^{\sqrt{-1}(t-a)s}\dif t}\dif s\\
&=&\frac1{2\pi}\int^\infty_{-\infty}
\frac{\sqrt{-1}\Pa{1-e^{\sqrt{-1}\frac{\epsilon^2}4s}}}{se^{\sqrt{-1}as}}\dif
s.
\end{eqnarray*}
Let $a=\sum^n_{j=1}\sin^2\Pa{\frac{\theta_j}2} = \frac
n2-\sum^n_{j=1}\frac12\cos\theta_j$. Indeed,
\begin{eqnarray*}
\frac n2-\sum^n_{j=1}\sin^2(\theta_j/2) &=& -\frac n2 +
\sum^n_{j=1}(1-\sin^2(\theta_j/2)) = -\frac n2 +
\sum^n_{j=1}\cos^2(\theta_j/2)\\
&=& -\frac n2 + \sum^n_{j=1}\frac12\Pa{1+\cos\theta_j}
=\sum^n_{j=1}\frac12\cos\theta_j.
\end{eqnarray*}
Therefore
\begin{eqnarray*}
\int^{\frac{\epsilon^2}4}_0
\delta\Pa{t-\sum^n_{j=1}\sin^2\Pa{\frac{\theta_j}2} }\dif t
&=&\frac1{2\pi}\int^\infty_{-\infty}
\frac{\sqrt{-1}\Pa{1-e^{\sqrt{-1}\frac{\epsilon^2}4s}}}{se^{\sqrt{-1}\frac
n2s}}e^{\sqrt{-1}s\sum^n_{j=1}\frac{\cos\theta_j}2}\dif s\\
&=&\frac1{2\pi}\int^\infty_{-\infty}
\frac{\sqrt{-1}\Pa{1-e^{\sqrt{-1}\frac{\epsilon^2}4s}}}{se^{\sqrt{-1}\frac
n2s}}\Pa{\prod^n_{j=1}e^{\sqrt{-1}s\frac{\cos\theta_j}2}}\dif s
\end{eqnarray*}
Inserting this formula into \eqref{eq:vol-ball-integral-1} and
performing the integration over $t$ first, we have
\begin{eqnarray*}
\vol(B_\epsilon) &=&
\frac1{2\pi}\int^\infty_{-\infty}\frac{\sqrt{-1}\Pa{1-e^{\sqrt{-1}\frac{\epsilon^2}4s}}}{se^{\sqrt{-1}\frac
n2s}}\Pa{\int_{[-\pi,\pi]^n}
p(\theta_1,\ldots,\theta_n)\prod^n_{j=1}e^{\sqrt{-1}s\frac{\cos\theta_j}2}\dif\theta_j}\\
&=&\frac1{2\pi}\int^\infty_{-\infty}\frac{\sqrt{-1}\Pa{1-e^{\sqrt{-1}\frac{\epsilon^2}4s}}}{se^{\sqrt{-1}\frac
n2s}}D_n(s)\dif s,
\end{eqnarray*}
where
\begin{eqnarray*}
D_n(s) = \int_{[-\pi,\pi]^n}
p(\theta_1,\ldots,\theta_n)\prod^n_{j=1}e^{\sqrt{-1}s\frac{\cos\theta_j}2}\dif\theta_j.
\end{eqnarray*}
Since
\begin{eqnarray*}
\prod_{1\leqslant i<j\leqslant n}\abs{e^{\sqrt{-1}\theta_i} -
e^{\sqrt{-1}\theta_j}}^2=
\det\Pa{e^{\sqrt{-1}(i-1)\theta_k}}\overline{\det\Pa{e^{\sqrt{-1}(i-1)\theta_k}}}
\end{eqnarray*}
is a product of two Vandermonde determinants where $i,k=1,\ldots,n$.
The following fact will be used.
\begin{prop}[Andr\'{e}ief's identity]
For two $n\times n$ matrices $M(\mathbf{x})$ and $N(\mathbf{x})$,
defined by the following:
\begin{eqnarray*}
M(\mathbf{x}) = \Br{\begin{array}{cccc}
                          M_1(x_1) & M_1(x_2) & \cdots & M_1(x_n) \\
                          M_2(x_1) & M_2(x_2) & \cdots & M_2(x_n)\\
                          \vdots & \vdots & \ddots & \vdots \\
                          M_n(x_1) & M_n(x_2) & \cdots & M_n(x_n)
                        \end{array}
},N(\mathbf{x}) = \Br{\begin{array}{cccc}
                          N_1(x_1) & N_1(x_2) & \cdots & N_1(x_n) \\
                          N_2(x_1) & N_2(x_2) & \cdots & N_2(x_n)\\
                          \vdots & \vdots & \ddots & \vdots \\
                          N_n(x_1) & N_n(x_2) & \cdots & N_n(x_n)
                        \end{array}
}
\end{eqnarray*}
and a function $w(\cdot)$ such that the integral
$$
\int^b_a M_i(x)N_j(x)w(x)\dif x
$$
exists, then the following multiple integral can be evaluated as
\begin{eqnarray}
\int\cdots\int_{\Delta_{a,b}}\det(M(\mathbf{x}))\det(N(\mathbf{x}))\prod^n_{j=1}w(x_j)\dif
x_j = \det\Pa{\int^b_a M_i(t)N_j(t)w(t)\dif t},
\end{eqnarray}
where $\Delta_{a,b}:=\set{\mathbf{x}=(x_1,\ldots,x_n)\in\real^n:
b\geqslant x_1\geqslant x_2\geqslant \cdots\geqslant x_n\geqslant
a}$.
\end{prop}
By invoking this identity, we know that
\begin{eqnarray*}
D_n(s) &=& \frac1{(2\pi)^nn!}\int\cdots\int_{[-\pi,\pi]^n}
\det\Pa{e^{\sqrt{-1}(i-1)\theta_k}}\det\Pa{e^{-\sqrt{-1}(i-1)\theta_k}}\prod^n_{j=1}e^{\sqrt{-1}s\frac{\cos\theta_j}2}\dif\theta_j\\
&=&
\frac1{(2\pi)^n}\det\Pa{\int^\pi_{-\pi}e^{\sqrt{-1}(i-j)\theta}e^{\sqrt{-1}s\frac{\cos\theta}2}\dif\theta}.
\end{eqnarray*}
In what follows, we need the Bessel function of the first kind. The
definition is
\begin{framed}
\begin{eqnarray*}
J_n(x) = \frac1{2\pi}\int^\pi_{-\pi}
e^{\sqrt{-1}(n\theta-x\sin\theta)}\dif\theta =
\sum^\infty_{j=0}\frac{(-1)^j}{\Gamma(j+n+1)j!}\Pa{\frac x2}^{2j+n}.
\end{eqnarray*}
\end{framed}
We can list some properties of this Bessel function:
\begin{enumerate}[(i)]
\item $\int^\pi_{-\pi} e^{\sqrt{-1}(n\theta+x\sin\theta)}\dif \theta = \int^\pi_{-\pi} e^{\sqrt{-1}(-n\theta-x\sin\theta)}\dif \theta$
\item $J_n(x) = \frac1{2\pi}\int^\pi_{-\pi}
e^{\sqrt{-1}(n\theta-x\sin\theta)}\dif\theta =
\frac1{2\pi}\int^\pi_{-\pi}
e^{\sqrt{-1}(-n\theta+x\sin\theta)}\dif\theta$
\item $J_{-n}(x) = \frac1{2\pi}\int^\pi_{-\pi}
e^{\sqrt{-1}(-n\theta-x\sin\theta)}\dif\theta =
\frac1{2\pi}\int^\pi_{-\pi}
e^{\sqrt{-1}(n\theta+x\sin\theta)}\dif\theta$
\end{enumerate}
We claim that
\begin{framed}
\begin{eqnarray}
\int^\pi_{-\pi} e^{\sqrt{-1}(n\theta+x\cos\theta)}\dif\theta =2\pi
e^{\sqrt{-1}\frac{n\pi}2} J_n(x).
\end{eqnarray}
\end{framed}
Indeed, since
\begin{eqnarray*}
\int^\pi_{-\pi} e^{\sqrt{-1}(n\theta+x\cos\theta)}\dif\theta
=e^{-\sqrt{-1}\frac{n\pi}2}\int^\pi_{-\pi}
e^{\sqrt{-1}(n\theta+x\sin\theta)}\dif\theta
=e^{-\sqrt{-1}\frac{n\pi}2}\cdot 2\pi J_{-n}(x)
\end{eqnarray*}
and $J_{-n}(x)=(-1)^nJ_n(x)$, the desired claim is obtained.

Thus
\begin{eqnarray*}
\int^\pi_{-\pi}e^{\sqrt{-1}(i-j)\theta}e^{\sqrt{-1}s\frac{\cos\theta}2}\dif\theta
&=& \int^\pi_{-\pi}e^{\sqrt{-1}\Br{(i-j)\theta+\frac
s2\cos\theta}}\dif\theta\\
&=& 2\pi e^{\sqrt{-1}\frac{i-j}2\pi}J_{i-j}\Pa{\frac s2}.
\end{eqnarray*}
We now have
\begin{framed}
\begin{eqnarray*}
D_n(s) = \frac1{(2\pi)^n}\det\Pa{2\pi
e^{\sqrt{-1}\frac{i-j}2\pi}J_{i-j}\Pa{\frac s2}} =
\frac1{(2\pi)^n}\det\Pa{J_{i-j}\Pa{\frac s2}}.
\end{eqnarray*}
\end{framed}
By the definition of $J_n(x)$, we have $J_n(-x)=(-1)^n J_n (x)$, and
furthermore,
$$
D_n(-s)=D_n(s).
$$
Therefore
\begin{eqnarray*}
\vol(B_\epsilon)
&=&\frac1{2\pi}\int^\infty_{-\infty}\frac{\sqrt{-1}\Pa{1-e^{\sqrt{-1}\frac{\epsilon^2}4s}}}{se^{\sqrt{-1}\frac
n2s}}D_n(s)\dif s\\
&=& \frac1\pi\int^\infty_0
\frac{\sin\Pa{\frac{ns}2}+\sin\Pa{\Pa{\frac{\epsilon^2}4-\frac
n2}s}}{s}\det\Pa{J_{i-j}\Pa{\frac s2}}\dif s,
\end{eqnarray*}
where we used Euler's formula. In fact, we have
\begin{eqnarray*}
&&\frac{1-e^{\sqrt{-1}\frac{\epsilon^2}4s}}{se^{\sqrt{-1}\frac n2
s}} = \frac{e^{-\sqrt{-1}\frac n2
s}-e^{\sqrt{-1}\Pa{\frac{\epsilon^2}4-\frac
n2 }s}}{s} \\
&&= \frac{\cos\Pa{\frac n2 s}-\cos\Pa{\Pa{\frac{\epsilon^2}4-\frac
n2}s}}s + \sqrt{-1}\frac{-\sin\Pa{\frac n2 s} -
\sin\Pa{\Pa{\frac{\epsilon^2}4-\frac n2}s}}s.
\end{eqnarray*}
Finally, we get the volume of a metric ball in unitary group
$\cU(n)$:
\begin{framed}
\begin{eqnarray}
\vol(B_\epsilon) = \frac1\pi\int^\infty_0
\frac{\sin\Pa{\frac{ns}2}+\sin\Pa{\Pa{\frac{\epsilon^2}4-\frac
n2}s}}{s}\det\Pa{J_{i-j}\Pa{\frac s2}}\dif s.
\end{eqnarray}
\end{framed}

\section{Appendix I: Volumes of a sphere and a ball}\label{AppI}
The following result is very important in deriving the volumes of a
sphere and a ball in $\mathbb{F}^n$ where
$\mathbb{F}=\real,\complex$.
\begin{prop}
It holds that
\begin{framed}
\begin{eqnarray}
\int_{\complex^n} \delta(t-\iinner{\psi}{\psi})[\dif\psi] =
\frac{\pi^n}{\Gamma(n)}t^{n-1}\quad(t\geqslant0).
\end{eqnarray}
\end{framed}
In particular,
\begin{eqnarray}
\int_{\complex^n} \delta(1-\iinner{\psi}{\psi})[\dif\psi] =
\frac{\pi^n}{\Gamma(n)}.
\end{eqnarray}
\end{prop}

\begin{proof}
Recall that
$$
[\dif\psi] = [\dif(\re(\psi))][\dif(\im(\psi))]=\prod^n_{j=1}\dif
x_j\dif y_j
$$
for $\ket{\psi}=\sum^n_{j=1}\psi_j\ket{j}$ and
$\psi_j=x_j+\mathrm{i}y_j\in\real+\mathrm{i}\real(j=1,\ldots,n)$. We
now have
\begin{eqnarray*}
\delta(1-\iinner{\psi}{\psi})[\dif\psi] =
\delta\Pa{1-\sum^n_{j=1}(x^2_j+y^2_j)}\prod^n_{j=1}\dif x_j\dif y_j.
\end{eqnarray*}
Let
$$
F(t)=\frac{\Gamma(n)}{\pi^n}\int\delta\Pa{t-\sum^n_{j=1}(x^2_j+y^2_j)}\prod^n_{j=1}\dif
x_j\dif y_j.
$$
Then the Laplace transform of $F(t)$ is given by:
\begin{eqnarray*}
\sL(F)(s) &=&
\frac{\Gamma(n)}{\pi^n}\int\Br{\int^\infty_0\delta\Pa{t-\sum^n_{j=1}(x^2_j+y^2_j)}e^{-st}\dif
t}\prod^n_{j=1}\dif x_j\dif y_j\\
&=&\frac{\Gamma(n)}{\pi^N}\Br{\int^\infty_{-\infty}
\prod^n_{j=1}e^{-sx^2_j}\dif x_j}\Br{\int^\infty_{-\infty}
\prod^n_{j=1}e^{-sy^2_j}\dif y_j} = s^{-n}\Gamma(n)
\end{eqnarray*}
implying via the inverse Laplace transform that $F(t)=t^{n-1}$. This
indicates that
\begin{eqnarray*}
\int\delta(t-\iinner{\psi}{\psi})[\dif\psi]
=\frac{\pi^n}{\Gamma(n)}t^{n-1}.
\end{eqnarray*}
In particular, for $t=1$, we get the desired result.
\end{proof}

In the following, we will give a method via Dirac delta function to
calculate the volume of $n$-dimensional ball in $\real^n$. Moreover,
based on this formula, we then give the surface of
$(n-1)$-dimensional sphere in $\real^n$. Denote the ball of radius
$R$ in $\real^n$ as
$$
\mathbb{B}_n(R):=\Set{(x_1,\ldots,x_n)\in\real^n:
\sum^n_{j=1}x^2_j\leqslant R^2}
$$
and the sphere of radius $R$ in $\real^n$ as
$$
\mathbb{S}^{n-1}(R):=\Set{(x_1,\ldots,x_n)\in\real^n:
\sum^n_{j=1}x^2_j= R^2}.
$$
Now the volume of the ball of radius $R$ is given by
\begin{eqnarray}
\vol(\mathbb{B}_n(R)) = \int^{R^2}_0\dif t \int_{\real^n}\delta\Pa{t
- \sum^n_{j=1}x^2_j}\prod^n_{j=1}\dif x_j.
\end{eqnarray}
Denote
$$
F(t) = \int\delta\Pa{t - \sum^n_{j=1}x^2_j}\prod^n_{j=1}\dif x_j.
$$
Then its Laplace transform is given by
$$
\sL(F)(s) = \pi^{\frac n2}s^{-\frac n2}.
$$
This implies that
$$
F(t) =\frac{\pi^{\frac n2}}{\Gamma\Pa{\frac n2}}t^{\frac n2-1}.
$$
Thus
\begin{eqnarray}
\vol(\mathbb{B}_n(R)) = \int^{R^2}_0 F(t)\dif t = \frac{\pi^{\frac
n2}}{\Gamma\Pa{\frac n2}}\frac{R^n}{\frac n2} = \frac{\pi^{\frac
n2}}{\Gamma\Pa{\frac n2+1}}R^n.
\end{eqnarray}
By differentiating the above formula with respect to the radius $R$,
we get that
\begin{eqnarray}
\vol\Pa{\mathbb{S}^{n-1}(R)}=\frac{\dif\vol(\mathbb{B}_n(R))}{\dif
R} = 2R\cdot F(R^2) = \frac{\pi^{\frac n2}}{\Gamma\Pa{\frac
n2+1}}nR^{n-1} = \frac{2\pi^{\frac n2}}{\Gamma\Pa{\frac n2}}R^{n-1}.
\end{eqnarray}
From this, we see that
\begin{eqnarray}
\vol\Pa{\mathbb{S}^{n-1}(R)}= 2R\cdot \int\delta\Pa{R^2 -
\sum^n_{j=1}x^2_j}\prod^n_{j=1}\dif x_j,
\end{eqnarray}
where
\begin{eqnarray}
\int\delta\Pa{R^2 - \sum^n_{j=1}x^2_j}\prod^n_{j=1}\dif x_j =
\frac{\pi^{\frac n2}}{\Gamma\Pa{\frac n2}}R^{n-2}.
\end{eqnarray}
In particular, for $R=1$, we see that
\begin{eqnarray}
\vol\Pa{\mathbb{S}^{n-1}(1)}= 2\int\delta\Pa{1 -
\sum^n_{j=1}x^2_j}\prod^n_{j=1}\dif x_j =2\int_{\real^n}\delta\Pa{1
- \iinner{u}{u}}[\dif u],
\end{eqnarray}
where
\begin{eqnarray}
\int\delta\Pa{1 - \sum^n_{j=1}x^2_j}\prod^n_{j=1}\dif x_j
=\int_{\real^n}\delta\Pa{1 - \iinner{u}{u}}[\dif u]=
\frac{\pi^{\frac n2}}{\Gamma\Pa{\frac n2}}.
\end{eqnarray}
We also see that the volume of the ball of radius $R$ in
$\complex^n\cong\real^{2n}$ is given by
\begin{eqnarray}
\vol(\mathbb{B}_n(R,\complex)) =\int^{R^2}_0\dif
t\int_{\complex^n}\delta\Pa{t-\iinner{\psi}{\psi}}[\dif\psi] =
\frac{\pi^n}{\Gamma\Pa{n+1}}R^{2n}.
\end{eqnarray}
The surface of the sphere of radius $R$ in $\complex^n$ is given by
\begin{eqnarray}
\vol\Pa{\mathbb{S}^{2n-1}(R,\complex)} = 2R\cdot
\int_{\complex^n}\delta\Pa{R^2-\iinner{\psi}{\psi}}[\dif\psi]=\frac{2\pi^n}{\Gamma\Pa{n}}R^{2n-1},
\end{eqnarray}
where
\begin{eqnarray}
\int_{\complex^n}\delta\Pa{R^2-\iinner{\psi}{\psi}}[\dif\psi]=\frac{\pi^n}{\Gamma\Pa{n}}R^{2n-2}.
\end{eqnarray}
Recall a property of Dirac delta function:
\begin{eqnarray}
\delta\Pa{a^2-x^2} = \frac1{2\abs{a}}(\delta(a-x)+\delta(a+x)).
\end{eqnarray}
From this, we see that
\begin{eqnarray}
\delta\Pa{R^2-\iinner{\psi}{\psi}} = \delta\Pa{R^2-\norm{\psi}^2} =
\frac1{2R}(\delta(R-\norm{\psi})+\delta(R+\norm{\psi})).
\end{eqnarray}
Clearly $\delta(R+\norm{\psi})=0$ since $R+\norm{\psi}>0$. Therefore
we obtain
\begin{framed}
\begin{eqnarray}
2R\cdot\delta\Pa{R^2-\iinner{\psi}{\psi}} = \delta(R-\norm{\psi}).
\end{eqnarray}
\end{framed}
In summary, we have the following result.
\begin{prop}
The surface of sphere of radius $R$ in $\complex^n$ can be
represented by
\begin{framed}
\begin{eqnarray}
\vol\Pa{\mathbb{S}^{2n-1}(R,\complex)} =
\int_{\complex^n}\delta\Pa{R-\norm{\psi}}[\dif\psi]=\frac{2\pi^n}{\Gamma\Pa{n}}R^{2n-1}.
\end{eqnarray}
\end{framed}
\end{prop}

We remark here that
\begin{framed}
\begin{eqnarray}
\int_{\real^n}f(x)[\dif x] = \int^\infty_0\dif r\int_{\real^n}[\dif
u] r^{n-1}f(r\cdot u)\delta(1-\norm{u}).
\end{eqnarray}
\end{framed}
In particular, if $f$ is independent of $\mathbf{u}$, i.e., $f(r
u)=f(r)$, then we have that
\begin{eqnarray}
\int_{\real^n}f(x)[\dif x] = \int^\infty_0r^{n-1}f(r)\dif
r\times\int_{\real^n}\delta(1-\norm{u})[\dif u] = \frac{2\pi^{\frac
n2}}{\Gamma\Pa{\frac n2}}\int^\infty_0r^{n-1}f(r)\dif r.
\end{eqnarray}
Indeed, For a Gaussian random vector
$\omega=[\omega_1,\ldots,\omega_n]^\t\in\real^n$ with i.i.d.
standard normal random variable $\omega_j\sim N(0,1)$, we have
\begin{eqnarray}
1=(2\pi)^{-\frac n2}\int_{\real^n}
\exp\Pa{-\frac12\norm{\omega}^2}[\dif \omega].
\end{eqnarray}
By using polar coordinate of $\omega$, we see
$\omega=\norm{\omega}\cdot u=r\cdot u$, where $\norm{u}=1$. Thus
\begin{eqnarray}
[\dif\omega] = \dif r\times\delta(r-\norm{\omega})[\dif\omega]=
r^{n-1}\dif r\times \delta(1-\norm{u})[\dif u],
\end{eqnarray}
where
\begin{eqnarray*}
\fbox{$\delta(1-\norm{u})[\dif u] = 2\delta(1-\iinner{u}{u})[\dif
u].$}
\end{eqnarray*}
Indeed, the truth is checked as follows,
\begin{eqnarray*}
&&(2\pi)^{-\frac
n2}\int_{\real^n}\exp\Pa{-\frac12\norm{\omega}^2}[\dif\omega]
\\&&=(2\pi)^{-\frac n2}\int^\infty_0\exp\Pa{-\frac12r^2}r^{n-1}\dif
r\int_{\real^n} \delta(1-\norm{u})[\dif u]\\
&&=(2\pi)^{-\frac n2} 2^{\frac n2-1}\Gamma\Pa{\frac
n2}\int_{\real^n} \delta(1-\norm{u})[\dif u]\\
&&=\frac{\Gamma\Pa{\frac n2}}{2\pi^{\frac n2}}\int_{\real^n}
\delta(1-\norm{u})[\dif u]=\frac{\Gamma\Pa{\frac n2}}{2\pi^{\frac
n2}}\frac{2\pi^{\frac n2}}{\Gamma\Pa{\frac n2}}=1.
\end{eqnarray*}

\section{Appendix II: Some useful facts}

\subsection{Matrices with simple eigenvalues form open dense sets of full measure}

The present subsection is written based on the Book by Deift and
Gioev \cite{Deift}.

Let $\sH_n(\complex)$ be the set of all $n\times n$ Hermitian
matrices with simple spectrum, i.e. the multiplicity of each
eigenvalue are just one. Next we show that $\sH_n(\complex)$ is an
\blue{open and dense set of full measure} (i.e. the Lebesgue measure
of the complement is vanished) in $\Her{\complex^n}\simeq
\real^{n^2}$, the set of all $n\times n$ Hermitian matrices, that
is, $\real^{n^2}\backslash \sH_n(\complex)$ is a set of zero-measure
in the sense of Lebesgue measure.

Assume that $\wtil M$ is an arbitrary Hermitian matrix in
$\sH_n(\complex)$. with simple spectrum $\mu_1<\cdots<\mu_n$, then
by standard perturbation theory, all matrices in a neighborhood of
$\wtil M$ have simple spectrum. Moreover if $\wtil
H\in\Her{\complex^n}$ with eigenvalues $\set{h_j:j=1,\ldots,n}$,
then by spectral theorem, $\wtil H=\wtil U\Lambda \wtil U^*$ for
some unitary $\wtil U$ and $\Lambda=\diag(h_1,\ldots,h_n)$. Now we
can always find $\epsilon_j$ arbitrarily small for all $j$ so that
$h_j+\epsilon_j$ are distinct for all $j$. Thus
$$
\wtil H_\epsilon := \wtil
U\diag(h_1+\epsilon_1,\ldots,h_n+\epsilon_n)\wtil U^*
$$
is a Hermitian matrix with simple spectrum, arbitrarily close to
$\wtil H$. The above two facts show that $\sH_n(\complex)$ is open
and dense. In order to show that $\sH_n(\complex)$ is of full
measure, consider the discriminant:
$$
\Delta(\lambda): = \prod_{1\leqslant i<j\leqslant
n}(\lambda_j-\lambda_i)^2.
$$
By the fundamental theorem of symmetric functions, $\Delta$ is a
polynomial function of the elementary symmetric functions of the
$\lambda_j$'s, and hence a polynomial function of the $n^2$ entries
$\wtil H_{11},\ldots,\wtil H_{nn},\re(\wtil H_{ij}), \im(\wtil
H_{ij})(i<j)$ of $\wtil H$. Now if $\sH_n(\complex)$ were not of
full measure, then $\Delta$ would vanish on a set of positive
measure in $\real^{n^2}$. It follows that $\Delta\equiv0$ because
$\Delta$ is polynomial in $\real^{n^2}$.

Let $H=\diag(1,2,\ldots,n)$ is a Hermitian matrix with distinct
spectrum and $\Delta(H)\neq0$. This gives a contradiction and so
$\sH_n(\complex)$ is of full measure in $\Her{\complex^n}$.

Similarly, the set $\sS_n(\real)$ of all $n\times n$ real symmetric
matrices with simple spectrum is an \blue{open and dense set of full
measure} in the set $\bS\Pa{\real^n}$ of all real symmetric
matrices.

From the above discussion, we conclude that the set of all density
matrices with distinct positive eigenvalues is an open and dense set
of full measure in the set of all density matrices.

We begin by considering real symmetric matrices of size $n$, which
is the simplest case. Because $\sS_n(\real)$ is of full measure, it
is sufficient for the purpose of integration to restrict our
attention to matrices $M\in\sS_n(\real)$. Let $\mu_1<\cdots<\mu_n$
denote the eigenvalues of $M$. By spectral theorem, $M=U\Lambda
U^\t$, where $\Lambda=\diag(\mu_1,\ldots,\mu_n)$. Clearly the
columns of $U$ are defined only up to multiplication by $\pm1$, and
so the map $M\mapsto (\Lambda, U)$:
$$
\sS_n(\real)\ni M\mapsto (\Lambda, U)\in\real^n_\uparrow\times
\cO(n)
$$
is not well-defined, where
$$
\real^n_\uparrow=\Set{(\mu_1,\ldots,\mu_n)\in\real^n:
\mu_1<\cdots<\mu_n},~~~\cO(n)=n\times n~\text{orthogonal group}.
$$
We consider instead the map
\begin{eqnarray}
\phi_1:\sS_n(\real)\ni M\mapsto (\Lambda,\widehat U)\in
\real^n_\uparrow\times (\cO(n)/K_1)
\end{eqnarray}
where $K_1$ is the closed subgroup of $\cO(n)$ containing $2^n$
elements of the form $\diag(\pm1,\ldots,\pm1)$, $\cO(n)/K_1$ is the
homogeneous manifold obtained by factoring $\cO(n)$ by $K_1$, and
$\widehat U=UK_1$ is the coset containing $U$. The map $\phi_1$ is
now clearly well-defined.

The differentiable structure on $\cO(n)/K_1$ is described in the
following general result about homogeneous manifolds:
\begin{prop}\label{prop:dif-structure}
Let $K$ be a closed subgroup of a Lie group $G$ and let $G/K$ be the
set $\Set{gK: g\in G}$ of left cosets module $K$. Let $\pi:G\to G/K$
denote the natural projection $\pi(g)=gK$. Then $G/K$ has a unique
manifold structure such that $\pi$ is $C^\infty$ and there exist
local smooth sections of $G/K$ in $G$, i.e., if $gK\in G/K$, there
is a neighborhood $W$ of $gK$ and a $C^\infty$ map $\tau: W\to G$
such that $\pi\circ\tau=\mathrm{id}$.
\end{prop}
In other words, for each $gK$, it is possible to choose a $g'\in
gK\subset G$ such that the map $gK\mapsto g'\equiv \tau(gK))$ is
locally defined and smooth.

For example, if $G=\real^\times$, the multiplication group of
nonzero real numbers, and if $K$ is the subgroup $\set{\pm1}$, then
$G/K\cong\set{x>0}$ and $\pi(a)=\abs{a}\cong\set{\pm a}$. If
$gK=\set{\pm a}$ for some $a>0$, then $\tau(\set{\pm a})=a$. Also,
if $G=\cO(2)$ and $K_1$ consists of four elements of the form
$$
\Br{\begin{array}{cc}
  \pm1 & 0 \\
  0 & \pm1
\end{array}},\quad g=\Br{\begin{array}{cc}
                           u_{11} & u_{12} \\
                           u_{21} & u_{22}
                         \end{array}
}\in\cO(2),
$$
then
$$
gK_1 = \Set{\Br{\begin{array}{cc}
                           u_{11} & u_{12} \\
                           u_{21} & u_{22}
                         \end{array}
},\Br{\begin{array}{cc}
                           -u_{11} & u_{12} \\
                           -u_{21} & u_{22}
                         \end{array}
},\Br{\begin{array}{cc}
                           u_{11} & -u_{12} \\
                           u_{21} & -u_{22}
                         \end{array}
},\Br{\begin{array}{cc}
                           -u_{11} & -u_{12} \\
                           -u_{21} & -u_{22}
                         \end{array}
}}.
$$
Since $u^2_{11}+u^2_{21}=1=u^2_{12}+u^2_{22}$, each column of $g$
contains at least one nonzero number: For example, suppose $u_{21}$
and $u_{22}$ are nonzero; then there exists a \emph{unique} $g'\in
gK$ such that the elements in the second row of $g'$ are positive.
The same is true for all $g''K$ close to $gK$. Then $g''K\mapsto
g''$ is the desired (local map) $\tau$. We can generalize the above
construction on $\cO(n)/K_1$. We will prove the following result:
\begin{prop}\label{prop:diffeomorphism-real}
$\phi_1$ is a diffeomorphism from $\sS_n(\real)$ onto
$\real^n_\uparrow\times (\cO(n)/K_1)$.
\end{prop}

\begin{proof}
Note that (here and below we always speak of the \blue{real}
dimensions)
\begin{eqnarray}
\dim\Pa{\real^n_\uparrow\times (\cO(n)/K_1)} =
N+\frac{N(N-1)}2=\frac{N(N+1)}2
\end{eqnarray}
as it should.

Define the map $\hat \phi_1:\real^n_\uparrow\times(\cO(n)/K_1)\to
\sS_n(\real)$ as follows: If $(\Lambda,\widehat U)$ lies in
$\real^n_\uparrow\times(\cO(n)/K_1)\to \sS_n(\real)$, then
\begin{eqnarray}
\hat \phi_1(\Lambda, \widehat U) = U\Lambda U^\t
\end{eqnarray}
where $U$ is \blue{any} matrix in the coset $\widehat U$. If $U'$ is
another such matrix, then $U'=Uh$ for some $h\in K_1$ and so
$$
U'\Lambda (U')^\t = Uh\Lambda h^\t U^\t = U\Lambda U^\t.
$$
Hence $\hat \phi_1$ is well-defined. We will show that
\begin{eqnarray}\label{eq:phi-1}
\phi_1\circ\hat \phi_1=\mathrm{id}_{\real^n_\uparrow\times
(\cO(n)/K_1)}
\end{eqnarray}
and
\begin{eqnarray}\label{eq:hat-phi-1}
\hat \phi_1\circ\phi_1=\mathrm{id}_{\sS_n(\real)}.
\end{eqnarray}
Indeed,
\begin{eqnarray}
\phi_1\Pa{\hat\phi_1\Pa{\Lambda,\widehat U}} &=& \phi_1\Pa{U\Lambda
U^\t},~~U\in\widehat U\\
&=& \Pa{\Lambda, UK_1} = \Pa{\Lambda, \widehat U}.
\end{eqnarray}
Conversely, if $M=U\Lambda U^\t\in \sS_n(\real)$, then
\begin{eqnarray}
\hat\phi_1(\phi_1(M)) = \hat \phi_1\Pa{\Lambda,\widehat
U=UK_1}=U\Lambda U^\t=M.
\end{eqnarray}
This proves Eq.~\eqref{eq:phi-1} and Eq.~\eqref{eq:hat-phi-1}. In
order to prove that $\phi_1$ is a diffeomorphism, it suffices to
show that $\phi_1$ and $\hat\phi_1$ are smooth.

The smoothness of $\phi_1$ follows from perturbation theory: Fix
$M_0=U_0\Lambda_0 U^\t_0\in\sS_n(\real)$, where
$\Lambda_0=\diag(\mu_{01},\ldots,\mu_{0n})$ for
$\mu_{01}<\cdots<\mu_{0n}$, and $U_0\in\cO(n)$. Then for $M$ near
$M_0$, $M\in\sS_n(\real)$, the eigenvalues of $M$,
$\mu_1(M)<\cdots<\mu_n(M)$, are smooth functions of $M$. Moreover,
the associated eigenvectors $u_j(M)$,
$$
Mu_j(M) = \lambda_j(M)u_j(M),~~1\leqslant j\leqslant n,
$$
can be chosen orthogonal
$$
\Inner{u_i(M)}{u_j(M)}=\delta_{ij},~~~(1\leqslant i,j\leqslant n),
$$
and smooth in $M$. Indeed, for any $j$ with $1\leqslant j\leqslant
n$, let $P_j$ be the orthogonal projection
\begin{eqnarray}
P_j(M) = \frac1{2\pi\mathrm{i}}\oint_{\Gamma_j}\frac{1}{s-M}\dif s
\end{eqnarray}
where $\Gamma_j$ is a small circle of radius $\epsilon$ around
$\mu_{0j}$, $\abs{\mu_{0i}-\mu_{0j}}>\epsilon$ for $i\neq j$. Then,
where $u_j(M_0)$ is the $j$-th column of $M_0$,
$$
u_j(M) =
\frac{P_j(M)u_j(M_0)}{\sqrt{\Inner{u_j(M_0)}{P_j(M)u_j(M_0)}}},~~1\leqslant
j\leqslant n,
$$
is the desired eigenvector of $M$.

Set $U(M)=[u_1(M),\ldots,u_n(M)]\in\cO(n)$. Then
$M\mapsto(\Lambda(M),U(M))$ is smooth and hence
$$
\phi_1(M) = (\Lambda(M),U(M))\equiv\pi(U(M))
$$
is smooth, as claimed.

Finally, we show that $\hat \phi_1$ is smooth. Fix
$\Pa{\Lambda_0,\widehat U_0}\in \real^n_\uparrow\times (\cO(n)/
K_1)$ and let $\tau$ be the lifting map from some neighborhood $W$
of $\widehat U_0$ to $\cO(n)$. Now for all $\widehat U\in W$,
$\tau(\widehat U)\in\widehat U$ by
Proposition~\ref{prop:dif-structure}. Hence
\begin{eqnarray}
\hat \phi_1\Pa{\Lambda,\widehat U} = \tau(\widehat
U)\Lambda\tau(\widehat U)^\t
\end{eqnarray}
from which it is clear that $\hat \phi_1$ is smooth near
$\Pa{\Lambda_0,\widehat U_0}$, and hence everywhere on
$\real^n_\uparrow\times (\cO(n)/K_1)$. This completes the proof.
\end{proof}

Now consider $\sH_n(\complex)$. The calculations are similar to
$\sS_n(\real)$. Define the map
\begin{eqnarray}
\phi_2:\sH_n(\complex)\ni M\mapsto \Pa{\Lambda,\widehat
U}\in\real^n_\uparrow\times (\cU(n)/K_2).
\end{eqnarray}
Here $K_2$ is the closed subgroup of $\cU(n)$ given by
$\mathbb{T}\times\cdots\times\mathbb{T}=\Set{\diag\Pa{e^{\mathrm{i}\theta_1},\ldots,e^{\mathrm{i}\theta_n}}:\theta_j\in\real}$,
$\Lambda=\diag(\lambda_1,\ldots,\lambda_n),\lambda_1<\cdots<\lambda_n$
as before, and
$$
M = U\Lambda U^*,~~\widehat U=\pi(U),
$$
where $\pi:\cU(n)\to \cU(n)/K_2$ is the natural projection $U\mapsto
\pi(U)=UK_2$ of the unitary group $\cU(n)$ onto the homogeneous
manifold $\cU(n)/K_2$. Let $\tau:\cU(n)/K_2\to \cU(n)$ denote the
(locally defined) lifting map as above. As before, $\phi_2$ is a
well-defined map that is one-to-one from $\sH_n(\complex)$ to
$\cU(n)\to \cU(n)/K_2$. Note that because
$\dim(\cU(n))=n^2,\dim(\cU(n)/K_2)=\dim(\cU(n))-\dim(K_2) =n^2-n$,
as it should.

As before, the proof of the following result is similar to that of
Proposition~\ref{prop:diffeomorphism-real}.
\begin{prop}\label{prop:diffeomorphism-complex}
$\phi_2$ is a diffeomorphism from $\sH_n(\complex)$ onto
$\real^n_\uparrow\times (\cU(n)/K_2)$.
\end{prop}

\subsection{Results related to orthogonal groups}

\begin{prop}\label{prop:LU-decom}
Let $Y,X,T\in\real^{n\times n}$, where $Y$ and $T$ are nonsingular,
$X$ is skew symmetric and $T$ is lower triangular of independent
real entries $y_{ij}$'s, $x_{ij}$'s and $t_{ij}$'s respectively. Let
$t_{jj}>0, j=1,\ldots,n-1, -\infty<t_{nn}<\infty$;
$-\infty<t_{jk}<\infty, j>k$; $-\infty<x_{jk}<\infty, j<k$ or
$-\infty<t_{jk}<\infty, j\geqslant k$ and the first row entries,
except the first one, of $(\I_n+X)^{-1}$ are negative. Then the
unique representation $Y=T\Br{2(\I_n+X)^{-1}-\I_n} =
T(\I_n-X)(\I_n+X)^{-1}$ implies that
\begin{eqnarray}
[\dif{Y}] =
2^{n(n-1)/2}\Pa{\prod^n_{j=1}\abs{t_{jj}}^{n-j}}\Pa{\det(\I_n+X)}^{-(n-1)}[\dif{T}][\dif{X}].
\end{eqnarray}
\end{prop}

\begin{proof}
Take the differentials to get
\begin{eqnarray*}
\dif Y &=& \dif T \cdot \Pa{2(\I_n+X)^{-1}-\I_n} + T\cdot
\Pa{-2(\I_n+X)^{-1}\cdot\dif X\cdot(\I_n+X)^{-1}}\\
&=&\dif T \cdot \Pa{2(\I_n+X)^{-1}-\I_n} + T\cdot
\Pa{-\frac12(\I_n+Z)\cdot\dif X\cdot(\I_n+Z)}\\
&=& \dif T\cdot Z -\frac12 T(\I_n+Z)\cdot \dif X\cdot (\I_n+Z).
\end{eqnarray*}
where $Z=2(\I_n+X)^{-1}-\I_n$. Then we have
\begin{eqnarray*}
T^{-1}\cdot \dif Y\cdot Z^\t = T^{-1}\cdot\dif T -
\frac12(\I_n+Z)\cdot\dif X\cdot (\I_n+Z^\t).
\end{eqnarray*}
The Jacobian of the transformation of $T$ and $X$ going to $Y$ is
equal to the Jacobian of $\dif T,\dif X$ going to $\dif Y$. Now
treat $\dif T,\dif X, \dif Y$ as variables and everything else as
constants. Let
\begin{eqnarray*}
\dif U = T^{-1}\cdot \dif Y\cdot Z^\t,~~\dif V = T^{-1}\cdot \dif T,
~~\dif W = (\I_n+Z)\cdot \dif X\cdot (\I_n+Z^\t).
\end{eqnarray*}
Thus
\begin{eqnarray*}
&&[\dif U] = \det(T)^{-n}\det(Z^\t)^n [\dif Y]=\det(T)^{-n} [\dif
Y]\\
&&\Longrightarrow[\dif Y] = \det(T)^n[\dif
U]=\Pa{\prod^n_{j=1}\abs{t_{jj}}^n}[\dif U].
\end{eqnarray*}
Note that $\det(Z^\t)=\pm1$ since $Z$ is orthogonal. Since $X$ is
skew symmetric, one has
\begin{eqnarray*}
[\dif W] = \det(\I_n+Z)^{n-1}[\dif X] =
2^{n(n-1)}\det(\I_n+X)^{-(n-1)}[\dif X].
\end{eqnarray*}
One also has
\begin{eqnarray*}
[\dif V] = \Pa{\prod^n_{j=1}\abs{t_{jj}}^{-j}}[\dif T].
\end{eqnarray*}
Now we see that
$$
\dif U = \dif V - \frac12\dif W\Longrightarrow U=V-\frac12W.
$$
Let $U=[u_{ij}],V=[v_{ij}], W=[w_{ij}]$. Then since $T$ is lower
triangular $t_{ij}=0, i<j$ and thus $V$ is lower triangular, and
since $X$ is skew symmetric $x_{jj}=0$ for all $j$ and
$x_{ij=-x_{ji}}$ for $i\neq j$ and thus $W$ is skew symmetric. Thus
we have
\begin{eqnarray*}
u_{ii} &=& v_{ii}, u_{ij}=-\frac12 w_{ij}, i<j\\
u_{ij} &=& v_{ij}+\frac12 w_{ij}, i>j.
\end{eqnarray*}
Take the $u$-variables in the order $u_{ii}, i=1,\ldots,n; u_{ij},
i<j; u_{ij}, i>j$ and $v_{ii}, i=1,\ldots,n; v_{ij}, i<j; v_{ij},
i>j$. Then the Jacobian matrix is of the following form:
\begin{eqnarray*}
\Br{\begin{array}{ccc}
      \I_n & 0 & 0 \\
      0 & \Pa{-\frac12}\I_{\binom{n}{2}} & 0 \\
      0 & \Pa{\frac12}\I_{\binom{n}{2}} & \I_{\binom{n}{2}}
    \end{array}
}
\end{eqnarray*}
and the determinant, in absolute value, is
$\Pa{\frac12}^{n(n-1)/2}$. That is,
$$
[\dif U]=\Pa{\frac12}^{n(n-1)/2}[\dif V][\dif W].
$$
Now substitute for $[\dif U], [\dif W]$ and $[\dif V]$, respectively
to obtain the result.
\end{proof}

\begin{prop}\label{prop:LU-spetral-decom}
Let $Y,X,D\in\real^{n\times n}$ be of independent real entries,
where $Y$ is symmetric with distinct and nonzero eigenvalues, $X$ is
skew symmetric with the entries of the first row of $(\I_n+X)^{-1}$,
except the first entry, negative and
$D=\diag(\lambda_1,\ldots,\lambda_n)$, with
$\lambda_1>\cdots>\lambda_n$. Then, excluding the sign, $Y =
\Br{2(\I_n+X)^{-1}-\I_n}D\Br{2(\I_n+X)^{-1}-\I_n}$ implies that
\begin{eqnarray}
[\dif{Y}] =
2^{n(n-1)/2}\Pa{\det(\I_n+X)}^{-(n-1)}\Pa{\prod_{i<j}\abs{\lambda_i-\lambda_j}}[\dif{X}][\dif{D}].
\end{eqnarray}
\end{prop}

\begin{proof}
Let $Z=2(\I_n+X)^{-1}-\I_n$. Take the differentials and reduce to
get
\begin{eqnarray*}
Z^\t \cdot\dif Y\cdot Z = -\frac12(\I_n+Z^\t)\cdot \dif X\cdot
(\I_n+Z)\cdot D + \dif D + \frac12D\cdot(\I_n+Z^\t)\cdot \dif X\cdot
(\I_n+Z).
\end{eqnarray*}
Put
\begin{eqnarray*}
\dif U = Z^\t\cdot \dif Y\cdot Z,~~\dif W = \dif D,~~\dif V =
(\I_n+Z^\t)\dif X(\I_n+Z).
\end{eqnarray*}
Thus
\begin{eqnarray*}
\dif U = -\frac12\dif V\cdot D+\frac12D\cdot \dif V + \dif W.
\end{eqnarray*}
But since $\dif Y$ is symmetric and $Z$ is orthogonal, $[\dif
U]=\det(Z)^{n+1}[\dif Y]=[\dif Y]$, excluding the sign. Clearly
$[\dif W]=[\dif D]$. Since $X$ is skew symmetric we have
\begin{eqnarray*}
\dif V = \det(\I_n+Z)^{n-1}[\dif
X]=2^{n(n-1)}\det(\I_n+X)^{-(n-1)}[\dif X].
\end{eqnarray*}
We see that
\begin{eqnarray*}
\dif u_{ii} &=& \dif w_{ii}, \\
\dif u_{ij} &=& \frac12(\lambda_i-\lambda_j)\dif v_{ij}, i<j.
\end{eqnarray*}
Take the $u$-variables in the order $u_{ii}, i=1,\ldots,n; u_{ij},
i<j$ and the $w$ and $v$-variables in the order
$w_{ii},i=1,\ldots,n; v_{ij}, i<j$. Then the matrix of partial
derivatives is of the following form:
$$
\Br{\begin{array}{cc}
      \I & 0 \\
      0 & M
    \end{array}
},
$$
where $M$ is a diagonal matrix with the diagonal elements
$\frac12(\lambda_i-\lambda_j),i<j$. There are $n(n-1)/2$ elements.
Hence the determinant of the above matrix, in absolute value, is
$2^{-n(n-1)/2}\prod_{i<j}\abs{\lambda_i-\lambda_j}$. That is,
\begin{eqnarray*}
[\dif U]
=2^{-n(n-1)/2}\Pa{\prod_{i<j}\abs{\lambda_i-\lambda_j}}[\dif V][\dif
D].
\end{eqnarray*}
Hence
\begin{eqnarray*}
[\dif Y]
=2^{n(n-1)/2}\det(\I_n+X)^{-(n-1)}\Pa{\prod_{i<j}\abs{\lambda_i-\lambda_j}}[\dif
X][\dif D].
\end{eqnarray*}
\end{proof}

\begin{remark}
When integrating over the skew symmetric matrix $X$ using the
transformation in Proposition~\ref{prop:LU-spetral-decom}, under the
unique choice for $Z=2(\I_n+X)^{-1}-\I_n$, observe that
\begin{eqnarray}
2^{n(n-1)/2}\int_X \det(\I_n+X)^{-(n-1)}[\dif X] =
\frac{\pi^{\frac{n^2}2}}{\Gamma_n\Pa{\frac n2}}.
\end{eqnarray}
Note that the $\lambda_j$'s are to be integrated out over
$\infty>\lambda_1>\cdots>\lambda_n$ and $X$ over a unique choice of
$Z$.
\end{remark}

\subsection{Results related to unitary groups}

When $\wtil X$ is skew hermitian, that is, $\wtil X^* = - \wtil X$,
it is not difficult to show that $\I\pm \wtil X$ are both
nonsingular and $\wtil Z=2(\I+\wtil X)^{-1}-\I$ is unitary, that is,
$\wtil Z\wtil Z^\ast=\I$. This property will be made use of in the
first result that will be discussed here. Also note that
\begin{eqnarray*}
2(\I+\wtil X)^{-1}-\I = (\I+\wtil X)^{-1}(\I-\wtil X) =(\I-\wtil
X)(\I+\wtil X)^{-1}.
\end{eqnarray*}

When $\wtil X$ is skew hermitian and of functionally independent
complex variables, then there are $p+2\frac{n(n-1)}2=n^2$ real
variables in $\wtil X$. Let $\wtil T$ be a lower triangular matrix
of functionally independent complex variables with the diagonal
elements being real. Then there are $n^2$ real variables in $\wtil
T$ also. Thus combined, there are $2n^2$ real variables in $\wtil T$
and $\wtil X$. It can be shown that
$$
\wtil Y = \wtil T\Pa{2(\I+\wtil X)^{-1}-\I}
$$
can produce a one-to-one transformation when the $t_{jj}$'s are real
and positive.
\begin{eqnarray*}
\wtil Y = \wtil T \wtil Z,
\end{eqnarray*}
where $\wtil Z \wtil Z^* = \I, \wtil T = [\wtil t_{jk}], \wtil
t_{jj}=t_{jj}>0, j=1,\ldots,n$. Then
\begin{eqnarray*}
\wtil Y\wtil Y^* = \wtil T \wtil T^*\Longrightarrow t^2_{11} =
\sum^n_{k=1} \abs{\wtil y_{1k}}^2.
\end{eqnarray*}
Note that when $t_{11}$ is real and positive it is uniquely
determined in terms of $\wtil Y$. Now consider the first row
elements of $\wtil T \wtil T^*$ that is $t^2_{11},t_{11}\wtil
t_{21},\ldots, t_{11}\wtil t_{n1}$. Hence $\wtil t_{21},\ldots,\wtil
t_{n1}$, that is, the first column of $\wtil T$ is uniquely
determined in terms of $\wtil Y$. Now consider the second row of
$\wtil T \wtil T^*$ and so on. Thus $\wtil T$ is uniquely determined
in terms of $\wtil Y$. But $\wtil Z=\wtil T^{-1} \wtil Y$ and hence
$\wtil Z$, thereby $\wtil X$ is uniquely determined in terms of
$\wtil Y$ with no additional restrictions imposed on the elements of
$\wtil Z$.

In this chapter the Jacobians will also be written ignoring the sign
as in the previous chapters.

\begin{prop}\label{prop:tilde-LU-complex}
Let $\wtil Y,\wtil X$ and $\wtil T = [\wtil t_{jk}]$ be $n\times n$
matrices of functionally independent complex variables where $\wtil
Y$ is nonsingular, $\wtil X$ is skew hermitian and $\wtil T$ is
lower triangular with real and positive diagonal elements. Ignoring
the sign, if
\begin{eqnarray*}
\wtil Y = \wtil T\Pa{2(\wtil X+\I)^{-1}-\I} = \wtil T(\I-\wtil
X)(\I+\wtil X)^{-1},
\end{eqnarray*}
then
\begin{framed}
\begin{eqnarray}
[\dif\wtil Y] = 2^{n^2}\cdot\Pa{\prod_{j=1}^n
t^{2(n-j)+1}_{jj}}\cdot\abs{\det((\I+\wtil X)(\I-\wtil
X))}^{-n}\cdot[\dif\wtil X][\dif\wtil T].
\end{eqnarray}
\end{framed}
\end{prop}

\begin{proof}
Taking differentials in $\wtil Y = \wtil T\Pa{2(\I+\wtil
X)^{-1}-\I}$, one has
\begin{eqnarray*}
\dif\wtil Y = \wtil T \Pa{-2(\I+\wtil X)^{-1}\cdot\dif\wtil
X\cdot(\I+\wtil X)^{-1}} + \dif\wtil T\cdot\Pa{2(\I+\wtil
X)^{-1}-\I}.
\end{eqnarray*}
Let
\begin{eqnarray*}
\wtil Z = 2(\I+\wtil X)^{-1}-\I\Longrightarrow (\I+\wtil X)^{-1} =
\frac12(\I+\wtil Z)
\end{eqnarray*}
and observe that $\wtil Z\wtil Z^*=\I$. Then
\begin{eqnarray}\label{eq:a}
\wtil T^{-1} \cdot\dif\wtil Y\cdot\wtil Z^* = -\frac12(\I+\wtil
Z)\cdot \dif\wtil X\cdot (\I+\wtil Z^*) + \wtil T^{-1}\cdot
\dif\wtil T.
\end{eqnarray}
Let $\dif\wtil W = (\I+\wtil Z)\cdot \dif\wtil X\cdot (\I+\wtil
Z^*)$. Then
\begin{eqnarray}\label{eq:b}
[\dif\wtil W] = \abs{\det((\I+\wtil Z)(\I+\wtil
Z^*))}^n\cdot[\dif\wtil X] = 2^{2n^2}\cdot\abs{\det((\I+\wtil
X)(\I-\wtil X))}^{-n}\cdot[\dif\wtil X].
\end{eqnarray}
Let $\dif\wtil U = \wtil T^{-1}\cdot \dif\wtil T$, then
\begin{eqnarray}\label{eq:c}
[\dif\wtil U] = \Pa{\prod^n_{j=1}t^{-(2j-1)}_{jj}}\cdot [\dif\wtil
T].
\end{eqnarray}
Let $\dif\wtil V = \wtil T^{-1}\cdot\dif\wtil Y\cdot\wtil Z^*$, then
\begin{eqnarray}\label{eq:d}
[\dif\wtil V] = \abs{\det(\wtil T\wtil T^*)}^{-n}\cdot[\dif\wtil Y].
\end{eqnarray}
Eq.~\eqref{eq:a} reduces to
\begin{eqnarray}\label{eq:e}
\dif\wtil V = -\frac12 \dif\wtil W + \dif\wtil U.
\end{eqnarray}
Note that $\dif\wtil W$ is skew hermitian. Denote
\begin{eqnarray*}
\wtil V &=& [\wtil v_{jk}] = [v^{(1)}_{jk}]+\sqrt{-1}[v^{(2)}_{jk}],\\
\wtil W &=& [\wtil w_{jk}] = [w^{(1)}_{jk}]+\sqrt{-1}[w^{(2)}_{jk}],\\
\wtil U &=& [\wtil u_{jk}] = [u^{(1)}_{jk}]+\sqrt{-1}[u^{(2)}_{jk}],
\end{eqnarray*}
where $v^{(m)}_{jk},w^{(m)}_{jk},u^{(m)}_{jk}, m=1,2$ are all real.
Then from Eq.~\eqref{eq:e}, we see that $\wtil V = -\frac12 \wtil
W+\wtil U$. Thus
\begin{eqnarray*}
v^{(m)}_{jk} &=& \frac12 w^{(m)}_{kj}+u^{(m)}_{jk}, j>k, m=1,2,\\
v^{(m)}_{jk} &=& -\frac12 w^{(m)}_{jk}, j<k, m=1,2,\\
v^{(1)}_{jj} &=& u^{(1)}_{jj},\\
v^{(2)}_{jj} &=& -\frac12w^{(2)}_{jj}.
\end{eqnarray*}
The matrices of partial derivatives are the following:
\begin{eqnarray*}
&&\frac{\partial\Pa{v^{(1)}_{jj}, v^{(2)}_{jj}; v^{(1)}_{jk},
v^{(2)}_{jk}(j<k); v^{(1)}_{jk},
v^{(2)}_{jk}(j>k)}}{\partial\Pa{u^{(1)}_{jj},w^{(2)}_{jj};
w^{(1)}_{jk},w^{(2)}_{jk}(j<k);
u^{(1)}_{jk},u^{(2)}_{jk}(j>k)}}\\
&&=\Br{\begin{array}{cccccc}
         \I_n & 0 & 0 & 0 & 0 & 0  \\
         0 & -\frac12\I_n & 0 & 0 & 0 & 0  \\
         0 & 0 & -\frac12\I_{\binom{n}{2}} & 0 & 0 & 0 \\
         0 & 0 & 0 & -\frac12\I_{\binom{n}{2}} & 0 & 0  \\
         0 & 0 & \frac12\I_{\binom{n}{2}} & 0 & \I_{\binom{n}{2}} & 0 \\
         0 & 0 & 0 & \frac12\I_{\binom{n}{2}} & 0 & \I_{\binom{n}{2}}

       \end{array}
}
\end{eqnarray*}
where
\begin{eqnarray*}
A_{11} &=& \Pa{\frac{\partial v^{(1)}_{jj}}{\partial
u^{(1)}_{jj}},j=1,\ldots,n}
=\I_n,\\
A_{22} &=& \Pa{\frac{\partial v^{(2)}_{jj}}{\partial
w^{(2)}_{jj}},j=1,\ldots,n} = -\frac12\I_n
\end{eqnarray*}
and
\begin{eqnarray*}
A_{33} &=& \Pa{\frac{\partial v^{(1)}_{jk}}{\partial
w^{(1)}_{jk}},j<k}
=-\frac12\I_{\binom{n}{2}},\\
A_{44} &=& \Pa{\frac{\partial v^{(2)}_{jk}}{\partial w^{(2)}_{jk}},j<k} =-\frac12\I_{\binom{n}{2}} \\
A_{55} &=& \Pa{\frac{\partial v^{(1)}_{jk}}{\partial
u^{(1)}_{jk}},j>k}
=\I_{\binom{n}{2}},\\
A_{66} &=& \Pa{\frac{\partial v^{(2)}_{jk}}{\partial
u^{(2)}_{jk}},j>k} =\I_{\binom{n}{2}}.
\end{eqnarray*}
$$
A_{53}=\Pa{\frac{\partial v^{(1)}_{jk}}{\partial w^{(1)}_{kj}}:
j>k}=\frac12\I_{\binom{n}{2}},~~A_{64}=\Pa{\frac{\partial
v^{(2)}_{jk}}{\partial w^{(2)}_{kj}}: j>k}=\frac12\I_{\binom{n}{2}}.
$$
The determinants of $A_{11},\ldots,A_{66}$ contribute towards the
Jacobian and the product of the determinants, in absolute value, is
$2^{-n(n-1)-n}=2^{-n^2}$. Without going through the above procedure
one may note from \eqref{eq:a} that since $\dif\wtil X$ has $n^2$
real variables, multiplication by $\frac12$ produces the factor
$2^{-n^2}$ in the Jacobian. From \eqref{eq:b},\eqref{eq:c},
\eqref{eq:d} and \eqref{eq:e} we have
\begin{eqnarray*}
[\dif\wtil Y] &=& \Pa{\prod^n_{j=1}t^{2n}_{jj}}\cdot[\dif\wtil V] =
\Pa{\prod^n_{j=1}t^{2n}_{jj}}\cdot2^{-n^2}\cdot[\dif\wtil
W][\dif\wtil U]\\
&=&
\Pa{\prod^n_{j=1}t^{2n}_{jj}}\cdot2^{2n^2-n^2}\cdot\abs{\det((\I+\wtil
X)(\I-\wtil X))}^{-n}\cdot\Pa{\prod^n_{j=1}t^{-(2j-1)}_{jj}}\cdot
[\dif\wtil X][\dif\wtil T]\\
&=&
2^{n^2}\Pa{\prod^n_{j=1}t^{2(n-j)+1}_{jj}}\cdot\abs{\det((\I+\wtil
X)(\I-\wtil X))}^{-n}\cdot [\dif\wtil X][\dif\wtil T].
\end{eqnarray*}
This completes the proof.
\end{proof}

\begin{exam}\label{exam:integral-det}
Let $\wtil X\in\complex^{n\times n}$ skew hermitian matrix of
independent complex variables. Then show that
\begin{framed}
\begin{eqnarray*}
\int_{\wtil X} [\dif\wtil X]\abs{\det((\I+\wtil X)(\I-\wtil
X))}^{-n} = \frac{\pi^{n^2}}{2^{n(n-1)}\wtil \Gamma_n(n)}.
\end{eqnarray*}
\end{framed}
Indeed, let $\wtil Y=[\wtil y_{jk}]$ be a $n\times n$ matrix of
independent complex variables. Consider the integral
\begin{eqnarray}\label{eq:aa}
\int_{\wtil Y} [\dif\wtil Y] e^{-\Tr{\wtil Y\wtil Y^*}}=
\prod^n_{j,k=1}\int^{+\infty}_{-\infty}e^{-\abs{\wtil y_{jk}}^2}
\dif\wtil y_{jk} = \pi^{n^2}.
\end{eqnarray}
Now consider a transformation $\wtil Y=\wtil T\Pa{2(\I+\wtil
X)^{-1}-\I}$, where $\wtil T=[\wtil t_{jk}]$ is lower triangular
with $t_{jj}$'s real and positive and no restrictions on $\wtil X$
other than that it is skew hermitian. Then since $\Tr{\wtil Y\wtil
Y^*}=\Tr{\wtil T\wtil T^*}$, from
Proposition~\ref{prop:tilde-LU-complex} and Eq.~\eqref{eq:aa}, we
have
\begin{eqnarray*}
\pi^{n^2} &=& \int_{\wtil T,\wtil X} [\dif\wtil X][\dif\wtil
T]e^{-\Tr{\wtil T\wtil T^*}}\cdot 2^{n^2}\cdot
\Pa{\prod^n_{j=1}t^{2(n-j)+1}_{jj}}\cdot
\abs{\det((\I+\wtil X)(\I-\wtil X))}^{-n}\\
&=&\int_{\wtil T}[\dif\wtil T]e^{-\Tr{\wtil T\wtil T^*}}\cdot
2^{n^2}\cdot \Pa{\prod^n_{j=1}t^{2(n-j)+1}_{jj}}\times \int_{\wtil
X}[\dif\wtil X]\abs{\det((\I+\wtil X)(\I-\wtil X))}^{-n}
\end{eqnarray*}
Note that
\begin{eqnarray*}
e^{-\Tr{\wtil T\wtil T^*}} = \exp\Pa{-\sum^n_{j=1}t^2_{jj} -
\sum_{j>k}\abs{\wtil t_{jk}}^2}.
\end{eqnarray*}
But
\begin{eqnarray*}
\int^{+\infty}_{-\infty} e^{-\abs{\wtil t_{jk}}^2}\dif\wtil t_{jk} =
\pi~~~\text{and}~~~\int^{+\infty}_0
t^{2(n-j)+1}_{jj}e^{-t^2_{jj}}\dif t_{jj}= \frac12\Gamma(n-j+1).
\end{eqnarray*}
Hence
\begin{eqnarray*}
\int_{\wtil T}[\dif\wtil T]e^{-\Tr{\wtil T\wtil
T^*}}\Pa{\prod^n_{j=1}t^{2(n-j)+1}_{jj}} = 2^{-n}\wtil \Gamma_n(n).
\end{eqnarray*}
Substituting this the result follows.
\end{exam}

\begin{remark}
When a skew hermitian matrix $\wtil X$ is used to parameterize a
unitary matrix such as $\wtil Z$ in
Proposition~\ref{prop:tilde-LU-complex}, can we evaluate the
Jacobian by direct integration? This will be examined here. Let
\begin{eqnarray*}
\sI_n:= \int_{\wtil X} [\dif\wtil X]\abs{\det((\I+\wtil X)(\I-\wtil
X))}^{-n}.
\end{eqnarray*}
Partition $\I+\wtil X$ as follows:
\begin{eqnarray*}
\I+\wtil X = \Br{\begin{array}{cc}
                        1+\wtil x_{11} & \wtil X_{12} \\
                        -\wtil X^*_{12} & \I+\wtil X_1
                      \end{array}
},
\end{eqnarray*}
where $\wtil X_{12}$ represents the first row of $\I+\wtil X$
excluding the first element $1+\wtil x_{11}$, and $\I+\wtil X_1$ is
obtained from $\I+\wtil X$ by deleting the first row and the first
column. Note that
\begin{eqnarray*}
\det(\I+\wtil X) = \det(\I+\wtil X_1)\Pa{1+\wtil x_{11} + \wtil
X_{12}(\I+\wtil X_1)^{-1}\wtil X^*_{12}}.
\end{eqnarray*}
Similarly,
\begin{eqnarray*}
\I-\wtil X = \Br{\begin{array}{cc}
                        1-\wtil x_{11} & -\wtil X_{12} \\
                        \wtil X^*_{12} & \I-\wtil X_1
                      \end{array}
}
\end{eqnarray*}
and
\begin{eqnarray*}
\det(\I-\wtil X) = \det(\I-\wtil X_1)\Pa{1-\wtil x_{11} + \wtil
X_{12}(\I+\wtil X_1)^{-1}\wtil X^*_{12}}.
\end{eqnarray*}
For fixed $(\I+\wtil X_1)$ let $\wtil U_{12} := \wtil
X_{12}(\I+\wtil X_1)^{-1}$, then
\begin{eqnarray*}
[\dif\wtil U_{12}] = \abs{\det((\I+\wtil X_1)(\I-\wtil
X_1))}^{-1}\cdot[\dif\wtil X_{12}]
\end{eqnarray*}
and observing that $\wtil X^*_1=-\wtil X_1$ we have
\begin{eqnarray*}
\wtil X_{12}(\I+\wtil X_1)^{-1}\wtil X^*_{12} = \wtil
U_{12}(\I-\wtil X_1)\wtil U^*_{12}.
\end{eqnarray*}
Let $\wtil Q$ be a unitary matrix such that
\begin{eqnarray*}
\wtil Q^* \wtil X_1\wtil Q =
\mathrm{diag}(\sqrt{-1}\lambda_1,\ldots,\sqrt{-1}\lambda_{n-1})
\end{eqnarray*}
where $\lambda_1,\ldots,\lambda_{n-1}$ are real. Let
\begin{eqnarray*}
\wtil V_{12}=\wtil U_{12}\wtil Q=[\wtil v_1,\ldots,\wtil v_{n-1}].
\end{eqnarray*}
Then
\begin{eqnarray*}
\wtil U_{12}(\I-\wtil X_1)\wtil
U^*_{12}=(1-\sqrt{-1}\lambda_1)\abs{\wtil
v_1}^2+\cdots+(1-\sqrt{-1}\lambda_{n-1})\abs{\wtil v_{n-1}}^2
\end{eqnarray*}
and
\begin{eqnarray*}
1+\wtil x_{11} + \wtil X_{12}(\I+\wtil X_1)^{-1}\wtil X^*_{12} =
a-\sqrt{-1}b
\end{eqnarray*}
where
\begin{eqnarray*}
a &=& 1+ \abs{\wtil v_1}^2 +\cdots+\abs{\wtil v_{n-1}}^2\\
b &=& -x^{(2)}_{11}+\lambda_1\abs{\wtil
v_1}^2+\cdots+\lambda_{n-1}\abs{\wtil v_{n-1}}^2
\end{eqnarray*}
observing that $\wtil x_{11}$ is purely imaginary, that is, $\wtil
x_{11}=\sqrt{-1}x^{(2)}_{11}$, where $x^{(2)}_{11}$ is real. Thus
$$
\abs{\det((\I+\wtil X)(\I-\wtil X))}^{-n}
$$
yields the factor
\begin{eqnarray*}
[(a - \sqrt{-1}b)(a + \sqrt{-1}b)]^{-n}=(a^2+b^2)^{-n}.
\end{eqnarray*}
So
\begin{eqnarray*}
\sI_n &=& \int_{\wtil X_1} \int_{\wtil X_{12}}\int_{\wtil x_{11}}
\Pa{\det(\I+\wtil X_1)(a-\sqrt{-1}b)\det(\I-\wtil
X_1)(a+\sqrt{-1}b)}^{-n}[\dif\wtil X_1][\dif\wtil X_{12}]\dif\wtil
x_{11}\\
&=&\int_{\wtil X_1}\int_{\wtil X_{12}}\int_{\wtil
x_{11}}\abs{\det\Pa{(\I+\wtil X_1)(\I-\wtil
X_1)}}^{-n}(a^2+b^2)^{-n}[\dif\wtil X_1][\dif\wtil X_{12}]\dif\wtil
x_{11}.
\end{eqnarray*}
Since
$$
[\dif\wtil V_{12}] = [\dif\wtil U_{12}]~~\text{and}~~[\dif\wtil
U_{12}]= \abs{\det\Pa{(\I+\wtil X_1)(\I-\wtil X_1)}}^{-1}[\dif\wtil
X_{12}],
$$
it follows that
$$
[\dif\wtil X_{12}] =\abs{\det\Pa{(\I+\wtil X_1)(\I-\wtil X_1)}}
[\dif\wtil V_{12}].
$$
Based on this, we have
\begin{eqnarray*}
\sI_n &=& \int_{\wtil X_1}\int_{\wtil V_{12}}\int_{\wtil
x_{11}}\abs{\det\Pa{(\I+\wtil X_1)(\I-\wtil
X_1)}}^{-(n-1)}(a^2+b^2)^{-n}[\dif\wtil X_1][\dif\wtil
V_{12}]\dif\wtil x_{11}.
\end{eqnarray*}
Then
\begin{eqnarray*}
\sI_n &=& \sI_{n-1}\int_{\wtil
V_{12}}\int_{x^{(2)}_{11}}(a^2+b^2)^{-n}[\dif\wtil
V_{12}]\dif x^{(2)}_{11}\\
&=&\sI_{n-1}\int_{\wtil V_{12}}[\dif\wtil V_{12}]
\Pa{a^{-2n}\int_{x^{(2)}_{11}}\Pa{1+\frac{b^2}{a^2}}^{-n}\dif
x^{(2)}_{11}}.
\end{eqnarray*}
Consider the integral over
$x^{(2)}_{11},-\infty<x^{(2)}_{11}<+\infty$. Change $x^{(2)}_{11}$
to $b$ and then to $c=b/a$. Then
\begin{eqnarray*}
\int_{x^{(2)}_{11}}\Pa{1+\frac{b^2}{a^2}}^{-n}\dif x^{(2)}_{11} &=&
\int_b \Pa{1+\frac{b^2}{a^2}}^{-n}\dif b =
a\int^{+\infty}_{c=-\infty}(1+c^2)^{-p}\dif c\\
&=& 2a\int^\infty_0(1+c^2)^{-n}\dif c =
a\frac{\Gamma\Pa{\frac12}\Gamma\Pa{n-\frac12}}{\Gamma(n)}:=k,
\end{eqnarray*}
by evaluating using a type-2 beta integral after transforming
$u=c^2$. Hence
\begin{eqnarray*}
\sI_n &=& k\sI_{n-1}\int_{\wtil V_{12}}[\dif\wtil
V_{12}] a^{-(2n-1)}\\
&=& k\sI_{n-1} \int^{+\infty}_{-\infty} \cdots
\int^{+\infty}_{-\infty}\Pa{1+ \abs{\wtil v_1}^2 +\cdots+\abs{\wtil
v_{n-1}}^2}^{-(2n-1)}\dif\wtil v_1\cdots\dif\wtil v_{n-1}.
\end{eqnarray*}
For evaluating the integral use the polar coordinates. Let $\wtil
v_j = v^{(1)}_{j}+\sqrt{-1}v^{(2)}_{j}$, where $v^{(1)}_{j}$ and
$v^{(2)}_{j}$ are real. Let
$$
\begin{cases}
v^{(1)}_{j}=r_j\cos\theta_j, \\
v^{(2)}_{j}=r_j\sin\theta_j,
\end{cases}~~~0\leqslant r_j<\infty, 0\leqslant \theta_j\leqslant
2\pi.
$$
Then denoting the multiple integral by $\cI_{n-1}$, we have
\begin{eqnarray*}
\cI_{n-1} &=& \int^{+\infty}_{-\infty} \cdots
\int^{+\infty}_{-\infty}\Pa{1+ \abs{\wtil v_1}^2 +\cdots+\abs{\wtil
v_{n-1}}^2}^{-(2n-1)}\dif\wtil
v_1\cdots\dif\wtil v_{n-1}\\
&=& (2\pi)^{n-1}\int^{+\infty}_{r_1=0} \cdots
\int^{+\infty}_{r_{n-1}=0}r_1\cdots r_{n-1}\Pa{1+r^2_1+\cdots
+r^2_{n-1}}^{-(2n-1)}\dif r_1\cdots\dif r_{n-1}.
\end{eqnarray*}
Evaluating this by a Dirichlet integral we have
\begin{eqnarray*}
\cI_{n-1} &=&
\pi^{n-1}\frac{\Gamma(n)}{\Gamma(2n-1)}~~\text{for}~~n\geqslant2.
\end{eqnarray*}
Hence for $n\geqslant2$,
\begin{eqnarray*}
\sI_n = \sI_{n-1}
\pi^{n-1}\sqrt{\pi}\frac{\Gamma(n)\Gamma\Pa{n-\frac12}}{\Gamma(2n-1)\Gamma(n)}.
\end{eqnarray*}
By using the duplication formula for gamma functions
\begin{eqnarray*}
\Gamma(2n-1)= \sqrt{\pi}2^{2n-2}\Gamma\Pa{n-\frac12}\Gamma(n).
\end{eqnarray*}
Hence
\begin{eqnarray*}
\sI_n = \sI_{n-1}\frac{\pi^n}{2^{2n-2}\Gamma(n)}.
\end{eqnarray*}
Repeating this process we have
\begin{eqnarray*}
\sI_n =
\frac{\pi^n}{2^{2n-2}\Gamma(n)}\frac{\pi^{n-1}}{2^{2(n-1)-2}\Gamma(p-1)}\cdots
\frac{\pi}{2^{2-2}\Gamma(1)} = \frac{\pi^{n^2}}{2^{n(n-1)}\wtil
\Gamma_n(n)}.
\end{eqnarray*}
This is what we obtained in Example~\ref{exam:integral-det}.
\end{remark}

Next we consider a representation of a hermitian positive definite
matrix $\wtil Y$ in terms of a skew hermitian matrix $\wtil X$ and a
diagonal matrix $D$ such that
\begin{eqnarray*}
\wtil Y = \Pa{2(\I+\wtil X)^{-1}-\I}D\Pa{2(\I+\wtil X)^{-1}-\I}^*
\end{eqnarray*}
where $D=\mathrm{diag}(\lambda_1,\ldots,\lambda_p)$ with the
$\lambda_j$'s real distinct and positive, and the first row elements
of $(\I+\wtil X)^{-1}$ real and of specified signs, which amounts to
require $2(\I+\wtil X)^{-1}-\I\in\cU(n)/\cU(1)^{\times n}$. In this
case it can be shown that the transformation is unique. Note that
\begin{eqnarray*}
\wtil Y = \wtil Z D\wtil Z^* = \lambda_1\wtil Z_1\wtil
Z_1^*+\cdots+\lambda_p\wtil Z_n\wtil Z_n^*
\end{eqnarray*}
it indicates that
\begin{eqnarray*}
(\wtil Y - \lambda_j\I)\wtil Z_j = 0, j=1,\ldots,n
\end{eqnarray*}
where $\wtil Z_1,\ldots,\wtil Z_n$ are the columns of $\wtil Z$ such
that $\inner{\wtil Z_j}{\wtil Z_k}=\delta_{jk}$. Since
$\lambda_1,\ldots,\lambda_n$ are the eigenvalues of $\wtil Y$, which
are assumed to be real distinct and positive, $D$ is uniquely
determined in terms of $\wtil Y$. Note that $\wtil Z_j$ is an
eigenvector corresponding to $\lambda_j$ such that $\inner{\wtil
Z_j}{\wtil Z_j}=1, j=1,\ldots,n$. Hence $\wtil Z_j$ is uniquely
determined in terms of $\wtil Y$ except for a multiple of $\pm
1,\pm\sqrt{-1}$. If any particular element of $\wtil Z_j$ is assumed
to be real and positive, for example the first element, then $\wtil
Z_j$ is uniquely determined. Thus if the first row elements of
$\wtil Z$ are real and of specified signs, which is equivalent to
saying that the first row elements of $(\I+\wtil X)^{-1}$ are real
and of specified signs, then the transformation is unique.

\begin{prop}\label{prop:spetral-decom-complex}
Let $\wtil Y$ and $\wtil X$ be $n\times n$ matrices of functionally
independent complex variables such that $\wtil Y$ is hermitian
positive definite, $\wtil X$ is skew hermitian and the first row
elements of $(\I+\wtil X)^{-1}$ are real and of specified signs. Let
$D=\mathrm{diag}(\lambda_1,\ldots,\lambda_n)$, where the
$\lambda_j$'s are real distinct and positive. Ignoring the sign, if
\begin{eqnarray*}
\wtil Y = \Pa{2(\wtil X+\I)^{-1}-\I} D \Pa{2(\wtil X+\I)^{-1}-\I}^*,
\end{eqnarray*}
then
\begin{framed}
\begin{eqnarray*}
[\dif\wtil Y] = 2^{n(n-1)}\cdot\Pa{\prod_{j>k}
\abs{\lambda_k-\lambda_j}^2}\cdot\abs{\det((\I+\wtil X)(\I-\wtil
X))}^{-n}\cdot[\dif\wtil X][\dif D].
\end{eqnarray*}
\end{framed}
\end{prop}

\begin{proof}
Let $\wtil Z=2(\I+\wtil X)^{-1}-\I, \wtil X^* = -\wtil X$. Taking
the differentials in $\wtil Y=\wtil Z D\wtil Z^*$ we have
\begin{eqnarray}\label{eq:aaa}
\dif\wtil Y = \dif\wtil Z\cdot D\cdot\wtil Z^* + \wtil Z\cdot \dif
D\cdot\wtil Z^* + \wtil Z\cdot D \cdot\dif\wtil Z^*.
\end{eqnarray}
But
\begin{eqnarray*}
\dif\wtil Z &=& -2(\I+\wtil
X)^{-1}\cdot\dif\wtil X\cdot(\I+\wtil X)^{-1}\\
&=& -\frac12(\I+\wtil Z)\cdot\dif\wtil X\cdot(\I+\wtil Z)
\end{eqnarray*}
and
\begin{eqnarray*}
\dif\wtil Z^* &=& 2(\I-\wtil
X)^{-1}\cdot\dif\wtil X\cdot(\I-\wtil X)^{-1}\\
&=& \frac12(\I+\wtil Z^*)\cdot\dif\wtil X\cdot(\I+\wtil Z^*).
\end{eqnarray*}
From \eqref{eq:aaa}, one has
\begin{eqnarray*}
\wtil Z^*\cdot\dif\wtil Y\cdot \wtil Z &=& -\frac12(\I+\wtil
Z^*)\cdot\dif\wtil X\cdot(\I+\wtil Z)\cdot D+\dif D+\frac12
D\cdot(\I+\wtil Z^*)\cdot\dif\wtil X\cdot(\I+\wtil Z)
\end{eqnarray*}
observing that $\wtil Z^*\wtil Z=\I$. Let
\begin{eqnarray}
\dif\wtil U &=& \wtil Z^*\cdot \dif\wtil Y\cdot\wtil
Z\Longrightarrow [\dif\wtil U] = [\dif\wtil Y]~~~\text{since}~\wtil
Z^*\wtil
Z=\I,\label{eq:bbb}\\
\dif\wtil V &=& (\I+\wtil Z^*)\cdot\dif\wtil X\cdot(\I+\wtil Z) =
(\I+\wtil X^*)^{-1}\cdot 4\dif\wtil
X\cdot(\I+\wtil X)^{-1}\Longrightarrow\notag\\
~[\dif\wtil V]  &=& 4^{n^2}\cdot \abs{\det((\I+\wtil X)(\I-\wtil
X))}^{-n}\cdot[\dif\wtil X]\label{eq:ccc}
\end{eqnarray}
if there are $n^2$ free real variables in $\wtil X$. But in our case
there are only $n^2-n$ real variables in $\wtil X$ when $\wtil X$ is
uniquely chosen and hence
\begin{eqnarray}\label{eq:ddd}
[\dif\wtil V] &=& 4^{n^2-n} \abs{\det((\I+\wtil X)(\I-\wtil
X))}^{-n}\cdot[\dif\wtil X],
\end{eqnarray}
and
\begin{eqnarray}\label{eq:eee}
\dif\wtil U &=& -\frac12\dif\wtil V\cdot D + \frac12D\cdot \dif\wtil
V+ \dif D.
\end{eqnarray}
From \eqref{eq:eee} and using the fact that $\dif\wtil V$ is skew
hermitian and $\dif\wtil U$ is hermitian we have
\begin{eqnarray*}
\dif u_{jj} = \dif\lambda_j,\quad \dif u^{(m)}_{jk} = \pm
\frac12(\lambda_k-\lambda_j)\dif v^{(m)}_{jk}, j>k, m=1,2.
\end{eqnarray*}
Thus the determinant of the Jacobian matrix, in absolute value, is
\begin{eqnarray*}
\Pa{\prod_{j>k}\frac12\abs{\lambda_k-\lambda_j}}^2=
2^{-n(n-1)}\prod_{j>k}\abs{\lambda_k-\lambda_j}^2.
\end{eqnarray*}
That is,
\begin{eqnarray*}
[\dif\wtil U] =
2^{-n(n-1)}\cdot\prod_{j>k}\abs{\lambda_k-\lambda_j}^2\cdot
[\dif\wtil V][\dif D].
\end{eqnarray*}
Substituting for $[\dif\wtil U]$ and $[\dif\wtil V]$ from
\eqref{eq:bbb} and \eqref{eq:ccc} the result follows.
\end{proof}

\begin{exam}
For $\wtil X$ an $n\times n$ skew hermitian matrix with the first
row elements of $(\I+\wtil X)^{-1}$ real and of specified signs show
that
\begin{eqnarray*}
\int_{\wtil X}[\dif\wtil X] \abs{\det((\I+\wtil X)(\I-\wtil
X))}^{-n} = \frac{\wtil \Gamma_n(n)}{\Delta},
\end{eqnarray*}
where
\begin{eqnarray*}
\Delta: = 2^{n(n-1)}\int_{\lambda_1>\cdots>\lambda_n>0} [\dif
D]\Br{\prod_{1\leqslant i< j\leqslant n}\abs{\lambda_i -
\lambda_j}^2} e^{-\Tr{D}},
\end{eqnarray*}
with
$D=\diag(\lambda_1,\ldots,\lambda_n),\lambda_1>\cdots>\lambda_n>0$.
Consider a $n\times n$ hermitian positive definite matrix $\wtil Y$
of functionally independent complex variables. Let
\begin{eqnarray*}
B &=& \int_{\wtil Y=\wtil Y^*>0} [\dif\wtil Y]e^{-\Tr{\wtil Y}} =
\int_{\wtil Y>0} [\dif\wtil Y] \abs{\det(\wtil
Y)}^{n-n}e^{-\Tr{\wtil Y}}\\
&=&\wtil\Gamma_n(n)=\pi^{\frac{n(n-1)}2}\Gamma(n)\Gamma(n-1)\cdots\Gamma(1)\\
&=&\pi^{\frac{n(n-1)}2}(n-1)!(n-2)!\cdots 1!
\end{eqnarray*}
evaluating the integral by using a complex matrix-variate gamma
integral. Put
\begin{eqnarray*}
\wtil Y = \wtil ZD\wtil Z^*,~~~\wtil Z = 2(\I+\wtil X)^{-1}-\I
\end{eqnarray*}
as in Proposition~\ref{prop:spetral-decom-complex}. Then
\begin{eqnarray*}
[\dif\wtil Y] =
2^{n(n-1)}\cdot\Pa{\prod_{j>k}\abs{\lambda_k-\lambda_j}^2}\cdot\abs{\det((\I+\wtil
X)(\I-\wtil X))}^{-n} \cdot [\dif\wtil X][\dif D]
\end{eqnarray*}
and
\begin{eqnarray*}
B &=& \int_{\wtil X}[\dif\wtil X]
\abs{\det((\I+\wtil X)(\I-\wtil X))}^{-n}\\
&&\times \int_{\lambda_1>\cdots>\lambda_n>0}[\dif D]
2^{n(n-1)}\cdot\Pa{\prod_{j>k}\abs{\lambda_k-\lambda_j}^2}e^{-\Tr{D}}.
\end{eqnarray*}
Hence the result. From (i) in Example~\ref{exam:two-integrals}, we
see that
$$
\Delta = \Pa{\frac{2}{\pi}}^{n(n-1)}\Pa{\wtil\Gamma_n(n)}^2,
$$
which implies that, when $\wtil X$ is taken over all skew hermitian
under the restriction that the first row elements of $(\I+\wtil
X)^{-1}$ are real and of specified signs,
\begin{framed}
\begin{eqnarray*}
\int_{\wtil X}[\dif\wtil X] \abs{\det((\I+\wtil X)(\I-\wtil
X))}^{-n} = \Pa{\frac{\pi}{2}}^{n(n-1)}\frac1{\wtil \Gamma_n(n)}.
\end{eqnarray*}
\end{framed}
\end{exam}

\begin{remark}
In fact, we can derive the volume formula
\eqref{eq:vol-formula-of-state} from
Proposition~\ref{prop:spetral-decom-complex}. The reasoning is as
follows:
\begin{eqnarray*}
\int_{\wtil Y>0:\Tr{\wtil Y}=1}[\dif\wtil Y] &=&
2^{n(n-1)}\int_{\lambda_1>\cdots>\lambda_n>0}[\dif
D]\delta\Pa{\sum^n_{j=1}\lambda_j-1}\Pa{\prod_{j>k}
\abs{\lambda_k-\lambda_j}^2}\\
&&\times\int_{\wtil X}[\dif\wtil X]\abs{\det((\I+\wtil X)(\I-\wtil
X))}^{-n}\\
&=& \frac{2^{n(n-1)}}{n!}\times\frac{\Gamma(1)\cdots\Gamma(n)\Gamma(1)\cdots\Gamma(n+1)}{\Gamma(n^2)}\times\Pa{\frac{\pi}{2}}^{n(n-1)}\frac1{\wtil \Gamma_n(n)} \\
&=&
\pi^{\frac{n(n-1)}2}\frac{\Gamma(1)\cdots\Gamma(n)}{\Gamma(n^2)}.
\end{eqnarray*}
\end{remark}

\begin{exam}\label{exam:exp}
Let $\wtil X\in\complex^{n\times n}$ be a hermitian matrix of
independent complex variables. Show that
\begin{eqnarray*}
\int_{\wtil X}[\dif\wtil X]e^{-\Tr{\wtil X\wtil X^*}} =
2^{-\frac{n(n-1)}2}\pi^{\frac{n^2}2}.
\end{eqnarray*}
Indeed, $\wtil X^*=\wtil X$ implies that
$$
\Tr{\wtil X\wtil X^*} = \sum^n_{j=1}x^2_{jj} + 2\sum_{i<j}\abs{\wtil
x_{ij}}^2.
$$
Thus
\begin{eqnarray*}
\int_{\wtil X}[\dif\wtil X]e^{-\Tr{\wtil X\wtil X^*}} &=&
\Pa{\prod^n_{j=1}\int^\infty_{-\infty}e^{-x^2_{jj}}\dif x_{jj}}
\times\Pa{\prod_{i<j}\int e^{-2\abs{\wtil x_{ij}}^2}\dif \wtil
x_{ij}}\\
&=& \pi^{\frac n2} \times 2^{-\frac{n(n-1)}2}\Pa{\prod_{i<j}\int
e^{-\abs{\wtil x_{ij}}^2}\dif \wtil
x_{ij}}\\
&=&\pi^{\frac n2} \times 2^{-\frac{n(n-1)}2}\times
\pi^{\frac{n(n-1)}2} = 2^{-\frac{n(n-1)}2}\pi^{\frac{n^2}2}.
\end{eqnarray*}
\end{exam}

\begin{exam}
By using Example~\ref{exam:exp} or otherwise show that
\begin{eqnarray*}
\int_{\infty>\lambda_1>\cdots>\lambda_n>-\infty}
\Pa{\prod_{j>k}\abs{\lambda_k-\lambda_j}^2}\exp\Pa{-\sum^n_{j=1}\lambda^2_j}\dif\lambda_1\cdots\dif\lambda_n
=  2^{-\frac{n(n-1)}2}\pi^{\frac{n}2}\prod^{n-1}_{j=1}j!.
\end{eqnarray*}
Indeed, in Example~\ref{exam:exp} letting $\wtil X = \wtil UD\wtil
U^*$, where $D=\diag(\lambda_1,\ldots,\lambda_n)$ and $\wtil
U\in\cU_1(n)$, gives rise to
$$
[\dif \wtil X] = \Pa{\prod_{i<j}(\lambda_i-\lambda_j)^2}[\dif
D][\dif \wtil G_1],~~\lambda_1>\cdots>\lambda_n,
$$
which means that
\begin{eqnarray*}
2^{-\frac{n(n-1)}2}\pi^{\frac{n^2}2}&=&\int_{\wtil X}[\dif\wtil
X]e^{-\Tr{\wtil X\wtil X^*}} \\
&=&
\int_{\infty>\lambda_1>\cdots>\lambda_n>-\infty}\Pa{\prod_{i<j}(\lambda_i-\lambda_j)^2}\exp\Pa{-\sum^n_{j=1}\lambda^2_j}[\dif
D]\times \int_{\cU_1(n)}[\dif \wtil G_1].
\end{eqnarray*}
That is,
\begin{eqnarray*}
\int_{\infty>\lambda_1>\cdots>\lambda_n>-\infty}\Pa{\prod_{i<j}(\lambda_i-\lambda_j)^2}\exp\Pa{-\sum^n_{j=1}\lambda^2_j}\prod^n_{j=1}\dif\lambda_j
=2^{-\frac{n(n-1)}2}\pi^{\frac n2} \prod^{n-1}_{j=1}j!.
\end{eqnarray*}
Therefore
\begin{eqnarray*}
\int\Pa{\prod_{i<j}(\lambda_i-\lambda_j)^2}\exp\Pa{-\sum^n_{j=1}\lambda^2_j}\prod^n_{j=1}\dif\lambda_j
=2^{-\frac{n(n-1)}2}\pi^{\frac n2} \prod^n_{j=1}j!.
\end{eqnarray*}
\end{exam}

\section{Appendix III: Some matrix factorizations}

The materials in this section are collected from Muirhead's book
\cite{Muirhead}. It is the necessary underlying basis for computing
some Jacobians .

Firstly, we recall the Gram-Schmidt orthogonalization process which
enables us to construct an orthonormal basis of $\real^m$ given any
other basis $X_1,\ldots,X_m$ of $\real^m$. We define
$$
\begin{cases}
Y_1 &= X_1,\\
Y_2 &= X_2 - \frac{\inner{Y_1}{X_2}}{\inner{Y_1}{Y_1}}Y_1,\\
Y_3 &= X_3 - \frac{\inner{Y_2}{X_3}}{\inner{Y_2}{Y_2}}Y_2-
\frac{\inner{Y_1}{X_3}}{\inner{Y_1}{Y_1}}Y_1,\\
&\cdots\cdots\cdots\cdots\\
Y_m &= X_m -
\sum^{m-1}_{j=1}\frac{\inner{Y_j}{X_m}}{\inner{Y_j}{Y_j}}Y_j,
\end{cases}
$$
and put $Z_j=\frac1{\inner{Y_j}{Y_j}}Y_j$, where $j=1,\ldots,m$.
Then $Z_1,\ldots,Z_m$ form an orthonormal basis for $\real^m$. Next
matrix factorization utilizes this process.

\begin{prop}
If $A$ is a real $m\times m$ matrix with \blue{real} characteristic
roots, then there exists an orthogonal matrix $H$ such that $H^\t
AH$ is an upper-triangular matrix whose diagonal elements are the
characteristic roots of $A$.
\end{prop}

\begin{proof}
Let $\lambda_1,\ldots,\lambda_m$ be the characteristic roots of
$A:=A_1$ and let $X_1$ be a characteristic vector of $A$
corresponding to $\lambda_1$. This is real since the characteristic
roots are real. Let $X_2,\ldots,X_m$ be any other vectors such that
$X_1, X_2,\ldots,X_m$ for a basis for $\real^m$. Using the
Gram-Schmidt orthogonalization process, construct from $X_1,
X_2,\ldots,X_m$ an orthonormal basis given as the columns of the
orthogonal matrix $H_1$, where the first column $h_1$ is
proportional to $X_1$, so that $h_1$ is also a characteristic vector
of $A$ corresponding to $\lambda_1$. Then the first column of $AH_1$
is $Ah_1=\lambda_1h_1$, and hence the first column of $H^\t_1A_1H_1$
is $\lambda_1 H^\t_1h_1$. Since this is the first column of
$\lambda_1 H^\t_1H_1=\lambda_1\I_m$, it is
$(\lambda_1,0,\ldots,0)^\t$. Hence
$$
H^\t_1A_1H_1 = \Br{\begin{array}{cc}
                   \lambda_1 & B_1 \\
                   0 & A_2
                 \end{array}
},
$$
where $A_2$ is $(m-1)\times (m-1)$. Since
$$
\det(A_1-\lambda\I_m) = (\lambda_1-\lambda)\det(A_2-\lambda\I_{m-1})
$$
and $A_1$ and $H^\t_1A_1H_1$ have the same characteristic roots, the
characteristic roots of $A_2$ are $\lambda_2,\ldots,\lambda_m$.

Now, using a construction similar to that above, find an orthogonal
$(m-1)\times (m-1)$ matrix $H_2$ whose first column is a
characteristic vector of $A_2$ corresponding to $\lambda_2$. Then
$$
H^\t_2A_2H_2 = \Br{\begin{array}{cc}
                   \lambda_2 & B_2 \\
                   0 & A_3
                 \end{array}
},
$$
where $A_3$ is $(m-2)\times (m-2)$ with characteristic roots
$\lambda_3,\ldots,\lambda_m$.

Repeating this procedure an additional $m-3$ times we now define the
orthogonal matrix
$$
H = H_1(1\oplus H_2)(\I_2\oplus H_3)\cdots (\I_{m-2}\oplus H_{m-1})
$$
and note that $H^\t AH$ is upper-triangular with diagonal elements
equal to $\lambda_1,\ldots,\lambda_m$.
\end{proof}

\begin{prop}
If $A$ is an $m\times m$ non-negative definite matrix of rank $r$
then:
\begin{enumerate}[(i)]
\item There exists an $m\times r$ matrix $B$ of rank $r$ such that
$A=BB^\t$.
\item There exists an $m\times m$ nonsingular matrix $C$ such that
$$
A = C\Br{\begin{array}{cc}
           \I_r & 0 \\
           0 & 0
         \end{array}
}C^\t.
$$
\end{enumerate}
\end{prop}

\begin{proof}
As for (i), let $D_1=\diag(\lambda_1,\ldots,\lambda_r)$ where
$\lambda_1,\ldots,\lambda_r$ are the nonzero characteristic roots of
$A$, and let $H$ be an $m\times m$ orthogonal matrix such that $H^\t
AH=\diag(\lambda_1,\ldots,\lambda_r,0,\ldots,0)$. Partition $H$ as
$H=[H_1,H_2]$,where $H_1$ is $m\times r$ and $H_2$ is $m\times
(m-r)$; then
$$
A = H\Br{\begin{array}{cc}
           D_1 & 0 \\
           0 & 0
         \end{array}
}H^\t = H_1D_1H^\t_1.
$$
Putting
$\sqrt{D_1}=\diag(\sqrt{\lambda_1},\ldots,\sqrt{\lambda_r})$, we
then have
$$
A = H_1\sqrt{D_1}\sqrt{D_1}H^\t_1 = BB^\t,
$$
where $B=H_1\sqrt{D_1}$ is $m\times r$ of rank $r$. As for (ii), let
$C$ bee an $m\times m$ nonsingular matrix whose first $r$ columns
are the columns of the matrix $B$ in (i). Then
$$
A = C\Br{\begin{array}{cc}
           \I_r & 0 \\
           0 & 0
         \end{array}
}C^\t.
$$
\end{proof}

The following result is used often in the text.
\begin{prop}[Vinograd, 1950]\label{prop:vinograd}
Suppose that $A$ and $B$ are real matrices, where $A$ is $k\times m$
and $B$ is $k\times n$, with $m\leqslant n$. Then $AA^\t=BB^\t$ if
and only if there exists an $m\times n$ matrix $H$ with $HH^\t=\I_m$
such that $AH=B$.
\end{prop}

\begin{proof}
First suppose there exists an $m\times n$ matrix $H$ with
$HH^\t=\I_m$ such that $AH=B$. Then $BB^\t=AHH^\t A^\t=AA^\t$.

Now suppose that $AA^\t=BB^\t$. Let $C$ be a $k\times k$ nonsingular
matrix such that
$$
AA^\t=BB^\t = C\Br{\begin{array}{cc}
                     \I_r & 0 \\
                     0 & 0
                   \end{array}
}C^\t,
$$
where $\rank(AA^\t)=r$. Now put $D=C^{-1}A, E=C^{-1}B$ and partition
these as
$$
D = \Br{\begin{array}{c}
          D_1 \\
          D_2
        \end{array}
},\quad E=\Br{\begin{array}{c}
                E_1 \\
                E_2
              \end{array}
},
$$
where $D_1$ is $r\times m$, $D_2$ is $(k-r)\times m$, $E_1$ is
$r\times n$, and $E_2$ is $(k-r)\times n$. Then
$$
EE^\t = \Br{\begin{array}{cc}
              E_1E^\t_1 & E_1E^\t_2 \\
              E_2E^\t_1 & E_2E^\t_2
            \end{array}
} = C^{-1}BB^\t C^{-1,\t} = \Br{\begin{array}{cc}
                     \I_r & 0 \\
                     0 & 0
                   \end{array}
}
$$
and
$$
DD^\t = \Br{\begin{array}{cc}
              D_1D^\t_1 & D_1D^\t_2 \\
              D_2D^\t_1 & D_2D^\t_2
            \end{array}
} = C^{-1}AA^\t C^{-1,\t} = \Br{\begin{array}{cc}
                     \I_r & 0 \\
                     0 & 0
                   \end{array}
}
$$
which imply that $E_1E^\t_1=D_1D^\t_1=\I_r$ and $D_2=0,E_2=0$, so
that
$$
D = \Br{\begin{array}{c}
          D_1 \\
          0
        \end{array}
},\quad E=\Br{\begin{array}{c}
                E_1 \\
                0
              \end{array}
}.
$$
Now let $\widetilde E_2$ be an $(n-r)\times n$ matrix such that
$$
\widetilde E = \Br{\begin{array}{c}
                     E_1 \\
                     \widetilde E_2
                   \end{array}
}
$$
is an $n\times n$ orthogonal matrix, and choose an $(n-r)\times m$
matrix $\widetilde D_2$ and an $(n-r)\times(n-m)$ matrix $\widetilde
D_3$ such that
$$
\widetilde D = \Br{\begin{array}{cc}
                     D_1 & 0 \\
                     \widetilde D_2 & \widetilde D_3
                   \end{array}
}
$$
is an $n\times n$ orthogonal matrix. Then
$$
E=\Br{\begin{array}{c}
                E_1 \\
                0
              \end{array}
}=\Br{\begin{array}{cc}
        \I_r & 0 \\
        0 & 0
      \end{array}
}\widetilde E, \quad [D,0]=\Br{\begin{array}{cc}
        D_1 & 0 \\
        0 & 0
      \end{array}
}=\Br{\begin{array}{cc}
        \I_r & 0 \\
        0 & 0
      \end{array}
}\widetilde D,
$$
and hence  $E=[D,0]\widetilde D^\t \widetilde E = [D,0]Q$, where
$Q=\widetilde D^\t\widetilde E$ is $n\times n$ orthogonal.
Partitioning $Q$ as
$$
Q = \Br{\begin{array}{c}
          H \\
          P
        \end{array}
},
$$
where $H$ is $m\times n$ and $P$ is $(n-m)\times n$, we then have
$HH^\t = \I_m$ and
$$
C^{-1}B = E=DH = C^{-1}AH
$$
so that $B=AH$, completing the proof.
\end{proof}

\begin{prop}
Let $A$ be an $n\times m$ real matrix of rank $m(\leqslant n)$.
Then:
\begin{enumerate}[(i)]
\item $A$ can be written as $A=H_1B$, where $H_1$ is $n\times m$
with $H^\t_1H_1=\I_m$ and $B$ is $m\times m$ positive definite.
\item $A$ can be written as
$$
A = H\Br{\begin{array}{c}
           \I_m \\
           0
         \end{array}
}B,
$$
where $H$ is $n\times n$ orthogonal and $B$ is $m\times m$ positive
definite.
\end{enumerate}
\end{prop}

\begin{proof}
As for (i), let $B:=\sqrt{A^\t A}$ be the positive definite square
root of the positive definite matrix $A^\t A$, so that
$$
A^\t A = B^2 = B^\t B.
$$
Now by using Theorem~\ref{prop:vinograd}, $A$ can be written as
$A=H_1B$, where $H_1$ is $n\times m$ with $H^\t_1H_1=\I_m$. As for
(ii), let $H_1$ be the matrix in (i) such that $A=H_1B$ and choose
an $n\times (n-m)$ matrix $H_2$ so that $H=[H_1,H_2]$ is $n\times n$
orthogonal. Then
$$
A =H_1B= H\Br{\begin{array}{c}
           \I_m \\
           0
         \end{array}
}B.
$$
We are done.
\end{proof}
We now turn to decompositions of positive definite matrices in terms
of triangular matrices.

\begin{thrm}\label{th:unique-factorization}
If $A$ is an $m\times m$ positive definite matrix, then there exists
a \red{unique} $m\times m$ upper-triangular matrix $T$ with positive
diagonal elements such that $A=T^\t T$.
\end{thrm}

\begin{proof}
An induction proof can easily be constructed. The stated result
holds trivially for $m = 1$. Suppose the result holds for positive
definite matrices of size $m - 1$. Partition the $m\times m$ matrix
$A$ as
$$
A = \Br{\begin{array}{cc}
          A_{11} & \mathbf{a}_{12} \\
          \mathbf{a}^\t_{12} & a_{22}
        \end{array}
},
$$
where $A_{11}$ is $(m-1)\times (m-1)$. By the induction hypothesis
there exists a unique $(m-1)\times (m-1)$ upper-triangular matrix
$T_{11}$ with positive diagonal elements such that
$A_{11}=T^\t_{11}T_{11}$. Now suppose that
$$
A = \Br{\begin{array}{cc}
          A_{11} & \mathbf{a}_{12} \\
          \mathbf{a}^\t_{12} & a_{22}
        \end{array}
} =  \Br{\begin{array}{cc}
          T^\t_{11} & 0 \\
          \mathbf{x}^\t & y
        \end{array}
}\Br{\begin{array}{cc}
          T_{11} & \mathbf{x} \\
          0 & y
        \end{array}
} =\Br{\begin{array}{cc}
          T^\t_{11}T_{11} & T^\t_{11}\mathbf{x} \\
          \mathbf{x}^\t T_{11} & \mathbf{x}^\t\mathbf{x}+y^2
        \end{array}
},
$$
where $\mathbf{x}$ is $(m-1)\times1$ and $y\in\real^1$. For this to
hold we must have $\mathbf{x}=\Pa{T^\t_{11}}^{-1}\mathbf{a}_{12}$,
and then
$$
y^2 = a_{22} - \mathbf{x}^\t\mathbf{x} = a_{22} -
\mathbf{a}^\t_{12}T^{-1}_{11}\Pa{T^\t_{11}}^{-1}\mathbf{a}_{12} =
a_{22} - \mathbf{a}^\t_{12}A^{-1}_{11}\mathbf{a}_{12}.
$$
Note that this is positive, and the unique $y>0$ satisfying this is
$$
y = \sqrt{a_{22} - \mathbf{a}^\t_{12}A^{-1}_{11}\mathbf{a}_{12}}.
$$
This completes the proof.
\end{proof}

\begin{thrm}\label{th:unique-for-rectangular}
If $A$ is an $n\times m$ real matrix of rank $m(\leqslant n)$, then
$A$ can be uniquely written as $A=H_1T$, where $H_1$ is $n\times m$
with $H^\t_1 H_1=\I_m$ and $T$ is $m\times m$ upper-triangular wit
positive diagonal elements.
\end{thrm}

\begin{proof}
Since $A^\t A$ is $m\times m$ positive definite it follows form
Theorem~\ref{th:unique-factorization} that there exists a unique
$m\times m$ upper-triangular matrix with positive diagonal elements
such that $A^\t A=T^\t T$. By Theorem~\ref{prop:vinograd}, there
exists an $n\times m$ matrix $H_1$ with $H^\t_1H_1=\I_m$ such that
$A=H_1T$. Note that $H_1$ is unique because $T$ is unique and
$\rank(T)=m$.
\end{proof}

\begin{thrm}
If $A$ is an $m\times m$ positive definite matrix and $B$ is an
$m\times m$ symmetric matrix, there exists an $m\times m$
nonsingular matrix $L$ such that $A=LL^\t$ and $B=LDL^\t$, where
$D=\diag(d_1,\ldots,d_m)$, with $d_1,\ldots,d_m$ being the
characteristic roots of $A^{-1}B$. If $B$ is positive definite and
$d_1,\ldots,d_m$ are all distinct, $L$ is unique up to sign changes
in the first row of $L$.
\end{thrm}

\begin{proof}
Let $\sqrt{A}$ be the positive definite square root of $A$. There
exists an $m\times m$ orthogonal matrix $H$ such that
$$
A^{-1/2}BA^{1/2} = HDH^\t,
$$
where $D=\diag(d_1,\ldots,d_m)$. Putting $L=A^{1/2}H$, we now have
$A=LL^\t$ and $B=LDL^\t$. Note that $d_1,\ldots,d_m$ are the
characteristic roots of $A^{-1}B$.

Now suppose that $B$ is positive definite and the $d_j$ are all
distinct. Assume that as well as $A=LL^\t$ and $B=LDL^\t$, we also
have $A=MM^\t$ and $B=MDM^\t$, where $M$ is $m\times m$ nonsingular.
Then
$$
\Pa{M^{-1}L}\Pa{M^{-1}L}^\t =M^{-1}LL^\t M^{-1,\t} =
M^{-1}AM^{-1,\t} = M^{-1}MM^\t M^{\t,-1}=\I_m
$$
so that the matrix $Q=M^{-1}L$ is orthogonal and $QD=DQ$. If
$Q=(q_{ij})$ we then have $q_{ij}d_i=q_{ij}d_j$ so that $q_{ij}=0$
for $i\neq j$. Since $Q$ is orthogonal it must then have the form
$\widetilde Q=\diag(\pm1,\ldots,\pm1)$, and $L=M\widetilde Q$.
\end{proof}

\begin{thrm}[SVD]
If $A$ is an $m\times n$ real matrix $(m\leqslant n)$, there exist
an $m\times m$ orthogonal matrix $H$ and an $n\times n$ orthogonal
matrix $Q$ such that
$$
HAQ^\t = [\Sigma_m,0],
$$
where $\Sigma_m = \diag(d_1,\ldots,d_m)$ for $d_j\geqslant0,
j=1,\ldots,m$ and $d^2_1,\ldots,d^2_m$ are the characteristic roots
of $AA^\t$.
\end{thrm}

\begin{proof}
Let $H$ be an orthogonal $m\times m$ matrix such that $AA^\t=H^\t
D^2H$, where $D^2=\diag(d^2_1,\ldots,d^2_m)$, with $d^2_j\geqslant0$
for $j=1,\ldots,m$ because $AA^\t$ is non-negative definite. Let
$D=\diag(d_1,\ldots,d_m)$ with $d_j\geqslant0$ for $j=1,\ldots,m$;
then $AA^\t=(H^\t D)(H^\t D)^\t$, and by
Theorem~\ref{prop:vinograd}, there exists an $m\times n$ matrix
$Q_1$ with $Q_1Q^\t_1=\I_m$ such that $A=H^\t DQ_1$. Choose an
$(n-m)\times n$ matrix $Q_2$ so that the $n\times n$ matrix
$$
Q = \Br{\begin{array}{c}
          Q_1 \\
          Q_2
        \end{array}
}
$$
is orthogonal; we now have $A=H^\t DQ_1 = H^\t[\Sigma_m,0]Q$ so that
$HAQ^\t = [\Sigma_m,0]$, and the proof is complete.
\end{proof}

\begin{thrm}
If $Z\in\rS\rO(m)$, i.e., $Z$ is an orthogonal matrix with
determinant one, then there exists an $m\times m$ skew-symmetric $X$
such that
$$
Z= e^X.
$$
\end{thrm}

\section{Appendix IV: Selberg's integral}

This section is rewritten based on Mehta's book \cite{Mehta}. The
well-known Selberg's integral is calculated, and some variants and
consequences are obtained as well.

\begin{thrm}[Selberg's integral]
For any positive integer $N$, let $[\dif x]=\dif x_1\cdots \dif
x_N$,
\begin{eqnarray}
\Delta(x)\equiv\Delta(x_1,\ldots,x_N) = \begin{cases}
\prod_{1\leqslant i<j\leqslant N}(x_i-x_j), & \text{if}~ N>1,\\
1, &\text{if}~ N=1,
\end{cases}
\end{eqnarray}
and
\begin{eqnarray}
\Phi(x)\equiv\Phi(x_1,\ldots,x_N)
=\Pa{\prod^N_{j=1}x^{\alpha-1}_j(1-x_j)^{\beta-1}}\abs{\Delta(x)}^{2\gamma}.
\end{eqnarray}
Then
\begin{eqnarray}\label{eq:Selberg-integral}
S_N(\alpha,\beta,\gamma) \equiv \int^1_0\cdots\int^1_0\Phi(x)[\dif
x] = \prod^{N-1}_{j=0}\frac{\Gamma(\alpha+\gamma
j)\Gamma(\beta+\gamma j)\Gamma(\gamma+1+\gamma
j)}{\Gamma(\alpha+\beta+\gamma(N+j-1))\Gamma(1+\gamma)},
\end{eqnarray}
and for $1\leqslant K\leqslant N$,
\begin{eqnarray}\label{eq:Aomoto-integral}
\int^1_0\cdots\int^1_0\Pa{\prod^K_{j=1}x_j}\Phi(x)[\dif x] =
\prod^K_{j=1}\frac{\alpha+\gamma(N-j)}{\alpha+\beta+\gamma(2N-j-1)}\int^1_0\cdots\int^1_0\Phi(x)[\dif
x],
\end{eqnarray}
valid for integer $N$ and complex $\alpha,\beta,\gamma$ with
\begin{eqnarray}
\mathrm{Re}(\alpha)>0,\quad\mathrm{Re}(\beta)>0,\quad\mathrm{Re}(\gamma)>
- \min\Pa{\frac1N, \frac{\mathrm{Re}(\alpha)}{N-1},
\frac{\mathrm{Re}(\beta)}{N-1}}.
\end{eqnarray}
\end{thrm}

\begin{proof}[Aomoto's proof]
For brevity, let us write
\begin{eqnarray}
\langle f(x_1,\ldots,x_N)\rangle :=\frac{
\int_{[0,1]^N}f(x_1,\ldots,x_N)\Phi(x_1,\ldots,x_N)\dif x_1\cdots
\dif x_N}{\int_{[0,1]^N}\Phi(x_1,\ldots,x_N) \dif x_1\cdots \dif
x_N}.
\end{eqnarray}
Firstly, we note that
\begin{eqnarray}\label{eq:id}
&&\frac{\dif}{\dif x_1}\Pa{x^a_1x_2\cdots x_K\Phi}\\
&&=(a+\alpha-1)x^{a-1}_1x_2\cdots x_K\Phi -
(\beta-1)\frac{x^a_1x_2\cdots x_K}{1-x_1}\Phi + 2\gamma\sum^N_{j=2}
\frac{x^a_1x_2\cdots x_K}{x_1-x_j}\Phi.
\end{eqnarray}
Indeed,
\begin{eqnarray}
\frac{\dif}{\dif x_1}\Pa{x^a_1x_2\cdots x_K\Phi} = a
x^{a-1}_1x_2\cdots x_K \Phi + x^a_1x_2\cdots x_K\frac{\dif}{\dif
x_1}\Phi(x_1,\ldots,x_N),
\end{eqnarray}
where
\begin{eqnarray}
\frac{\dif}{\dif x_1}\Phi = \frac{\dif}{\dif
x_1}\Pa{\abs{\Delta(x)}^{2\gamma}\prod^N_{j=1}x^{\alpha-1}_j(1-x_j)^{\beta-1}}.
\end{eqnarray}
Since
\begin{eqnarray}
\Br{x^a_1x_2\cdots x_K\Phi(x_1,\ldots,x_N)}^{x_1=1}_{x_1=0} =
x_2\cdots x_K\Phi(1,x_2,\ldots,x_N) - 0 = 0,
\end{eqnarray}
it follows from integrating between 0 and 1 in \eqref{eq:id}, we get
\begin{eqnarray}
0 = (a+\alpha-1)\langle x^{a-1}_1x_2\cdots x_K\rangle -
(\beta-1)\left\langle \frac{x^a_1x_2\cdots x_K}{1-x_1}\right\rangle
+ 2\gamma\sum^N_{j=2}\left\langle \frac{x^a_1x_2\cdots
x_K}{x_1-x_j}\right\rangle.
\end{eqnarray}
Now for $a=1$ and $a=2$, we get
\begin{eqnarray}
0&=& \alpha\langle x_2\cdots x_K\rangle - (\beta-1)\left\langle
\frac{x_1x_2\cdots x_K}{1-x_1}\right\rangle +
2\gamma\sum^N_{j=2}\left\langle \frac{x_1x_2\cdots
x_K}{x_1-x_j}\right\rangle,\label{eq:a=1}\\
0&=& (\alpha+1)\langle x_1 x_2\cdots x_K\rangle -
(\beta-1)\left\langle \frac{x^2_1x_2\cdots x_K}{1-x_1}\right\rangle
+ 2\gamma\sum^N_{j=2}\left\langle \frac{x^2_1x_2\cdots
x_K}{x_1-x_j}\right\rangle.\label{eq:a=2}
\end{eqnarray}
Clearly
\begin{eqnarray}
\left\langle \frac{x_1x_2\cdots x_K}{x_1-x_j}\right\rangle  -
\left\langle\frac{x^2_1x_2\cdots x_K}{x_1-x_j}\right\rangle&=&
\left\langle \frac{x_1x_2\cdots x_K}{x_1-x_j} - \frac{x^2_1x_2\cdots
x_K}{x_1-x_j}\right\rangle\notag\\
&=&\langle x_1x_2\cdots x_K\rangle;
\end{eqnarray}
interchanging $x_1$ and $x_j$ and observing the symmetry,
\begin{eqnarray}
\left\langle \frac{x_1x_2\cdots x_K}{x_1-x_j} \right\rangle =
-\left\langle \frac{x_jx_2\cdots x_K}{x_1-x_j}\right\rangle =
\begin{cases}
0,&\text{if}~2\leqslant j\leqslant K,\\
\frac12\langle x_2\cdots x_K \rangle,&\text{if}~K< j\leqslant N,
\end{cases}
\end{eqnarray}
and
\begin{eqnarray}
\left\langle \frac{x^2_1x_2\cdots x_K}{x_1-x_j} \right\rangle =
-\left\langle \frac{x^2_jx_2\cdots x_K}{x_1-x_j}\right\rangle =
\begin{cases}
\frac12\langle x_1x_2\cdots x_K\rangle,&\text{if}~2\leqslant j\leqslant K,\\
\langle x_1x_2\cdots x_K \rangle,&\text{if}~K< j\leqslant N.
\end{cases}
\end{eqnarray}
Performing the difference: \eqref{eq:a=1} $-$ \eqref{eq:a=2} and
using the above, we get:
\begin{eqnarray}
(\alpha+1+\gamma(K-1)+2\gamma(N-K)+\beta-1)\langle x_1x_2\cdots
x_K\rangle = (\alpha+\gamma(N-K))\langle x_2\cdots x_K\rangle,
\end{eqnarray}
or repeating the process
\begin{eqnarray}
\langle x_1x_2\cdots x_K\rangle &=&
\frac{(\alpha+\gamma(N-K))}{\alpha+\beta+\gamma(2N-K-1)}\langle
x_1\cdots x_{K-1}\rangle = \cdots\\
&=&
\prod^K_{j=1}\frac{\alpha+\gamma(N-j)}{\alpha+\beta+\gamma(2N-j-1)}.
\end{eqnarray}
For $K=N$, \eqref{eq:Aomoto-integral} can be written as
\begin{eqnarray}\label{eq:recurrence}
\frac{S_N(\alpha+1,\beta,\gamma)}{S_N(\alpha,\beta,\gamma)} =
\langle x_1x_2\cdots x_N\rangle =
\prod^K_{j=1}\frac{\alpha+\gamma(N-j)}{\alpha+\beta+\gamma(2N-j-1)}
\end{eqnarray}
or for a positive integer $\alpha$,
\begin{eqnarray}
\frac{S_N(\alpha,\beta,\gamma)}{S_{N}(1,\beta,\gamma)} &=&
\frac{S_N(\alpha,\beta,\gamma)}{S_{N}(\alpha-1,\beta,\gamma)}\cdots
\frac{S_N(2,\beta,\gamma)}{S_{N}(1,\beta,\gamma)}\\
&=&
\prod^K_{j=1}\frac{\alpha-1+\gamma(N-j)}{\alpha-1+\beta+\gamma(2N-j-1)}\cdots
\prod^K_{j=1}\frac{1+\gamma(N-j)}{1+\beta+\gamma(2N-j-1)}\\
&=& \prod^K_{j=1}\frac{(\alpha-1+\gamma(N-j))\cdots
(1+\gamma(N-j))}{(\alpha-1+\beta+\gamma(2N-j-1))\cdots
(1+\beta+\gamma(2N-j-1))}.
\end{eqnarray}
Note that $\Gamma(m+\zeta)=(\zeta)_m\Gamma(\zeta)$, where
$(\zeta)_m:=\zeta(\zeta+1)\cdots (\zeta+m-1)$. By using this
identity, we get
$$
(\alpha-1+\gamma(N-j))\cdots (1+\gamma(N-j)) =
\frac{\Gamma(\alpha+\gamma(N-j))}{\Gamma(1+\gamma(N-j))}
$$
and
$$
(\alpha-1+\beta+\gamma(2N-j-1))\cdots (1+\beta+\gamma(2N-j-1)) =
\frac{\Gamma(\alpha+\beta+\gamma(2N-j-1))}{\Gamma(1+\beta+\gamma(2N-j-1))}.
$$
It follows that
\begin{eqnarray*}
&&s_N(\alpha,\beta,\gamma) =S_{N}(1,\beta,\gamma)
\prod^N_{j=1}\frac{\Gamma(\alpha+\gamma(N-j))\Gamma(1+\beta+\gamma(2N-j-1))}{\Gamma(\alpha+\beta+\gamma(2N-j-1))\Gamma(1+\gamma(N-j))}\\
&&=\Pa{S_N(1,\beta,\gamma)\prod^N_{j=1}\frac{\Gamma(1+\beta+\gamma(2N-j-1))}{\Gamma(\beta+\gamma(N-j))\Gamma(1+\gamma(N-j))}}
\prod^N_{j=1}\frac{\Gamma(\alpha+\gamma(N-j))\Gamma(\beta+\gamma(N-j))}{\Gamma(\alpha+\beta+\gamma(2N-j-1))}.
\end{eqnarray*}
As $S_N(\alpha,\beta,\gamma)=S_N(\beta,\alpha,\gamma)$ by the fact
that $\Delta(1-x)=\pm \Delta(x)$, that is,
$S_N(\alpha,\beta,\gamma)$ is a symmetric function of $\alpha$ and
$\beta$, at the same time, the following factor is also symmetric in
$\alpha$ and $\beta$:
$$
\prod^N_{j=1}\frac{\Gamma(\alpha+\gamma(N-j))\Gamma(\beta+\gamma(N-j))}{\Gamma(\alpha+\beta+\gamma(2N-j-1))}.
$$
It follows that the factor
$$
S_N(1,\beta,\gamma)\prod^N_{j=1}\frac{\Gamma(1+\beta+\gamma(2N-j-1))}{\Gamma(\beta+\gamma(N-j))\Gamma(1+\gamma(N-j))}
$$
should be a symmetric function of $\alpha$ and $\beta$. But,
however, this factor is independent of $\alpha$, therefore it should
be also independent of $\beta$ by the symmetry. Denote this factor
by $c(\gamma,N)$, we get
\begin{eqnarray}\label{eq:aomoto}
S_N(\alpha,\beta,\gamma) &=&c(\gamma,N)
\prod^N_{j=1}\frac{\Gamma(\alpha+\gamma(N-j))\Gamma(\beta+\gamma(N-j))}{\Gamma(\alpha+\beta+\gamma(2N-j-1))}\notag\\
&=&c(\gamma,N) \prod^N_{j=1}\frac{\Gamma(\alpha+\gamma
(j-1))\Gamma(\beta+\gamma
(j-1))}{\Gamma(\alpha+\beta+\gamma(N+j-2))}
\end{eqnarray}
where $c(\gamma,N)$ is independent of $\alpha$ and $\beta$.

To determine $c(\gamma,N)$, put $\alpha=\beta=1$;
\begin{eqnarray}
S_N(1,1,\gamma) = \int_{[0,1]^N} \abs{\Delta(x)}^{2\gamma}\dif x =
c(\gamma,N)\prod^N_{j=1}\frac{\Gamma(1+\gamma(j-1))^2}{\Gamma(2+\gamma(N+j-2))}.
\end{eqnarray}

Let $y$ be the largest of the $x_1,\ldots,x_N$ and replace the other
$x_j$ by $x_j=yt_j$, where $0\leqslant t_j\leqslant 1$. Without loss
of generality, we assume that $y=x_N$. Then $x_j=yt_j$ for
$j=1,\ldots,N-1$. Then
\begin{eqnarray}
\abs{\Delta(x_1,\ldots,x_N)}^{2\gamma} = y^{\gamma
N(N-1)}\cdot\abs{\Delta(t_1,\ldots,t_{N-1})}^{2\gamma}\cdot\prod^{N-1}_{j=1}(1-t_j)^{2\gamma},
\end{eqnarray}
and the Jacobian of change of variables is
\begin{eqnarray}
\abs{\det\Pa{\frac{\partial(x_1,\ldots,x_N)}{\partial(y,t_1,\ldots,t_{N-1})}}}=\abs{\abs{\begin{array}{ccccc}
                                                                                t_1 & t_2 & \cdots & t_{N-1} & 1 \\
                                                                                y & 0 & 0 & \cdots & 0 \\
                                                                                0 & y & 0 & \cdots & 0 \\
                                                                                \vdots & \vdots & \ddots & \vdots & \vdots \\
                                                                                0 & 0 & \cdots & y &
                                                                                0
                                                                              \end{array}
}}=y^{N-1}.
\end{eqnarray}
Now we have
\begin{eqnarray*}
&&S_N(1,1,\gamma) = N!\int_{0\leqslant x_1\leqslant\cdots\leqslant
x_N\leqslant 1}\abs{\Delta(x)}^{2\gamma}\dif x \\
&&= N!\int^1_0 y^{\gamma N(N-1)} y^{N-1}\dif y\cdot \int_{0\leqslant
t_1\leqslant\cdots\leqslant t_{N-1}\leqslant
1}\abs{\Delta(t_1,\ldots,t_{N-1})}^{2\gamma}\prod^{N-1}_{j=1}(1-t_j)^{2\gamma}\dif
t_1\ldots
\dif t_{N-1}\\
&&=\Pa{N\int^1_0y^{\gamma N(N-1)} y^{N-1}\dif
y}\Pa{\int_{[0,1]^{N-1}}\abs{\Delta(t_1,\ldots,t_{N-1})}^{2\gamma}\prod^{N-1}_{j=1}(1-t_j)^{2\gamma}\dif
t_1\ldots
\dif t_{N-1}}\\
&&=\frac1{\gamma(N-1)+1}S_{N-1}(1,2\gamma+1,\gamma),
\end{eqnarray*}
that is
\begin{eqnarray}
S_N(1,1,\gamma) &=&
\frac1{\gamma(N-1)+1}S_{N-1}(1,2\gamma+1,\gamma)\\
&=& \frac{c(\gamma,N-1)}{\gamma(N-1)+1}\cdot\prod^{N-1}_{j=1}
\frac{\Gamma(1+\gamma(j-1))\Gamma(2\gamma+1+\gamma(j-1))}{\Gamma(1+2\gamma+1+\gamma(N-1+j-2))}\\
&=&\frac{c(\gamma,N-1)}{\gamma(N-1)+1}\cdot\prod^{N-1}_{j=1}
\frac{\Gamma(1+\gamma(j-1))\Gamma(1+\gamma+\gamma
j)}{\Gamma(2+\gamma(N+j-1))}.
\end{eqnarray}
Thus
\begin{eqnarray*}
c(\gamma,N)\prod^N_{j=1}\frac{\Gamma(1+\gamma(j-1))^2}{\Gamma(2+\gamma(N+j-2))}
= S_N(1,1,\gamma) =
\frac{c(\gamma,N-1)}{\gamma(N-1)+1}\cdot\prod^{N-1}_{j=1}
\frac{\Gamma(1+\gamma(j-1))\Gamma(1+\gamma+\gamma
j)}{\Gamma(2+\gamma(N+j-1))}.
\end{eqnarray*}
This implies that
\begin{eqnarray}
\frac{c(\gamma,N)}{c(\gamma,N-1)} =\frac{\Gamma(1+\gamma
N)}{\Gamma(1+\gamma)},
\end{eqnarray}
or
\begin{eqnarray}
c(\gamma,N) =
c(\gamma,1)\cdot\Pa{\frac{c(\gamma,N)}{c(\gamma,N-1)}\cdots
\frac{c(\gamma,2)}{c(\gamma,1)}} = c(\gamma,1)\cdot\Pa{
\frac{\Gamma(1+\gamma N)}{\Gamma(1+\gamma)}\cdots
\frac{\Gamma(1+\gamma 2)}{\Gamma(1+\gamma)}}.
\end{eqnarray}
Finally we get
\begin{eqnarray}
c(\gamma,N) = c(\gamma,1)\cdot\prod^N_{j=2}\frac{\Gamma(1+\gamma
j)}{\Gamma(1+\gamma)} = \prod^N_{j=2}\frac{\Gamma(1+\gamma
j)}{\Gamma(1+\gamma)},
\end{eqnarray}
where the fact that $c(\gamma,1)=1$ is trivial.
\end{proof}

\begin{remark}
Note that the conclusion is derived here for integers $\alpha,\beta$
and complex $\gamma$.

A slight change of reasoning due to Askey gives \eqref{eq:aomoto}
directly for complex $\alpha,\beta$ and $\gamma$ as follows.
\eqref{eq:recurrence} and the symmetry identity that
$S_N(\alpha,\beta,\gamma)=S_N(\beta,\alpha,\gamma)$ give the ratio
of $S_N(\alpha,\beta,\gamma)$ and $S_N(\alpha,\beta+m,\gamma)$ for
any integer $m$,
\begin{eqnarray}
&&\frac{S_N(\alpha,\beta+m,\gamma)}{S_N(\alpha,\beta,\gamma)}=
\frac{S_N(\alpha,\beta+m,\gamma)}{S_N(\alpha,\beta+m-1,\gamma)}\cdots\frac{S_N(\alpha,\beta+1,\gamma)}{S_N(\alpha,\beta,\gamma)}\\
&&=\prod^N_{j=1}\frac{(\beta+m-1)+\gamma(N-j)}{\alpha+(\beta+m-1)+\gamma(2N-j-1)}\cdots
\prod^N_{j=1}\frac{\beta+\gamma(N-j)}{\alpha+\beta+\gamma(2N-j-1)}\\
&&= \prod^N_{j=1}
\frac{(\beta+\gamma(N-j))_m}{(\alpha+\beta+\gamma(2N-j-1))_m},
\end{eqnarray}
where we have used the notation
$$
(a)_m = \frac{\Gamma(a+m)}{\Gamma(m)};\quad m\geqslant0;
$$
i.e. $(a)_0=1$ and $(a)_m:=a(a+1)\cdots (a+m-1)$ for $m\geqslant1$.
Now
\begin{eqnarray}
S_N(\alpha,\beta+m,\gamma) &=& \int_{[0,1]^N}
\abs{\Delta(x)}^{2\gamma}\prod^N_{j=1}x^{\alpha-1}_j(1-x_j)^{\beta+m-1}\dif x_j\\
&=& m^{-\alpha N-\gamma N(N-1)}
\int_{[0,m]^N}\abs{\Delta(x)}^{2\gamma}\prod^N_{j=1}x^{\alpha-1}_j\Pa{1-\frac{x_j}m}^{\beta+m-1}\dif
x_j.
\end{eqnarray}
Thus
\begin{eqnarray}\label{eq:depends-on-m}
&&S_N(\alpha,\beta,\gamma) =
S_N(\alpha,\beta+m,\gamma)\cdot\prod^N_{j=1}
\frac{(\alpha+\beta+\gamma(2N-j-1))_m}{(\beta+\gamma(N-j))_m}\\
&&= m^{-\alpha N-\gamma N(N-1)}
\int_{[0,m]^N}\abs{\Delta(x)}^{2\gamma}\prod^N_{j=1}x^{\alpha-1}_j\Pa{1-\frac{x_j}m}^{\beta+m-1}\dif
x_j
\prod^N_{j=1}\frac{(\alpha+\beta+\gamma(2N-j-1))_m}{(\beta+\gamma(N-j))_m}\\
&&=\prod^N_{j=1}\Pa{\frac{(\alpha+\beta+\gamma(2N-j-1))_m}{(\beta+\gamma(N-j))_m}m^{-\alpha
-\gamma(N-1)}}
\int_{[0,m]^N}\abs{\Delta(x)}^{2\gamma}\prod^N_{j=1}x^{\alpha-1}_j\Pa{1-\frac{x_j}m}^{\beta+m-1}\dif
x_j.
\end{eqnarray}
Denote $a_j=\alpha+\gamma(N-j), b_j=\beta+\gamma(N-j)$ and
$c_j=\alpha+\beta+\gamma(2N-j-1)$, we have
$$
\frac{(c)_m}{(b)_m}m^{b-c} = \frac{\Gamma(b)}{\Gamma(c)}
\frac{\Gamma(c+m)}{\Gamma(b+m)}m^{b-c} = \frac{\Gamma(b)}{\Gamma(c)}
\frac{\frac{\Gamma(c+m)}{\Gamma(m)m^c}}{\frac{\Gamma(b+m)}{\Gamma(m)m^b}}.
$$
By using the fact that
$$
\lim_{m\to\infty}\frac{\Gamma(m+c)}{\Gamma(m)m^c} = 1~~~(\forall
c\in\real),
$$
it follows that
\begin{eqnarray}
\lim_{m\to\infty}\frac{(c)_m}{(b)_m}m^{b-c} =
\frac{\Gamma(b)}{\Gamma(c)},
\end{eqnarray}
therefore
\begin{eqnarray}
\lim_{m\to\infty}\prod^N_{j=1}\frac{(c_j)_m}{(b_j)_m}m^{b_j-c_j} =
\prod^N_{j=1}\frac{\Gamma(b_j)}{\Gamma(c_j)}.
\end{eqnarray}
Taking $m\to\infty$ in \eqref{eq:depends-on-m} gives rise to the
following:
\begin{eqnarray}
S_N(\alpha,\beta,\gamma) =\Pa{
\prod^N_{j=1}\frac{\Gamma(b_j)}{\Gamma(c_j)}} \int^\infty_0\cdots
\int^\infty_0\abs{\Delta(x)}^{2\gamma}\prod^N_{j=1}x^{\alpha-1}_j\exp\Pa{-x_j}\dif
x_j.
\end{eqnarray}
Furthermore,
\begin{eqnarray}
S_N(\alpha,\beta,\gamma) =\Pa{
\prod^N_{j=1}\frac{\Gamma(a_j)\Gamma(b_j)}{\Gamma(c_j)}}
\frac{\int^\infty_0\cdots
\int^\infty_0\abs{\Delta(x)}^{2\gamma}\prod^N_{j=1}x^{\alpha-1}_j\exp\Pa{-x_j}\dif
x_j}{\prod^N_{j=1}\Gamma(a_j)}.
\end{eqnarray}
By the symmetry of $\alpha$ and $\beta$, it follows that the factor
\begin{eqnarray}
\frac{\int^\infty_0\cdots
\int^\infty_0\abs{\Delta(x)}^{2\gamma}\prod^N_{j=1}x^{\alpha-1}_j\exp\Pa{-x_j}\dif
x_j}{\prod^N_{j=1}\Gamma(a_j)}
\end{eqnarray}
is a symmetric function of $\alpha$ and $\beta$. Since it is
independent of $\beta$, it is also independent of $\alpha$ by the
symmetry of $\alpha$ and $\beta$. Thus
\begin{eqnarray}
\frac{\int^\infty_0\cdots
\int^\infty_0\abs{\Delta(x)}^{2\gamma}\prod^N_{j=1}x^{\alpha-1}_j\exp\Pa{-x_j}\dif
x_j}{\prod^N_{j=1}\Gamma(a_j)} = \frac{\int^\infty_0\cdots
\int^\infty_0\abs{\Delta(x)}^{2\gamma}\prod^N_{j=1}x^{\beta-1}_j\exp\Pa{-x_j}\dif
x_j}{\prod^N_{j=1}\Gamma(b_j)},
\end{eqnarray}
which is denoted by $c(\gamma,N)$. This indicates that
\begin{eqnarray}
&&\int^\infty_0\cdots
\int^\infty_0\abs{\Delta(x)}^{2\gamma}\prod^N_{j=1}x^{\alpha-1}_j\exp\Pa{-x_j}\dif x_j=c(\gamma,N)\prod^N_{j=1}\Gamma(a_j)\\
&&=\prod^N_{j=2}\frac{\Gamma(1+\gamma
j)}{\Gamma(1+\gamma)}\prod^N_{j=1}\Gamma(a_j)=\prod^N_{j=1}\frac{\Gamma(\alpha+\gamma(j-1))\Gamma(1+\gamma
j)}{\Gamma(1+\gamma)}.
\end{eqnarray}
\end{remark}

\begin{remark}
Now we show that
$$
\lim_{n\to\infty}
\frac{\Gamma(n+\alpha)}{\Gamma(n)n^\alpha}=1~(\forall
\alpha\in\real^+\cup\set{0}).
$$
Indeed, in order to prove this fact, we need a limit representation
of the gamma function given by Carl Friedrich Gauss via Euler's
representation of
$n!=\prod^\infty_{k=1}\frac{\Pa{1+\frac1k}^n}{1+\frac nk}$:
$$
\Gamma(z) = \lim_{n\to\infty} \frac{n!n^z}{z(z+1)\cdots (z+n)}.
$$
\begin{eqnarray}
\lim_{n\to\infty} \frac{\Gamma(n+\alpha)}{\Gamma(n)n^\alpha} &=&
\lim_{n\to\infty}
\frac{\Gamma(\alpha)\alpha(\alpha+1)\cdots(\alpha+n-1)}{(n-1)!n^\alpha}\\
&=&\Gamma(\alpha)\lim_{n\to\infty}
\frac{\alpha(\alpha+1)\cdots(\alpha+n)}{n!n^\alpha}\frac{n}{n+\alpha}\\
&=&\Gamma(\alpha)\lim_{n\to\infty}
\frac{\alpha(\alpha+1)\cdots(\alpha+n)}{n!n^\alpha}\lim_{n\to\infty}\frac{n}{n+\alpha}\\
&=&\Gamma(\alpha)\cdot\frac1{\Gamma(\alpha)}\cdot 1 = 1.
\end{eqnarray}
Another short and elementary proof of this fact can be derived from
a result related to inequalities for Gamma function ratios
\cite{Wendel}:
\begin{eqnarray}\label{eq:Wendel}
x(x+a)^{a-1}\leqslant \frac{\Gamma(x+a)}{\Gamma(x)}\leqslant
x^a~~(\forall a\in[0,1]).
\end{eqnarray}
Indeed,
\begin{eqnarray}
\frac{\Gamma(x+\alpha)}{\Gamma(x)}\sim x^\alpha
\end{eqnarray}
as $x\to\infty$ with $\alpha$ fixed.(This was the objective of
Wendel's article.) To show this, first suppose that
$\alpha\in[0,1]$. Then Eq.~\eqref{eq:Wendel} gives
$$
\Pa{1+\frac{\alpha}x}^{\alpha-1}\leqslant
\frac{\Gamma(x+\alpha)}{\Gamma(x)x^\alpha}\leqslant1,
$$
leasing to
$\lim_{x\to\infty}\frac{\Gamma(x+\alpha)}{\Gamma(x)x^\alpha}=1$
since
$\lim_{x\to\infty}\Pa{1+\frac{\alpha}x}^{\alpha-1}=\lim_{x\to\infty}1=1$.
For $\alpha>1$, the statement now follows from the fact that
$$
\frac{\Gamma(x+\alpha)}{\Gamma(x)}=(x+\alpha-1)\frac{\Gamma(x+\alpha-1)}{\Gamma(x)}.
$$
We are done.
\end{remark}

\begin{cor}[Laguerre's integral]\label{cor:Lag-int}
By letting $x_j=y_j/L$ and $\beta=L+1$ in Selberg's integral and
taking the limit $L\to \infty$, we obtain
\begin{eqnarray}
\int^\infty_0\cdots
\int^\infty_0\abs{\Delta(x)}^{2\gamma}\prod^N_{j=1}x^{\alpha-1}_j\exp\Pa{-x_j}\dif
x_j=\prod^N_{j=1}\frac{\Gamma(\alpha+\gamma(j-1))\Gamma(1+\gamma
j)}{\Gamma(1+\gamma)}.
\end{eqnarray}
\end{cor}

\begin{cor}[Hermite's integral]\label{cor:Her-int}
By letting $x_j=y_j/L$ and $\alpha=\beta=\lambda L^2+1$ in Selberg's
integral and taking the limit $L\to \infty$, we obtain
\begin{eqnarray}\label{eq:Hermite-integral}
\int^{+\infty}_{-\infty}\cdots
\int^{+\infty}_{-\infty}\abs{\Delta(x)}^{2\gamma}\prod^N_{j=1}\exp\Pa{-\lambda
x^2_j}\dif
x_j=(2\pi)^{N/2}(2\lambda)^{-N(\gamma(N-1)+1)/2}\prod^N_{j=1}\frac{\Gamma(1+\gamma
j)}{\Gamma(1+\gamma)}.
\end{eqnarray}
\end{cor}
\begin{remark}
For an integer $\gamma$, the last equation can also be written as a
finite algebraic identity. Firstly, we note that
\begin{eqnarray}
&&\int^{+\infty}_{-\infty}
\exp\Pa{-a^2x^2-2\mathrm{i}ax\lambda}x^n\dif x =
\Pa{\frac{\mathrm{i}}{2a}\frac{\dif}{\dif\lambda}}^n
\int^{+\infty}_{-\infty} \exp\Pa{-a^2x^2-2\mathrm{i}ax\lambda}\dif x\\
&&=\Pa{\frac{\mathrm{i}}{2a}\frac{\dif}{\dif\lambda}}^n
\Pa{\exp(-\lambda^2)\int^{+\infty}_{-\infty}
\exp\Pa{-(ax+\mathrm{i}\lambda)^2}\dif x}\\
&&=\frac{\sqrt{\pi}}{a}\Pa{\frac{\mathrm{i}}{2a}\frac{\dif}{\dif\lambda}}^n
\exp(-\lambda^2),
\end{eqnarray}
letting $\lambda=0$ in the above reasoning, we get
\begin{eqnarray}
\int^{+\infty}_{-\infty} \exp\Pa{-a^2x^2}x^n\dif x
=\frac{\sqrt{\pi}}{a}\left.\Pa{\frac{\mathrm{i}}{2a}\frac{\dif}{\dif\lambda}}^n
\exp(-\lambda^2)\right|_{\lambda=0}
=\frac{\sqrt{\pi}}{a}\left.\Pa{\frac{\mathrm{i}}{2a}\frac{\dif}{\dif
x}}^n \exp(-x^2)\right|_{x=0} ,
\end{eqnarray}
we replace $a$ by $\sqrt{a}$ in the last equation, we get
\begin{eqnarray}
\int^{+\infty}_{-\infty} \exp\Pa{-ax^2}x^n\dif x
=\sqrt{\frac{\pi}{a}}\left.\Pa{\frac{\mathrm{i}}{2\sqrt{a}}\frac{\dif}{\dif
x}}^n \exp(-x^2)\right|_{x=0},
\end{eqnarray}
thus \eqref{eq:Hermite-integral} therefore takes the form:
\begin{eqnarray}
\Pa{\frac{\mathrm{i}}{2\sqrt{a}}}^{\gamma
N(N-1)}\left.\prod_{1\leqslant p<q\leqslant
N}\Pa{\frac{\partial}{\partial x_p} - \frac{\partial}{\partial
x_q}}^{2\gamma}
\exp\Pa{-\sum^N_{j=1}x^2_j}\right|_{(x_1,\ldots,x_N)=0} =
(2a)^{-\gamma N(N-1)/2}\prod^N_{j=1}\frac{(j\gamma)!}{\gamma!}.
\end{eqnarray}
Replacing the exponential by its power series expansion, one notes
that the term $(-\sum^N_{j=1}x^2_j)^\ell$ gives zero on
differentiation if $\ell<\gamma N(N-1)/2$, and leaves a homogeneous
polynomial of order $\ell-\gamma N(N-1)/2$ in the variables
$x_1,\ldots,x_N$, if $\ell>\gamma N(N-1)/2$. On setting
$x_j=0,j=1,\ldots,N$, one sees that therefore that there is only one
term, corresponding to $\ell=\gamma N(N-1)/2$, which gives a
non-zero contribution. So
\begin{eqnarray}
\prod_{1\leqslant p<q\leqslant N}\Pa{\frac{\partial}{\partial x_p} -
\frac{\partial}{\partial x_q}}^{2\gamma}
\Pa{\sum^N_{j=1}x^2_j}^\ell= 2^\ell
\ell!\prod^N_{j=1}\frac{(j\gamma)!}{\gamma!},
\end{eqnarray}
where $\ell = \gamma N(N-1)/2$. If $P(x):=P(x_1,\ldots,x_N)$ and
$Q(x):=Q(x_1,\ldots,x_N)$ are homogeneous polynomials in
$x:=(x_1,\ldots,x_N)$ of the same degree, then a little reflection
shows that $P(\partial/\partial x)Q(x)$ is a constant which is also
equal to $Q(\partial/\partial x)P(x)$. Thus one can interchange the
roles of $x_j$ and $\partial/\partial x_j$ to get
\begin{eqnarray}
\Pa{\sum^N_{k=1}\frac{\partial^2}{\partial x^2_k}}^\ell
\prod_{1\leqslant i<j\leqslant N}(x_i-x_j)^{2\gamma} =2^\ell
\ell!\prod^N_{j=1}\frac{(j\gamma)!}{\gamma!},
\end{eqnarray}
where $\ell = \gamma N(N-1)/2$.
\end{remark}

\begin{cor}
It holds that
\begin{eqnarray}
\int_{[0,2\pi]^N}
\abs{\Delta(e^{\mathrm{i}\theta_1},\ldots,e^{\mathrm{i}\theta_N})}^{2\gamma}\prod^N_{k=1}\frac{\dif\theta_k}{2\pi}&=&\int^{2\pi}_0\cdots
\int^{2\pi}_0 \prod_{1\leqslant i<j\leqslant N}
\abs{e^{\mathrm{i}\theta_i} - e^{\mathrm{i}\theta_j}}^{2\gamma}
\prod^N_{k=1}\frac{\dif\theta_k}{2\pi}\\
&=&\frac{(N\gamma)!}{(\gamma!)^N},
\end{eqnarray}
where $\gamma$ is non-negative integer.
\end{cor}

\begin{cor}
It holds that when there is no overlap in the two sets of factors,
the result is
\begin{eqnarray}
\cB(K_1,K_2) &=& \int^1_0\cdots \int^1_0
\prod^{K_1}_{i=1}x_i\prod^{K_1+K_2}_{j=K_1+1}(1-x_j)\Phi(x)\dif x\\
&=& \cI_N(\alpha,\beta,\gamma)
\frac{\prod^{K_1}_{i=1}(\alpha+\gamma(N-i))\prod^{K_2}_{j=1}(\beta+\gamma(N-j))}{\prod^{K_1+K_2}_{k=1}(\alpha+\beta+\gamma(2N-k-1))},
\end{eqnarray}
where $K_1,K_2\geqslant0, K_1+K_2\leqslant N$, and when there is
overlap
\begin{eqnarray}
\cC(K_1,K_2,K_3) &=& \int^1_0\cdots \int^1_0
\prod^{K_1}_{i=1}x_i\prod^{K_1+K_2-K_3}_{j=K_1+1-K_3}(1-x_j)\Phi(x)\dif x\\
&=&
\cB(K_1,K_2)\prod^{K_3}_{k=1}\frac{\alpha+\beta+\gamma(N-k-1)}{\alpha+\beta+1+\gamma(2N-k-1)},
\end{eqnarray}
where $K_1,K_2,K_3\geqslant0, K_1+K_2-K_3\leqslant N$.
\end{cor}

\begin{remark}
Still another integral of interest is the average value of the
product of traces of the matrix in the circular ensembles. For
example,
\begin{eqnarray}
S_N(p,\gamma):=\frac1{(2\pi)^N}\frac{(\gamma!)^N}{(N\gamma)!}\int^{2\pi}_0\cdots
\int^{2\pi}_0\abs{\sum^N_{k=1}
e^{\mathrm{i}\theta_k}}^{2p}\prod_{1\leqslant i<j\leqslant N}
\abs{e^{\mathrm{i}\theta_i} - e^{\mathrm{i}\theta_j}}^{2\gamma}
\prod^N_{k=1}\dif\theta_k
\end{eqnarray}
is known for $\gamma=1$ that $S_N(p,1)$ gives the number of
permutations of $(1,\ldots,k)$ in which the length of the longest
increasing subsequence is less than or equal to $N$. One has in
particular,
$$
S_N(p,1)=p!,\quad 0\leqslant p\leqslant N.
$$
It is desirable to know the integrals $S_N(p,\gamma)$ for a general
$\gamma$.
\end{remark}
The following short proof of Selberg's formula is from \cite{aar}.
\begin{proof}[Anderson's proof of Selberg's Integral]
Anderson's proof depends on Dirichlet's generalization of the beta
integral given in the following: For $\re(\alpha_j)>0$,
\begin{eqnarray}\label{eq:Dirichlet-integral}
\int\cdots\int_V p^{\alpha_0-1}_0p^{\alpha_1-1}_1\cdots
p^{\alpha_n-1}_n\dif p_0\cdots\dif p_{n-1} =
\frac{\prod^n_{j=1}\Gamma(\alpha_j)}{\Gamma\Pa{\sum^n_{j=1}\alpha_j}},
\end{eqnarray}
where $V$ is the set $p_j\geqslant0,\sum^n_{j=0}p_j=1$. The formula
is used after a change of variables. To see this, first consider
Selberg's integral, which may be written as
\begin{eqnarray}
S_n = n!A_n(\alpha,\beta,\gamma):= n!\int^1_0\int^{x_n}_0\cdots
\int^{x_2}_0\abs{\phi(0)}^{\alpha-1}\abs{\phi(1)}^{\beta-1}\abs{\Delta_\phi}^{\gamma}\dif
x_1\cdots\dif x_n,
\end{eqnarray}
where $0<x_1<x_2<\cdots<x_n<1$,
\begin{eqnarray}
\phi(t)=\prod^n_{j=1}(t-x_j) =
t^n-\phi_{n-1}t^{n-1}+\cdots+(-1)^n\phi_0
\end{eqnarray}
and $\Delta_\phi$ is the \emph{discriminant} of $\phi$, so that
$$
\abs{\Delta_\phi} = \abs{\prod^n_{j=1}\phi'(x_j)} =
\abs{\prod_{1\leqslant i<j\leqslant n}(x_i-x_j)}^2.
$$
We now change the variables from $x_1,\ldots,x_n$ to
$\phi_0,\ldots,\phi_{n-1}$, which are the elementary symmetric
functions of the $x_i$'s. In fact, we have:
\begin{eqnarray}
A_n(\alpha,\beta,\gamma) = \int
\abs{\phi(0)}^{\alpha-1}\abs{\phi(1)}^{\beta-1}\abs{\Delta_\phi}^{\gamma-\frac12}\dif
\phi_0\dif\phi_1\cdots\dif \phi_{n-1},
\end{eqnarray}
where the integration is over all points
$(\phi_0,\phi_1,\ldots,\phi_{n-1})$ in which the $\phi_j$ are
elementary symmetric functions of $x_1,\ldots,x_n$ with
$0<x_1<\cdots<x_n$. Indeed, it is sufficient to prove that the
Jacobian
$$
\det\Pa{\Br{\frac{\partial \phi_i}{\partial x_j}}} =
\sqrt{\abs{\Delta_\phi}}.
$$
Observe that two columns of the Jacobian are equal when $x_i=x_j$.
Thus $\prod_{i<j}(x_i-x_j)$ is a factor of the determinant.
Moreover, the Jacobian and $\prod_{i<j}(x_i-x_j)$ are homogeneous
and of the same degree. This proves the above Jacobian.

We make a similar change of variables in
Eq.~\eqref{eq:Dirichlet-integral}. To accomplish this, set
$$
\varphi(t)=\prod^n_{j=0}(t-\zeta_j)~~(0\leqslant
\zeta_0<\zeta_1<\cdots<\zeta_n<1)
$$
and let
\begin{eqnarray}
\cD = \Set{\prod^n_{j=1}(t-x_j): \zeta_{j-1}<x_j<\zeta_j;
j=1,\ldots,n}.
\end{eqnarray}
Next, we show that: For all
$\phi(t)=t^n-\phi_{n-1}t^{n-1}+\cdots+(-1)^n\phi_0\in\cD$, the
following map
$$
(\phi_0,\phi_1,\ldots,\phi_{n-1})\mapsto
\Pa{\frac{\phi(\zeta_0)}{\varphi'(\zeta_0)},\ldots,\frac{\phi(\zeta_n)}{\varphi'(\zeta_n)}}\equiv(p_0,p_1,\ldots,p_n)\in\real^{n+1}
$$
where $\varphi'(t)$ denotes the derivative of $\varphi(t)$, is a
bijection and $p_j>0$ with $\sum^n_{j=0}p_j=1$.

Observe that
\begin{eqnarray}
p_j=\frac{\phi(\zeta_j)}{\varphi'(\zeta_j)}=\frac{(\zeta_j-x_1)(\zeta_j-x_2)\cdots
(\zeta_j-x_n)}{(\zeta_j-\zeta_0)\cdots
(\zeta_j-\zeta_{j-1})(\zeta_j-\zeta_{j+1})\cdots
(\zeta_j-\zeta_n)}>0
\end{eqnarray}
since the numerator and denominator have exactly $n-j$ negative
factors. Now let $\varphi_j(t)=\frac{\varphi(t)}{t-\zeta_j}$. By
Lagrange's interpolation formula
\begin{eqnarray}\label{eq:phi}
\phi(t) = \sum^n_{j=0}p_j
\varphi_j(t)\equiv\sum^n_{j=0}\frac{\varphi_j(t)}{\varphi'(\zeta_j)}\phi(\zeta_j).
\end{eqnarray}
One can directly verify this by checking that both sides of the
equation are polynomials of degree $n$ and are equal at $n+1$ points
$t=\zeta_j,j=0,\ldots,n$. Equate the coefficients of $t^n$ on both
sides to get $\sum^n_{j=0}p_j=1$. Now for a given point
$(p_0,p_1,\ldots,p_n)$ with $\sum^n_{j=1}p_j=1$ and
$p_j>0,j=1,\ldots,n$, define $\phi(t)$ by Eq.~\eqref{eq:phi}. The
expressions
$$
\phi(\zeta_j)=p_j\varphi_j(\zeta_j) =
p_j(\zeta_j-\zeta_0)\cdots(\zeta_j-\zeta_{j-1})(\zeta_j-\zeta_{j+1})\cdots
(\zeta_j-\zeta_n)
$$
and
$$
\phi(\zeta_{j+1})=p_{j+1}\varphi_{j+1}(\zeta_{j+1}) =
p_{j+1}(\zeta_{j+1}-\zeta_0)\cdots(\zeta_{j+1}-\zeta_j)(\zeta_{j+1}-\zeta_{j+2})\cdots
(\zeta_{j+1}-\zeta_n)
$$
show that $\phi(\zeta_j)$ and $\phi(\zeta_{j+1})$ have different
signs and $\phi$ vanishes at some point $x_{j+1}$ between $\zeta_j$
and $\zeta_{j+1}$. Thus $\phi\in\cD$. This proves the bijection.

We can now restate Dirichlet's formula
Eq.~\eqref{eq:Dirichlet-integral} as:
\begin{eqnarray}\label{eq:varphi-int}
\int_{\phi(t)\in\cD}\prod^n_{j=0}\abs{\phi(\zeta_j)}^{\alpha_j-1}\dif\phi_0\cdots\dif\phi_{n-1}
=
\frac{\prod^n_{j=0}\abs{\varphi'(\zeta_j)}^{\alpha_j-\frac12}\Gamma(\alpha_j)}{\Gamma\Pa{\sum^n_{j=1}\alpha_j}}.
\end{eqnarray}
Indeed,
$$
p^{\alpha_j-1}_j =
\Pa{\frac{\phi(\zeta_j)}{\varphi'(\zeta_j)}}^{\alpha_j-1} =
\frac{\abs{\phi(\zeta_j)}^{\alpha_j-1}}{\abs{\varphi'(\zeta_j)}^{\alpha_j-1}},
$$
hence
\begin{eqnarray}
&&\int\cdots\int_V p^{\alpha_0-1}_0p^{\alpha_1-1}_1\cdots
p^{\alpha_n-1}_n\dif p_0\cdots\dif p_{n-1} \\
&&=
\frac1{\prod^n_{j=0}\abs{\varphi'(\zeta_j)}^{\alpha_j-1}}\int_{\phi(t)\in\cD}\prod^n_{j=0}\abs{\phi(\zeta_j)}^{\alpha_j-1}\dif
p_0\cdots\dif p_{n-1}.
\end{eqnarray}
We need to verify that the Jacobian
$$
\det\Pa{\frac{\partial(p_0,\ldots,p_{n-1})}{\partial(\phi_0,\ldots,\phi_{n-1})}}
= \prod^n_{j=0} \abs{\varphi'(\zeta_j)}^{-\frac12},
$$
that is, $\dif p_0\cdots\dif p_{n-1}=\prod^n_{j=0}
\abs{\varphi'(\zeta_j)}^{-\frac12}\dif\phi_0\cdots\dif\phi_{n-1}$ or
$\dif\phi_0\cdots\dif\phi_{n-1}=\prod^n_{j=0}\abs{\varphi'(\zeta_j)}^{\frac12}\dif
p_0\cdots\dif p_{n-1}$. Since
$$
p_j =
\frac1{\varphi'(\zeta_j)}\Pa{\zeta^n_j-\phi_{n-1}\zeta^{n-1}_j+\cdots+(-1)^n\phi_0},
$$
the Jacobian is
$$
\abs{\det\Pa{\frac{\partial(p_0,\ldots,p_{n-1})}{\partial(\phi_0,\ldots,\phi_{n-1})}}}
=
\abs{\frac{\det\Pa{\zeta^j_i}_{n-1}}{\prod^{n-1}_{j=0}\varphi'(\zeta_j)}}
=\abs{\frac{\det\Pa{\zeta^j_i}_n}{\prod^n_{j=0}\varphi'(\zeta_j)}}=\frac{\prod^n_{j=0}\abs{\varphi'(\zeta_j)}^{\frac12}}{\prod^n_{j=0}\abs{\varphi'(\zeta_j)}}.
$$
The numerator is a Vandermonde determinant and therefore the result
follows:
\begin{eqnarray*}
&&\int\cdots\int_V p^{\alpha_0-1}_0p^{\alpha_1-1}_1\cdots
p^{\alpha_n-1}_n\dif p_0\cdots\dif p_{n-1} \\
&&=
\frac1{\prod^n_{j=0}\abs{\varphi'(\zeta_j)}^{\alpha_j-\frac12}}\int_{\phi(t)\in\cD}\prod^n_{j=0}\abs{\phi(\zeta_j)}^{\alpha_j-1}\Pa{\prod^n_{j=0}\abs{\varphi'(\zeta_j)}^{\frac12}\dif
p_0\cdots\dif p_{n-1}}\\
&&=\frac1{\prod^n_{j=0}\abs{\varphi'(\zeta_j)}^{\alpha_j-\frac12}}\int_{\phi(t)\in\cD}\prod^n_{j=0}\abs{\phi(\zeta_j)}^{\alpha_j-1}\dif
\phi_0\cdots\dif \phi_{n-1}.
\end{eqnarray*}
The final step is to obtain the $(2n-1)$-dimensional integral. Let
$\phi(t)$ and $\Phi(t)$ be two polynomials such that
\begin{eqnarray}\label{eq:phi-Phi}
&&\phi(t) = \prod^{n-1}_{i=1}(t-x_i)~~\text{and}~~\Phi(t) =
\prod^n_{j=1}(t-y_j),\\
&&0<y_1<x_1<y_2<\cdots<x_{n-1}<y_n<1.
\end{eqnarray}
The resultant of $\phi$ and $\Phi$, denoted $R(\phi,\Phi)$, is given
by
\begin{eqnarray}
\abs{R(\phi,\Phi)} = \abs{\prod_{i\in[n-1];j\in[n]}(x_i-y_j)} =
\abs{\prod^n_{j=1}\phi(y_j)} = \abs{\prod^{n-1}_{i=1}\Phi(x_i)}.
\end{eqnarray}
The absolute value of the discriminant of $\phi$ can be written as
$\abs{R(\phi,\phi')}$. That is,
$$
\abs{\Delta_\phi} = \abs{\Delta(x)}^2 = \prod^{n-1}_{j=1}\phi'(x_j).
$$
The $(2n-1)$-dimensional integral is
\begin{eqnarray}\label{eq:Phi-integral}
&&\int_{(\phi,\Phi)}
\abs{\Phi(0)}^{\alpha-1}\abs{\Phi(1)}^{\beta-1}\abs{R(\phi,\Phi)}^{\gamma-1}\dif\phi_0
\cdots\dif \phi_{n-2}\dif\Phi_0\cdots\dif\Phi_{n-1}\notag\\
&&=\int_{(\phi,\Phi)}\abs{\Phi(0)}^{\alpha-1}\abs{\Phi(1)}^{\beta-1}\abs{\prod^n_{j=1}\phi(y_j)}^{\gamma-1}\dif\phi_0
\cdots\dif \phi_{n-2}\dif\Phi_0\cdots\dif\Phi_{n-1}.
\end{eqnarray}
Here the integration is over all $\phi$ and $\Phi$ defined by
Eq.~\eqref{eq:phi-Phi}. Then we show that Selberg's integral
$A_n\Pa{\alpha,\beta,\gamma}$ satisfies the recurrence relation:
\begin{eqnarray}\label{eq:recurrence}
A_n\Pa{\alpha,\beta,\gamma} =
\frac{\Gamma(\alpha)\Gamma(\beta)\Gamma(\gamma
n)}{\Gamma(\alpha+\beta+\gamma(n-1))}A_{n-1}\Pa{\alpha+\gamma,\beta+\gamma,\gamma}.
\end{eqnarray}
In fact, integrate the $(2n-1)$-dimensional integral
Eq.~\eqref{eq:Phi-integral} with respect to
$\dif\phi_0\cdots\dif\phi_{n-2}$ and use $\Phi(t)$ instead of
$\varphi(t)$ in Eq.~\eqref{eq:varphi-int} to get
\begin{eqnarray*}
&&\int_{(\phi,\Phi)}\abs{\Phi(0)}^{\alpha-1}\abs{\Phi(1)}^{\beta-1}\abs{\prod^n_{j=1}\phi(y_j)}^{\gamma-1}\dif\phi_0
\cdots\dif \phi_{n-2}\dif\Phi_0\cdots\dif\Phi_{n-1}\\
&&=\int_\Phi\abs{\Phi(0)}^{\alpha-1}\abs{\Phi(1)}^{\beta-1}\Pa{\int_\phi\prod^n_{j=1}\abs{\phi(y_j)}^{\gamma-1}\dif\phi_0
\cdots\dif \phi_{n-2}}\dif\Phi_0\cdots\dif\Phi_{n-1}\\
&&=\int_\Phi\abs{\Phi(0)}^{\alpha-1}\abs{\Phi(1)}^{\beta-1}\Pa{\frac{\prod^n_{j=1}\abs{\Phi'(y_j)}^{\gamma-\frac12}\Gamma(\gamma)}{\Gamma(\gamma
n)}}\dif\Phi_0\cdots\dif\Phi_{n-1}\\
&&=\frac{\Gamma(\gamma)^n}{\Gamma(\gamma n)}\int_\Phi
\abs{\Phi(0)}^{\alpha-1}\abs{\Phi(1)}^{\beta-1}\abs{\prod^n_{j=1}\Phi'(y_j)}^{\gamma-\frac12}\dif
\Phi_0\cdots\dif\Phi_{n-1} =\frac{\Gamma(\gamma)^n}{\Gamma(\gamma
n)}A_n(\alpha,\beta,\gamma).
\end{eqnarray*}
It remains to compute Eq.~\eqref{eq:Phi-integral} in another way,
set $\widetilde \phi(t)=t\prod^n_{j=1}(t-x_j)$, and
$$
\alpha_0=\alpha,\alpha_j=\gamma(j=1,\ldots,n-1),\alpha_n=\beta;~~~x_0=0,x_n=1
$$
so that Eq.~\eqref{eq:Phi-integral} is equal to
\begin{eqnarray}
&&\int_{(\phi,\Phi)}\abs{\Phi(0)}^{\alpha-1}\abs{\Phi(1)}^{\beta-1}\abs{\prod^{n-1}_{j=1}\Phi(x_j)}^{\gamma-1}\dif\Phi_0\cdots\dif\Phi_{n-1}\dif\phi_0
\cdots\dif \phi_{n-2}\notag\\
&&=\int_{(\phi,\Phi)}\prod^n_{j=0}\abs{\Phi(x_j)}^{\alpha_j-1}\dif\Phi_0\cdots\dif\Phi_{n-1}\dif\phi_0
\cdots\dif \phi_{n-2}.\label{eq:0-j-n}
\end{eqnarray}
Now integrate Eq.~\eqref{eq:0-j-n} with respect to $\dif
\Phi_0\cdots\dif\Phi_{n-1}$ and use $\widetilde \phi$ instead of
$\varphi(t)$ in Eq.~\eqref{eq:varphi-int} to obtain
\begin{eqnarray*}
\frac{\Gamma(\alpha)\Gamma(\beta)\Gamma(\gamma)^{n-1}}{\Gamma(\alpha+\beta+\gamma(n-1))}\int_\phi
\abs{\prod^{n-1}_{j=1}\widetilde
\phi'(x_j)}^{\gamma-\frac12}\abs{\widetilde
\phi'(0)}^{\alpha-\frac12}\abs{\widetilde
\phi'(1)}^{\beta-\frac12}\dif\phi_0\cdots\dif\phi_{n-2}.
\end{eqnarray*}
Since
\begin{eqnarray*}
\abs{\widetilde
\phi'(0)}&=&\abs{\prod^{n-1}_{j=1}x_j},\\
\abs{\widetilde \phi'(1)}&=&\abs{\prod^{n-1}_{j=1}(1-x_j)},\\
\prod^n_{j=1}\abs{\phi'(x_j)}
&=&\prod^{n-1}_{j=1}\abs{x_j}\prod^{n-1}_{j=1}\abs{1-x_j}\abs{\Delta_\phi},
\end{eqnarray*}
the last integral can be written as
\begin{eqnarray*}
&&\frac{\Gamma(\alpha)\Gamma(\beta)\Gamma(\gamma)^{n-1}}{\Gamma(\alpha+\beta+\gamma(n-1))}\int_\phi
\Pa{\prod^{n-1}_{j=1}x^{\alpha+\gamma-1}_j(1-x_j)^{\beta+\gamma-1}}\abs{\Delta_\phi}^{\gamma-\frac12}\dif\phi_0\cdots\dif\phi_{n-2}\\
&&=\frac{\Gamma(\alpha)\Gamma(\beta)\Gamma(\gamma)^{n-1}}{\Gamma(\alpha+\beta+\gamma(n-1))}A_{n-1}(\alpha,\beta,\gamma).
\end{eqnarray*}
Equate the two different evaluations of the $(2n-1)$-dimensional
integral to obtain the result. Finally, Selberg's formula is
obtained by iterating Eq.~\eqref{eq:recurrence} $(n-1)$ times.
\end{proof}

\subsection*{Acknowledgement}

Partial material of the present work is completed during a research
visit to Chern Institute of Mathematics, at Nankai University. The
author would like, in particular, to thank Seunghun Hong for his
remarks about the approach toward the volume of unitary group via
Macdonald's method for the volume of a compact Lie group. Both
Zhen-Peng Xu and Zhaoqi Wu are acknowledged for valuable comments
for the earlier version of the paper. LZ is grateful to the
financial support from National Natural Science Foundation of China
(No.11301124).



\end{document}